\documentclass[preprint,3p]{elsarticle}

\usepackage[utf8]{inputenc}
\usepackage{lipsum}
\usepackage{lineno}

\usepackage{hyperref}
\usepackage{tabularx}

\usepackage{booktabs}   %
\usepackage{subcaption} %
\usepackage{amsthm}
\usepackage{amssymb}
\usepackage{mathrsfs}

\usepackage{etoolbox}%

\usepackage[dvipsnames]{xcolor}%

\usepackage[nomargin,inline,index,%
  status=draft %
]{fixme} %
\fxusetheme{colorsig}
\FXRegisterAuthor{fxAS}{anfxAS}{AS}%
\FXRegisterAuthor{fxAB}{anfxAB}{AB}%
\FXRegisterAuthor{fxNY}{anfxNY}{NY}%
\FXRegisterAuthor{fxFZ}{anfxFZ}{FZ}
\FXRegisterAuthor{fxPH}{anfxPH}{PH}

\makeatletter
\DeclareFontFamily{OMX}{MnSymbolE}{}
\DeclareSymbolFont{MnLargeSymbols}{OMX}{MnSymbolE}{m}{n}
\SetSymbolFont{MnLargeSymbols}{bold}{OMX}{MnSymbolE}{b}{n}
\DeclareFontShape{OMX}{MnSymbolE}{m}{n}{
    <-6>  MnSymbolE5
   <6-7>  MnSymbolE6
   <7-8>  MnSymbolE7
   <8-9>  MnSymbolE8
   <9-10> MnSymbolE9
  <10-12> MnSymbolE10
  <12->   MnSymbolE12
}{}
\DeclareFontShape{OMX}{MnSymbolE}{b}{n}{
    <-6>  MnSymbolE-Bold5
   <6-7>  MnSymbolE-Bold6
   <7-8>  MnSymbolE-Bold7
   <8-9>  MnSymbolE-Bold8
   <9-10> MnSymbolE-Bold9
  <10-12> MnSymbolE-Bold10
  <12->   MnSymbolE-Bold12
}{}

\let\llangle\@undefined
\let\rrangle\@undefined
\DeclareMathDelimiter{\llangle}{\mathopen}%
                     {MnLargeSymbols}{'164}{MnLargeSymbols}{'164}
\DeclareMathDelimiter{\rrangle}{\mathclose}%
                     {MnLargeSymbols}{'171}{MnLargeSymbols}{'171}
\makeatother

\usepackage[inline]{enumitem}%

\usepackage{nicefrac}%

\usepackage{proof}%

\usepackage[most]{tcolorbox}%
\newtcolorbox{myframe}[1][]{
  enhanced,
  arc=0pt,
  outer arc=0pt,
  colback=white,
  boxrule=0.5pt,
  boxsep=0mm,
  left=1mm,
  right=1mm,
  top=0.5mm,
  bottom=0.5mm,
  #1
}

\lstdefinelanguage{Scribble}{%
  basicstyle=\footnotesize\ttfamily,
  stringstyle=\color{Blue},
  showstringspaces=false,
  keywords={nested,new,calls,and,as,at,by,catches,choice,continue,do,from,global,import,instantiates,interruptible,local,module,or,par,protocol,rec,role,sig,throws,to,type,with,int,aux,reliable,crash},
  morestring=[b]",
  morestring=[b]',
  morecomment=[l][\color{greencomments}]{//},
}

\lstdefinelanguage{nuScr}{%
  basicstyle=\footnotesize\ttfamily,
  stringstyle=\color{Blue},
  showstringspaces=false,
  keywords={
    nested,new,calls,and,as,at,by,catches,choice,continue,do,from,global,import,instantiates,interruptible,local,module,or,par,protocol,rec,role,sig,throws,to,type,with,int,aux,
    safe
  },
  morestring=[b]",
  morestring=[b]',
  morecomment=[l][\color{greencomments}]{//},
  morecomment=[s][\color{magenta}]{(*}{*)},
}

\lstdefinelanguage{effpi}{
  keywords=[1]{
    case,class,sealed,abstract,object,extends,type,def,val,if,else,new,var,match
  },
  keywords=[2]{
    InChan,OutChan,RecVar,Rec,Out,In,InErr,Loop,
  },
  keywords=[3]{
    rec,send,receive,receiveErr,eval,par,Channel
  },
  keywordstyle=[1]{\color{blue}},
  keywordstyle=[2]{\color{ImperialIris}}, %
  keywordstyle=[3]{\color{OliveGreen}},
  otherkeywords={=>,.type,<:,>>:},
  morecomment=[l][\color{darkgray}]{//},
}

\usepackage{stmaryrd}
\usepackage{thmtools}%
\usepackage{arydshln}

\usepackage{xifthen}%
\usepackage{xspace}%

\usepackage{mathtools}%

\usepackage{balance}%

\usepackage{hyperref}
\hypersetup{hidelinks}
\usepackage[capitalise]{cleveref}

\definecolor{ImperialBlue}{HTML}{003E74}
\definecolor{ImperialDarkGreen}{HTML}{02893B}
\definecolor{ImperialTangerine}{HTML}{EC7300}
\definecolor{ImperialIris}{HTML}{751E66}

\definecolor{RYB1}{RGB}{141, 211, 199}
\definecolor{RYB2}{RGB}{255, 255, 179}
\definecolor{RYB3}{RGB}{190, 186, 218}
\definecolor{RYB4}{RGB}{251, 128, 114}
\definecolor{RYB5}{RGB}{128, 177, 211}
\definecolor{RYB6}{RGB}{253, 180, 98}
\definecolor{RYB7}{RGB}{179, 222, 105}

\usetikzlibrary{automata, positioning, arrows, calc, fit}
\tikzset{
  >=stealth,
  node distance=2cm,
  every state/.style={thick, fill=gray!10},
  initial text=$ $,
}
\usepackage{pgfplots}
\pgfplotscreateplotcyclelist{stackexchange}{
  {RYB1!50!black,fill=RYB1},
  {RYB4!50!black,fill=RYB4},
  {RYB3!50!black,fill=RYB3},
  {RYB2!50!black,fill=RYB2},
  {RYB5!50!black,fill=RYB5},
  {RYB6!50!black,fill=RYB6},
  {RYB7!50!black,fill=RYB7},
}
\pgfplotscreateplotcyclelist{stackexchangePlusOne}{
  {RYB4!50!black,fill=RYB4},
  {RYB3!50!black,fill=RYB3},
  {RYB2!50!black,fill=RYB2},
  {RYB5!50!black,fill=RYB5},
  {RYB6!50!black,fill=RYB6},
  {RYB7!50!black,fill=RYB7},
}
\pgfplotsset{
  compat=1.8,
  /pgfplots/bar cycle list/.style={/pgfplots/cycle list={%
    {brown!60!black,fill=brown!30!white,mark=none},
    {red,fill=red!30!white,mark=none},
    {blue,fill=blue!30!white,mark=none},
    {black,fill=gray,mark=none},
    }
  },
}
\newcolumntype{L}{>{$}l<{$}}
\newcolumntype{C}{>{$}c<{$}}
\newcolumntype{P}[1]{>{\centering\arraybackslash$}p{#1}<{$}}
\usepackage{multirow}

\usepackage[export]{adjustbox}

\Crefname{section}{\S\!}{\S\!}%
\Crefname{subsection}{\S\!}{\S\!}%
\Crefname{subsubsection}{\S\!}{\S\!}%
\Crefname{appendix}{Appendix \S\!}{Appendix \S\!}
\Crefname{definition}{Def.\@}{Defs.\@}%
\Crefname{figure}{Fig.\@}{Figs.\@}%
\Crefname{example}{Ex.\@}{Exs.\@}%
\Crefname{corollary}{Cor.\@}{Cors.\@}%
\Crefname{theorem}{Thm.\@}{Thms.\@}%
\Crefname{proposition}{Prop.\@}{Props.\@}%
\Crefname{lemma}{Lem.\@}{Lems.\@}
\Crefname{equation}{Eq.\@}{Eqs.\@}

\crefname{section}{\S\!}{\S\!}%
\crefname{subsection}{\S\!}{\S\!}%
\crefname{subsubsection}{\S\!}{\S\!}%
\crefname{appendix}{Appendix \S\!}{Appendix \S\!}
\crefname{definition}{Def.\@}{Defs.\@}%
\crefname{figure}{Fig.\@}{Figs.\@}%
\crefname{example}{Ex.\@}{Exs.\@}%
\crefname{corollary}{Cor.\@}{Cors.\@}%
\crefname{theorem}{Thm.\@}{Thms.\@}%
\crefname{proposition}{Prop.\@}{Props.\@}%
\crefname{lemma}{Lem.\@}{Lems.\@}
\crefname{equation}{Eq.\@}{Eqs.\@}

\usepackage{magicvariables}

\newif\ifdraft%
  \draftfalse%
  \drafttrue%

\newcommand{\ifempty}[3]{%
  \ifthenelse{\isempty{#1}}{#2}{#3}%
}%

\newtoggle{techreport}%

\newtoggle{cruft}%

\newcommand{\dom}[1]{{\color{black}\operatorname{dom}\!\left({#1}\right)}}%
\newcommand{\gtFv}[1]{\gtFmt{\operatorname{fv}}\!\left({#1}\right)}%
\newcommand{\fcv}[1]{\operatorname{fcv}\!\left({#1}\right)}%
\newcommand{\fev}[1]{\operatorname{fev}\!\left({#1}\right)}%
\newcommand{\fc}[1]{\operatorname{fc}\!\left({#1}\right)}%
\newcommand{\dpv}[1]{\operatorname{dpv}\!\left({#1}\right)}%
\newcommand{\fpv}[1]{\operatorname{fpv}\!\left({#1}\right)}%
\newcommand{\unfoldOne}[1]{%
  {\color{black}\operatorname{unf}\!\left({#1}\right)}}%
\newcommand{\notImplies}{\mathrel{\,\,\,\not\!\!\!\!\implies}}%
\newcommand{\notImpliedBy}{\mathrel{{\kern .5em}{\not{\kern -1.2em}\impliedby}}}%

\newcommand{\coloncolonequals}{\Coloneqq}%
\newcommand{\bnfdef}{\coloncolonequals}%
\newcommand{\bnfsep}{\mathbin{\;\big|\;}}%

\newcommand{\eval}[2]{#1 \downarrow #2}

\def\aka{a.k.a.\@\xspace}%
\def\cf{cf.\@\xspace}%
\def\eg{e.g.\@\xspace}%
\def\ie{i.e.\@\xspace}%

\definecolor{ruleColor}{rgb}{0.1, 0.3, 0.1}%
\newcommand{\inferrule}[1]{{\color{ruleColor}\textsc{\scriptsize [#1]}}}%
\newcommand{\inference}[3][]{\infer[\ifempty{#1}{}{\inferrule{#1}}]{#3}{#2}}%
\newcommand{\cinference}[3][]{\infer=[\ifempty{#1}{}{\inferrule{#1}}]{#3}{#2}}%

\newcommand{\setenum}[1]{\mathord{{\color{black}\left\{#1\right\}}}}%
\newcommand{\setcomp}[2]{\mathord{%
  {\color{black}\left\{{#1} \,\middle|\, {#2}\right\}}}}%

\newcommand{\predP}[1][]{\ifempty{#1}{\varphi}{\varphi_{#1}}}%
\newcommand{\predPApp}[2][]{\ifempty{#1}{\predP}{\predP[{#1}]}\!\left({#2}\right)}%

\newcommand{\bind}[2]{\nicefrac{#2}{#1}}%
\newcommand{\substenum}[1]{\mathord{\left\{{#1}\right\}}}%
\newcommand{\subst}[2]{\substenum{\bind{#1}{#2}}}%

\definecolor{hlColor}{rgb}{0.65, 1.0, 0.65}%

\newcommand{\lbbar}{\{\kern-0.2em|}
\newcommand{\rbbar}{|\kern-0.2em\}}

\definecolor{tyColorCustom}{rgb}{0.0, 0.0, 0.85}%
\newcommand{\tyCol}[1]{{\color{tyColorCustom}{#1}}}%
\newcommand{\tyFont}[1]{{#1}}%
\newcommand{\tyFmt}[1]{\tyCol{\tyFont{#1}}}%
\newcommand{\tyFontC}[1]{\operatorname{#1}}%
\newcommand{\tyFmtC}[1]{\tyCol{\tyFontC{#1}}}%

\newcommand{\tyGroundSet}{\stFmt{\mathcal{B}}}  
\newcommand{\tyGround}[1][]{\tyFmt{\ifempty{#1}{B}{B_{#1}}}}%
\newcommand{\tyNat}{\tyFmtC{nat}}%
\newcommand{\tyBool}{\tyFmtC{bool}}%
\newcommand{\tyUnit}{\tyFmtC{unit}}%
\newcommand{\tyInt}{\tyFmtC{int}}%
\newcommand{\tyString}{\tyFmtC{str}}%

\newcommand{\tyS}[1][]{\tyCol{\ifempty{#1}{\tyFont{S}}{\tyFont{S}_{#1}}}}%
\newcommand{\tySi}[1][]{\tyCol{\ifempty{#1}{\tyFont{S}'}{\tyFont{S}'_{#1}}}}%
\newcommand{\tySii}[1][]{\tyCol{\ifempty{#1}{\tyFont{S}''}{\tyFont{S}''_{#1}}}}%
\newcommand{\tyT}[1][]{\tyCol{\ifempty{#1}{\tyFont{T}}{\tyFont{T}_{#1}}}}%
\newcommand{\tyTi}[1][]{\tyCol{\ifempty{#1}{\tyFont{T}'}{\tyFont{T}'_{#1}}}}%

\newcommand{\muCol}[1]{{\color{red}#1}}%
\newcommand{\muWordEmpty}[1][]{\muCol{\epsilon}}%
{\centerline{\bf --- Begin Copied From Previous Paper ---} \hrule}%
{\hrule \centerline{\bf --- End Copied From Previous Paper ---}}%

{\centerline{\bf --- Begin Discussion\ifempty{#1}{}{: {#1}} ---}
  \hrule\vspace{1mm}}%
{\hrule\vspace{1mm}\centerline{\bf --- End Discussion ---}}%

\newtcolorbox{cross}{blank,breakable,parbox=false,
  overlay={\draw[red,line width=5pt] (interior.south west)--(interior.north east);
    \draw[red,line width=5pt] (interior.north west)--(interior.south east);}}

\definecolor{roleColor}{rgb}{0.5, 0.0, 0.0}%
\newcommand{\roleCol}[1]{{\color{roleColor}#1}}%
\newcommand{\roleSet}{\roleCol{\mathcal{R}}}%
\newcommand{\roleFmt}[1]{\ensuremath{{\boldsymbol{\roleCol{\mathtt{#1}}}}}\xspace}%

\newcommand{\roleP}[1][]{%
  \ifempty{#1}{{\color{roleColor}\roleFmt{p}}}{{\color{roleColor}\roleFmt{p}_{#1}}}%
}%
\newcommand{\rolePi}[1][]{%
  \ifempty{#1}{{\color{roleColor}\roleFmt{p}'}}{{\color{roleColor}\roleFmt{p}'_{#1}}}%
}%
\newcommand{\roleQ}[1][]{%
  \ifempty{#1}{{\color{roleColor}\roleFmt{q}}}{{\color{roleColor}\roleFmt{q}_{#1}}}%
}%
\newcommand{\roleQi}[1][]{%
  \ifempty{#1}{{\color{roleColor}\roleFmt{q}'}}{{\color{roleColor}\roleFmt{q}'_{#1}}}%
}%
\newcommand{\roleR}[1][]{%
  \ifempty{#1}{{\color{roleColor}\roleFmt{r}}}{{\color{roleColor}\roleFmt{r}_{\!#1}}}%
}%
\newcommand{\roleS}[1][]{%
  \ifempty{#1}{{\color{roleColor}\roleFmt{s}}}{{\color{roleColor}\roleFmt{s}_{\!#1}}}%
}%
\newcommand{\roleT}[1][]{%
  \ifempty{#1}{{\color{roleColor}\roleFmt{t}}}{{\color{roleColor}\roleFmt{t}_{\!#1}}}%
}%

\newcommand{\roleU}[1][]{%
  \ifempty{#1}{{\color{roleColor}\roleFmt{u}}}{{\color{roleColor}\roleFmt{u}_{\!#1}}}%
}%

\definecolor{gtColor}{rgb}{0.43, 0.21, 0.1}%
\newcommand{\gtFmt}[1]{\ensuremath{{\color{gtColor}#1}}\xspace}%
\newcommand{\gtMsgFmt}[1]{\gtFmt{\labFmt{#1}}}%
\newcommand{\gtLab}[1][]{%
  \ifempty{#1}{\gtMsgFmt{m}}{{\color{gtColor}\gtMsgFmt{m}_{#1}}}%
}%
\newcommand{\gtLabi}[1][]{%
  \ifempty{#1}{\gtMsgFmt{m}'}{{\color{gtColor}\gtMsgFmt{m}'_{#1}}}%
}%
\newcommand{\gtLabii}[1][]{%
  \ifempty{#1}{\gtMsgFmt{m}''}{{\color{gtColor}\gtMsgFmt{m}''_{#1}}}%
}%

\newcommand{\gtG}[1][]{\gtFmt{\ifempty{#1}{G}{G_{#1}}}}%
\newcommand{\gtGi}[1][]{\gtFmt{\ifempty{#1}{G'}{G'_{#1}}}}%
\newcommand{\gtGii}[1][]{\gtFmt{\ifempty{#1}{G''}{G''_{#1}}}}%
\newcommand{\gtGiii}[1][]{\gtFmt{\ifempty{#1}{G'''}{G'''_{#1}}}}%

\newcommand{\gtSeq}{\mathbin{\gtFmt{.}}}%

\newcommand{\gtCommRaw}[3]{%
  \gtFmt{%
    {#1} {\to} {#2}{:}%
    \left\{%
      {#3}%
    \right\}%
  }%
}%
\newcommand{\gtComm}[6]{%
  \gtFmt{%
    \gtCommRaw{#1}{#2}{%
      \gtCommChoice{#4}{#5}{#6}%
    }_{#3}%
  }%
}%
\newcommand{\gtCommSmall}[6]{%
  \gtFmt{%
    \gtCommRaw{#1}{#2}{%
      \gtCommChoiceSmall{#4}{#5}{#6}%
    }_{#3}%
  }%
}%
\newcommand{\gtCommSingle}[5]{%
  \gtFmt{%
    {#1} {\to} {#2}{:}%
    \gtCommChoice{#3}{#4}{#5}%
  }%
}%
\newcommand{\gtCommChoice}[3]{%
  \gtFmt{%
    \gtMsgFmt{#1}\ifempty{#2}{}{({#2})}%
    \ifempty{#3}{}{\vphantom{x} \!\gtSeq\! {#3}}%
  }%
}%
\newcommand{\gtCommChoiceSmall}[3]{%
  \gtFmt{%
    \gtMsgFmt{#1}\ifempty{#2}{}{({#2})}%
    \ifempty{#3}{}{\vphantom{x} \!\gtSeq\! {#3}}%
  }%
}%

\newcommand{\gtEnd}{\gtFmt{\mathbf{end}}}%

\newcommand{\gtRec}[2]{\gtFmt{\mu{#1}.{#2}}}%
\newcommand{\gtRecVarBase}{\gtFmt{\mathbf{t}}}%
\newcommand{\gtRecVar}[1][]{\gtFmt{\ifempty{#1}{\gtRecVarBase}{\gtRecVarBase_{#1}}}}%

\newcommand{\gtRoles}[1]{{\color{roleColor} \operatorname{roles}(\gtFmt{#1})}}%

\newcommand{\gtProj}[3][]{%
  {\color{stColor}\gtFmt{#2}\!\ifempty{#1}{\upharpoonright}{\upharpoonright_{#1}}\!\roleFmt{#3}}%
}%

\newcommand{\gtMove}[1][\phantom{\stEnvAnnotGenericSym}]{\gtFmt{\xrightarrow{#1}}} %
\newcommand{\gtMoveStar}[1][]{\ifempty{#1}{\gtMove[]{}^{\!\gtFmt{*}}}{\gtMove[]{#1}^{\!\gtFmt{*}}}} %
\newcommand{\gtNotMove}[1]{\gtFmt{#1}\!\not{\!\!\gtMove}}%

\newcommand{\iruleGtMove}[1]{GR-{#1}}

\newcommand{\iruleGtMoveComm}[0]{\iruleGtMove{$\oplus\&$}}
\newcommand{\iruleGtMoveRec}[0]{\iruleGtMove{$\mu$}}

\newcommand{\iruleGtMoveCtx}[0]{\iruleGtMove{Ctx}}

\newcommand{\labFmt}[2][]{\ensuremath{\ifempty{#1}{\mathtt{#2}}{\mathtt{#2}\textsubscript{#1}}}\xspace}%

\definecolor{stColor}{rgb}{0, 0, 0.9}%
\newcommand{\stFmt}[1]{\ensuremath{{\color{stColor}#1}}\xspace}%
\newcommand{\stFmtC}[1]{\stFmt{\operatorname{#1}}}

\newcommand{\stOut}[3]{\ifempty{#1}{}{\roleFmt{#1}}\stFmt{\oplus{#2}\ifempty{#3}{}{({#3})}}}%

\newcommand{\stChoice}[2]{\stLabFmt{#1}\ifempty{#2}{}{\stFmt{({#2})}}}%

\newcommand{\stSeq}{\mathbin{\!\stFmt{.}\!}}%
\newcommand{\stIntC}{\mathbin{\stFmt{\oplus}}}%
\newcommand{\stIntSum}[3]{\roleFmt{#1}\stFmt{\oplus\!\left\{#3\right\}_{#2}}}%
\newcommand{\stExtC}{\mathbin{\stFmt{\&}}}%
\newcommand{\stExtSum}[3]{\roleFmt{#1}\stFmt{\&\!\left\{#3\right\}_{#2}}}%
\newcommand{\stRec}[2]{\stFmt{\mu{#1}.{#2}}}%
\newcommand{\stEnd}{\stFmt{\mathbf{end}}}%

\newcommand{\stLabFmt}[1]{\stFmt{\labFmt{#1}}}%
\newcommand{\stLab}[1][]{%
  \ifempty{#1}{\stLabFmt{m}}{\stLabFmt{m}_{{\color{stColor}#1}}}
}%
\newcommand{\stLabi}[1][]{%
  \ifempty{#1}{\stLabFmt{m'}}{\stLabFmt{m'}_{{\color{stColor}#1}}}
}%
\newcommand{\stLabii}[1][]{%
  \ifempty{#1}{\stLabFmt{m''}}{\stLabFmt{m''}_{{\color{stColor}#1}}}
}%

\newcommand{\stS}[1][]{\stFmt{\ifempty{#1}{S}{S_{#1}}}}%
\newcommand{\stSi}[1][]{\stFmt{\ifempty{#1}{S'}{S'_{#1}}}}%
\newcommand{\stSii}[1][]{\stFmt{\ifempty{#1}{S''}{S''_{#1}}}}%
\newcommand{\stSiii}[1][]{\stFmt{\ifempty{#1}{S'''}{S'''_{#1}}}}%

\newcommand{\stT}[1][]{\stFmt{\ifempty{#1}{T}{T_{#1}}}}%
\newcommand{\stTi}[1][]{\stFmt{\ifempty{#1}{T'}{T'_{#1}}}}%
\newcommand{\stTii}[1][]{\stFmt{\ifempty{#1}{T''}{T''_{#1}}}}%
\newcommand{\stTiii}[1][]{\stFmt{\ifempty{#1}{T'''}{T'''_{#1}}}}%

\newcommand{\stU}[1][]{\stFmt{\ifempty{#1}{U}{U_{#1}}}}%
\newcommand{\stUi}[1][]{\stFmt{\ifempty{#1}{U'}{U'_{#1}}}}%

\newcommand{\stB}[1][]{\stFmt{\ifempty{#1}{B}{B_{#1}}}}%

\newcommand{\stRecVarBase}{\stFmt{\mathbf{t}}}%
\newcommand{\stRecVar}[1][]{\stFmt{\ifempty{#1}{\stRecVarBase}{\stRecVarBase_{#1}}}}%

\newcommand{\stMerge}[2]{\stFmt{\bigsqcap_{#1}{#2}}}%
\newcommand{\stBinMerge}{\mathbin{\stFmt{\sqcap}}}%

\newcommand{\stSub}{\mathrel{\stFmt{\leqslant}}}%

\newcommand{\stNotSub}{\mathrel{\stFmt{\not\leqslant}}}%

\newcommand{\stepsto}[1][\quad]{\xrightarrow{#1}} %

\newcommand{\ltsSendRecv}[4]{\mpChanRole{#1}{#2}\mpFmt{[\roleFmt{#3}]}\labFmt{#4}}

\newcommand{\ltsSubject}[1]{{\color{roleColor} \operatorname{subject}({#1})}}%

\definecolor{mpColor}{rgb}{0, 0, 0}%
\newcommand{\mpFmt}[1]{{\color{mpColor}#1}}%

\newcommand{\mpLab}[1][]{%
  \mpFmt{\ifempty{#1}{\labFmt{m}}{{\labFmt{m}}_{\mathnormal #1}}}%
}%
\newcommand{\mpLabi}[1][]{%
  \mpFmt{\ifempty{#1}{\labFmt{m}'}{\labFmt{m}'_{\mathnormal #1}}}%
}%
\newcommand{\mpLabCrash}[0]{\mpFmt{\labFmt{\mathbf{crash}}}}
\newcommand{\mpLabFmt}[1]{\mpFmt{\labFmt{#1}}}%

\newcommand{\mpInt}{\mpFmt{\text{\texttt{i}}}}
\newcommand{\mpNat}{\mpFmt{\text{\texttt{n}}}}%
\newcommand{\mpTrue}{\mpFmt{\text{\texttt{true}}}}%
\newcommand{\mpFalse}{\mpFmt{\text{\texttt{false}}}}%
\newcommand{\mpString}[1]{\text{\texttt{"{#1}"}}}%
\newcommand{\mpUnit}[1]{\text{\texttt{({#1})}}}%
\newcommand{\mpSucc}[1]{\text{\texttt{succ}}({#1})}%
\newcommand{\mpNeg}[1]{\text{\texttt{neg}}({#1})}%

\newcommand{\mpChanRole}[2]{\mpFmt{{#1}[{#2}]}}%

\newcommand{\mpNil}{\mpFmt{\mathbf{0}}}%
\newcommand{\mpSeq}{\mathbin{\mpFmt{\!.\!}}}%
\newcommand{\mpIf}[3]{%
  \mpFmt{\mathsf{if}\,{#1}\,\mathsf{then}}\,{#2}\,\mathsf{else}\,{#3}%
}%
\newcommand{\mpChoice}[3]{%
  \mpFmt{%
    \mpLabFmt{#1}\ifempty{#2}{}{({#2})}\ifempty{#3}{}{\vphantom{x}\mpSeq {#3}}%
  }%
}%
\newcommand{\mpChoiceNoBind}[3]{%
  \mpFmt{%
    \mpLabFmt{#1}\ifempty{#2}{}{\langle{#2}\rangle}\ifempty{#3}{}{\vphantom{x}\mpSeq {#3}}%
  }%
}%
\newcommand{\mpBranchRaw}[3]{%
  \mpFmt{%
    {#1}[\roleFmt{#2}]  \& \,
      {#3}%
  }%
}%
\newcommand{\mpBranch}[7]{%
  \mpFmt{%
    {#1}[\roleFmt{#2}] \mathbin{\!\ifempty{\sum}{\sum}{\&}\!}%
    \ifempty{#3}{%
      \{\mpChoice{#4}{#5}{#6}\ifempty{#7}{}{,\;\mpChoice{\mpLabCrash}{}{#7}}\}_{#3}
    }{%
      \{\mpChoice{#4}{#5}{#6}\ifempty{#7}{}{,\;\mpChoice{\mpLabCrash}{}{#7}}\}_{#3}
    }%
  }%
}%
\newcommand{\mpBranchSingle}[5]{%
  \mpFmt{%
    {#1}[\roleFmt{#2}] \mathbin{\!\ifempty{\sum}{\sum}{\&}\!}%
    \mpChoice{#3}{#4}{#5}
  }%
}%
\newcommand{\mpSel}[5]{%
  \mpFmt{%
    {#1}[\roleFmt{#2}] \mathbin{\!\oplus\!}%
    \mpChoiceNoBind{#3}{#4}{#5}%
  }%
}%
\newcommand{\mpPar}{\mathbin{\mpFmt{\mid}}}%
\newcommand{\mpBigPar}[2]{\mathbin{\mpFmt{\Pi_{#1}}{#2}}}%
\newcommand{\mpRes}[2]{\mpFmt{\left(\mathbf{\nu}{#1}\right){#2}}}%
\newcommand{\mpJustDef}[3]{%
  \mpFmt{{#1}(#2) = {#3}}%
}%
\newcommand{\mpDef}[4]{%
  \mpFmt{\mathbf{def}\;\mpJustDef{#1}{#2}{#3}\;\mathbf{in}\;{#4}}%
}%
\newcommand{\mpDefAbbrev}[2]{%
  \mpFmt{\mathbf{def}\;{#1}\;\mathbf{in}\;{#2}}%
}%
\newcommand{\mpCall}[2]{\mpFmt{{#1}\!\left\langle{#2}\right\rangle}}%
\newcommand{\mpCallSmall}[2]{\mpFmt{{#1}\langle{#2}\rangle}}
\newcommand{\mpErr}{\mpFmt{\boldsymbol{\mathtt{err}}}}%

\newcommand{\mpCtx}[1][]{\mpFmt{\ifempty{#1}{\mathbb{C}}{\mathbb{C}_{#1}}}}%
\newcommand{\mpCtxi}[1][]{\mpFmt{\ifempty{#1}{\mathbb{C}'}{\mathbb{C}'_{#1}}}}%
\newcommand{\mpCtxHole}{[\,]}%
\newcommand{\mpCtxApp}[2]{{#1}\!\left[{#2}\right]}%

\newcommand{\mpV}[1][]{\mpFmt{\ifempty{#1}{v}{v_{#1}}}}
\newcommand{\mpVi}[1][]{\mpFmt{\ifempty{#1}{v'}{v'_{#1}}}}

\newcommand{\mpW}[1][]{\mpFmt{\ifempty{#1}{w}{w_{#1}}}}

\newcommand{\mpU}[1][]{\mpFmt{\ifempty{#1}{u}{u_{#1}}}}

\makeVars{mpE}{mpFmt}{e}

\newcommand{\mpC}[1][]{\mpFmt{\ifempty{#1}{c}{c_{#1}}}}%
\newcommand{\mpCi}[1][]{\mpFmt{\ifempty{#1}{c'}{c'_{#1}}}}%
\newcommand{\mpD}[1][]{\mpFmt{\ifempty{#1}{d}{d_{#1}}}}%

\newcommand{\mpS}[1][]{\mpFmt{\ifempty{#1}{s}{s_{#1}}}}%
\newcommand{\mpSi}[1][]{\mpFmt{\ifempty{#1}{s'}{s'_{#1}}}}%

\newcommand{\mpX}[1][]{\mpFmt{\ifempty{#1}{X}{X_{#1}}}}%
\newcommand{\mpXi}[1][]{\mpFmt{\ifempty{#1}{X'}{X'_{#1}}}}%
\newcommand{\mpY}[1][]{\mpFmt{\ifempty{#1}{Y}{Y_{#1}}}}%
\newcommand{\mpYi}[1][]{\mpFmt{\ifempty{#1}{Y'}{Y'_{#1}}}}%

\newcommand{\mpP}[1][]{\mpFmt{\ifempty{#1}{P}{P_{#1}}}}%
\newcommand{\mpPi}[1][]{\mpFmt{\ifempty{#1}{P'}{P'_{#1}}}}%
\newcommand{\mpPii}[1][]{\mpFmt{\ifempty{#1}{P''}{P''_{#1}}}}%
\newcommand{\mpQ}[1][]{\mpFmt{\ifempty{#1}{Q}{Q_{#1}}}}%
\newcommand{\mpQi}[1][]{\mpFmt{\ifempty{#1}{Q'}{Q'_{#1}}}}%
\newcommand{\mpR}[1][]{\mpFmt{\ifempty{#1}{R}{R_{#1}}}}%

\newcommand{\mpDefD}[1][]{\mpFmt{\ifempty{#1}{D}{D_{#1}}}}%
\newcommand{\mpDefDi}[1][]{\mpFmt{\ifempty{#1}{D'}{D'_{#1}}}}%

\newcommand{\mpMove}{\to}%
\newcommand{\mpMoveStar}{\mathrel{\mpMove{}^{\!\!\!*}}}%
\newcommand{\mpNotMoveP}[1]{\mpFmt{#1}\!\not{\!\!\!\mpMove}}%

\newcommand{\mpCtxMove}{\mpFmt{\rightsquigarrow}}%

\newcommand{\mpCtxMoveStar}{\mathrel{\mpCtxMove{}^{\!*}}}%
\newcommand{\iruleMPRedCommExp}{R-$\oplus\&$-V}%
\newcommand{\iruleMPRedCommChannel}{R-$\oplus\&$-D}%
\newcommand{\iruleMPRedCondTrue}{R-Cond-T}%
\newcommand{\iruleMPRedCondFalse}{R-Cond-F}%
\newcommand{\iruleMPRedCall}{R-$\mpX$}%
\newcommand{\iruleMPRedCongr}{R-$\equiv$}%
\newcommand{\iruleMPCtx}[1][]{R-Ctx\ifempty{#1}{}{#1}}
\newcommand{\iruleMPRedCtx}{\iruleMPCtx[]}%
\newcommand{\iruleMPErrLabel}{R-Err}%

\newcommand{\iruleSafeComm}{S-${\stIntC}{\stExtC}$}%
\newcommand{\iruleSafeRec}{S-$\stFmt{\mu}$}%
\newcommand{\iruleMoveSession}{S-$\stEnvMoveWithSession[\mpS]$}%

\newcommand{\iruleStSubEnd}{Sub-$\stEnd$}

\newcommand{\iruleStSubRecL}{Sub-$\stFmt{\mu}$L}
\newcommand{\iruleStSubRecR}{Sub-$\stFmt{\mu}$R}
\newcommand{\iruleStSubOut}{Sub-$\stFmt{\oplus}$}
\newcommand{\iruleStSubIn}{Sub-$\stFmt{\&}$}
\newcommand{\iruleStSubGround}{Sub-$\tyGround$}

\newcommand{\iruleMPEnd}{T-$\stEnvEndPred$}%
\newcommand{\iruleMPSub}{T-Sub}%
\newcommand{\iruleMPNil}{T-$\mpNil$}%
\newcommand{\iruleMPDef}{T-$\mpFmt{\mathbf{def}}$}%
\newcommand{\iruleMPCall}{T-Call}%
\newcommand{\iruleMPPar}{T-$\mpPar$}%
\newcommand{\iruleMPRes}{T-$\mpFmt{\mathbf{\nu}}$Classic}%
\newcommand{\iruleMPGlobalRes}{T-$\gtG$-$\mpFmt{\mathbf{\nu}}$}
\newcommand{\iruleMPBranch}{T-$\mpFmt{\&}$}%
\newcommand{\iruleMPSelV}{T-$\mpFmt{\oplus}$-V}%
\newcommand{\iruleMPSelD}{T-$\mpFmt{\oplus}$-D}%
\newcommand{\iruleMPIf}{T-If}

\newcommand{\iruleTCtxOut}{$\stEnvNew$-$\stFmt{\oplus}$}%
\newcommand{\iruleTCtxIn}{$\stEnvNew$-$\stFmt{\&}$}%
\newcommand{\iruleTCtxCom}{$\stEnvNew$-$\stFmt{\oplus\&}$}%
\newcommand{\iruleTCtxRec}{$\stEnvNew$-$\mu$}%
\newcommand{\iruleTCtxCong}{$\stEnvNew$-$\stEnvComp$}%
\newcommand{\iruleTCtxCongBasic}{$\stEnvNew$-$\stEnvComp$$\tyGround$}%

\iffalse
\newcommand{\stEnv}[1][]{\stFmt{\ifempty{#1}{\Gamma}{\Gamma_{#1}}}}%
\newcommand{\stEnvi}[1][]{\stFmt{\ifempty{#1}{\Gamma'}{\Gamma'_{#1}}}}%
\newcommand{\stEnvii}[1][]{\stFmt{\ifempty{#1}{\Gamma''}{\Gamma''_{#1}}}}%
\newcommand{\stEnviii}[1][]{\stFmt{\ifempty{#1}{\Gamma'''}{\Gamma'''_{#1}}}}%
\fi

\newcommand{\stEnv}[1][]{\stFmt{\ifempty{#1}{\Delta}{\Delta_{#1}}}}%
\newcommand{\stEnvi}[1][]{\stFmt{\ifempty{#1}{\Delta'}{\Delta'_{#1}}}}%
\newcommand{\stEnvii}[1][]{\stFmt{\ifempty{#1}{\Delta''}{\Delta''_{#1}}}}%
\newcommand{\stEnviii}[1][]{\stFmt{\ifempty{#1}{\Delta'''}{\Delta'''_{#1}}}}%
\newcommand{\stEnvEmpty}{\stFmt{\emptyset}}%
\newcommand{\stEnvMap}[2]{\stFmt{\mpFmt{#1}\mathbin{:}{#2}}}%
\newcommand{\stEnvComp}{\mathpunct{\stFmt{,}}}%
\newcommand{\stEnvApp}[2]{\stFmt{#1\!\left(\mpFmt{#2}\right)}}%

\newcommand{\stEnvNew}[1][]{\stFmt{\ifempty{#1}{\Delta}{\Delta_{#1}}}}%
\newcommand{\stEnvNewi}[1][]{\stFmt{\ifempty{#1}{\Delta'}{\Delta'_{#1}}}}%
\newcommand{\stEnvNewii}[1][]{\stFmt{\ifempty{#1}{\Delta''}{\Delta''_{#1}}}}%

\newcommand{\stEnvAssoc}[3]{\stFmt{{#2} \mathrel{\stFmt{\sqsubseteq}_{#3}} {#1}}}

\newcommand{\stEnvMove}{\mathrel{\stFmt{\to}}}%

\newcommand{\stEnvMoveWithSession}[1][]{
\ifempty{#1}{}{
    \mathrel{\stFmt{\to}_{\!#1}}}}

\newcommand{\stEnvAnnotOutSym}{\stFmt{\oplus}}%
\newcommand{\stEnvAnnotInSym}{\stFmt{\&}}%
\newcommand{\stEnvAnnotGenericSym}[1][]{\stFmt{\ifempty{#1}{\alpha}{\alpha_{#1}}}}%
\newcommand{\stEnvAnnotGenericSymi}[1][]{\stFmt{\ifempty{#1}{\alpha'}{\alpha'_{#1}}}}%
\newcommand{\stEnvMoveAnnot}[1]{\mathrel{\stFmt{\xrightarrow{#1}}}}
\newcommand{\stEnvMoveGenAnnot}{\stEnvMoveAnnot{\stEnvAnnotGenericSym}}%
\newcommand{\stEnvMoveInAnnot}[3]{%
  \stEnvMoveAnnot{\stEnvInAnnot{#1}{#2}{#3}}%
}%
\newcommand{\stEnvMoveOutAnnot}[3]{%
  \stEnvMoveAnnot{\stEnvOutAnnot{#1}{#2}{#3}}%
}%
\newcommand{\stEnvMoveCommAnnot}[4]{%
  \stEnvMoveAnnot{\ltsSendRecv{#1}{#2}{#3}{#4}}%
}%

\newcommand{\stEnvInAnnot}[3]{\mpChanRole{\mpS}{#1}:{#2}{\stEnvAnnotInSym}{#3}}%
\newcommand{\stEnvOutAnnot}[3]{\mpChanRole{\mpS}{#1}:{#2}{\stEnvAnnotOutSym}{#3}}%
\newcommand{\stEnvInAnnotSmall}[3]{\mpChanRole{\mpS}{#1}{:}{#2}{\stEnvAnnotInSym}{#3}}%
\newcommand{\stEnvOutAnnotSmall}[3]{\mpChanRole{\mpS}{#1}{:}{#2}{\stEnvAnnotOutSym}{#3}}%
\newcommand{\stEnvCommAnnotSmall}[3]{\ltsSendRecv{\mpS}{#1}{#2}{#3}}%
\newcommand{\stEnvMoveP}[1]{{#1}\!\!\stEnvMove}%
\newcommand{\stEnvNotMoveP}[1]{{#1}\!\!\not\stEnvMove}%
\newcommand{\stEnvNotMoveWithSessionP}[2][]{{#2}\!\!\not\stEnvMoveWithSession[#1]}%
\newcommand{\stEnvMoveStar}{\mathrel{\stFmt{\stEnvMove{}^{\!\!\!*}}}}%

\newcommand{\stEnvMoveWithSessionStar}[1][]{%
  \ifempty{#1}{}{%
    \mathrel{\stFmt{\to^{\!*}}_{\!\!\!{#1}}}}}

\newcommand{\stEnvMoveAnnotP}[2]{{#1}\!\!\stEnvMoveAnnot{#2}}%

\newcommand{\stEnvEndPred}{\operatorname{end}}%
\newcommand{\stEnvEndP}[1]{\stEnvEndPred(\stFmt{#1})}%

\newcommand{\stEnvDFPred}{\operatorname{df}}%
\newcommand{\stEnvDFSessP}[2]{\stEnvDFPred({#1},\stFmt{#2})}%

\newcommand{\stEnvLivePred}{\operatorname{live}}%
\newcommand{\stEnvLiveSessP}[2]{\stEnvLivePred({#1},\stFmt{#2})}%

\newcommand{\stEnvSafePred}{\operatorname{safe}}%
\newcommand{\stEnvSafeP}[1]{\stEnvSafePred(\stFmt{#1})}%
\newcommand{\stEnvSafeSessP}[2]{\stEnvSafePred({#1},\stFmt{#2})}%

\iffalse
\newcommand{\mpEnv}[1][]{\stFmt{\ifempty{#1}{\Theta}{\Theta_{#1}}}}%
\fi

\newcommand{\mpEnv}[1][]{\stFmt{\ifempty{#1}{\Gamma}{\Gamma_{#1}}}}%
\newcommand{\mpEnvEmpty}{\stFmt{\emptyset}}%
\newcommand{\mpEnvMap}[2]{\stFmt{\mpFmt{#1}\mathbin{:}\stFmt{#2}}}%
\newcommand{\mpEnvComp}{\mathpunct{\stFmt{,}}}%

\newcommand{\mpEnvNew}[1][]{\stFmt{\ifempty{#1}{\Gamma}{\Gamma_{#1}}}}%

\newcommand{\stEnvEntails}[3]{%
  \stFmt{#1} \vdash \stFmt{\mpFmt{#2} \mathbin{:} {#3}}%
}%
\newcommand{\mpEnvEntails}[3]{%
  \stFmt{#1} \vdash \stFmt{\mpFmt{#2} \mathbin{:} \stFmt{#3}}%
}%

\iffalse
\newcommand{\stJudge}[3]{%
  \stFmt{\ifempty{#1}{#2}{{#1} \cdot {#2}}%
  \mathrel{\mpFmt{\vdash}} \mpFmt{#3}}%
}%
\fi

\newcommand{\stJudge}[3]{%
  \stFmt{\ifempty{#1}{#2}{{#1}}%
  \mathrel{\mpFmt{\vdash}} \mpFmt{#3} \; \mpFmt{\triangleright}\; #2}%
}%

\newcommand{\stJudgeNew}[3]{%
  \stFmt{\ifempty{#1}{#2}{{#1}}%
  \mathrel{\mpFmt{\vdash}} \mpFmt{#3} \; \mpFmt{\triangleright}\; #2}%
}%

\newcommand{\gtExChoice}{\gtG[\text{auth}]}%
\newcommand{\roleAuth}{{\color{roleColor}\roleFmt{a}}}%
\newcommand{\roleClient}{{\color{roleColor}\roleFmt{c}}}%
\newcommand{\roleServer}{{\color{roleColor}\roleFmt{s}}}%

\newtheorem{theorem}{Theorem}[section]
\newtheorem{lemma}[theorem]{Lemma}

\newtheorem{proposition}[theorem]{Proposition}
\newdefinition{definition}{Definition}[section]
\newdefinition{example}{Example}[section]                                                                       
\newdefinition{remark}{Remark}[section]
\newproof{pf}{Proof}
\newproof{pot}{Proof of Theorem \ref{thm2}}

\journal{Theoretical Computer Science}

\begin{document}
\begin{frontmatter}

\title{Less is More Revisited  \\ %
{\large Association with Global Protocols and Multiparty Sessions}} %

\author[label1]{Ping Hou\corref{cor1}}
\ead{ping.hou@cs.ox.ac.uk}

\author[label1]{Nobuko Yoshida}
\ead{nobuko.yoshida@cs.ox.ac.uk}

\author[label2]{Iona Kuhn}
\ead{ioku00001@stud.uni-saarland.de}

\cortext[cor1]{Corresponding author}

\affiliation[label1]{organization={Department of Computer Science, University of Oxford},%
            country={United Kingdom}}
 
 \affiliation[label2]{organization={Department of Computer Science, Saarland University},%
            country={Germany}}

\begin{abstract} 
Ensuring correctness of communication in distributed systems remains challenging. 
To address this, Multiparty session types (MPST), initially introduced by~\citet{HYC08,HYC16}, offer a type discipline in which a programmer or architect specifies an overall view of communication 
as a \emph{global protocol} (\emph{global type}), and each distributed program 
is locally type-checked against its \emph{end-point projection}.
In practice, the MPST framework has been integrated into over 25 programming
languages or tools. 
Ten years after the emergence of MPST,  
\citet{POPL19LessIsMore} 
discovered that 
existing \emph{proofs} of type safety 
using end-point projection with \emph{mergeability}  are
flawed, where the mergeability operator enlarges the 
typability of %
MPST end-point programs,  
admits easy implementation, 
and is more efficient than alternative approaches, including model checking. %
Nevertheless, following the result in~\cite{POPL19LessIsMore}, 
the soundness of end-point projection (with mergeability) has been interpreted in the literature as \emph{problematic}.
We clarify this concern by 
proposing a new general proof 
technique for \emph{type soundness} (\emph{subject reduction}) of multiparty session $\pi$-calculus,
which relies on an \emph{association} relation between the behavioural 
semantics of a global type and its
end-point projection. With this approach, \emph{behavioural properties}, namely \emph{session fidelity}, \emph{deadlock freedom}, and \emph{liveness}, are also guaranteed based on global types.
Additionally, we %
provide detailed comparisons  
with existing MPST  typing systems
and 
discuss their respective proof methods for type soundness.
\end{abstract}

\end{frontmatter}

\section{Introduction}
\label{sec:intro}

Distributed systems are built upon interactions between concurrent processes, 
implemented using \emph{message-passing} communication abstractions. 
In this model, an interaction between processes can be interpreted as the exchange of messages, forming a \emph{protocol} that consists of sending and receiving values,  making choices 
between multiple possible paths, and repeating or terminating  
the interaction. Such protocols are carried out over communication transports used in web applications, ranging from request-response client-server interactions via HTTP to
full-duplex communication channels
via the WebSocket protocol~\cite{WebSocketRFC}.

\emph{Multiparty session types} (MPST)~\cite{HYC08,HYC16} constitute %
a type formalism inspired by the 
Web Service Choreography Description Language (WS-CDL)~\cite{CDL}, 
originating as \emph{abstract choreographies} or \emph{choreography APIs} that extract communication flows and message types while abstracting away from program-level constructs such as conditionals, assignments, and concrete values. %
This formalism facilitates the description, specification, and verification 
of communication in %
concurrent and distributed systems
based on \emph{multiparty protocols}.
Intuitively, MPST describe structured multiparty interactions closely related to interaction-based formalisms such as message sequence charts (MSCs) and choreographies; 
for instance, multiparty  protocols are commonly presented in an MSC-like form (\eg \cite[Figs.~1 and 2]{HYC16}), 
while additionally providing a typing discipline and projection mechanism that support  implementability and protocol compliance.

The methodology of MPST begins with specifying a protocol, called a \emph{global type}, which
describes a sequence of communication actions, choices, and
recursions between two or more participants.
The global type is then \emph{projected} into
a set of \emph{end-point local} (\emph{session}) \emph{types}, 
each representing a participant's viewpoint. 
Well-typed \emph{end-point implementations} (\emph{processes})
that conform to a global type are
guaranteed to be \emph{correct by construction},
ensuring safety, deadlock-freedom, and liveness of interactions. 
On the practical side, MPST are \emph{language-agnostic}, \ie specifications and 
end-point projection algorithms do not depend on specific programming languages.
The top-down approach 
has been implemented in a variety of mainstream programming languages, including Java~\cite{FASE16EndpointAPI,HY2017,DBLP:journals/scp/KouzapasDPG18,DBLP:conf/tacas/BoumaGJ23},
Go~\cite{DBLP:journals/pacmpl/CastroHJNY19}, 
Rust~\cite{CYV2022,DBLP:conf/ecoop/LagaillardieNY22,DBLP:conf/ecoop/HouLY24,DBLP:conf/ecoop/VassorY24}, 
TypeScript~\cite{MFYZ2021,DBLP:conf/ecoop/GheriLSTY22},
PureScript~\cite{DBLP:journals/corr/abs-1904-01287}, 
Scala~\cite{SDHY2017,OOPSLA21FaultTolerantMPST,ECOOP22MPSTScala,BHYZ2023,DBLP:conf/issta/FerreiraJ23,DBLP:conf/tacas/Jongmans25}, 
OCaml~\cite{YZF2021,INYY2020},
MPI-C~\cite{NCY2015},
F$\star$~\cite{ZFHNY2020},
F$\sharp$~\cite{NHYA2018},
Python~\cite{DBLP:conf/rv/NeykovaYH13,NY2017Actor,DBLP:journals/fmsd/DemangeonHHNY15},
Erlang~\cite{DBLP:journals/corr/Fowler16,NY2017,DBLP:conf/icsoft/EgidiGV22},
C~\cite{NYH2012,NYNTL2012}, 
C$\sharp$~\cite{DBLP:journals/corr/abs-2004-01325},
and domain-specific actor languages~\cite{DBLP:journals/corr/abs-1208-4632,DBLP:journals/scp/CharalambidesDA16,DBLP:conf/ecoop/00020DG21}.  
These implementations adopt diverse styles,  such as code generation, API generation,
static type checking,
protocol compliance, and  
runtime monitoring, facilitating their application in real-world programs~\cite{Yoshida24}.

Figure~\ref{fig:overview-of-topdown} illustrates the MPST workflow. 
The design of multiparty protocols begins with a  global type $\gtG$
(top of Figure~\ref{fig:overview-of-topdown}),
and each participant's implementation (process) $P_{\roleP}$ (bottom) 
relies on its local (session) type $\stT_{\roleP}$ (middle), obtained via  
end-point projection of the global type $\gtProj{\gtG}{\roleP}$. 
The global and local types capture the global and local communication
behaviours, respectively. 
Since each process $P_{\roleP}$ conforms to its type $\stT_{\roleP}$, 
the resulting set of processes collectively 
makes progress 
according to the global type.

\begin{figure}[t]
\centering
{\footnotesize
  \begin{tikzpicture}
 \node (Gtext) {A Global Type $\gtG$};
  \node[below= 0.5mm of Gtext, xshift=-27mm, align=center, pos=0.5] (proj) {
    \small {\bf{projection}}\,($\upharpoonright$)};
  \node[below=6mm of Gtext, xshift=-40mm] (LA) {\footnotesize Local Type for $\roleP$
  \small \boxed{\stT_{\roleP}}};
  \node[below=6mm of Gtext] (LB) {\footnotesize Local Type for $\roleQ$  \small \boxed{\stT_{\roleQ}}};
  \node[below=6mm of Gtext, xshift=40mm] (LC) {\footnotesize Local Type for $\roleR$
   \small  \boxed{\stT_{\roleR}}};
  \draw[->] (Gtext) -- (LA);
  \draw[->] (Gtext) -- (LB);
  \draw[->] (Gtext) -- (LC);
   \node[below= 1mm of LA, xshift=-8.4mm, align=center] (typ) {
    \small {\bf{typing}}\,($\vdash$)};
   \node[below=6mm of LA, xshift=0mm] (PA) {\, \, \, \footnotesize Process for $\roleP$ \small
    \boxed{P_{\roleP}}};
   \node[below=6mm of LB, xshift=0mm] (PB) {\, \, \, \footnotesize Process for $\roleQ$ \small
    \boxed{P_{\roleQ}}};
     \node[below=6mm of LC, xshift=0mm] (PC) {\, \, \, \footnotesize Process for $\roleR$ \small
    \boxed{P_{\roleR}}};
   \draw[->] (LA) -- (PA);
   \draw[->] (LB) -- (PB);
   \draw[->] (LC) -- (PC);
\end{tikzpicture}
}
\caption{Top-down methodology of multiparty session types}
\label{fig:overview-of-topdown}
\end{figure}

\paragraph{\bf Top-Down Multiparty Session Types} 
In this paper, we demonstrate the  rigour of the MPST \emph{top-down} 
approach by revisiting the work of Scalas and Yoshida~\cite{POPL19LessIsMore}.  
That work proposed a general MPST framework, known as the \emph{bottom-up} approach, which does \emph{not} require global types. 
Within the bottom-up framework, process safety is enforced by directly verifying a relevant set of local types.  %
\citet{DBLP:journals/pacmpl/UdomsrirungruangY25} reveal that, compared with the top-down methodology, the bottom-up offers greater typability but incurs higher computational cost, particularly for liveness checking and type inference. 
Moreover, as shown in~\cite{POPL19LessIsMore}, in asynchronous MPST -- where processes communicate via unbounded FIFO queues -- type checking becomes undecidable under the bottom-up approach, whereas decidability is preserved by the top-down procedure, as type checking is performed with respect to end-point types obtained via a decidable projection of a 
well-formed global type~\cite{HYC08,HYC16}\footnote{Type checking discussed here does not involve subtyping checks; in particular, asynchronous subtyping checking is generally undecidable (\eg~\cite{GPPSY2023}).}. Intuitively, the global type fully characterises the communication structure in advance, thereby reducing type checking to verifying local conformance rather than inferring global compatibility among independently specified local behaviours.

The work in~\cite{POPL19LessIsMore} also 
identified flaws in certain published type-safety proofs for top-down typing systems based on projection with \emph{mergeability}. 
In subsequent literature, this issue has been referred to as 
``\emph{broken proofs}" or ``\emph{unsound results}", as well as the claim that ``\emph{several versions of classical projection with the full merge are flawed}", framed in terms of the ``\emph{brittleness}" of merging-based projection mechanisms and the architectural limitations of existing MPST frameworks (\eg~\cite{DBLP:conf/pldi/Castro-Perez0GY21,DBLP:journals/pacmpl/JacobsBK22a,DBLP:conf/ecoop/JongmansF23,DBLP:conf/ecoop/Stutz23,DBLP:conf/esop/StutzD25}). 
In particular, this has led to a characterisation of these proof flaws as ``\emph{unsoundness under a specific formal definition}" (\eg \cite{DBLP:conf/itp/LiW25}), and to accounts emphasising the error-proneness of the theory~(\eg \cite{DBLP:journals/pacmpl/JacobsBK22a,DBLP:journals/pacmpl/HinrichsenJK24}).  
Without careful qualification, such assessments may be over-generalised into conclusions about the soundness of the top-down approach (with mergeability), or even about global types being \emph{problematic} in general.

In this paper, we show that a sound typing system can indeed be built using 
end-point projection with mergeability.  
More importantly, we demonstrate that global types enable a clear and structured proof method for type soundness.
To clarify the statement in~\cite{POPL19LessIsMore} precisely, 
we summarise global types, end-point projections, and merging operators, 
following their chronological development in the literature. 

\paragraph{End-point Projection with Plain Merging}
Theories and implementations based on MPST require  ``\emph{correct by construction}" protocols that 
prevent deadlocks and type errors during the interaction of endpoint programs. 
This is achieved by imposing a \emph{well-formedness condition} on global types, known as \emph{projectability}. A global type $\gtG$ is \emph{projectable} if a set of end-point types can be generated from it, following formal rules or algorithms.  

The MPST framework
with global types and end-point projection %
was first introduced by~\citet{HYC08}, %
with \emph{linearity conditions} %
ensuring 
both the well-formedness of global types annotated with type-level channel declarations 
and %
the \emph{projectability} of local types. %
Subsequently, \citet{BCDDDY08} proposed a simplified MPST system without channel declarations, 
which has since been widely adopted in both theory and practice; 
the global types used in this paper also follow this channel-less style.  
The end-point projection defined in~\cite{HYC08,BCDDDY08} employs \emph{plain merging}, which requires that, in a branching, the projections for uninvolved roles yield identical local types across all branches.

To illustrate end-point projection (with plain merging), consider a simple global type involving three participants: 
$\roleFmt{A}$lice,
$\roleFmt{B}$ob, and $\roleFmt{C}$arol. 
In the initial choice,
$\roleFmt{A}$ sends to $\roleFmt{B}$
\emph{either} a request for  $\gtMsgFmt{add}$ition or
$\gtMsgFmt{sub}$traction (each carrying an $\stFmtC{int}$eger). 
In both cases, %
$\roleFmt{B}$ continues by $\gtMsgFmt{forward}$ing 
the answer  $\stFmtC{int}$ 
to $\roleFmt{C}$, and   
the session then $\gtEnd$s. This protocol is specified 
by the following global type $\gtG[\text{p}]$:
  \begin{equation}
  \label{global:type:plain}
    \gtG[\text{p}] =%
    \gtCommRaw{\roleFmt{A}}{\roleFmt{B}}{
      \begin{array}{l}%
        \gtCommChoice{add}{\stFmtC{int}}{%
          \gtCommSingle{\roleFmt{B}}{\roleFmt{C}}{forward}{\stFmtC{int}}{%
           \gtEnd}
           }%
        \\%
        \gtCommChoice{sub}{\stFmtC{int}}{%
          \gtCommSingle{\roleFmt{B}}{\roleFmt{C}}{forward}{\stFmtC{int}}{%
           \gtEnd}
          }%
      \end{array}
    }
    \end{equation}
\noindent
Following \cite{HYC08,BCDDDY08}, the global type $\gtG[\text{p}]$ is projected onto three local types
(one for each role $\roleFmt{A}$, $\roleFmt{B}$, $\roleFmt{C}$):
\begin{equation}
\label{plain:local:types}
  \begin{array}{c}
    \stT[\roleFmt{A}] =%
      \stIntSum{\roleFmt{B}}{}{%
        \begin{array}{@{\hskip 0mm}l@{\hskip 0mm}}%
          \stChoice{\stLabFmt{add}}{\stFmtC{int}}\stSeq\stEnd %
          \\%
          \stChoice{\stLabFmt{sub}}{\stFmtC{int}}\stSeq{\stEnd} %
        \end{array}
      }%
      \quad 
       \stT[\roleFmt{B}]  =%
      \stExtSum{\roleFmt{A}}{}{%
        \begin{array}{@{\hskip 0mm}l@{\hskip 0mm}}%
          \stChoice{\stLabFmt{add}}{\stFmtC{int}} \stSeq%
          \stOut{\roleFmt{C}}{%
            \stChoice{\stLabFmt{forward}}{\stFmtC{int}}\stSeq \stEnd %
          }{}%
          \\%
          \stChoice{\stLabFmt{sub}}{\stFmtC{int}} \stSeq%
          \stOut{\roleFmt{C}}{%
            \stChoice{\stLabFmt{forward}}{\stFmtC{int}}
\stSeq \stEnd %
          }{}%
        \end{array}
      }%
\quad       
      \stT[\roleFmt{C}] =%
      \stExtSum{\roleFmt{B}}{}{%
        \begin{array}{@{\hskip 0mm}l@{\hskip 0mm}}%
          \stChoice{\stLabFmt{forward}}{\stFmtC{int}}
           \stSeq  \stEnd
        \end{array}
      }%
  \end{array}
\end{equation}
  \noindent
Here, $\stT[\roleFmt{A}]$ represents the 
interface of $\roleFmt{A}$ in $\gtG[\text{p}]$: %
it must send ($\stIntC$) to $\roleFmt{B}$ either %
$\stChoice{\stLabFmt{add}}{}$ %
or $\stChoice{\stLabFmt{sub}}{}$. %
In the first case, %
$\roleFmt{B}$ 
receives the $\stChoice{\stLabFmt{add}}{}$ message with $\stExtC$, 
performs an addition with some specific integer, 
$\stChoice{\stLabFmt{forward}}{}$s the result  to $\roleFmt{C}$,  
and the session ends.  %
Otherwise, 
the message $\stChoice{\stLabFmt{sub}}{}$ with $\stExtC$ is received
at $\roleFmt{B}$; after the subtraction at $\roleFmt{B}$, 
the result is
$\stChoice{\stLabFmt{forward}}{}$ed to $\roleFmt{C}$,
and the session likewise ends. 

We highlight the projection onto $\roleFmt{C}$: 
its local behaviour is determined by merging the 
two choices at $\roleFmt{B}$. 
In this example, both branches deliver the same message  
$\stChoice{\stLabFmt{forward}}{}$  to $\roleFmt{C}$, so the merge succeeds and the projection is well-defined. 
If instead the messages or payload types differed, the merge would fail, illustrating the constraint imposed by plain merging on the projectability of global types.

\paragraph{End-point Projection with Full Merging}
We now modify $\gtG[\text{p}]$ to obtain the following global type  $\gtG[\text{f}]$, 
where the message sent from $\roleFmt{B}$ to $\roleFmt{C}$ depends on the choice made by $\roleFmt{A}$:

  \begin{equation}
  \label{global:type:full}
    \gtG[\text{f}]    =%
    \gtCommRaw{\roleFmt{A}}{\roleFmt{B}}{
      \begin{array}{l}%
    \gtCommChoice{add}{\stFmtC{int}}{%
    \gtCommRaw{\roleFmt{B}}{\roleFmt{C}}{
      \begin{array}{l}%
       \gtCommChoice{add}{\stFmtC{int}}{%
           \gtEnd}
        \\%
        \gtCommChoice{forward}{\stFmtC{int}}{%
        \gtEnd
        }%
      \end{array}
    }        
           }%
    \\%
   \gtCommChoice{sub}{\stFmtC{int}}{%
   \gtCommRaw{\roleFmt{B}}{\roleFmt{C}}{
     \begin{array}{l}%
      \gtCommChoice{sub}{\stFmtC{int}}{%
          \gtEnd}
        \\%
       \gtCommChoice{forward}{\stFmtC{int}}{%
       \gtEnd
       }%
      \end{array}
   }%
   }
   \end{array}
  }
\end{equation}
\noindent
Under \emph{plain merging}, 
$\gtG[\text{f}]$ is no longer projectable onto $\roleFmt{C}$, 
since the projections for 
$\roleFmt{C}$ differ across branches.

However,  if $\roleFmt{C}$ prepares for all three message labels, $\stLabFmt{add}$, 
$\stLabFmt{sub}$ and $\stLabFmt{forward}$, its implementation will not
get stuck. 
That is, we should be able to merge not only a common label like 
$\stLabFmt{forward}$, but also input choice~(branching) types from the same role with disjoint
labels $\stLabFmt{add}$ and $\stLabFmt{sub}$,  
into a single input type that offers both options. 
This motivates \emph{full merging}, which generalises plain merging.  
Under the full merging strategy, projection of the global type  $\gtG[\text{f}]$
yields the following local types $\stTi[\roleFmt{A}], \stTi[\roleFmt{B}]$, and $\stTi[\roleFmt{C}]$: 

\medskip
  \centerline{\(
  \begin{array}{c}
    \stTi[\roleFmt{A}] =%
      \stIntSum{\roleFmt{B}}{}{%
        \begin{array}{@{\hskip 0mm}l@{\hskip 0mm}}%
          \stChoice{\stLabFmt{add}}{\stFmtC{int}} \stSeq \stEnd %
          \\%
          \stChoice{\stLabFmt{sub}}{\stFmtC{int}} \stSeq \stEnd %
        \end{array}
      }%
      \quad 
       \stTi[\roleFmt{B}]  =%
      \stExtSum{\roleFmt{A}}{}{%
        \begin{array}{@{\hskip 0mm}l@{\hskip 0mm}}%
          \stChoice{\stLabFmt{add}}{\stFmtC{int}} \stSeq%
          \stIntSum{\roleFmt{C}}{}{%
           \begin{array}{@{\hskip 0mm}l@{\hskip 0mm}}%
           \stChoice{\stLabFmt{add}}{\stFmtC{int}}\stSeq \stEnd %
            \\
            \stChoice{\stLabFmt{forward}}{\stFmtC{int}}\stSeq \stEnd
            \end{array}
          }%
          \\%
          \stChoice{\stLabFmt{sub}}{\stFmtC{int}} \stSeq%
          \stIntSum{\roleFmt{C}}{}{%
           \begin{array}{@{\hskip 0mm}l@{\hskip 0mm}}%
           \stChoice{\stLabFmt{sub}}{\stFmtC{int}}\stSeq \stEnd %
            \\
            \stChoice{\stLabFmt{forward}}{\stFmtC{int}}\stSeq \stEnd
            \end{array}
          }%
        \end{array}
      }%
\quad       
      \stTi[\roleFmt{C}] =%
      \stExtSum{\roleFmt{B}}{}{%
        \begin{array}{@{\hskip 0mm}l@{\hskip 0mm}}%
          \stChoice{\stLabFmt{add}}{\stFmtC{int}} \stSeq%
          \stEnd %
              {}%
          \\%
         \stChoice{\stLabFmt{sub}}{\stFmtC{int}} \stSeq%
          \stEnd\\
          \stChoice{\stLabFmt{forward}}{\stFmtC{int}} \stSeq\stEnd
        \end{array}
      }%
  \end{array}
  \)}

\medskip

To enlarge the set of projectable global types, 
end-point projection with full merging was first introduced by~\citet{ParameterisedYDBH10} and further developed  in~\cite{ParameterisedYDBH12,Denielou2012,CHEN2015708,TY2016}. 
Its aim is to broaden the typability of processes 
by allowing more global types to be considered projectable (well-formed) than under the 
original, more limited projection in~\cite{HYC08}. 
Mergeability is implemented in the Scribble protocol description language
\cite{scribble10,YHNN2013,YZF2021}
and other related MPST tools. 
As illustrated by $\gtG[\text{f}]$~\eqref{global:type:full},
without  full merging,  
a broad range of global protocols, including those with branch-dependent continuations,  are not projectable.
We define the projection and mergeability formally 
in Definition~\ref{def:global-proj}.

\paragraph{Subject Reduction Theorem}
We now revisit the issue highlighted %
in~\cite{POPL19LessIsMore}. 
While the top-down approach to MPST guarantees \emph{type soundness}, also known as \emph{subject reduction}, 
a subtle problem arises in the \emph{proofs} of the Subject Reduction Theorem given in~\cite{ParameterisedYDBH10,Denielou2012,CHEN2015708,TY2016}. 
These proofs rely on an \emph{invariance property} that is valid under projection with plain merging but \emph{invalid} under full merging. 
Consequently, although both the theorem and the projection algorithm under full merging remain correct, the existing proofs are \emph{unsound} in the general setting with full merging.

Intuitively, the \emph{Subject Reduction Theorem} 
states that typed processes reduce only to typed processes and, therefore, 
no (untypable) error state can be reached, \ie 
\emph{``typed processes never go wrong''}.

Formally, subject reduction is stated in terms of the typing judgement:

\smallskip
\centerline{\(
\stJudgeNew{\mpEnvNew}{\stEnvNew}{\mpP}
\)}

\smallskip
\noindent
which asserts that the process $\mpP$ conforms to the standard \emph{typing context} $\mpEnvNew$ 
and the \emph{session typing context} $\stEnvNew$. Specifically, $\mpEnvNew$ assigns base types to variables, 
while $\stEnvNew$ is a collection of session types, recording for each communication channel its current protocol state.

 Typically, one expects a formulation of subject reduction similar to that of the simply typed $\lambda$-calculus: 
\begin{quote}
\label{con:subject-reduction-1}%
\textbf{(SR1)}\ 
\em 
Assume $\stJudgeNew{\mpEnvNew}{\stEnvNew}{\mpP}$. 
If $\mpP \!\mpMove\! \mpPi$, then  
we have $\stJudgeNew{\mpEnvNew}{\stEnvNew}{\mpPi}$. 
\end{quote} 
However, this statement is \emph{too strong}, 
since the session typing context $\stEnvNew$ may change  when $\mpP$ reduces to $\mpPi$ 
due to communication. Session types are \emph{behavioural}: channel types can evolve during interactions.

A weaker property is therefore: 
\begin{quote}
  \label{con:subject-reduction-2}%
\textbf{(SR2)}\ 
\em 
Assume $\stJudgeNew{\mpEnvNew}{\stEnvNew}{\mpP}$. 
If $\mpP \!\mpMove\! \mpPi$, then  
there exists $\stEnvNewi$ such that 
$\stJudgeNew{\mpEnvNew}{\stEnvNewi}{\mpPi}$. 
\end{quote} 
Unfortunately, this condition is \emph{too general (\ie, too weak)} to ensure 
type safety, since $\stEnvNewi$ may be \emph{arbitrary}. 

To recover type safety, subject reduction must be strengthened by imposing an appropriate invariance condition:

\begin{quote}
\textbf{(SR3)}\ 
\label{con:subject-reduction-3}%
\em Assume $\stJudgeNew{\mpEnvNew}{\stEnvNew}{\mpP}$
and \underline{$\stEnvNew$ satisfies property $\varphi$}. 
If $\mpP \!\mpMove\! \mpPi$, then  
there exists $\stEnvNewi$ such that 
$\stJudgeNew{\mpEnvNew}{\stEnvNewi}{\mpPi}$
and \underline{$\stEnvNewi$ satisfies property $\varphi$}. 
\end{quote} 

The invariance condition required by {\bf (SR3)} is introduced at the typing of session restriction.
In top-down MPST, 
a channel-restricted process 
$\mpRes{\mpS}{}{\mpP}$ denotes a complete session $\mpS$, 
where each participant playing role $\roleP$ is typed by the projection 
$\gtProj{\gtG}{\roleP}$. 
Restricted processes are typed using the following rule: 
\begin{equation}
\label{eq:global}
    \inference[\iruleMPRes$\gtG$]{%
   \stEnvNewi = \setenum{%
   \stEnvMap{\mpChanRole{\mpS}{\roleP}}{\gtProj{\gtG}{\roleP}}%
    }_{\roleP \in \gtRoles{\gtG}}
      \qquad%
      \mpS \not\in \stEnvNew%
      \qquad%
      \stJudgeNew{\mpEnvNew}{%
      \stEnvNew \stEnvComp \stEnvNewi %
 }{%
        \mpP%
      }%
    }{%
      \stJudgeNew{\mpEnvNew}{%
        \stEnvNew %
      }{%
       \mpRes{\stEnvMap{\mpS}{\stEnvNewi}}\mpP
      }%
    }%
\end{equation}
This rule is crucial in subject reduction proofs, 
as it determines the initial session typing context $\stEnvNewi$ 
for a newly created session $\mpS$, and thereby fixes the invariance property required by {\bf (SR3)} to hold initially.

\paragraph{Problem: An Incorrect Invariant $\varphi$ for the Type System with Projection Using Full Merging} 
In binary session types, such an invariance property $\varphi$ is known as \emph{balancedness}~\cite{GH05}, which 
requires 
the two endpoints of %
a channel to carry \emph{dual types}: 
one %
offering selection ($\stIntC$) while the other %
branching ($\stExtC$). %
For multiparty session types, the challenge is therefore to identify an appropriate invariance property $\varphi$ that %
supports a sound subject reduction argument. 

Following the intuition of balancedness, a natural attempt in the multiparty setting is 
to define the invariance property $\varphi$ as a form of duality between local types. 
The system in \cite{BCDDDY08} introduced \emph{consistency} (also called \emph{coherence}, \cf~\cite{ParameterisedYDBH12}) as a generalisation of binary duality: 
a global type $\gtG$ %
projects to \emph{consistent} local types if 
projecting $\gtG$ onto any two participants yields local types whose \emph{partial projections} -- capturing their respective behaviours towards each other -- are dual.

As an illustration, consider the local types $\stT[\roleFmt{A}]$, $\stT[\roleFmt{B}]$, and 
$\stT[\roleFmt{C}]$ in~\eqref{plain:local:types}, projected from 
$\gtG[\text{p}]$ in~\eqref{global:type:plain}.   
 
For the pair $(\roleFmt{A}, \roleFmt{B})$, the corresponding partial projections

\smallskip 
\centerline{\(
\begin{array}{c}
\gtProj{\stT[\roleFmt{A}]}{\roleFmt{B}}\ = \ 
      \stIntSum{}{}{%
        \begin{array}{@{\hskip 0mm}l@{\hskip 0mm}}%
          \stChoice{\stLabFmt{add}}{\stFmtC{int}} \stSeq \stEnd %
          \\%
          \stChoice{\stLabFmt{sub}}{\stFmtC{int}} \stSeq \stEnd %
        \end{array}
     }%
      \quad  \text{and}     \quad     
\gtProj{\stT[\roleFmt{B}]}{\roleFmt{A}}
\ = \ 
     \stExtSum{}{}{%
        \begin{array}{@{\hskip 0mm}l@{\hskip 0mm}}%
          \stChoice{\stLabFmt{add}}{\stFmtC{int}} \stSeq \stEnd %
          \\%
          \stChoice{\stLabFmt{sub}}{\stFmtC{int}} \stSeq \stEnd %
        \end{array}
     }%
  \end{array}
\)}

\smallskip
\noindent
 are dual. The same holds for the remaining pairs; hence, the projections of $\gtG[\text{p}]$ are consistent.

Under plain merging, \emph{consistency} is sufficient to establish subject reduction. 
Plain merging enforces structural alignment of local types -- sharing the same labels, payloads, and mergeable continuations -- so that projection yields consistent local types, which are preserved under reduction. 
However, this guarantee does not extend to full merging:  %
{\bf full merging admits inconsistent local types}, thereby violating \textbf{(SR3)}. 

As a counterexample, consider the global type $\gtG[\text{w}]$ updated from 
$\gtG[\text{p}]$ in~\eqref{global:type:plain} as follows: 

 \begin{equation}
 \label{eq:counter_example_global_inconsistent}
    \gtG[\text{w}] =%
    \gtCommRaw{\roleFmt{A}}{\roleFmt{B}}{
      \begin{array}{l}%
        \gtCommChoice{\gtMsgFmt{add}}{\stFmtC{int}}{%
          \gtCommSingle{\roleFmt{B}}{\roleFmt{C}}{\gtMsgFmt{add}}{\stFmtC{int}}{%
            \gtCommSingle{\roleFmt{C}}{\roleFmt{A}}{forward}{\stFmtC{int}}{%
              \gtEnd%
            }%
          }%
        }%
        \\%
        \gtCommChoice{sub}{\stFmtC{int}}{%
          \gtCommSingle{\roleFmt{B}}{\roleFmt{C}}{forward}{\stFmtC{int}}{%
              \gtEnd%
          }%
        }%
      \end{array}
    }
    \end{equation}

\smallskip
\noindent
The local types projected from $\gtG[\text{w}]$ using full merging are: 

\begin{equation}
 \label{eq:counter_example_projection_inconsistent}
 \begin{array}{c}
    \stTii[\roleFmt{A}] =%
      \stIntSum{\roleFmt{B}}{}{%
        \begin{array}{@{\hskip 0mm}l@{\hskip 0mm}}%
          \stChoice{\stLabFmt{add(\stFmtC{int})}}{} \stSeq%
          \roleFmt{C}
          \stFmt{\&}
          \stLabFmt{forward(\stFmtC{int})}%
          \\%
          \stChoice{\stLabFmt{sub(\stFmtC{int})}}{}%
        \end{array}
      }%
      \quad 
       \stTii[\roleFmt{B}]  =%
      \stExtSum{\roleFmt{A}}{}{%
        \begin{array}{@{\hskip 0mm}l@{\hskip 0mm}}%
          \stChoice{\stLabFmt{add(\stFmtC{int})}}{} \stSeq%
          \stOut{\roleFmt{C}}{%
            \stChoice{\stLabFmt{add}}{\stFmtC{int}}%
          }{}%
          \\%
          \stChoice{\stLabFmt{sub}}{\stFmtC{int}} \stSeq%
          \stOut{\roleFmt{C}}{%
            \stChoice{\stLabFmt{forward(\stFmtC{int})}}{}%
          }{}%
        \end{array}
      }%
     \\[2mm] 
      \stTii[\roleFmt{C}] =%
      \stExtSum{\roleFmt{B}}{}{%
        \begin{array}{@{\hskip 0mm}l@{\hskip 0mm}}%
          \stChoice{\stLabFmt{add}}{\stFmtC{int}} \stSeq%
          \stOut{\roleFmt{A}}{%
            \stChoice{\stLabFmt{forward}}{\stFmtC{int}}%
          }{}%
          \\%
          \stChoice{\stLabFmt{forward}}{\stFmtC{int}}%
        \end{array}
      }%
  \end{array}
  \end{equation}
  
\smallskip
\noindent
The resulting local types %
fail to be consistent, since the partial projections $\gtProj{\stTii[\roleFmt{A}]}{\roleFmt{C}}$ and 
$\gtProj{\stTii[\roleFmt{C}]}{\roleFmt{A}}$ are \emph{undefined}. 
Intuitively, the interaction between 
$\roleFmt{C}$ and $\roleFmt{A}$ is branch-dependent on $\roleFmt{B}$. 
Such inter-role dependencies are not captured by the syntactic nature of projection and duality checks, and may therefore lead to inconsistency.

More specifically, the issue in the subject reduction proofs of~\cite{ParameterisedYDBH10,ParameterisedYDBH12,Denielou2012,CHEN2015708,TY2016} results from  their reliance on consistency as the underlying invariance assumption. Moreover, \cite[p. 28]{ParameterisedYDBH12} and \cite[Prop. 2]{CHEN2015708} explicitly claim that projecting a global type with full merging yields a consistent typing context. 
Therefore, a broader invariance condition is required for fully mergeable global types.

\paragraph{Solution: Association}
In this paper, we prove the Subject Reduction Theorem under full merging by applying an invariance notion, namely \emph{association}, which captures compatibility between a global type and its local types. 
Intuitively, association allows local types to safely conform to the global protocol, rather than requiring them to be exact projections -- a relaxation that is essential under full merging.

Formally,  association relates a global type and a typing context for a given multiparty session 
 via \emph{subtyping}. 
 The subtyping relation $\stSub$ on local types is typically used to enhance 
 expressiveness by enabling more flexible processes to be typed while preserving behavioural soundness. %

For instance, consider the global type $\gtG[\text{w}]$~\eqref{eq:counter_example_global_inconsistent} and 
its projection on $\roleFmt{B}$, $\stTii[\roleFmt{B}]$~\eqref{eq:counter_example_projection_inconsistent}, where 
$\roleFmt{B}$ branches on inputs from $\roleFmt{A}$. 
Under our subtyping discipline, branching is contravariant: 
a subtype of $\stTii[\roleFmt{B}]$ may accept additional labels, \eg \stFmt{\stLabFmt{multiplus}}, beyond 
\stFmt{\stLabFmt{add}} and \stFmt{\stLabFmt{sub}}.   
A process implementing $\roleFmt{B}$ can therefore support extra operations, 
\eg by offering a further branch labelled $\mpFmt{\mpLabFmt{multiplus}}$, and remains typable by 
$\stTii[\roleFmt{B}]$ via subsumption. %
In this way,  subtyping permits such local refinements without altering the  
communication pattern fixed by $\gtG[\text{w}]$; association formalises precisely this flexibility.  

 More precisely, a typing context 
 $\stEnvNew$ is \emph{associated} with a global type $\gtG$ for a session $\mpS$, 
 written $\stEnvAssoc{\gtG}{\stEnvNew}{\mpS}$, 
if for every role $\roleP$ of $\gtG$, the endpoint $\mpChanRole{\mpS}{\roleP}$ in $\stEnv$ 
has a type satisfying $\stEnvApp{\stEnvNew }{\mpChanRole{\mpS}{\roleP}} \stSub \gtProj{\gtG}{\roleP}$, 
and all other endpoints of that session are terminated (\ie assigned $\stEnd$). 
Moreover, we say that \emph{$\stEnvNew$ is associated} if,  for every session $\mpS$ occurring in $\stEnv$,  
there exists a global type $\gtG$ such that $\stEnvAssoc{\gtG}{\stEnvNew_s}{\mpS}$, where 
$\stEnvNew_{\mpS}$ denotes the restriction of $\stEnvNew$ to $\mpS$.

With association in place, we prove the following result:
\begin{quote}
\textbf{(SR)}\ 
\em Assume $\stJudgeNew{\mpEnvNew}{\stEnvNew}{\mpP}$
and $\stEnvNew$ is associated.  
If $\mpP \!\mpMove\! \mpPi$, then  
there exists $\stEnvNewi$ such that 
$\stJudgeNew{\mpEnvNew}{\stEnvNewi}{\mpPi}$
with $\stEnvNewi$ associated. 
\end{quote} 
This is achieved by demonstrating a \emph{sound} and \emph{complete}  
operational correspondence 
between global and local types 
with respect to association.   
Finally, by applying the standard subsumption rule, which allows typing contexts to be widened, 
we conclude that the typing system equipped with rule~\eqref{eq:global} satisfies the Subject Reduction Theorem.

\paragraph{\bf Outline}  \textbf{\S\ref{sec:processes}} introduces a multiparty session $\pi$-calculus, including its syntax and operational semantics. The calculus builds on that of~\cite{POPL19LessIsMore}, extended with values, expressions, and conditional statements.
\vspace{-5.5pt}
\begin{itemize}[leftmargin=0pt,itemsep=-1.5pt]
\item[] \textbf{\S\ref{sec:sessiontypes}} presents a multiparty session type theory. 
\textbf{\S\ref{sec:gtype:syntax}} defines the syntax of global and local types, as well as projection and subtyping.  \textbf{\S\ref{sec:gtype:lts-gt}} and \textbf{\S\ref{sec:gtype:lts-context}} provide the semantics of global types and typing contexts (sets of local types), respectively. 
\textbf{\S\ref{sec:gtype:relating}} introduces the \emph{association relation} $\stEnvAssoc{\gtG}{\stEnv}{\mpS}$, which relates a global type $\gtG$ with a typing context $\stEnv$ on a session $\mpS$ via projection and subtyping. 
The relation establishes an operational correspondence between global types and typing contexts~(\Cref{thm:gtype:proj-sound,thm:gtype:proj-comp}). 
\textbf{\S\ref{sec:gtype:counterexample_projection}}
motivates the need for association through an illustrative example. 
\textbf{\S\ref{sec:gtype:pbp}} demonstrates that association ensures key typing context properties -- communication safety, deadlock-freedom, and liveness~(\Cref{cor:allproperties}).  
\textbf{\S\ref{sec:sync:relation}} clarifies the relationships among these properties~(\Cref{lem:liveness-safety-implic}). 

\item[] \textbf{\S\ref{sec:typesystem}} develops a typing system for our multiparty session $\pi$-calculus. 
\textbf{\S\ref{sec:type-system:tyrules}} formalises the typing rules. 
\textbf{\S\ref{sec:subject_reduction}}  establishes subject reduction: first with respect to association~(\Cref{lem:subject-reduction}), and subsequently using projection~(\Cref{lem:final-subject-reduction}), with type safety following as a corollary~(\Cref{cor:type-safety}). 
\textbf{\S\ref{sec:session_fidelity}} demonstrates session fidelity, also known as protocol conformance, \ie 
that well-typed processes behave according to their session types, in an analogous  way: initially via association (\Cref{lem:session-fidelity}), and thereafter by projection~(\Cref{lem:final-session-fidelity}).
\textbf{\S\ref{sec:typed-process-property}} shows that process properties -- deadlock-freedom and liveness -- are guaranteed by construction~(\Cref{lem:stenv-proc-properties}). 
\end{itemize}
\vspace{-5.5pt}
We discuss related work in \textbf{\S\ref{sec:related}} and conclude 
in \textbf{\S\ref{sec:conclusion}}. 
Detailed proofs are provided in the appendix. 

\paragraph{Extensions and Refinements of~\citet{DBLP:series/lncs/YoshidaH24}} 
This work extends~\citet{DBLP:series/lncs/YoshidaH24}, which introduced association as a proof technique for establishing type soundness in MPST.  
We advance this line by clarifying certain misunderstandings and challenges surrounding type soundness proofs in the MPST community, making these explicit in the Introduction. 
A dedicated section~(\Cref{sec:gtype:counterexample_projection}) further illustrates the necessity of association within the framework. 
On the technical side, we present a full multiparty session $\pi$-calculus that extends~\cite{DBLP:series/lncs/YoshidaH24} with constructs such as expressions and conditionals. Our system adopts a distinct subtyping discipline and typing judgements, yielding a substantially different yet more expressive type system. The definition of association is refined to align with the subtyping order, while the session restriction rule~(\inferrule{\iruleMPGlobalRes} in~\Cref{fig:typing_rules_all},\Cref{sec:typesystem}) is reformulated to use projection directly rather than association. In this setting, we formalise subject reduction and session fidelity through projection -- results not developed in~\cite{DBLP:series/lncs/YoshidaH24} -- in addition to the association-based theorems. 
Moreover,  the definition of process liveness~(\Cref{def:proc-properties} in~\Cref{sec:typed-process-property}) 
is strengthened by incorporating  an additional side condition. 
Finally, we include a comprehensive discussion of related work~(\Cref{sec:related}) and provide complete proofs.

\paragraph{\bf  Comparison with~\citet{POPL19LessIsMore}}
Since \cite{POPL19LessIsMore} serves as a key benchmark for this work, 
we summarise below the %
differences between the two frameworks, providing a clear and quick reference.
\begin{itemize}[leftmargin=*,itemsep=0pt]
\item Top-down vs. bottom-up methodology: the main difference lies in the underlying paradigm.  
\citet{POPL19LessIsMore} develop a bottom-up MPST theory that is independent of global types and instead founded on a  parametric safety invariant.  In contrast, 
this work follows a top-down MPST approach, in which correctness is ensured by construction through global types.  This leads to distinct technical structures and guarantees, as follows. 
\begin{itemize}[leftmargin=*,itemsep=0pt]
\item Global types and projection: the proposed framework provides explicit global types together with their syntax and semantics, and supports end-point projection with full merging (\Cref{sec:gtype:syntax,sec:gtype:lts-gt});  such mechanisms are absent from the theory in~\cite{POPL19LessIsMore}. 

\item Core invariant: type soundness in this work (established via the subject reduction theorem, \Cref{lem:subject-reduction}) is based on a global-to-local association invariant induced by projection and subtyping, whereas~\cite{POPL19LessIsMore} relies on preserving a parametric typing-context safety invariant that is independent of global types (\cite[Theorem 4.6]{POPL19LessIsMore}).

\item Behavioural property guarantees: in this framework, once a process is typed under a typing context associated with a global type, key behavioural properties, such as session fidelity, communication safety, deadlock-freedom, and liveness, follow by construction (\Cref{sec:session_fidelity,sec:typed-process-property}). 
In contrast, the guarantees in~\cite{POPL19LessIsMore}, established in \cite[\S 5.1, 5.2 and 5.5]{POPL19LessIsMore},  depend on the particular instantiations of the safety invariant (\cite[\S 5.3 and 5.4]{POPL19LessIsMore}). 
To support this approach, \cite{POPL19LessIsMore} provides a model-checking–based mechanism (\cite[\S 6]{POPL19LessIsMore}) to verify specific invariant instantiations at the type level, from which the corresponding process-level guarantees are derived. 
\end{itemize} 

\item Technical correction: the definition of process liveness in~\cite[Def. 5.1]{POPL19LessIsMore} contains a subtle issue:  
it allows the liveness condition to hold without requiring executable communication. 
We address this by refining the definition (\Cref{def:proc-properties}) to include an explicit side condition and presenting a counterexample in \Cref{sec:typed-process-property}. 

\item Minor differences (without impact on the theoretical results). 
\begin{itemize}[leftmargin=*,itemsep=0pt]
\item Multiparty session $\pi$-calculus: the calculus presented in~\Cref{sec:processes} 
is an extension of that in~\cite[\S 2.1]{POPL19LessIsMore}, incorporating values, expressions, and conditional statements. 

\item Basic types: to support values and expressions,  basic types are introduced in \Cref{sec:gtype:syntax}, which are absent in~\cite{POPL19LessIsMore}.

\item Subtyping: the subtyping discipline formalised in ~\Cref{def:subtyping} 
follows the \emph{process-oriented} approach \cite{DBLP:journals/toplas/CarboneHY12,DBLP:conf/concur/DemangeonH11,DBLP:conf/tlca/MostrousY09,ChenDY14,ChenDSY17},  
in which branching is contravariant and selection covariant. 
This contrasts with the \emph{channel-oriented} subtyping order of~\cite[Def. 2.5]{POPL19LessIsMore}, 
where the variance is reversed. 
A detailed comparison of these two approaches  is provided in~\cite{Gay2016}. 
Our association method is parametric in the chosen subtyping discipline and applies to \emph{both}.

\item Typing contexts, judgements, and rules: the typing contexts (\Cref{def:mpst-env}) are extended to account for expression variables and basic types; consequently, separate judgements (\Cref{sec:type-system:tyrules}) are applied to the typing of expressions and processes, with process typing formulated differently from that in~\cite{POPL19LessIsMore}. The typing rules (\Cref{fig:typing_rules_all}) are reformulated accordingly, resulting in a more expressive typing discipline.

\item Typing context liveness:  the notion of typing context liveness adopted here (\Cref{def:stenv-live}) is stronger than the liveness condition in~\cite[Fig. 5]{POPL19LessIsMore} and is closer  to $\text{liveness}^{+}$.

\item Synchronous vs. asynchronous MPST: \citet{POPL19LessIsMore} extend their theory to asynchronous MPST and establish a corresponding subject reduction result  (\cite[\S 7]{POPL19LessIsMore}), whereas this work focuses on the synchronous setting.

\end{itemize}
\end{itemize}

\section{Multiparty Session $\pi$-Calculus}
\label{sec:processes}
This section presents the syntax of the synchronous multiparty session $\pi$-calculus, 
and provides a formalisation of its operational semantics. 

\paragraph{Syntax of Processes in Multiparty Session $\pi$-Calculus} 
A \emph{session} is a sequence of interactions, typically 
including send and receive operations, 
performed by a set of \emph{roles}~(\emph{participants}) in a communication protocol. %
The multiparty session $\pi$-calculus 
models the behaviour of processes that interact
using multiparty channels. 

We use the following basic notations: \emph{basic values}, denoted by $\mpV, \mpVi, \mpV[i], \ldots$, 
\emph{expressions}, denoted by $\mpE, \mpEi, \mpE[i], \ldots$, 
\emph{expression variables}, denoted by $\mpFmt{x}, \mpFmt{x'}, \mpFmt{x_i}, \ldots$, 
\emph{channels}, denoted by $\mpC, \mpCi, \mpC[i], \ldots$, 
\emph{channel variables}, denoted by $\mpFmt{y}, \mpFmt{y'}, \mpFmt{y_i}, \ldots$, 
\emph{roles}, denoted by $\roleP, \rolePi, \roleQ, \roleQi, \roleP[i], \ldots$, 
\emph{sessions}, denoted by $\mpS, \mpSi, \mpS[i], \ldots$, 
\emph{message labels}, denoted by $\mpLab, \mpLabi, \mpLab[i], \ldots$, 
\emph{processes}, denoted by $\mpP, \mpPi, \mpQ, \mpQi, \mpP[i], \ldots$, 
and 
\emph{process variables}, denoted by $\mpX, \mpXi, \mpY, \mpYi, \mpX[i], \ldots$.

\begin{definition}[Syntax of Multiparty Session $\pi$-Calculus]%
\label{def:mpst-syntax-terms}
The \emph{multiparty session $\pi$-calculus} syntax is defined as follows:%

\smallskip
\centerline{\(
\begin{array}{r@{\hskip 2mm}c@{\hskip 2mm}l@{\hskip 5mm}l}
  \textstyle%
  \mpV 
  &\coloncolonequals&
  \mpNat \bnfsep \mpInt \bnfsep \mpTrue \bnfsep \mpFalse \bnfsep \mpString{} \bnfsep \mpUnit{} \bnfsep \cdots
  & \mbox{\footnotesize(natural number, integer, boolean, string, or unit, \ldots)}
  \\[.5mm]
  \mpE 
  &\coloncolonequals& 
  \mpFmt{x} \bnfsep \mpV \bnfsep \mpSucc{\mpE} \bnfsep \mpNeg{\mpE} \bnfsep
  \\
 &&  \neg \mpE \bnfsep \mpE \otimes \mpEi \bnfsep \mpE < \mpEi 
  \bnfsep \cdots 
  & \mbox{\footnotesize(expression variable, basic value, or expression term)}
  \\[.5mm]
  \mpC
  &\coloncolonequals&%
  \mpFmt{y} \bnfsep \mpChanRole{\mpS}{\roleP}%
  & \mbox{\footnotesize(channel variable or channel for session $\mpS$ with role $\roleP$)}
  \\[.5mm]
   \mpFmt{z}
  &\coloncolonequals&
  \mpFmt{x} \bnfsep \mpFmt{y}
  & \mbox{\footnotesize(expression variable or channel variable)}
  \\[.5mm]
  \mpD
  &\coloncolonequals&
  \mpE \bnfsep \mpC
  & \mbox{\footnotesize(expression, channel variable, or channel with role)}
  \\[.5mm]
  \mpW
  &\coloncolonequals&
  \mpE \bnfsep \mpChanRole{\mpS}{\roleP}
  & \mbox{\footnotesize(expression or channel with role)}
  \\[.5mm]
   \mpU
  &\coloncolonequals&
  \mpV \bnfsep \mpChanRole{\mpS}{\roleP}
  & \mbox{\footnotesize(value or channel with role)}
  \\[.5mm]
  \mpP, \mpQ
  &\coloncolonequals&%
  \mpNil
    \bnfsep \mpRes{\mpS}{\mpP}%
  &
  \mbox{\footnotesize(inaction, restriction)}
  \\[.3mm]
  &&
  \mpSel{\mpC}{\roleQ}{\mpLab}{\mpD}{\mpP}
  &
  \mbox{\footnotesize(selection towards role $\roleQ$)}
  \\[.3mm]
  &&
  \mpBranch{\mpC}{\roleQ}{i \in I}{\mpLab[i]}{z_i}{\mpP[i]}{}%
  &
  \mbox{\footnotesize(branching from role $\roleQ$ with an index set $I \neq \emptyset$)}
  \\[.3mm]
  &&
  \mpDefAbbrev{\mpJustDef{\mpX}{\mpFmt{x_1}, \ldots, \mpFmt{x_n}, \mpFmt{y_1}, \ldots, \mpFmt{y_m}}{P}}{\mpQ}%
  &
  \mbox{\footnotesize(process definition)}%
  \\[.3mm]
  &&
  \mpCallSmall{\mpX}{\mpE[1], \ldots, \mpE[n], \mpC[1], \ldots,  \mpC[m]}
  &
  \mbox{\footnotesize (process call)}
  \\[.3mm]
   && 
   \mpIf{\mpE}{\mpP}{\mpQ}
   &  
   \mbox{\footnotesize(conditional)}
   \\[.3mm]
  &&
  \mpP \mpPar \mpQ
  \bnfsep
   \mpErr
   &
   \mbox{\footnotesize(parallel composition, error)}
\end{array}
\)}

\smallskip
\noindent%
Restriction, branching, and process definitions and declarations act as
binders, as expected. 
$\fc{\mpP}$ denotes the set of \emph{free channels with roles} in $\mpP$, 
$\fev{\mpP}$ denotes the set of \emph{free expression variables} in $\mpP$, 
and $\fcv{\mpP}$ denotes 
the set of \emph{free channel variables} in $\mpP$. 
 We adopt a form of Barendregt convention: bound sessions and process 
 variables are assumed pairwise distinct, and different from free ones. 
We write $\mpBigPar{i \in I}{\mpP[i]}$ for the parallel composition of processes $\mpP[i]$.
\end{definition}

The syntax of our session $\pi$-calculus~(\Cref{def:mpst-syntax-terms}) is mostly
standard~\citep{POPL19LessIsMore}, with extensions for expressions and 
``if\ldots then\ldots else'' statements. 
A basic value $\mpV$ can be a natural number $\mpNat$,   
an integer $\mpInt$, 
a boolean $\mpTrue$ or $\mpFalse$, 
a string $\mpString{}$, 
a unit $\mpUnit{}$~(often omitted for brevity), 
or any other specific tailored value. 
An expression $\mpE$ can be an expression variable 
$\mpFmt{x}$, a basic value $\mpV$, or 
a term built from expressions by applying operators, \eg 
$\text{\texttt{succ}}, \text{\texttt{neg}}, \neg, \otimes, <$.  
A channel $\mpC$ can be either a channel variable $\mpFmt{y}$ or 
a \emph{channel with role}~(\aka \emph{session endpoint}) $\mpChanRole{\mpS}{\roleP}$, 
a multiparty communication
endpoint whose user plays role $\roleP$ in the session $\mpS$.

A process acts as a communication agent within a session, 
representing the behaviour and actions of a role in the session. 
 \emph{Inaction} $\mpNil$ represents a terminated process~(often omitted for brevity). 
\emph{Session restriction} $\mpRes{\mpS}{\mpP}$ 
declares a new session $\mpS$, 
with its scope restricted to the process $\mpP$. 
\emph{Selection} (\aka \emph{internal choice})
$\mpSel{\mpC}{\roleQ}{\mpLab}{\mpD}{\mpP}$
sends a message $\mpLab$ with payload
$\mpD$ to role $\roleQ$ via endpoint $\mpC$, 
where $\mpD$ may be either an expression or a channel. 
\emph{Branching} (\aka \emph{external choice})
$\mpBranch{\mpC}{\roleQ}{i \in I}{\mpLab[i]}{z_i}{\mpP[i]}{}$
expects to receive a message $\mpLab[i]$
(for some $i \in I$) %
from role $\roleQ$ via endpoint $\mpC$,
and then continues
as $\mpP[i]$. 

The \emph{conditional} process $\mpIf{\mpE}{\mpP}{\mpQ}$ 
represents the internal choice between processes $\mpP$ and $\mpQ$, 
with the selection of the branch determined by the \emph{evaluation} of expression $\mpE$~(as explained later in the semantics). 
\emph{Process definition} $\mpDefAbbrev{\mpJustDef{\mpX}{\mpFmt{x_1}, \ldots, \mpFmt{x_n}, \mpFmt{y_1}, \ldots, \mpFmt{y_m}}{P}}{\mpQ}$ and 
\emph{process call} $\mpCallSmall{\mpX}{\mpE[1], \ldots, \mpE[n], \mpC[1], \ldots,  \mpC[m]}$ capture recursion: 
the call invokes $\mpX$ by expanding it into $\mpQ$ and replacing its formal parameters with the actual arguments.
\emph{Parallel composition} $\mpP \mpPar \mpQ$ denotes two processes 
 capable of concurrent execution and potential communication. 
Finally, $\mpErr$ represents the \emph{error} process.  

\begin{figure}[t]
  \centerline{\(
    \begin{array}{@{\hskip 0mm}c@{\hskip 0mm}}
    \eval{\mpV}{\mpV} 
    \quad 
    \eval{\mpSucc{\mpNat}}{(\mpNat + 1)}
    \quad 
    \eval{\mpNeg{\mpInt}}{(-\mpInt)}
    \quad 
    \eval{\neg \mpTrue}{\mpFalse}
    \quad 
    \eval{\neg \mpFalse}{\mpTrue}
     \\[1.5mm]%
  \inference[]{\eval{\mpE[1]}{\mpV}}{\eval{\mpE[1] \otimes \mpE[2]}{\mpV}} 
  \quad 
   \inference[]{\eval{\mpE[2]}{\mpV}}{\eval{\mpE[1] \otimes \mpE[2]}{\mpV}}   
   \quad 
   \eval{\mpInt_{1} < \mpInt_{2}}{\left\{%
    \begin{array}{@{}l@{\hskip 5mm}l@{}}
      \mpTrue
      &\text{\small %
        if\, $\mpInt_{1} < \mpInt_{2}$}%
      \\[1mm]%
     \mpFalse
      &\text{\small%
      otherwise
      }
    \end{array}
    \right.
}
        \end{array}
    \)}%
\caption{%
Expression evaluation. 
}%
  \label{fig:expression_eval}%
\end{figure}

\begin{figure}[t]
  \centerline{\(%
    \begin{array}{c}
    \begin{array}{rl@{}l}
     \inferrule{\iruleMPRedCommExp}& %
      \mpBranch{\mpChanRole{\mpS}{\roleP}}{\roleQ}{i \in I}{%
        \mpLab[i]}{z_i}{\mpP[i]}{}%
      \,\mpPar\,%
      \mpSel{\mpChanRole{\mpS}{\roleQ}}{\roleP}{\mpLab[k]}{%
       \mpE
      }{\mpQ}%
      \;\;\mpMove\;\;%
      \mpP[k]\subst{\mpFmt{z_k}}{\mpV}
      \,\mpPar\,%
      \mpQ%
      &
      \text{\small{}$k \!\in\! I$ and $\eval{\mpE}{\mpV}$}%
      \\[2mm]
     \inferrule{\iruleMPRedCommChannel}& %
      \mpBranch{\mpChanRole{\mpS}{\roleP}}{\roleQ}{i \in I}{%
        \mpLab[i]}{z_i}{\mpP[i]}{}%
      \,\mpPar\,%
      \mpSel{\mpChanRole{\mpS}{\roleQ}}{\roleP}{\mpLab[k]}{%
       \mpChanRole{\mpSi}{\roleR}
      }{\mpQ}%
      \;\;\mpMove\;\;%
      \mpP[k]\subst{\mpFmt{z_k}}{\mpChanRole{\mpSi}{\roleR}}
      \,\mpPar\,%
      \mpQ%
   &
      \text{\small{}$k \!\in\! I$}%
      \\[2mm]%
      \inferrule{\iruleMPRedCall}& %
      \mpDef{\mpX}{\mpFmt{x_1}, \ldots, \mpFmt{x_n}, \mpFmt{y_1}, \ldots, \mpFmt{y_m}}{\mpP}{(%
        \mpCall{\mpX}{%
          \mpE[1], \ldots,%
          \mpE[n], \mpChanRole{\mpS[1]}{\roleP[1]}, \ldots, 
          \mpChanRole{\mpS[m]}{\roleP[m]}
        }%
        \,\mpPar\,%
        \mpQ%
        )%
      }%
      \\%
      &\hspace{10mm}%
      \;\mpMove\;%
      \mpDef{\mpX}{\mpFmt{x_1}, \ldots, \mpFmt{x_n}, \mpFmt{y_1}, \ldots, \mpFmt{y_m}}{\mpP}{
      }%
      \\
      &
      \hspace{16mm}
      (%
        \mpP\subst{\mpFmt{x_1}}{\mpV[1]}%
        \cdots%
        \subst{\mpFmt{x_n}}{\mpV[n]}
        \subst{\mpFmt{y_1}}{\mpChanRole{\mpS[1]}{\roleP[1]}}
        \cdots
        \subst{\mpFmt{y_m}}{\mpChanRole{\mpS[m]}{\roleP[m]}}%
        \,\mpPar\,%
          \mpQ%
          )%
      &
       \text{\small{}$\forall i \!\in\! \setenum{1, \ldots, n}: \eval{\mpE[i]}{\mpV[i]}$}%
      \\[2mm]%
      \inferrule{\iruleMPRedCondTrue} & 
    \mpIf{e}{\mpP}{\mpQ}
      \;\;\mpMove\;\;%
      \mpP
 &
      \text{\small{}$\eval{\mpE}{\mpTrue}$}
	\\[2mm] 
    \inferrule{\iruleMPRedCondFalse} & 
    \mpIf{\mpE}{\mpP}{\mpQ}
      \;\;\mpMove\;\;%
      \mpQ
     &
      \text{\small{}$\eval{\mpE}{\mpFalse}$}
      \\[2mm]
      \inferrule{\iruleMPRedCtx}&%
      \mpP \mpMove \mpPi%
      \;\;\text{implies}\;\;%
      \mpCtxApp{\mpCtx}{\mpP} \mpMove \mpCtxApp{\mpCtx}{\mpPi}%
      &
      \\[2mm]%
      \inferrule{\iruleMPErrLabel}&%
      \mpBranch{\mpChanRole{\mpS}{\roleP}}{\roleQ}{i \in I}{%
        \mpLab[i]}{z_i}{\mpP[i]}{}%
      \,\mpPar\,%
      \mpSel{\mpChanRole{\mpS}{\roleQ}}{\roleP}{\mpLab}{%
        \mpW%
      }{\mpQ}%
      \;\;\mpMove\;\;%
      \mpErr%
      &
      \text{\small{}$\forall i \!\in\! I: \mpLab[i] \!\neq\! \mpLab$}%
      \\[2mm]%
    \inferrule{\iruleMPRedCongr} &
    \mpPi \equiv \mpP \mpMove \mpQ \equiv \mpQi
    \;\;\text{implies}\;\;
    \mpPi \mpMove \mpQi  
    &
    \end{array}
   \end{array}
   \)}
     \caption{%
    Operational semantics of multiparty session $\pi$-calculus.}
  \label{fig:mpst-pi-semantics}%
\end{figure}
\begin{figure}[t]
  \centerline{\(
    \begin{array}{@{\hskip 0mm}c@{\hskip 0mm}}
        \mpP \mpPar \mpQ
        \equiv
        \mpQ \mpPar \mpP
      \quad %
        \mpFmt{(\mpP \mpPar \mpQ) \mpPar \mpR}%
        \equiv%
        \mpP \mpPar \mpFmt{(\mpQ \mpPar \mpR)}%
      \quad%
        \mpP \mpPar \mpNil%
        \equiv%
        \mpP%
 \quad%
        \mpRes{\mpS}{\mpNil}%
        \equiv%
        \mpNil%
 \\[2mm]%
        \mpRes{\mpS}{%
          \mpRes{\mpSi}{%
            \mpP%
          }%
        }%
        \equiv%
        \mpRes{\mpSi}{%
          \mpRes{\mpS}{%
            \mpP%
          }%
        }%
    \qquad 
        \mpRes{\mpS}{(\mpP \mpPar \mpQ)}%
        \equiv%
        \mpP \mpPar \mpRes{\mpS}{\mpQ}%
        \;\;%
        \text{\small{}if\, $\mpS \!\not\in\! \fc{\mpP}$}%
      \\[2mm]
        \mpDefAbbrev{\mpDefD}{\mpNil}%
        \equiv%
        \mpNil%
      \qquad%
        \mpDefAbbrev{\mpDefD}{\mpRes{\mpS}{\mpP}}%
        \,\equiv\,%
        \mpRes{\mpS}{(
          \mpDefAbbrev{\mpDefD}{\mpP}%
          )}%
        \;\;%
        \text{\small{}if\, $\mpS \!\not\in\! \fc{\mpDefD}$}%
      \\[2mm]%
        \mpDefAbbrev{\mpDefD}{(\mpP \mpPar \mpQ)}%
        \,\equiv\,%
        \mpFmt{(\mpDefAbbrev{\mpDefD}{\mpP})} \mpPar \mpQ%
        \;\;
        \text{\small{}if\, $\dpv{\mpDefD} \cap \fpv{\mpQ} = \emptyset$}%
      \\[2mm]%
        \mpDefAbbrev{\mpDefD}{%
          (\mpDefAbbrev{\mpDefDi}{\mpP})%
        }%
        \;\equiv\;%
        \mpDefAbbrev{\mpDefDi}{%
          (\mpDefAbbrev{\mpDefD}{\mpP})%
        }%
      \\%
      \text{\small{}%
        if\, %
        \(%
        (\dpv{\mpDefD} \cup \fpv{\mpDefD}) \cap \dpv{\mpDefDi}%
        \,=\,%
        (\dpv{\mpDefDi} \cup \fpv{\mpDefDi}) \cap \dpv{\mpDefD}%
        \,=\,%
        \emptyset%
        \)%
      }%
    \end{array}
    \)}%
\caption{%
 Standard structural congruence rules, 
    where $\fpv{\mpDefD}$ is the set of \emph{free process variables} in $\mpDefD$, 
    and $\dpv{\mpDefD}$ is the set of \emph{declared process
variables} in $\mpDefD$.
}%
  \label{fig:mpst-pi-congruence}%
\end{figure}

\paragraph{Operational Semantics of Multiparty Session $\pi$-Calculus}%
\label{sec:session-calculus:semantics}
The value of an expression is computed as shown in~\Cref{fig:expression_eval}, where $\eval{\mpE}{\mpV}$ denotes that an expression $\mpE$ evaluates to a value $\mpV$. Note that~\Cref{fig:expression_eval} can be extended to include additional expressions. 
The successor operation $\text{\texttt{succ}}$ is defined for natural numbers, 
negation $\text{\texttt{neg}}$ and comparison $<$  for integers, 
and logical negation $\neg$ for boolean values. 
The nondeterministic choice $\mpE[1] \otimes \mpE[2]$ evaluates to 
either the value of $\mpE[1]$ or the value of $\mpE[2]$.

We provide the operational semantics of our multiparty 
session $\pi$-calculus in~\Cref{def:mpst-pi-semantics}, 
using a standard \emph{structural congruence} $\equiv$ 
defined  in~\Cref{fig:mpst-pi-congruence}.  %

\begin{definition}[Semantics of Multiparty Session $\pi$-Calculus]
  \label{def:mpst-proc-context}%
  \label{def:mpst-pi-reduction-ctx}%
  \label{def:mpst-pi-semantics}%
  \label{def:mpst-pi-error}%
  \label{def:calculus_semantics}
  A \emph{reduction context} $\mpCtx$ is %
  defined as: 
  
  \smallskip
  \centerline{\( 
  \mpCtx \;\coloncolonequals\;%
  \mpCtx \mpPar \mpP%
  \bnfsep%
  \mpRes{\mpS}{\mpCtx}%
  \bnfsep%
  \mpDefAbbrev{\mpDefD}{\mpCtx}%
  \bnfsep%
  \mpCtxHole
  \)}
  
  \smallskip
  \noindent
  The \emph{reduction} relation $\mpMove$ on processes is inductively defined %
  by the rules in~\cref{fig:mpst-pi-semantics}. %
  A process $\mpP$ is said to \emph{have an error} if $\mpP = \mpCtxApp{\mpCtx}{\mpErr}$ 
  for some reduction context $\mpCtx$. 
A relation $\mpCtxMove$ on reduction contexts is defined by 
$\mpCtx \,\mpCtxMove\, \mpCtxi  \text{ iff } 
 \forall \mpP.  \, \mpCtxApp{\mpCtx}{\mpP} \mpMove \mpCtxApp{\mpCtxi}{\mpP}$. 
  \end{definition}

Our operational semantics for the multiparty session $\pi$-calculus retains the standard rules from~\citep{POPL19LessIsMore}, while incorporating additional rules for 
conditionals and value exchange.  
A reduction context $\mpCtx$ %
denotes  a process with a single hole $\mpCtxHole$, occurring in place of a subterm $\mpP$. 
In addition, a relation $\mpCtxMove$ on reduction contexts is introduced to capture 
context-level reductions that hold uniformly for all process instantiations; it will be used in the definition of process liveness in~\Cref{sec:typed-process-property}. 
The reflexive-transitive closures of $\mpMove$ and $\mpCtxMove$ are denoted by $\mpMoveStar$  and $\mpCtxMoveStar$, respectively. 

Rules~\inferrule{\iruleMPRedCommExp} and~\inferrule{\iruleMPRedCommChannel} 
describe communication on session $\mpS$ involving 
receiver $\roleP$ and sender $\roleQ$, where a value or an endpoint is exchanged, respectively, 
provided that  the transmitted message $\mpLab[k]$ is 
handled by the receiver~($k \in I$).  In case of a message label mismatch, rule \inferrule{\iruleMPErrLabel} 
is invoked to trigger an $\mpErr$or. 

Rules~\inferrule{\iruleMPRedCondTrue} and~\inferrule{\iruleMPRedCondFalse} 
pertain to conditionals, which are self-explanatory. 
Rule \inferrule{\iruleMPRedCall} initiates the expansion of process definitions when called. 
Additionally, rules \inferrule{\iruleMPRedCtx} and \inferrule{\iruleMPRedCongr} facilitate the reduction of processes 
under reduction contexts and modulo structural congruence~($\equiv$), respectively.

\begin{example}[Syntax and Semantics of Session $\pi$-Calculus, adapted from~\citep{BSYZ2022}]
\label{ex:cal_syntax_semantics}
  Processes $\mpP$ and $\mpQ$ below communicate over session $\mpS$: 
  $\mpP$ uses endpoint $\mpChanRole{\mpS}{\roleP}$ to
  send endpoint $\mpChanRole{\mpS}{\roleR}$ to role $\roleQ$, while 
  $\mpQ$ uses endpoint $\mpChanRole{\mpS}{\roleQ}$ to receive it, 
  and then sends a message to role $\roleP$ via
  $\mpChanRole{\mpS}{\roleR}$.

\smallskip
\centerline{\(
\begin{array}{c}
\mpP = \mpSel{\mpChanRole{\mpS}{\roleP}}{\roleQ}
      {\mpLabi}{\mpChanRole{\mpS}{\roleR}}{%
        \mpBranchSingle{\mpChanRole{\mpS}{\roleP}}{\roleR}{\mpLab}{z}{\mpNil}
} \qquad
\mpQ = \mpBranchSingle{\mpChanRole{\mpS}{\roleQ}}{\roleP}{\mpLabi}{z}{
  \mpSel{z}{\roleP}{\mpLab}{42}{\mpNil}
}{}
\end{array}
\)}

\smallskip
\noindent
By \Cref{def:mpst-pi-semantics}, successful reductions yield:

\smallskip
\centerline{\(
  \begin{array}{rl}
  &\mpRes{\mpS}{(\mpP \mpPar \mpQ)} 
  \\
  =
  &
  \mpRes{\mpS}{(
    \mpSel{\mpChanRole{\mpS}{\roleP}}{\roleQ}{\mpLabi}{\mpChanRole{\mpS}{\roleR}}{%
        \mpBranchSingle{\mpChanRole{\mpS}{\roleP}}{\roleR}{\mpLab}{z}{\mpNil}}
    \mpPar
    \mpBranchSingle{\mpChanRole{\mpS}{\roleQ}}{\roleP}{\mpLabi}{z}{\mpSel{z}{\roleP}{\mpLab}{42}{\mpNil}}
  )}
  \\
  \mpMove
  &
  \mpRes{\mpS}{(
    \mpBranchSingle{\mpChanRole{\mpS}{\roleP}}{\roleR}{\mpLab}{z}{\mpNil}{}
    \mpPar
    \mpSel{\mpChanRole{\mpS}{\roleR}}{\roleP}{\mpLab}{42}{\mpNil}%
  )}
  \;\mpMove\;
  \mpRes{\mpS}{(\mpNil \mpPar \mpNil)} \equiv \mpNil
  \end{array}
\)}
\end{example}

\begin{example}[OAuth Process, extended from~\citep{POPL19LessIsMore}]
\label{ex:process_oauth}
The following process interacts on session $\mpS$ 
using channels with role 
$\mpChanRole{\mpS}{\roleFmt{s}}$, %
$\mpChanRole{\mpS}{\roleFmt{c}}$, $\mpChanRole{\mpS}{\roleFmt{a}}$, 
corresponding respectively to roles $\roleFmt{s}$, $\roleFmt{c}$, and $\roleFmt{a}$.
For brevity,  irrelevant message payloads are omitted. %

\medskip%
\centerline{\(%
  \mpRes{\mpS}{\left(%
    \mpP[\roleFmt{s}] \mpPar \mpP[\roleFmt{c}] \mpPar \mpP[\roleFmt{a}]%
  \right)}%
  \quad \text{where:}
  \left\{%
  \scalebox{0.9}{%
  \begin{minipage}{0.75\linewidth}
  \(%
  \begin{array}{@{\hskip 0mm}r@{\hskip 1mm}c@{\hskip 1mm}l}
    \mpP[\roleFmt{s}] &=&%
    \mpSel{\mpChanRole{\mpS}{\roleFmt{s}}}{\roleFmt{c}}{\mpLabFmt{cancel}}{}{}%
    \\[1mm]%
    \mpP[\roleFmt{c}] &=&%
    \mpBranchRaw{\mpChanRole{\mpS}{\roleFmt{c}}}{\roleFmt{s}}{\!\left\{%
    \begin{array}{l}
      \mpChoice{login}{}{%
        \mpSel{\mpChanRole{\mpS}{\roleFmt{c}}}{\roleFmt{a}}{passwd}{\mpString{XYZ}}{}%
      }\\
      \mpChoice{cancel}{}{%
        \mpSel{\mpChanRole{\mpS}{\roleFmt{c}}}{\roleFmt{a}}{quit}{}{}%
      }, \, 
      \mpChoice{fail}{}{%
        \mpSel{\mpChanRole{\mpS}{\roleFmt{c}}}{\roleFmt{a}}{fatal}{}{}%
      }%
      \end{array}
      \right\}%
    }%
    \\[3mm]%
    \mpP[\roleFmt{a}] &=&%
    \mpBranchRaw{\mpChanRole{\mpS}{\roleFmt{a}}}{\roleFmt{c}}{\!\left\{%
		    \begin{array}{l}
      \mpChoice{passwd}{z}{%
        \mpSel{\mpChanRole{\mpS}{\roleFmt{a}}}{\roleFmt{s}}{auth}{\mpString{secret}}{}%
      }\\
      \mpChoice{quit}{}{}, \, 
      \mpChoice{fatal}{}{}%
		\end{array}
      \right\}%
    }%
  \end{array}
  \)%
  \end{minipage}
  }%
  \right.%
\)}%

\medskip%
\noindent%
Here, %
$\mpRes{\mpS}{%
  \left(\mpP[\roleFmt{s}] \mpPar \mpP[\roleFmt{c}] \mpPar \mpP[\roleFmt{a}]\right)%
}$ %
is %
the parallel composition of processes %
$\mpP[\roleFmt{s}]$, $\mpP[\roleFmt{c}]$, $\mpP[\roleFmt{a}]$ %
in the scope of session $\mpS$. %
In $\mpP[\roleFmt{s}]$, %
``$\mpSel{\mpChanRole{\mpS}{\roleFmt{s}}}{\roleFmt{c}}{\mpLabFmt{cancel}}{}{}$'' %
means using $\mpChanRole{\mpS}{\roleFmt{s}}$ %
to send $\mpLabFmt{cancel}$ to $\roleFmt{c}$. %
Process $\mpP[\roleFmt{c}]$ %
uses $\mpChanRole{\mpS}{\roleFmt{c}}$ to receive\, %
$\mpLabFmt{login}$, 
$\mpLabFmt{cancel}$, or $\mpLabFmt{fatal}$ 
from $\roleFmt{s}$;  %
then, %
in the first case, it uses $\mpChanRole{\mpS}{\roleFmt{c}}$ %
to send $\mpLabFmt{passwd}$ to $\roleFmt{a}$; %
in the second case, %
it uses $\mpChanRole{\mpS}{\roleFmt{c}}$ %
to send $\mpLabFmt{quit}$ to $\roleFmt{a}$; 
in the third case, 
it uses $\mpChanRole{\mpS}{\roleFmt{c}}$ %
to send $\mpLabFmt{fatal}$ to $\roleFmt{a}$. 
By \Cref{def:mpst-pi-semantics}, %
we have the reductions: %

\smallskip%
\centerline{\(
  \mpRes{\mpS}{\left(%
  \mpP[\roleFmt{s}] \mpPar \mpP[\roleFmt{c}] \mpPar \mpP[\roleFmt{a}]%
  \right)}%
  \;\mpMove\;%
  \mpRes{\mpS}{\left(%
    \mpNil \,\mpPar\,%
    \mpSel{\mpChanRole{\mpS}{\roleFmt{c}}}{\roleFmt{a}}{quit}{}{}%
    \,\mpPar\, \mpP[\roleFmt{a}]%
  \right)}%
  \;\mpMove\;%
  \mpRes{\mpS}{\left(%
    \mpNil \!\mpPar\! \mpNil \!\mpPar\! \mpNil%
  \right)}%
  \;\equiv\; \mpNil
\)}%
\end{example}

\section{Multiparty Session Types}
\label{sec:sessiontypes}
This section introduces multiparty session types. In~\Cref{sec:gtype:syntax}, 
we provide an extensive exploration of global and local types, 
including their syntax, projection, and subtyping. 
In~\Cref{sec:gtype:lts-gt}, we establish a Labelled Transition System (LTS) semantics for global types, 
and in~\Cref{sec:gtype:lts-context}, for typing contexts. 
 The operational relationship between these semantics is explained in~\Cref{sec:gtype:relating} via association. 
In~\Cref{sec:gtype:counterexample_projection}, we show the necessity of association:  %
an example illustrates that projection alone does not preserve this operational correspondence. 
Finally, in~\Cref{sec:gtype:pbp}, 
we demonstrate  that a typing context associated with a global type  ensures safety, deadlock-freedom, and liveness.

\subsection{Global and Local Types}
\label{sec:gtype:syntax}
The Multiparty Session Type (MPST) theory utilises \emph{global types} to 
provide a comprehensive overview of communication between 
\emph{roles}, such as $\roleP, \roleQ, \roleS, \roleT, \ldots$, belonging to a set $\roleSet$.
In contrast, it employs \emph{local types}, derived via \emph{projection} from a global type,  
to describe the communication of an \emph{individual} role from a local perspective. 
The syntax for global and local types is presented in~\cref{fig:syntax-mpst}, 
where most constructs are standard~\citep{POPL19LessIsMore}.

\begin{figure}[t]%
  \centerline{\(
    \begin{array}{r@{\quad}c@{\quad}l@{\quad}l}
     \tyGround & \bnfdef & \tyNat \bnfsep \tyInt \bnfsep \tyBool \bnfsep %
     \tyString \bnfsep \tyUnit \bnfsep \ldots
        & \text{\small Basic types} \\
      \stS & \bnfdef & \tyGround \bnfsep \stT
        & \text{\small Basic type or Session type} \\
      \gtG & \bnfdef &
        \gtComm{\roleP}{\roleQ}{i \in I}{\gtLab[i]}{\tyS[i]}{\gtG[i]}
        &
        \text{\small Transmission} \\
        & \bnfsep & \gtRec{\gtRecVar}{\gtG} \quad \bnfsep \quad \gtRecVar 
        \quad \bnfsep \quad \gtEnd
        &
        \text{\small Recursion, Type variable, Termination}  \\[1ex]
      \stT
        & \bnfdef & \stExtSum{\roleP}{i \in I}{\stChoice{\stLab[i]}{\stS[i]} \stSeq \stT[i]} 
          & \text{\small External choice} \\
        & \bnfsep & \stIntSum{\roleP}{i \in I}{\stChoice{\stLab[i]}{\stS[i]} \stSeq \stT[i]}
          & \text{\small Internal choice}\\
        & \bnfsep & \stRec{\stRecVar}{\stT} \quad \bnfsep \quad \stRecVar 
        \quad \bnfsep \quad \stEnd
        &
        \text{\small Recursion, Type variable, Termination}
    \end{array}
  \)}
\caption{Syntax of types.} 
  \label{fig:syntax-global-type}%
  \label{fig:syntax-local-type}%
  \label{fig:syntax-mpst}
\end{figure}

\vspace{1.5mm}
\noindent
{\bf Basic types},  taken from a set $\tyGroundSet$, 
describe the types of values, including natural numbers, integers, booleans, %
strings, units, and others.  %

\vspace{1.5mm}
\noindent
{\bf Global types}, represented as  
$\gtG, \gtGi, \gtG[i], \ldots$, describe the high-level behaviour of all roles. 
The set of roles 
in a global
type $\gtG$ is denoted by $\gtRoles{\gtG}$.  
Each syntactic construct of global types is explained as follows:

\begin{itemize}
\item  
$\gtComm{\roleP}{\roleQ}{i \in
I}{\gtLab[i]}{\tyS[i]}{\gtG[i]}$: a \emph{transmission} of a message from role $\roleP$ to role $\roleQ$, 
consisting of a label $\gtLab[i]$, a payload of type $\tyS[i]$~(either a basic or local type), 
and a continuation $\gtG[i]$, where $i$ is taken from a non-empty index
set $I$. %
The labels
$\gtLab[i]$ must be pairwise distinct, and self receptions are excluded (\ie
$\roleP \neq \roleQ$).

\item $\gtRec{\gtRecVar}{G}$: a \emph{recursive} global type, subject to 
contractive requirements~\citep[\S 21.8]{PierceTAPL}, meaning that  
each recursion variable $\gtRecVar$ is bound within a $\gtRec{\gtRecVar}{\ldots}$ 
and is guarded. 
\item $\gtEnd$: a \emph{terminated} global type, omitted when unambiguous.
\end{itemize}
The set of free variables in $\gtG$, denoted by $\gtFv{\gtG}$, is defined as usual.

\vspace{1.5mm}
\noindent
{\bf Local types} 
(or {\bf session types}), represented as  $\stT, \stU, \stTi, \stT[i], \stUi, \stU[i], \ldots$, describe
the behaviour of a single role. Each syntactic construct of local types is elucidated below: 
\begin{itemize}
\item 
$\stIntSum{\roleP}{i \in I}{\stChoice{\stLab[i]}{\stS[i]} \stSeq \stT[i]}$: 
an \emph{internal choice}~(\emph{selection}), indicating 
that the \emph{current} role is expected to \emph{send} a message to role $\roleP$. 

\item $\stExtSum{\roleP}{i \in I}{\stChoice{\stLab[i]}{\stS[i]} \stSeq \stT[i]}$: 
an \emph{external choice} (\emph{branching}), indicating 
that the \emph{current} role is expected to \emph{receive} a message  from role $\roleP$. 

\item  $\stRec{\stRecVar}{\stT}$: a \emph{recursive} local type, following a pattern analogous to 
recursive global types. 

\item $\stEnd$: a \emph{termination}, omitted when unambiguous. 
\end{itemize}
Similar to global types, we require pairwise-distinct,
non-empty labels in local types.

\vspace{1.5mm}
\noindent
{\bf Projection } %
In the  top-down approach of MPST, local types are obtained 
by projecting a global type onto roles.  \emph{Projection} is a \emph{partial} function 
that yields the local type associated with a participating role in a global type. 
Specifically, it takes a global type $\gtG$ and a role $\roleP$,  returning the corresponding local type. 
Our definition of projection, provided in~\Cref{def:global-proj} below, follows the standard approach~\citep{POPL19LessIsMore}.

\begin{definition}[Global Type Projection]%
  \label{def:global-proj}%
  \label{def:local-type-merge}%
  The \emph{projection of a global type $\gtG$ onto a role $\roleP$}, %
  written $\gtProj{\gtG}{\roleP}$,~is:

\smallskip
  \centerline{\(%
  \begin{array}{c}
    \gtProj{\left(%
      \gtCommSmall{\roleQ}{\roleR}
      {i \in I}{\gtLab[i]}{\stS[i]}{\gtG[i]}%
      \right)}{\roleP}%
    =\!%
    \left\{%
    \begin{array}{@{}l@{\hskip 5mm}l@{}}
      \stIntSum{\roleR}{i \in I}{ %
        \stChoice{\stLab[i]}{\stS[i]} \stSeq (\gtProj{\gtG[i]}{\roleP})%
      }%
      &\text{\small %
        if\, $\roleP = \roleQ$%
      }%
      \\[2mm]%
      \stExtSum{\roleQ}{i \in I}{%
        \stChoice{\stLab[i]}{\stS[i]} \stSeq (\gtProj{\gtG[i]}{\roleP})%
      }%
      &\text{\small%
        if\, $\roleP = \roleR$
      }
      \\[2mm]%
      \stMerge{i \in I}{\gtProj{\gtG[i]}{\roleP}}%
      &%
      \text{\small%
        if\, $\roleP \neq \roleQ$ \,and \,$\roleP \neq \roleR$
      }%
    \end{array}
    \right.
    \\[8mm]%
    \gtProj{(\gtRec{\gtRecVar}{\gtG})}{\roleP}%
    \;=\;%
    \left\{%
    \begin{array}{@{\hskip 0.5mm}l@{\hskip 3mm}l@{}}
      \stRec{\stRecVar}{(\gtProj{\gtG}{\roleP})}%
      &%
      \text{\small%
        if\,
        $\roleP \in \gtRoles{\gtG}$ \,or\,
        $\gtFv{\gtRec{\gtRecVar}{\gtG}} \neq \emptyset$
      }%
      \\%
      \stEnd%
      &%
      \text{\small%
        otherwise}
    \end{array}
    \right.%
    \quad\qquad%
    \begin{array}{@{}r@{\hskip 1mm}c@{\hskip 1mm}l@{}}
      \gtProj{\gtRecVar}{\roleP}%
      &=&%
      \stRecVar%
      \\%
      \gtProj{\gtEnd}{\roleP}%
      &=&%
      \stEnd%
    \end{array}
  \end{array}
  \)}%
  
  \smallskip%
\noindent%
   where
  $\stMerge{}{}$ is %
  the \emph{merge operator for session types} %
  (\emph{full merging}), defined as:

  \smallskip%
  \centerline{\(%
  \begin{array}{c}%
    \textstyle%
     \stIntSum{\roleP}{i \in I}{\stChoice{\stLab[i]}{\stS[i]} \stSeq \stT[i]}%
    \,\stBinMerge\,%
    \stIntSum{\roleP}{i \in I}{\stChoice{\stLab[i]}{\stS[i]} \stSeq \stTi[i]}%
    \;=\;%
    \stIntSum{\roleP}{i \in I}{\stChoice{\stLab[i]}{\stS[i]} \stSeq (\stT[i]
    \stBinMerge \stTi[i])}%
    \\[1mm]%
    \stRec{\stRecVar}{\stT} \,\stBinMerge\, \stRec{\stRecVar}{\stTi}%
    \,=\,%
    \stRec{\stRecVar}{(\stT \stBinMerge \stTi)}%
    \quad%
    \stRecVar \,\stBinMerge\, \stRecVar%
    \,=\,%
    \stRecVar%
    \quad%
    \stEnd \,\stBinMerge\, \stEnd%
    \,=\,%
    \stEnd%
    \\[1mm]
    \stExtSum{\roleP}{i \in I}{\stChoice{\stLab[i]}{\stS[i]} \stSeq \stT[i]}%
    \,\stBinMerge\,%
    \stExtSum{\roleP}{\!j \in J}{\stChoice{\stLab[j]}{\stS[j]} \stSeq \stTi[j]}%
   \;=\;  \stExtSum{\roleP}{\!k \in I \cup J}{\stChoice{\stLab[k]}{\stS[k]} \stSeq \stTii[k]}
    \hspace{2mm} 
    \text{where}  \hspace{1.5mm} 
   \stTii[k]
    \;=\;%
    \left\{%
    \begin{array}{@{\hskip 0.5mm}l@{\hskip 3mm}l@{}}
     \stT[k]  \stBinMerge \stTi[k]
      &%
      \text{\footnotesize%
        if\, 
       $\stFmt{k} \in \stFmt{I \cap J}$
      }%
      \\%
       \stT[k]
      &%
      \text{\footnotesize%
        if\, 
        $\stFmt{k} \in \stFmt{I \setminus J}$
      }%
      \\%
     \stTi[k]
      &%
      \text{\footnotesize%
        if\, 
        $\stFmt{k} \in \stFmt{J \setminus I}$
      }
    \end{array}
    \right.%
  \end{array}
  \)}%
\end{definition}

Note that some global types cannot be projected onto all of their participants. 
This occurs when certain global types describe inherently meaningless protocols, 
leading to undefined merging operations, as illustrated 
in~\Cref{ex:merge}.

If a global type $\gtG$ begins with a transmission from role $\roleP$ to role $\roleQ$, projecting it onto 
role $\roleP$~(resp. $\roleQ$) results in an internal (resp. external) choice, provided that the continuation 
of each branching of $\gtG$ is also projectable. When projecting $\gtG$ onto other participants $\roleR$~($\roleR \neq \roleP$ and $\roleR \neq \roleQ$), a merge operator, as defined in~\Cref{def:global-proj} and exemplified in~\Cref{ex:merge},  
is used to ensure that the projections of all continuations are ``compatible''. 
If a global type $\gtG$ is a termination or a type variable, projecting it onto any role results in a termination or type variable, respectively. 
Finally, the projection of a recursive global type $\gtRec{\gtRecVar}{\gtG}$ preserves its recursive structure when projected onto a role in $\gtG$, or if $\gtRec{\gtRecVar}{\gtG}$ contains no free type variables; 
otherwise, the projection yields termination.

\begin{example}[Types of OAuth 2.0, from~\citep{POPL19LessIsMore}]
\label{ex:types_oath}
Consider a protocol  from OAuth 2.0~\citep{OAUTH}: 
the $\roleFmt{s}$ervice sends %
to the $\roleFmt{c}$lient \emph{either} a request 
to $\gtMsgFmt{login}$, 
\emph{or} $\gtMsgFmt{cancel}$; %
in the first case, %
$\roleFmt{c}$ continues by sending $\gtMsgFmt{passwd}$ %
(carrying a $\stFmtC{str}$ing) %
to the $\roleFmt{a}$uthorisation server, %
who in turn sends $\gtMsgFmt{auth}$ to $\roleFmt{s}$ %
(with a $\stFmtC{bool}$ean, %
telling whether the client is authorised), %
and the session $\gtEnd$s; %
in the second case, $\roleFmt{c}$ sends $\gtMsgFmt{quit}$ to $\roleFmt{a}$, %
and the session $\gtEnd$s. This protocol can be represented by 
the global type $\gtExChoice$: 
  
  \smallskip
  \centerline{\(
    \gtExChoice =%
    \gtCommRaw{\roleServer}{\roleClient}{
      \begin{array}{l}%
        \gtCommChoice{\gtMsgFmt{login}}{}{%
          \gtCommSingle{\roleClient}{\roleAuth}{\gtMsgFmt{passwd}}{\stFmtC{str}}{%
            \gtCommSingle{\roleAuth}{\roleServer}{auth}{\stFmtC{bool}}{%
              \gtEnd%
            }%
          }%
        }%
        \\%
        \gtCommChoice{cancel}{}{%
          \gtCommSingle{\roleClient}{\roleAuth}{quit}{}{%
              \gtEnd%
          }%
        }%
      \end{array}
    }
    \)}

\smallskip
\noindent%
Following the MPST top-down methodology, 
$\gtExChoice$  is projected onto three local types (one for each role
$\roleFmt{s}$, $\roleFmt{c}$, $\roleFmt{a}$):

\smallskip
  \centerline{\(
  \begin{array}{c}
    \stT[\roleFmt{s}] =%
      \stIntSum{\roleClient}{}{%
        \begin{array}{@{\hskip 0mm}l@{\hskip 0mm}}%
          \stChoice{\stLabFmt{login}}{} \stSeq%
          \roleFmt{a}
          \stFmt{\&}
          \stLabFmt{auth(\stFmtC{bool})}%
          \\%
          \stChoice{\stLabFmt{cancel}}{}%
        \end{array}
      }%
      \quad 
       \stT[\roleFmt{c}]  =%
      \stExtSum{\roleServer}{}{%
        \begin{array}{@{\hskip 0mm}l@{\hskip 0mm}}%
          \stChoice{\stLabFmt{login}}{} \stSeq%
          \stOut{\roleFmt{a}}{%
            \stChoice{\stLabFmt{passwd}}{\stFmtC{str}}%
          }{}%
          \\%
          \stChoice{\stLabFmt{cancel}}{} \stSeq%
          \stOut{\roleFmt{a}}{%
            \stChoice{\stLabFmt{quit}}{}%
          }{}%
        \end{array}
      }%
      \\[2mm] 
      \stT[\roleFmt{a}] =%
      \stExtSum{\roleFmt{c}}{}{%
        \begin{array}{@{\hskip 0mm}l@{\hskip 0mm}}%
          \stChoice{\stLabFmt{passwd}}{\stFmtC{str}} \stSeq%
          \stOut{\roleServer}{%
            \stChoice{\stLabFmt{auth}}{\stFmtC{bool}}%
          }{}%
          \\%
          \stChoice{\stLabFmt{quit}}{}%
        \end{array}
      }%
  \end{array}
  \)}
  
  \smallskip
  \noindent
 Here,  $\stT[\roleFmt{s}]$ represents the 
 interface of $\roleFmt{s}$ in $\gtExChoice$: %
it must send ($\stIntC$) to $\roleFmt{c}$ either %
$\stChoice{\stLabFmt{login}}{}$ %
or $\stChoice{\stLabFmt{cancel}}{}$; %
in the first case, %
$\roleFmt{s}$ must then %
receive ($\stExtC$) %
message $\stChoice{\stLabFmt{auth}}{\stFmtC{bool}}$ from $\roleFmt{a}$, %
and the session ends; %
otherwise, %
in the second case, %
the session just ends. %
Types $\stT[\roleFmt{c}]$ and $\stT[\roleFmt{a}]$ follow the same intuition.
\end{example}

\begin{example}[Merge and Projection, originated from~\citep{MFYZ2021}]
\label{ex:merge}
Two external choice~(branching) types from the same role with disjoint labels can be 
merged into a type
carrying both labels, 
\eg 

\smallskip
\centerline{\(
\roleFmt{A} 
 \stFmt{\&}
  \stLabFmt{greet(\stFmtC{str})} 
  \stBinMerge 
  \roleFmt{A} 
 \stFmt{\&}
  \stLabFmt{farewell(\stFmtC{bool})} 
   = 
    \stExtSum{\roleFmt{A}}{}{
  \begin{array}{@{}l@{}}
   \stLabFmt{greet(\stFmtC{str})},  
   \stLabFmt{farewell(\stFmtC{bool})} 
  \end{array}
  }
 \)}
 
 \smallskip
 \noindent
 However, this does not apply to internal choices~(selections), \eg 
 $\roleFmt{A} 
 \stFmt{\oplus}
  \stLabFmt{greet(\stFmtC{str})} 
  \stBinMerge 
  \roleFmt{A} 
 \stFmt{\oplus}
 \stLabFmt{farewell(\stFmtC{bool})}$ is undefined.

Two external choices from different roles cannot be merged. Same 
for internal choices;  \eg %
 $\roleFmt{A} 
 \stFmt{\oplus}
  \stLabFmt{greet(\stFmtC{str})} 
  \stBinMerge 
  \roleFmt{B} 
 \stFmt{\oplus}
 \stLabFmt{greet(\stFmtC{str})}
   $
   is undefined. 

Furthermore, two external choices from the same role with same labels but different payloads cannot 
be merged;  \eg %
  $\roleFmt{A} 
 \stFmt{\&}
  \stLabFmt{greet(\stFmtC{str})} 
  \stBinMerge 
  \roleFmt{A} 
 \stFmt{\&}
 \stLabFmt{greet(\stFmtC{bool})}
   $ is undefined. This also applies to internal choices. %

Additionally, two local types with same prefixes 
but unmergeable continuations 
cannot be merged;  \eg %
$\roleFmt{A} 
 \stFmt{\oplus}
  \stLabFmt{greet(\stFmtC{str})} 
  \stBinMerge 
  \roleFmt{A} 
 \stFmt{\oplus}
 \stLabFmt{greet(\stFmtC{str})}
   \stSeq
   \roleFmt{B} 
 \stFmt{\&}
 \stLabFmt{farewell(\stFmtC{bool})}
 $ is undefined as 
 $\stEnd
  \stBinMerge 
  \roleFmt{B} 
 \stFmt{\&}
 \stLabFmt{farewell(\stFmtC{bool})}
   $ is not mergeable. 
   
   Consider a global type:
  
  \smallskip
  \centerline{\(
\gtG = \gtCommRaw{\roleFmt{A}}{\roleFmt{B}}{
    \begin{array}{@{}l@{}}
    \gtCommChoice{\gtMsgFmt{greet}}{\stFmtC{str}}{
    \gtCommSingle{\roleFmt{C}}{\roleFmt{A}}
    {\gtMsgFmt{farewell}}{\stFmtC{bool}}{
    }}
   \\
     \gtCommChoice{\gtMsgFmt{farewell}}{\stFmtC{bool}}{
    \gtCommSingle{\roleFmt{C}}{\roleFmt{A}}{\gtMsgFmt{greet}}{\stFmtC{str}}{
    }
    }
    \end{array}
    }
   \)}

\smallskip
\noindent
$\gtG$ cannot be projected onto role $\roleFmt{C}$ 
since:

\smallskip
\centerline{\(
\begin{array}{lll}
\gtProj{\gtG}{\roleFmt{C}} 
&= &
\gtProj{\gtCommSingle{\roleFmt{C}}{\roleFmt{A}}{\gtMsgFmt{farewell}}{\stFmtC{bool}}{
  }}{\roleFmt{C}} 
   \,\stBinMerge\, 
  \gtProj{\gtCommSingle{\roleFmt{C}}{\roleFmt{A}}{\gtMsgFmt{greet}}{\stFmtC{str}}{
  }}{\roleFmt{C}} 
  \\
  &= &
  \roleFmt{A} 
 \stFmt{\oplus}
  \stLabFmt{farewell(\stFmtC{bool})} 
  \stBinMerge 
  \roleFmt{A} 
 \stFmt{\oplus}
  \stLabFmt{greet(\stFmtC{str})} 
   \text{\;\;(\textbf{undefined})}
   \end{array}
   \)}

\smallskip
\noindent
$\gtProj{\gtG}{\roleFmt{C}}$ is undefined because in $\gtG$, depending on 
whether $\roleFmt{A}$ and $\roleFmt{B}$ transmit $\gtMsgFmt{greet}$ or 
$\gtMsgFmt{farewell}$, $\roleFmt{C}$ is expected to send either 
$\gtMsgFmt{farewell}$ or 
$\gtMsgFmt{greet}$ to $\roleFmt{A}$. 
However, since $\roleFmt{C}$ is not involved in the interactions between $\roleFmt{A}$ and 
$\roleFmt{B}$, $\roleFmt{C}$ is not aware of which message to send; 
that is $\gtG$ does not provide a valid specification for $\roleFmt{C}$. 
 \end{example}

\vspace{1.5mm}
\noindent
{\bf Subtyping } %
We introduce a \emph{subtyping} relation $\stSub$ on local types, as defined in~\Cref{def:subtyping}. 
Intuitively, subtyping captures safe substitutability: when a local type $\stT$ is a subtype of $\stTi$, 
the behaviour specified  by $\stT$ and $\stTi$ is compatible; consequently, 
a process typed by 
$\stT$  can be safely used wherever one typed by $\stTi$ is expected, without introducing communication errors. %

\begin{definition}[Subtyping]
\label{def:subtyping}
  The \emph{session subtyping relation $\stSub$} is coinductively defined: 

\smallskip
\centerline{\(
\begin{array}{c}
\cinference[\iruleStSubGround]{
}{
  \tyGround \stSub \tyGround
}
\qquad 
\cinference[\iruleStSubIn]{
  \forall i \in I
  &
  \tyS[i] \stSub \tySi[i]
  &
  \stT[i] \stSub \tyTi[i]
}{
  \stExtSum{\roleP}{i \in I\cup{}J}{\stChoice{\stLab[i]}{\tyS[i]} \stSeq \stT[i]}%
  \stSub
  \stExtSum{\roleP}{i \in I}{\stChoice{\stLab[i]}{\tySi[i]} \stSeq \stTi[i]}%
}
\\[2ex]
\cinference[\iruleStSubEnd]{}{
  \stEnd \stSub \stEnd
}
\qquad
\cinference[\iruleStSubOut]{
  \forall i \in I
  &
  \tySi[i] \stSub \tyS[i]
  &
  \stT[i] \stSub \stTi[i]
}{
  \stIntSum{\roleP}{i \in I}{\stChoice{\stLab[i]}{\tyS[i]} \stSeq \stT[i]}
  \stSub
  \stIntSum{\roleP}{i \in I \cup J}{\stChoice{\stLab[i]}{\tySi[i]} \stSeq \stTi[i]}
}
\\[2ex]
\cinference[\iruleStSubRecL]{
  \stT{}[\stRec{\stRecVar}{\stT}/\stRecVar] \stSub \stTi
}{
  \stRec{\stRecVar}{\stT} \stSub \stTi
}
\qquad
\cinference[\iruleStSubRecR]{
  \stT \stSub \stTi{}[\stRec{\stRecVar}{\stTi}/\stRecVar]
}{
  \stT \stSub \stRec{\stRecVar}{\stTi}
}
\end{array}
\)}
\end{definition}

Rule \inferrule{\iruleStSubGround} states that any basic type is its own subtype. 
Rule \inferrule{\iruleStSubOut} allows the subtype of an internal choice to encompass a narrower range of message labels while permitting the sending of more generic payloads. 
Conversely, rule \inferrule{\iruleStSubIn} dictates 
that the subtype of an external choice must support a broader set of input message labels 
and less generic payloads. 
Rule \inferrule{\iruleStSubEnd} specifies that the type $\stEnd$ is its own subtype. 
Finally, rules \inferrule{\iruleStSubRecL} and \inferrule{\iruleStSubRecR} 
define the relationship between recursive types up to their unfolding.

\begin{example}[Subtyping]
\label{ex:subtyping}
  Recall the local types $\stT[\roleFmt{s}]$,  $\stT[\roleFmt{c}]$, and $\stT[\roleFmt{a}]$ from~\Cref{ex:types_oath}: 
  
  \smallskip
  \centerline{\(
  \begin{array}{c}
    \stT[\roleFmt{s}] =%
      \stIntSum{\roleClient}{}{%
        \begin{array}{@{\hskip 0mm}l@{\hskip 0mm}}%
          \stChoice{\stLabFmt{login}}{} \stSeq%
          \roleFmt{a}
          \stFmt{\&}
          \stLabFmt{auth(\stFmtC{bool})}%
          \\%
          \stChoice{\stLabFmt{cancel}}{}%
        \end{array}
      }%
      \quad 
       \stT[\roleFmt{c}]  =%
      \stExtSum{\roleServer}{}{%
        \begin{array}{@{\hskip 0mm}l@{\hskip 0mm}}%
          \stChoice{\stLabFmt{login}}{} \stSeq%
          \stOut{\roleFmt{a}}{%
            \stChoice{\stLabFmt{passwd}}{\stFmtC{str}}%
          }{}%
          \\%
          \stChoice{\stLabFmt{cancel}}{} \stSeq%
          \stOut{\roleFmt{a}}{%
            \stChoice{\stLabFmt{quit}}{}%
          }{}%
        \end{array}
      }%
      \\[2mm] 
      \stT[\roleFmt{a}] =%
      \stExtSum{\roleFmt{c}}{}{%
        \begin{array}{@{\hskip 0mm}l@{\hskip 0mm}}%
          \stChoice{\stLabFmt{passwd}}{\stFmtC{str}} \stSeq%
          \stOut{\roleServer}{%
            \stChoice{\stLabFmt{auth}}{\stFmtC{bool}}%
          }{}%
          \\%
          \stChoice{\stLabFmt{quit}}{}%
        \end{array}
      }%
  \end{array}
  \)}

\smallskip
\noindent
  Additionally, consider the following local types: 

\centerline{\(
\begin{array}{c}
    \stTi[\roleFmt{s}] = 
  \stEnvMap{\mpChanRole{\mpS}{\roleFmt{s}}}{\roleFmt{c} \stFmt{\oplus} \stLabFmt{cancel} }  
  \qquad
   \stTi[\roleFmt{c}] = \stEnvMap{\mpChanRole{\mpS}{\roleFmt{c}}}{\stExtSum{\roleServer}{}{%
        \begin{array}{@{\hskip 0mm}l@{\hskip 0mm}}%
          \stChoice{\stLabFmt{login}}{} \stSeq%
          \stOut{\roleFmt{a}}{%
            \stChoice{\stLabFmt{passwd}}{\stFmtC{str}}
          }{}%
          \\%
          \stChoice{\stLabFmt{cancel}}{} \stSeq%
          \stOut{\roleFmt{a}}{%
            \stChoice{\stLabFmt{quit}}{}
          }{}
          \\
             \stChoice{\stLabFmt{fail}}{} \stSeq%
          \stOut{\roleFmt{a}}{%
            \stChoice{\stLabFmt{fatal}}{}}{}
        \end{array}
      }}
   \\
    \stTi[\roleFmt{a}] =  \stEnvMap{\mpChanRole{\mpS}{\roleFmt{a}}}{\stExtSum{\roleFmt{c}}{}{%
		  \begin{array}{l}
          \stChoice{\stLabFmt{passwd}}{\stFmtC{str}} \stSeq%
          \stOut{\roleServer}{%
            \stChoice{\stLabFmt{auth}}{\stFmtC{bool}}%
          }{}\\
          \stChoice{\stLabFmt{quit}}{} \\
           \stChoice{\stLabFmt{fatal}}{}     
   \end{array}
      }}
      \end{array}
\)}
  
  \smallskip
  \noindent
  It holds that $\stTi[\roleFmt{s}] \stSub \stT[\roleFmt{s}]$, because a subtype is allowed to offer fewer options in an internal choice, while $\stTi[\roleFmt{c}] \stSub \stT[\roleFmt{c}]$ and  $\stTi[\roleFmt{a}] \stSub \stT[\roleFmt{a}]$, as a subtype in an  external choice may accept more options. 
  \end{example}

\subsection{Semantics of Global Types}
\label{sec:gtype:lts-gt}
We now present a Labelled Transition System (LTS) semantics for global types. 
First, we introduce the transition labels in~\Cref{def:mpst-env-reduction-label},
which are also used in the LTS semantics of typing contexts, discussed later in~\Cref{sec:gtype:lts-context}.

\begin{definition}[Transition Labels]
  \label{def:mpst-env-reduction-label}%
  \label{def:mpst-label-subject}%
  Let $\stEnvAnnotGenericSym$ %
  be a transition label of the form:

\medskip
 \centerline{\(
  \begin{array}{rcll}
  \stEnvAnnotGenericSym &\bnfdef&
    \stEnvInAnnotSmall{\roleP}{\roleQ}{\stChoice{\stLab}{\stS}}&
    \text{(in session $\mpS$, $\roleP$ receives $\stChoice{\stLab}{\stS}$ from $\roleQ$)}
    \\
    &\bnfsep&\stEnvOutAnnotSmall{\roleP}{\roleQ}{\stChoice{\stLab}{\stS}}&
    \text{(in session $\mpS$, $\roleP$ sends $\stChoice{\stLab}{\stS}$ to $\roleQ$)}
    \\
    &\bnfsep&\stEnvCommAnnotSmall{\roleP}{\roleQ}{\stLab}&
   \text{(in session $\mpS$, message $\stLab$ is transmitted from $\roleP$ to $\roleQ$)} 
  \end{array}
  \)}

\medskip
\noindent
The subject(s) of a transition label, written
\,$\ltsSubject{\stEnvAnnotGenericSym}$,\, are defined as follows:

\medskip
\centerline{\(
  \begin{array}{rcl}
    \ltsSubject{\stEnvInAnnotSmall{\roleP}{\roleQ}{\stChoice{\stLab}{\stS}}}
    &=&
    \ltsSubject{\stEnvOutAnnotSmall{\roleP}{\roleQ}{\stChoice{\stLab}{\stS}}}
    \;=\;
    \setenum{\roleP} \\
    \ltsSubject{\stEnvCommAnnotSmall{\roleP}{\roleQ}{\stLab}}
    &=&
    \setenum{\roleP, \roleQ}
  \end{array}
\)}%
\end{definition}

The label $\stEnvCommAnnotSmall{\roleP}{\roleQ}{\stLab}$ denotes a
synchronising communication between $\roleP$ and $\roleQ$ via a message label 
$\stLab$; %
the subject of this label includes \emph{both} roles.
The labels
$\stEnvOutAnnotSmall{\roleP}{\roleQ}{\stChoice{\stLab}{\stS}}$ 
and
$\stEnvInAnnotSmall{\roleP}{\roleQ}{\stChoice{\stLab}{\stS}}$
describe sending and receiving actions, respectively, and are specifically used in the 
typing context semantics in~\Cref{sec:gtype:lts-context}.

 We proceed to give a Labelled Transition System (LTS) semantics to a global type $\gtG$ in~\Cref{def:gtype:lts-gt}.

\begin{figure}[t]
 \centerline{\(
\begin{array}{c}

  \inference[\iruleGtMoveRec]{
    \gtG{}[\gtRec{\gtRecVar}{\gtG}/\gtRecVar] 
    \,\gtMove[\stEnvAnnotGenericSym]\, 
    \gtGi}
  {
    \gtRec{\gtRecVar}{\gtG} 
    \,\gtMove[\stEnvAnnotGenericSym]\, 
    \gtGi
  }
  \qquad 
  \inference[\iruleGtMoveComm]{
    j \in I
  }{
    \gtCommSmall{\roleP}{\roleQ}{i \in I}{\gtLab[i]}{\stS[i]}{\gtGi[i]}
    \,\gtMove[\ltsSendRecv{\mpS}{\roleP}{\roleQ}{\gtLab[j]}]\, 
    \gtGi[j]
  }
  \\[2ex]

  \inference[\iruleGtMoveCtx]{
    \forall i \in I: 
    \gtGi[i] 
    \,\gtMove[\stEnvAnnotGenericSym]\, 
    \gtGii[i] 
    &
    \ltsSubject{\stEnvAnnotGenericSym} \cap \setenum{\roleP, \roleQ} =
    \emptyset
  }{
      \gtCommSmall{\roleP}{\roleQ}{i \in I}{\gtLab[i]}{\stS[i]}{\gtGi[i]}
   \,\gtMove[\stEnvAnnotGenericSym]\, 
      \gtCommSmall{\roleP}{\roleQ}{i \in I}{\gtLab[i]}{\stS[i]}{\gtGii[i]}   
  }
\end{array}
\)}
\caption{Global type transition rules.}
\label{fig:gtype:red-rules}
\end{figure}

\begin{definition}[Global Type Semantics]
  \label{def:gtype:lts-gt}
  The global type transition %
  $\gtMove[\stEnvAnnotGenericSym]$ is inductively
  defined by the rules in~\cref{fig:gtype:red-rules}. 
  We use 
  $\gtG \,\gtMove\, 
  \gtGi$
  if there
  exists $\stEnvAnnotGenericSym$ such that
  $\gtG
  \,\gtMove[\stEnvAnnotGenericSym]\, 
  \gtGi$; 
  we write
  $\gtG \,\gtMove$
  if there
  exists $\gtGi$ such that
  $\gtG
  \,\gtMove\,  
  \gtGi$, 
  and $\gtNotMove{\gtG}$ for its negation (\ie there is no $\gtGi$ 
  such that  $\gtG
  \,\gtMove\,  
  \gtGi$). 
  Finally, 
 $\gtMoveStar$ denotes the transitive and reflexive closure of
  $\gtMove$.
\end{definition}

\Cref{fig:gtype:red-rules} depicts  the standard global type transition %
rules~\citep{DBLP:conf/icdcit/YoshidaG20}. 
Rule $\inferrule{\iruleGtMoveComm}$ models communication between two roles, 
while rule $\inferrule{\iruleGtMoveRec}$ addresses recursion.
Finally, rule $\inferrule{\iruleGtMoveCtx}$ permits the transition %
of a global type that is 
causally independent of its prefix, provided that all continuations can 
perform the transition %
with that label.  
This causal independence is indicated by the subjects of the label 
being disjoint from the prefix of the global type. For example, consider the global type 
$\gtG = \gtCommSingle{\roleP}{\roleQ}
    {\gtLab[1]}{}{
    \gtCommSingle{\roleR}{\roleFmt{u}}
    {\gtLab[2]}{}{\gtEnd
   }}$.    
   Given that the subjects of the label $\stEnvCommAnnotSmall{\roleR}{\roleFmt{u}}{\gtLab[2]}$ 
   are disjoint from $\roleP$ and $\roleQ$, \ie 
    $\ltsSubject{\stEnvCommAnnotSmall{\roleR}{\roleFmt{u}}{\gtLab[2]}} \cap \{\roleP, \roleQ\} = \emptyset$, 
    and $\gtCommSingle{\roleR}{\roleFmt{u}}{\gtLab[2]}{}{\gtEnd} \,\gtMove[\stEnvCommAnnotSmall{\roleR}{\roleFmt{u}}{\gtLab[2]}]\, \gtEnd$, rule $\inferrule{\iruleGtMoveCtx}$ allows $\gtG$ to perform the transition 
    $\gtG \,\gtMove[\stEnvCommAnnotSmall{\roleR}{\roleFmt{u}}{\gtLab[2]}]\, \gtCommSingle{\roleP}{\roleQ}{\gtLab[1]}{}{\gtEnd}$. %

\subsection{Semantics of Typing Context}
\label{sec:gtype:lts-context}

Following the introduction of global type semantics, 
we present an LTS semantics for \emph{typing contexts}, which are collections of local types. 
The formal definition of a typing context is provided in \Cref{def:mpst-env}, 
with its transition %
rules in \Cref{def:mpst-env-reduction}.

\begin{definition}[Typing Contexts]%
  \label{def:mpst-env}%
  \label{def:mpst-env-closed}%
  \label{def:mpst-env-comp}%
  \label{def:mpst-env-subtype}%
  \label{def:typing_context}
  $\mpEnvNew$ denotes a partial mapping
 from expression variables to basic types,  and from process variables to %
  sequences of  basic types and session types,
  while $\stEnvNew$ denotes a partial mapping from channels to session types. 
  Their syntax is defined as:

  \smallskip
  \centerline{\(%
  \mpEnvNew
  \;\;\coloncolonequals\;\;
  \mpEnvEmpty
  \bnfsep \mpEnvNew \mpEnvComp\, \mpEnvMap{\mpFmt{x}}{\tyGround} 
  \bnfsep
  \mpEnvNew \mpEnvComp\, \mpEnvMap{\mpX}{\tyGround[1], \ldots, \tyGround[n], \stT[1],\ldots,\stT[m]} 
   \quad\quad
  \stEnvNew
  \,\coloncolonequals\,
  \stEnvEmpty
  \bnfsep
   \stEnvNew \stEnvComp \stEnvMap{\mpC}{\stT}
  \)}%
  
  \smallskip
\noindent
  The \emph{context composition} $\stEnvNew[1] \stEnvComp \stEnvNew[2]$ 
  is defined iff $\dom{\stEnvNew[1]} \cap \dom{\stEnvNew[2]} = \emptyset$. %
  We write $\bigcup_{i \in I}\stEnv[i]$ for the composition of contexts $\stEnv[i]$. 
  We write  $\mpS \!\in\! \stEnvNew$ 
  iff 
  $\exists \roleP: \mpChanRole{\mpS}{\roleP} \!\in\! \dom{\stEnvNew}$
  (\ie session $\mpS$ occurs in $\stEnvNew$). %
  Conversely, we write 
  $\mpS \!\not\in\! \stEnvNew$ 
  iff 
  $\forall \roleP: \mpChanRole{\mpS}{\roleP} \!\not\in\! \dom{\stEnvNew}$
  (\ie session $\mpS$ does not occur in~$\stEnvNew$). %
  We write $\dom{\stEnvNew} = \setenum{\mpS}$ iff $\forall \mpC \!\in\! \dom{\stEnvNew}$ there exists 
  $\roleP$ such that $\mpC = \mpChanRole{\mpS}{\roleP}$~(\ie $\stEnvNew$ only contains session~$\mpS$). 
  We write
  $\stEnvNew \!\stSub\! \stEnvNewi$
  iff
  $\dom{\stEnvNew} \!=\! \dom{\stEnvNewi}$
  and
  $\forall \mpC \!\in\! \dom{\stEnvNew}{:}\,
  \stEnvApp{\stEnvNew}{\mpC} \!\stSub\! \stEnvApp{\stEnvNewi}{\mpC}$. 
  We write $\stEnvNew[\mpS]$ iff $\dom{\stEnvNew[\mpS]} = \setenum{\mpS}$, $\dom{\stEnvNew[\mpS]} \!\subseteq\! \dom{\stEnvNew}$, 
  and $\forall \mpChanRole{\mpS}{\roleP} \!\in\! \dom{\stEnvNew}: \stEnvApp{\stEnvNew}{\mpChanRole{\mpS}{\roleP}} = \stEnvApp{\stEnvNew[\mpS]}{\mpChanRole{\mpS}{\roleP}}$~(\ie restriction of  $\stEnvNew$ to session 
  $\mpS$). 
\end{definition}

\begin{figure}[t]
\centerline{\(%
  \begin{array}{c}
    \inference[\iruleTCtxOut]{%
      k \in I%
    }{%
      \stEnvMap{%
        \mpChanRole{\mpS}{\roleP}%
      }{%
        \stIntSum{\roleQ}{i \in I}{\stChoice{\stLab[i]}{\stS[i]} \stSeq \stT[i]}%
      }%
      \,\stEnvMoveOutAnnot{\roleP}{\roleQ}{\stChoice{\stLab[k]}{\stS[k]}}\,%
      \stEnvMap{%
        \mpChanRole{\mpS}{\roleP}%
      }{\stT[k]}%
    }%
  \quad 
    \inference[\iruleTCtxIn]{%
      k \in I%
    }{%
      \stEnvMap{%
        \mpChanRole{\mpS}{\roleP}%
      }{%
        \stExtSum{\roleQ}{i \in I}{\stChoice{\stLab[i]}{\stS[i]} \stSeq \stT[i]}%
      }%
      \,\stEnvMoveInAnnot{\roleP}{\roleQ}{\stChoice{\stLab[k]}{\stS[k]}}\,%
      \stEnvMap{%
        \mpChanRole{\mpS}{\roleP}%
      }{\stT[k]}%
    }%
    \\[2mm]
    \inference[\iruleTCtxCom]{%
      \stEnvNew[1]%
      \stEnvMoveOutAnnot{\roleP}{\roleQ}{\stChoice{\stLab}{\stS}}%
      \stEnvNewi[1]%
      &%
      \stEnvNew[2]%
      \stEnvMoveInAnnot{\roleQ}{\roleP}{\stChoice{\stLab}{\stSi}}%
      \stEnvNewi[2]%
      &%
      \stSi \!\stSub\! \stS%
    }{%
      \stEnvNew[1] \stEnvComp \stEnvNew[2]%
      \,\stEnvMoveCommAnnot{\mpS}{\roleP}{\roleQ}{\stLab}\,%
      \stEnvNewi[1] \stEnvComp \stEnvNewi[2]%
    }%
    \quad 
    \inference[\iruleTCtxRec]{%
      \stEnvMap{%
        \mpChanRole{\mpS}{\roleP}%
      }{%
        \stT\subst{\stRecVar}{\stRec{\stRecVar}{\stT}}%
      }%
      \stEnvMoveGenAnnot \stEnvNewi%
    }{%
      \stEnvMap{%
        \mpChanRole{\mpS}{\roleP}%
      }{%
        \stRec{\stRecVar}{\stT}%
      }%
      \stEnvMoveGenAnnot \stEnvNewi%
    }%
    \quad 
    \inference[\iruleTCtxCong]{%
      \stEnvNew \stEnvMoveGenAnnot \stEnvNewi%
    }{%
      \stEnvNew \stEnvComp \stEnvMap{\mpC}{\stT}
      \stEnvMoveGenAnnot%
      \stEnvNewi \!\stEnvComp \stEnvMap{\mpC}{\stT}
    }%
     \end{array}
  \)}%
    \caption{Typing context transition rules.}
  \label{fig:gtype:tc-red-rules}
\end{figure}

\begin{definition}[Typing Context Semantics]
  \label{def:mpst-env-reduction}%
  \label{def:typing_context_reduction}
  The \emph{typing context transition\; $\stEnvMoveGenAnnot$} %
  \;is inductively defined by the rules in~\cref{fig:gtype:tc-red-rules}. %
We write $\stEnvNew \!\stEnvMoveGenAnnot$ if there exists $\stEnvNewi$ 
such that $\stEnvNew \!\stEnvMoveGenAnnot\! \stEnvNewi$. %
 We define two \emph{reductions} $\stEnvNew \!\stEnvMoveWithSession[\mpS]\! \stEnvNewi$ and 
 $\stEnvNew \!\stEnvMove\! \stEnvNewi$, as follows:
\begin{itemize}
 \item $\stEnvNew \!\stEnvMoveWithSession[\mpS]\! \stEnvNewi$ holds iff 
    $\stEnvNew \stEnvMoveGenAnnot \stEnvNewi$
    with $\stEnvAnnotGenericSym = \ltsSendRecv{\mpS}{\roleP}{\roleQ}{\stLab}$ 
    for any $\roleP, \roleQ \in \roleSet$~(recall that $\roleSet$ is the set of all roles): 
    this means that $\stEnvNew$ can progress via message transmission on session $\mpS$, 
    involving any roles $\roleP$ and $\roleQ$.  
    We write $\stEnvNew \!\stEnvMoveWithSession[\mpS]$ %
    iff $\stEnvNew \!\stEnvMoveWithSession[\mpS]\! \stEnvNewi$ for some $\stEnvNewi$, %
    and $\stEnvNotMoveWithSessionP[\mpS]{\stEnvNew}$ for its negation 
    (\ie  there is no $\stEnvNewi$ such that %
    $\stEnvNew \!\stEnvMoveWithSession[\mpS]\!  \stEnvNewi$),   %
    and we denote $\stEnvMoveWithSessionStar[\mpS]$ as the reflexive and transitive closure of 
    $\stEnvMoveWithSession[\mpS]$; %
    \item $\stEnvNew \!\stEnvMove\! \stEnvNewi$ holds iff %
    $\stEnvNew \stEnvMoveWithSession[\mpS] \stEnvNewi$ for some $\mpS$: 
    this means that $\stEnvNew$ can progress via message transmission on any session.  
    We write $\stEnvMoveP{\stEnvNew}$ iff $\stEnvNew \!\stEnvMove\! \stEnvNewi$ for some $\stEnvNewi$, %
    and $\stEnvNotMoveP{\stEnvNew}$ for its negation,  %
  and we denote $\stEnvMoveStar$ %
    as the reflexive and transitive closure of $\stEnvMove$.  
  \end{itemize}
\end{definition}

\Cref{fig:gtype:tc-red-rules} presents rules for 
typing context transitions, %
which are similar to those in~\citep{POPL19LessIsMore}. 
Rule \inferrule{\iruleTCtxOut} (resp. \inferrule{\iruleTCtxIn}) 
states that a typing context entry can execute an output (resp. input) transition. 
Rule \inferrule{\iruleTCtxCom} 
ensures synchronised matching of input and output transitions, with payload compatibility through subtyping; 
consequently, the context progresses via a message transmission label 
$\stEnvCommAnnotSmall{\roleP}{\roleQ}{\stLab}$. Rule \inferrule{\iruleTCtxRec} handles recursion, while 
rule \inferrule{\iruleTCtxCong} lifts transitions %
to larger contexts.

\subsection{Relating Semantics between Global Types and Typing Contexts}
\label{sec:gtype:relating}

With the 
LTS semantics for global types~(\Cref{def:gtype:lts-gt}) 
and typing contexts~(\Cref{def:mpst-env-reduction}) defined, 
we establish a relationship between these semantics 
using the projection operator $\gtProj[]{}{}$~(\Cref{def:global-proj}) 
and the subtyping relation~$\stSub$~(\Cref{def:subtyping}).

\begin{definition}[Association of Global Types and Typing Contexts]
\label{def:assoc}
\label{def:association}
  A typing context $\stEnvNew$ is associated with %
  a global type $\gtG$ for %
  a multiparty session $\mpS$, 
  written %
  $\stEnvAssoc{\gtG}{\stEnvNew}{\mpS}$, 
  iff  $\stEnvNew$ 
  can be split into two disjoint (possibly empty) sub-contexts 
  $\stEnvNew = \stEnvNew[\,\gtProj{}{\stSub}] \stEnvComp \stEnvNew[\stEnd]$ where: 
  \begin{enumerate}
  \item $\stEnvNew_{\,\gtProj{}{\stSub}}$ contains subtypes of projections of $\gtG$: 
  $\dom{\stEnvNew[\,\gtProj{}{\stSub}]}
        =
        \setcomp{\mpChanRole{\mpS}{\roleP}}{\roleP \in \gtRoles{\gtG}}$, and 
        $ \forall \roleP \in \gtRoles{\gtG}: \stEnvApp{\stEnvNew[\,\gtProj{}{\stSub}]}{\mpChanRole{\mpS}{\roleP}} \stSub \gtProj{\gtG}{\roleP}$; 
 \item $\stEnvNew[\stEnd]$ contains only terminated endpoints: either $\stEnvNew[\stEnd] = \emptyset$, or $\dom{\stEnv[\stEnd]} = \mpS $ and $\forall \mpChanRole{\mpS}{\roleQ}  \in \dom{\stEnvNew[\stEnd]}: \stEnvApp{\stEnvNew[\stEnd]}{\mpChanRole{\mpS}{\roleQ}} = \stEnd$. 
        \end{enumerate}
\end{definition}

The association $\stEnvAssoc{\gtFmt{\cdot}}{\stFmt{\cdot}}{\mpFmt{\cdot}}$ is a
binary relation over typing contexts $\stEnvNew$ and global types $\gtG$,
parameterised by  multiparty sessions $\mpS$. 
 The association relation requires that,  
 for each role $\roleP$ in the global type, its corresponding entry in the
typing context~($\stEnvApp{\stEnvNew}{\mpChanRole{\mpS}{\roleP}}$) must be a subtype~(\Cref{def:subtyping}) 
of the projection of the global type
onto that role~($\gtProj{\gtG}{\roleP}$). 

\begin{remark}
It is evident that a typing context
$\stEnvNew =
  \setenum{\stEnvMap{\mpChanRole{\mpS}{\roleP}}{\gtProj{\gtG}{\roleP}}}_{%
    \roleP \in \gtRoles{\gtG}}$, which contains the projections of all roles in $\gtG$, 
    is associated with the global type $\gtG$ for $\mpS$. 
    \hfill $\blacktriangleleft$
\end{remark}
    
\begin{example}[Association]
\label{ex:assoc}
Recall the global type $\gtG[\text{auth}]$ and its projected local types 
$\stT[\roleFmt{s}], \stT[\roleFmt{c}], \stT[\roleFmt{a}]$ from \Cref{ex:types_oath}, as well as the local types 
$\stTi[\roleFmt{s}], \stTi[\roleFmt{c}], \stTi[\roleFmt{a}]$ from~\Cref{ex:subtyping}. 
Consider the typing context:

\smallskip
\centerline{\(
\stEnvNew[\text{auth}] = \stEnvNew[\text{auth}_{\roleFmt{s}}] \stEnvComp \stEnvNew[\text{auth}_{\roleFmt{c}}] 
\stEnvComp
\stEnvNew[\text{auth}_{\roleFmt{a}}]
\)}

\smallskip
\noindent
where $\stEnvNew[\text{auth}_{\roleFmt{s}}]  = \stEnvMap{\mpChanRole{\mpS}{\roleFmt{s}}}{\stTi[\roleFmt{s}]}$, $\stEnvNew[\text{auth}_{\roleFmt{c}}]  = \stEnvMap{\mpChanRole{\mpS}{\roleFmt{s}}}
{\stTi[\roleFmt{c}]}$, and 
$\stEnvNew[\text{auth}_{\roleFmt{a}}]  = \stEnvMap{\mpChanRole{\mpS}{\roleFmt{s}}}
{\stTi[\roleFmt{a}]}$.

Intuitively, $\stEnvNew[\text{auth}]$ is associated with $\gtG[\text{auth}]$, as $\roleFmt{s}$ sends only 
$\stLabFmt{cancel}$, while $\roleFmt{c}$ and $\roleFmt{a}$ expect 
to receive additional messages $\stLabFmt{fail}$ and $\stLabFmt{fatal}$, respectively. 
Indeed,  the entries in $\stEnvNew[\text{auth}]$ adhere to the communication behaviour patterns of each role  
of $\gtG[\text{auth}]$, 
though with fewer output messages and more 
input ones.

We can formally verify the association of $\stEnvNew[\text{auth}]$ with $\gtG[\text{auth}]$ for session $\mpS$ by: 
\begin{itemize}
\item %
$\gtRoles{\gtG[\text{auth}]} = \setenum{\roleFmt{s}, 
\roleFmt{c}, \roleFmt{a}}$ and 
$\dom{\stEnvNew[\text{auth}]} = \setenum{\mpChanRole{\mpS}{\roleFmt{s}}, \mpChanRole{\mpS}{\roleFmt{c}}, \mpChanRole{\mpS}{\roleFmt{a}}}$ %
\item $\stEnvApp{\stEnvNew[\text{auth}]}{\mpChanRole{\mpS}{\roleFmt{s}}} =  \stTi[\roleFmt{s}] 
\stSub 
\gtProj{\gtG[\text{auth}]}{\roleFmt{s}} =  \stT[\roleFmt{s}] $

\item $\stEnvApp{\stEnvNew[\text{auth}]}{\mpChanRole{\mpS}{\roleFmt{c}}} =  \stTi[\roleFmt{c}] 
\stSub 
\gtProj{\gtG[\text{auth}]}{\roleFmt{c}} =  \stT[\roleFmt{c}] $

\item $\stEnvApp{\stEnvNew[\text{auth}]}{\mpChanRole{\mpS}{\roleFmt{a}}} =  \stTi[\roleFmt{a}] 
\stSub 
\gtProj{\gtG[\text{auth}]}{\roleFmt{a}} =  \stT[\roleFmt{a}] $

\end{itemize}

\end{example}

We demonstrate the \emph{operational correspondence} between a global type 
and any \emph{associated} typing context through 
two main theorems: 
\Cref{thm:gtype:proj-sound} shows that the transition behaviour of a global type corresponds to that of 
its associated typing context, while  
\Cref{thm:gtype:proj-comp} illustrates that each possible transition %
of a typing context is reflected by 
an action in the transitions %
of the associated global type.

\begin{restatable}[Soundness of Association]{theorem}{thmProjSoundness}%
  \label{thm:gtype:proj-sound}
  \label{thm:soundness_association}
  Given associated global type $\gtG$ and typing context $\stEnvNew$ for session~$\mpS$: 
  $\stEnvAssoc{\gtG}{\stEnvNew}{\mpS}$. 
  If $\gtG \,\gtMove[\stEnvAnnotGenericSym]\, \gtGi$ where 
 $\stEnvAnnotGenericSym =   \stEnvCommAnnotSmall{\roleP}{\roleQ}{\stLab}$, 
  then there exist  $\stLabi$, $\stEnvAnnotGenericSymi$, $\stEnvNewi$, 
  and $\gtGii$, such that 
  $\stEnvAnnotGenericSymi = \stEnvCommAnnotSmall{\roleP}{\roleQ}{\stLabi}$, 
  $\gtG \,\gtMove[\stEnvAnnotGenericSymi]\, \gtGii$, 
  $\stEnvAssoc{\gtGii}{\stEnvNewi}{\mpS}$, and $\stEnvNew \stEnvMoveAnnot{\stEnvAnnotGenericSymi} \stEnvNewi$. 
\end{restatable}
\begin{proof}
  By induction on global type transitions (\Cref{def:gtype:lts-gt}).
\end{proof}

\begin{restatable}[Completeness of Association]{theorem}{thmProjCompleteness}%
  \label{thm:gtype:proj-comp}
  \label{thm:completeness_association}
  Given 
  associated global type $\gtG$ and typing context $\stEnvNew$ for session~$\mpS$:  
  $\stEnvAssoc{\gtG}{\stEnvNew}{\mpS}$. 
  If $\stEnvNew \stEnvMoveGenAnnot \stEnvNewi$ where 
  $\stEnvAnnotGenericSym = \stEnvCommAnnotSmall{\roleP}{\roleQ}{\stLab}$, %
  then there exists $\gtGi$ such that
  $\stEnvAssoc{\gtGi}{\stEnvNewi}{\mpS}$
  and
  $\gtG  \,\gtMove[\stEnvAnnotGenericSym]\, 
    \gtGi$.
\end{restatable}
\begin{proof}
  By induction on typing context transitions (\Cref{def:mpst-env-reduction}).
\end{proof}

\begin{example}[Soundness and Completeness of Association]
\label{ex:sound_comp}
Consider the associated global type $\gtG[\text{auth}]$
and typing context $\stEnvNew[\text{auth}]$ in~\Cref{ex:assoc}. 
We have the global type transition $\gtG[\text{auth}] \,\,\gtMove[\ltsSendRecv{\mpS}{\roleFmt{s}}{\roleFmt{c}}{\gtMsgFmt{login}}]$, and,  by 
the soundness of association~(\Cref{thm:gtype:proj-sound}), 
there exist  $\stEnvAnnotGenericSym = 
\ltsSendRecv{\mpS}{\roleFmt{s}}{\roleFmt{c}}{\gtMsgFmt{cancel}}$,  
$\gtGi[\text{auth}]$, and $\stEnvNewi[\text{auth}]$ 
such that 

\smallskip
\centerline{\( 
\gtG[\text{auth}] \,\,\gtMove[\ltsSendRecv{\mpS}{\roleFmt{s}}{\roleFmt{c}}{\gtMsgFmt{cancel}}]\,\,  \gtGi[\text{auth}]
= \gtCommSingle{\roleFmt{c}}{\roleFmt{a}}
    {\gtMsgFmt{quit}}{}{} \, \text{ and } \,
    \stEnvNew[\text{auth}] \stEnvMoveCommAnnot{\mpS}{\roleFmt{s}}{\roleFmt{c}}{\stLabFmt{cancel}} 
\stEnvNewi[\text{auth}] = \stEnvNewi[\text{auth}_\roleFmt{s}] \stEnvComp \stEnvNewi[\text{auth}_\roleFmt{c}]
 \stEnvComp \stEnvNewi[\text{auth}_\roleFmt{a}]
\)}

\smallskip 
\noindent
where: 

\smallskip
\centerline{\(
\begin{array}{lll}
\stEnvNewi[\text{auth}_\roleFmt{s}] &=& \stEnvMap{\mpChanRole{\mpS}{\roleFmt{s}}}
{\stEnd
}
\qquad \quad 
\stEnvNewi[\text{auth}_\roleFmt{c}] = \stEnvMap{\mpChanRole{\mpS}{\roleFmt{c}}}
{ \roleFmt{a}
 \stFmt{\oplus}
  \stLabFmt{quit} }
\\
\stEnvNewi[\text{auth}_\roleFmt{a}] &=& 
\stEnvMap{\mpChanRole{\mpS}{\roleFmt{a}}}{\stExtSum{\roleFmt{c}}{}{%
          \stChoice{\stLabFmt{passwd}}{\stFmtC{str}} \stSeq%
          \stOut{\roleServer}{%
            \stChoice{\stLabFmt{auth}}{\stFmtC{bool}}%
          }{}, 
          \stChoice{\stLabFmt{quit}}{}, 
           \stChoice{\stLabFmt{fatal}}{}     
      }}
      \end{array}
\)}

\smallskip 
\noindent
By~\Cref{def:assoc}, it is straightforward to check that  $\stEnvAssoc{\gtGi[\text{auth}]}{\stEnvNewi[\text{auth}]}{\mpS}$.
Further, $\stEnvNewi[\text{auth}]$ admits the transition: %

\smallskip
\centerline{\(
\stEnvNewi[\text{auth}] \stEnvMoveCommAnnot{\mpS}{\roleFmt{c}}{\roleFmt{a}}{\stLabFmt{quit}} 
\stEnvNewii[\text{auth}] = 
\stEnvMap{\mpChanRole{\mpS}{\roleFmt{s}}}{\stEnd}, 
\stEnvMap{\mpChanRole{\mpS}{\roleFmt{c}}}{\stEnd},  
\stEnvMap{\mpChanRole{\mpS}{\roleFmt{a}}}{\stEnd}
\)}

\smallskip
\noindent
which, by the completeness of association~(\Cref{thm:gtype:proj-comp}), relates to the transition: 
$\gtGi[\text{auth}] \,\,\gtMove[\ltsSendRecv{\mpS}{\roleFmt{c}}{\roleFmt{a}}{\gtMsgFmt{quit}}]\,\,  
\gtEnd$, with $\stEnvAssoc{\gtEnd}{\stEnvNewii[\text{auth}]}{\mpS}$. 

\end{example}

\begin{remark}[Sufficiency of Soundness Theorem]
  Inquisitive readers may question why we opted to  formulate a 
  soundness theorem that does not directly mirror the completeness theorem, 
  as in common existing literature~\citep{ICALP13CFSM}.
 This decision stems from our use of the  subtyping~(particularly 
  rule \inferrule{\iruleStSubOut}).   
  A local type in the typing context may offer fewer branches for selection 
  compared to the projected local type, 
  leading to transmission actions in the global type that remain uninhabited.

  For example, consider the global type $\gtG[\text{auth}]$ 
  and its associated typing context $\stEnvNew[\text{auth}]$ from~\Cref{ex:assoc}. 
 While the global type $\gtG[\text{auth}]$ can transition via   
$\ltsSendRecv{\mpS}{\roleFmt{s}}{\roleFmt{c}}{\gtMsgFmt{login}}$, 
no corresponding transition is available for the associated typing context $\stEnvNew[\text{auth}]$. 

Importantly, the message $\stLabi$ appearing in the soundness theorem is not arbitrary: 
it must correspond to a branch enabled by the global type and realisable by the associated local types, 
although it need not coincide with the particular label $\stLab$. Consequently, soundness of   
association requires the existence of a \emph{compatible} typing context transition, 
rather than preservation of the identical message transmission. 

 Despite this formulation, the soundness theorem is \emph{sufficient} to ensure the key properties
    guaranteed by association, including safety, deadlock-freedom, and liveness,
  as demonstrated  in~\cref{sec:gtype:pbp}.
  \hfill $\blacktriangleleft$
  \end{remark}

\subsection{Association vs Projection: Necessity of Association}
\label{sec:gtype:counterexample_projection}
As introduced in~\Cref{sec:intro}, the association defined in~\Cref{sec:gtype:relating} is applied as the invariant in the proof of subject reduction, where typing contexts used to type processes are incrementally constructed along transitions  while maintaining the association. A natural question, therefore, concerns the necessity of association: why cannot  projection be used directly in its place?

The difficulty lies in the fact that projection is not preserved under transitions.  More specifically,  if the conditions 
$\stEnvAssoc{\gtG}{\stEnv}{\mpS}$ and 
$\stEnvAssoc{\gtGi}{\stEnvi}{\mpS}$ in~\Cref{thm:gtype:proj-sound,thm:gtype:proj-comp} are replaced with 
$\stEnv = \setenum{\stEnvMap{\mpChanRole{\mpS}{\roleR}}{\gtProj{\gtG}{\roleR}}}_{\roleR \in \gtRoles{\gtG}}$ and $\stEnvi = \setenum{\stEnvMap{\mpChanRole{\mpS}{\roleR}}{\gtProj{\gtGi}{\roleR}}}_{\roleR \in \gtRoles{\gtG}}$, so that typing contexts are derived directly from projections, the operational correspondence no longer holds. The following example illustrates this issue.

Consider the global type $\gtG[\times]$: 
	\begin{align*}
		\gtCommRaw{\roleFmt{A}}{\roleFmt{B}}{
			\begin{array}{l}
				\gtCommChoice{add}{\stFmtC{int}}{\gtCommRaw{\roleFmt{B}}{\roleFmt{C}}{
					\begin{array}{l}
				\gtCommChoice{add}{\stFmtC{int}}{\gtEnd} \\
				\gtCommChoice{forward}{\stFmtC{int}}{\gtEnd}
			\end{array}
			}
			}\\
			\gtCommChoice{sub}{\stFmtC{int}}{\gtCommRaw{\roleFmt{B}}{\roleFmt{C}}{
					\begin{array}{l}
				\gtCommChoice{sub}{\stFmtC{int}}{\gtEnd}\\
				\gtCommChoice{forward}{\stFmtC{int}}{\gtEnd}
			\end{array}
		}}
		\end{array}
	} 
	\end{align*}
	and the typing context $\stEnv[\times]$: 
	\begin{align*}
		\begin{array}{c}
		\begin{array}{l}
		\stEnvMap{\mpChanRole{\mpS}{\roleFmt{A}}}{\stIntSum{\roleFmt{B}}{}{
			\begin{array}{l}
				\stChoice{\stLabFmt{add}}{\stFmtC{int}} \stSeq \stEnd \\
				\stChoice{\stLabFmt{sub}}{\stFmtC{int}} \stSeq \stEnd
			\end{array}
	}}
	\\[1.5em] 
	\stEnvMap{\mpChanRole{\mpS}{\roleFmt{C}}}{\stExtSum{\roleFmt{B}}{}{
			\begin{array}{l}
				\stChoice{\stLabFmt{add}}{\stFmtC{int}} \stSeq \stEnd \\
				\stChoice{\stLabFmt{sub}}{\stFmtC{int}} \stSeq \stEnd \\
				\stChoice{\stLabFmt{forward}}{\stFmtC{int}} \stSeq \stEnd
			\end{array}
	}}
	\end{array}
	\quad 
	\stEnvMap{\mpChanRole{\mpS}{\roleFmt{B}}}{\stExtSum{\roleFmt{A}}{}{
			\begin{array}{l}
				\stChoice{\stLabFmt{add}}{\stFmtC{int}} 
				\stSeq \stIntSum{\roleFmt{C}}{}{
			\begin{array}{l}
				\stChoice{\stLabFmt{add}}{\stFmtC{int}} \stSeq \stEnd \\
				\stChoice{\stLabFmt{forward}}{\stFmtC{int}} \stSeq \stEnd
			\end{array}}
				\\
				\stChoice{\stLabFmt{sub}}{\stEnd} 
				\stSeq \stIntSum{\roleFmt{C}}{}{
			\begin{array}{l}
				\stChoice{\stLabFmt{sub}}{\stFmtC{int}} \stSeq \stEnd \\
				\stChoice{\stLabFmt{forward}}{\stFmtC{int}} \stSeq \stEnd 
			\end{array}}
			\end{array}
	\!}}
	\end{array}
	\end{align*}
	
It is straightforward to verify that $\stEnv[\times]$ is a projection of $\gtG[\times]$ onto session $\mpS$. 
	There are two possible transitions between $\roleFmt{A}$ and $\roleFmt{B}$ in $\gtG[\times]$:  
	either $\gtMsgFmt{add}$ or $\gtMsgFmt{sub}$. 
Since both cases are symmetrical, we focus on the transition labeled $\gtMsgFmt{add}$, 
	i.e. $\stEnvAnnotGenericSym = \ltsSendRecv{\mpS}{\roleFmt{A}}{\roleFmt{B}}{\gtMsgFmt{add}}$.
	Consequently, $\gtG[\times]$ reduces to 
	\begin{align*}
		\gtGi[\times] = \gtCommRaw{\roleFmt{B}}{\roleFmt{C}}{
			\begin{array}{l}
				\gtCommChoice{add}{\stFmtC{int}}{\gtEnd} \\
				\gtCommChoice{forward}{\stFmtC{int}}{\gtEnd}
			\end{array}
			}
	\end{align*}
	The typing context obtained through projection from $\gtGi[\times]$ is: 
	\begin{align*}
		\stEnvi[\times] =~&\mpChanRole{\mpS}{\roleFmt{A}}: \stEnd \stEnvComp \;
		\mpChanRole{\mpS}{\roleFmt{B}}: \stIntSum{\roleFmt{C}}{}{
			\begin{array}{l}
				\stChoice{\stLabFmt{add}}{\stFmtC{int}} \stSeq \stEnd \\
				\stChoice{\stLabFmt{forward}}{\stFmtC{int}} \stSeq \stEnd
			\end{array}
		}\stEnvComp \;
		\mpChanRole{\mpS}{\roleFmt{C}}: \stExtSum{\roleFmt{B}}{}{
			\begin{array}{l}
				\stChoice{\stLabFmt{add}}{\stFmtC{int}} \stSeq \stEnd \\
				\stChoice{\stLabFmt{forward}}{\stFmtC{int}} \stSeq \stEnd
			\end{array}
		}
	\end{align*}
	However,  $\stEnv[\times] \stEnvMoveGenAnnot \stEnvi[\times]$ does not hold. 
	The only possible transition of $\stEnv[\times]$ with  
	$\stEnvAnnotGenericSym = \ltsSendRecv{\mpS}{\roleFmt{A}}{\roleFmt{B}}{\gtMsgFmt{add}}$ leads to:  
		\begin{align*}
		\stEnvii[\times] =~&\mpChanRole{\mpS}{\roleFmt{A}}: \stEnd\stEnvComp \;
		\mpChanRole{\mpS}{\roleFmt{B}}: \stIntSum{\roleFmt{C}}{}{
			\begin{array}{l}
				\stChoice{\stLabFmt{add}}{\stFmtC{int}} \stSeq \stEnd \\
				\stChoice{\stLabFmt{forward}}{\stFmtC{int}} \stSeq \stEnd 
			\end{array}
		}\stEnvComp \;
		\mpChanRole{\mpS}{\roleFmt{C}}: \stExtSum{\roleFmt{B}}{}{
			\begin{array}{l}
				\stChoice{\stLabFmt{add}}{\stFmtC{int}} \stSeq \stEnd \\
				\stChoice{\stLabFmt{sub}}{\stFmtC{int}} \stSeq \stEnd \\
				\stChoice{\stLabFmt{forward}}{\stFmtC{int}} \stSeq \stEnd 
			\end{array}
		}
	\end{align*}

Thus, direct projection cannot, in general, be substituted for association in~\Cref{thm:gtype:proj-sound}, and 
a similar limitation arises in~\Cref{thm:gtype:proj-comp}. 
Association therefore provides the required transition-preserving relation between global types and typing contexts, acting as the intermediate invariant for  subject reduction.

\subsection{Typing Context Properties Guaranteed via Association}
\label{sec:gtype:pbp}

The design of multiparty session type theory 
provides substantial advantages in guaranteeing  
desirable properties. 
Processes that adhere to the local types obtained from projections are 
inherently \emph{correct by construction}. 
This subsection highlights three key properties: 
\emph{communication safety}, \emph{deadlock-freedom}, 
and \emph{liveness}, 
which are ensured by typing contexts associated with global types.

\paragraph*{\bf{Communication Safety}}
\label{sec:type-system-safety}
We begin by introducing communication safety for typing contexts -- a behavioural property that ensures each role exchanges compatible messages, preventing label mismatches.

\begin{definition}[Typing Context Safety]
\label{def:mpst-env-safe}%
\label{def:safety_typing_context}
  Given a session $\mpS$, we say that
  $\predP$ is an \emph{$\mpS$-safety property} of typing contexts %
  iff, whenever $\predPApp{\stEnvNew}$, we have:

\smallskip
  \noindent%
  \begin{tabular}{@{\;\;}r@{\hskip 2mm}l}
    \inferrule{\iruleSafeComm}%
    &%
    $\stEnvMoveAnnotP{\stEnvNew}{\stEnvOutAnnot{\roleP}{\roleQ}{\stChoice{\stLab}{\stS}}}$
    \,and\,
    $\stEnvMoveAnnotP{\stEnvNew}{\stEnvInAnnot{\roleQ}{\roleP}{\stChoice{\stLabi}{\stSi}}}$
    \;\;implies\;\; %
    $\stEnvMoveAnnotP{\stEnvNew}{\stEnvCommAnnotSmall{\roleP}{\roleQ}{\stLab}}$;
    \\%
   \inferrule{\iruleSafeRec}%
    &%
    $\stEnvNew = \stEnvNewi \stEnvComp\, 
    \stEnvMap{%
        \mpChanRole{\mpS}{\roleP}%
      }{%
        \stRec{\stRecVar}{\stT}%
      }$ 
    \;\;implies\;\; %
    $\predPApp{%
      \stEnvNewi \stEnvComp\, \stEnvMap{%
        \mpChanRole{\mpS}{\roleP}%
      }{%
        \stT\subst{\stRecVar}{\stRec{\stRecVar}{\stT}}%
      }%
    }$;
    \\%
    \inferrule{\iruleMoveSession}%
    &%
    $\stEnvNew \stEnvMoveWithSession[\mpS] \stEnvNewi$
    \;\;implies\;\; %
    $\predPApp{\stEnvNewi}$.
  \end{tabular}

  \medskip
  \noindent%
  \upshape
  We say \emph{$\stEnvNew$ is $\mpS$-safe}, %
  written $\stEnvSafeSessP{\mpS}{\stEnvNew}$, %
  if $\predPApp{\stEnvNew}$  holds %
  for some $\mpS$-safety property $\predP$. %
  We say \emph{$\stEnvNew$ is safe}, %
  written $\stEnvSafeP{\stEnvNew}$, %
  if $\predPApp{\stEnvNew}$ holds %
  for some property $\predP$ which is $\mpS$-safe 
  for all sessions $\mpS$ occurring in $\dom{\stEnvNew}$.
\end{definition}

Our safety property is derived from a fundamental feature of generalised MPST systems~\citep[Def. 4.1]{POPL19LessIsMore}. 
 According to~\Cref{def:mpst-env-safe}, safety is a \emph{coinductive} property~\citep{SangiorgiBiSimCoInd}. 
 This implies that,  
  for a given a session $\mpS$, $\mpS$-safe is the largest $\mpS$-safety property, 
  including the union of all $\mpS$-safety properties. 
 To demonstrate the $\mpS$-safety of a typing context $\stEnvNew$, 
 we need to identify a property $\predP$ such that $\stEnvNew \in \predP$, 
 and then prove that $\predP$ qualifies as an $\mpS$-safety property. 
Specifically, if such a $\predP$ exists, it can be formulated as a set containing $\stEnvNew$ 
and all its reductums (via reduction $\stEnvMoveWithSessionStar[\mpS]$). 
We then verify whether all elements of $\predP$ satisfy each clause of~\Cref{def:mpst-env-safe}.

As specified by 
clause~\inferrule{\iruleSafeComm}, 
whenever two roles $\roleP$ and $\roleQ$ 
attempt communication, the receiving role $\roleQ$ must 
accommodate all output messages of the
sending role $\roleP$ with compatible payload types, ensuring the feasibility of the communication. 
Clause~\inferrule{\iruleSafeRec} unfolds %
any recursive entry, while 
clause \inferrule{\iruleMoveSession}
states that any typing context $\stEnvNewi$ reachable from $\stEnvNew$ through reduction on
session $\mpS$, must also belong to $\predP$, indicating that $\stEnvNewi$ is $\mpS$-safe as well. 
Note that any entry $\stEnvMap{\mpChanRole{\mpS}{\roleP}}{\stEnd}$ in $\stEnvNew$  
satisfies all clauses.

\begin{example}[Typing Context Safety]
\label{ex:ctx_safety}
Consider the typing context $\stEnvNew[\text{auth}]$ from~\Cref{ex:assoc}. 
We show that $\stEnvNew[\text{auth}]$ is $\mpS$-safe by inspecting 
its transitions. %
Specifically,  we have %

\smallskip
\centerline
{\(
\stEnvNew[\text{auth}]
\stEnvMoveCommAnnot{\mpS}{\roleFmt{s}}{\roleFmt{c}}{\stLabFmt{cancel}}
\cdot
\stEnvMoveCommAnnot{\mpS}{\roleFmt{c}}{\roleFmt{a}}{\stLabFmt{quit}}
\stEnvNewii[\text{auth}] = 
\stEnvMap{\mpChanRole{\mpS}{\roleFmt{s}}}{\stEnd} 
\stEnvComp 
\stEnvMap{\mpChanRole{\mpS}{\roleFmt{c}}}{\stEnd}  
\stEnvComp
\stEnvMap{\mpChanRole{\mpS}{\roleFmt{a}}}{\stEnd}
\)} 

\smallskip
\noindent
where  
each transition %
complies with all clauses of~\Cref{def:mpst-env-safe}.

The typing context $\stEnvNew[A] =  
\stEnvMap{\mpChanRole{\mpS}{\roleP}}
{\roleQ
 \stFmt{\oplus}
  \stLab[1]
   \stSeq
   \roleR
 \stFmt{\oplus}
  \stLab[3]
}
\stEnvComp 
\stEnvMap{\mpChanRole{\mpS}{\roleQ}}
{ \roleP
 \stFmt{\&}
 \stLab[2]
 }
 \stEnvComp 
\stEnvMap{\mpChanRole{\mpS}{\roleR}}
{ \roleP
 \stFmt{\&}
 \stLab[4]
}$ is \emph{not} $\mpS$-safe. 
Any property $\predP$ containing such a typing
context is \emph{not} a $\mpS$-safety property, as it violates 
\inferrule{\iruleSafeComm} of~\Cref{def:mpst-env-safe}: 
$\stEnvMoveAnnotP{\stEnvNew[A]}{\stEnvOutAnnot{\roleP}{\roleQ}
{\stChoice{\stLab[1]}{}}}$ and 
    $\stEnvMoveAnnotP{\stEnvNew[A]}{\stEnvInAnnot{\roleQ}{\roleP}
    {\stChoice{\stLab[2]}{}}}$, but 
    $\stEnvMoveAnnotP{\stEnv[A]}{\stEnvCommAnnotSmall{\roleP}{\roleQ}
   {\stLab[1]}}$ 
    does not hold.  

Finally, let's consider the typing context $\stEnvNew[B] = 
\stEnvMap{\mpChanRole{\mpS}{\roleP}}
{\roleQ
 \stFmt{\oplus}
  \stLabFmt{m(\stFmtC{real})}} 
  \stEnvComp 
  \stEnvMap{\mpChanRole{\mpS}{\roleQ}}
{\roleP
 \stFmt{\&}
  \stLabFmt{m(\stFmtC{int})}} 
  $. 
  $\stEnvNew[B]$ is \emph{not} $\mpS$-safe, as any property $\predP$ containing $\stEnvNew[B]$ 
contradicts \inferrule{\iruleSafeComm}: 
$\stEnvMoveAnnotP{\stEnvNew[B]}{\stEnvOutAnnot{\roleP}{\roleQ}
{\stChoice{\stLabFmt{m}}{\stFmtC{real}}}}$ and 
    $\stEnvMoveAnnotP{\stEnvNew[B]}{\stEnvInAnnot{\roleQ}{\roleP}
    {\stChoice{\stLabFmt{m}}{\stFmtC{int}}}}$, but 
    $\stEnvMoveAnnotP{\stEnvNew[B]}{\stEnvCommAnnotSmall{\roleP}{\roleQ}
   {\stLabFmt{m}}}$ 
    is not uphold due to $\stFmtC{int} \not \stSub  \stFmtC{real}$.  

\end{example}

\paragraph*{\bf{Deadlock-Freedom}}
\label{sec:type-system-df}
The property of deadlock-freedom ensures that a typing context does not get ``stuck'' during reduction -- that is, it either always has a transition available or has reached a terminal state.

\begin{definition}[Deadlock-Free Typing Contexts]
  \label{def:stenv-deadlock-free}
Given a session $\mpS$, a typing context $\stEnvNew$ is \emph{$\mpS$-deadlock-free}, 
written $\stEnvDFSessP{\mpS}{\stEnvNew}$,  
iff %
$\stEnvNew \!\stEnvMoveWithSessionStar[\mpS]\! \stEnvNotMoveWithSessionP[\mpS]{\stEnvNewi}$ 
implies $\forall \mpChanRole{\mpS}{\roleP} \!\in\! \dom{\stEnvNewi}: 
\stEnvApp{\stEnvNewi}{\mpChanRole{\mpS}{\roleP}}  =  %
\stEnd$. 
\end{definition}

Notably, 
a typing context that reduces infinitely adheres to deadlock-freedom, 
as it consistently undergoes further reductions. 
Alternatively, when a terminal typing context is reached, 
all entries must successfully terminate with $\stEnd$. 
Consequently, a deadlock-free typing context either continues to reduce perpetually or terminates successfully.

\begin{example}[Typing Context Deadlock-Freedom]
\label{ex:ctx_df}
 The typing 
 context $\stEnvNew[\text{auth}]$ from~\Cref{ex:assoc} is $\mpS$-deadlock-free, 
as any terminal typing context $\stEnvNewi[\text{auth}]$, reached from $\stEnvNew[\text{auth}]$ 
 via 
$\stEnvNew[\text{auth}] \!\stEnvMoveWithSessionStar[\mpS]\! 
\stEnvNotMoveWithSessionP[\mpS]{\stEnvNewi[\text{auth}]}$, contains only $\stEnd$ entries, 
which can be easily verified. 

The typing 
context 
$\stEnvNew[C] = \stEnvMap{\mpChanRole{\mpS}{\roleP}}
{ \roleQ
 \stFmt{\oplus}
 \stLab[1] 
 \stSeq
\roleR
 \stFmt{\&}
 \stLab[3] 
 } \stEnvComp 
 \stEnvMap{\mpChanRole{\mpS}{\roleQ}}
{ \roleR
 \stFmt{\oplus}
 \stLab[2] 
 \stSeq
\roleP
 \stFmt{\&}
 \stLab[1] 
 } \stEnvComp 
 \stEnvMap{\mpChanRole{\mpS}{\roleR}}
{ \roleP
 \stFmt{\oplus}
 \stLab[3] 
 \stSeq
\roleQ
 \stFmt{\&}
 \stLab[2] 
 }$ 
 is $\mpS$-safe but \emph{not} $\mpS$-deadlock-free, as 
 its inputs and outputs, despite being dual, are arranged in the incorrect order. 
Specifically, there are no possible reductions for $\stEnvNew[C]$, \ie  
$\stEnvNotMoveWithSessionP[\mpS]{\stEnvNew[C]}$, 
while none of the entries in $\stEnvNew[C]$ is a termination. 

Finally,  the typing 
context 
$
\stEnvNew[D] = \stEnvMap{\mpChanRole{\mpS}{\roleP}} 
{\stIntSum{\roleQ}{}{
\stLab[1]
\stSeq 
\roleR
 \stFmt{\oplus}
 \stLab[2], 
 \stLab[3]
}}  \stEnvComp 
 \stEnvMap{\mpChanRole{\mpS}{\roleQ}}
{ \stExtSum{\roleP}{}{
\stLab[1], 
 \stLab[2]
}
 } \stEnvComp 
 \stEnvMap{\mpChanRole{\mpS}{\roleR}}
{
\roleP
 \stFmt{\&}
 \stLab[2] 
 }$
 is $\mpS$-deadlock-free but \emph{not} $\mpS$-safe. While it consistently reaches a 
 successful termination by transmitting $\stLab[1]$ between $\roleP$ and $\roleQ$, 
 $\roleQ$ is unable to receive the message $\stLab[3]$ when $\roleP$ attempts to send it, due to a message mismatch. 
 \end{example}

\paragraph*{\bf{Liveness}}
\label{sec:type-system-live}
The liveness property ensures that every pending internal  or external 
choice eventually gets triggered through a message transmission. 
It relies on fairness -- specifically, the \emph{strong fairness of components}~\citep[Fact 2]{VanGlabbeekLICS2021} -- to guarantee that every enabled message transmission is eventually executed. 
Definitions of fair and live paths for typing contexts are provided in~\Cref{def:stenv-fairness}, and these paths are used to formalise the liveness for typing contexts in~\Cref{def:stenv-live}.

\begin{definition}[Fair, Live Paths]%
  \label{def:stenv-fairness}
  A \emph{path} %
  is %
  a
  possibly infinite sequence of typing contexts %
  $(\stEnvNew[n])_{n \in N}$, %
  where $N = \{0,1,2,\ldots\}$ is a (finite or infinite) set of consecutive natural numbers, %
  and, $\forall n, n+1 \!\in\! N$, $\stEnvNew[n] \stEnvMove \stEnvNew[n+1]$. 
  
  We say that a path $(\stEnvNew[n])_{n \in N}$ is \emph{fair for session $\mpS$} iff, %
  $\forall n \!\in\! N$:
    $\stEnvMoveAnnotP{\stEnvNew[n]}{\ltsSendRecv{\mpS}{\roleP}{\roleQ}{\stLab}}$ %
    implies $\exists k, \stLabi$ %
    such that $N \!\ni\! k \ge n$, %
    and $\stEnvNew[k] \stEnvMoveCommAnnot{\mpS}{\roleP}{\roleQ}{\stLabi} \stEnvNew[k+1]$.

  We say that a path $(\stEnvNew[n])_{n \in N}$ is \emph{live for session $\mpS$} iff, %
  $\forall n \!\in\! N$:
  \begin{enumerate}[label={(L\arabic*)}, leftmargin=*, nosep]
  \item
  \label{item:liveness:send}%
    $\stEnvMoveAnnotP{\stEnvNew[n]}{\stEnvOutAnnot{\roleP}{\roleQ}{\stChoice{\stLab}{\stS}}}$ %
    implies $\exists k, \stLabi$ %
    such that $N \ni k \ge n$ %
    and $\stEnvNew[k] \stEnvMoveCommAnnot{\mpS}{\roleP}{\roleQ}{\stLabi}
      \stEnvNew[k+1]$;
  \item
  \label{item:liveness:rcv}%
    $\stEnvMoveAnnotP{\stEnvNew[n]}{\stEnvInAnnot{\roleQ}{\roleP}{\stChoice{\stLab}{\stS}}}$ %
    implies $\exists k, \stLabi$ %
    such that $N \ni k \ge n$ %
    and $\stEnvNew[k] \!\stEnvMoveCommAnnot{\mpS}{\roleP}{\roleQ}{\stLabi}\! \stEnvNew[k+1]$. 
  \end{enumerate}
\end{definition}

A path is a (possibly infinite) sequence of transitions of a typing context. 
A path is fair for session $\mpS$ if, along the path,  
every enabled message transmission  
is eventually performed on $\mpS$. 
Similarly, a path is live for session $\mpS$ if, along the path, 
every pending internal or external choice is eventually triggered on $\mpS$.

\begin{definition}[Live Typing Contexts]
  \label{def:stenv-live}
  \label{def:typing-ctx-live}
Given a session $\mpS$, a typing context $\stEnvNew$ is \emph{$\mpS$-live},  
written $\stEnvLiveSessP{\mpS}{\stEnvNew}$, iff %
all paths starting with $\stEnvNew$ that are fair for session $\mpS$ are also live for $\mpS$.
\end{definition}

A typing context $\stEnvNew$ is $\mpS$-live if it consistently generates a live path for $\mpS$ under fairness.

\begin{example}[Fairness and Typing Context Liveness, originated from~\citep{GPPSY2023,POPL19LessIsMore}]
\label{ex:ctx_liveness}
Consider the typing context: 

\smallskip
\centerline{\(
\stEnvNew[E] = 
\stEnvMap{\mpChanRole{\mpS}{\roleP}}
{ \rolePi
 \stFmt{\oplus}
 \stLab[1] 
 }
 \stEnvComp
 \stEnvMap{\mpChanRole{\mpS}{\rolePi}}
{ \roleP
 \stFmt{\&}
 \stLab[1] 
 }
 \stEnvComp
  \stEnvMap{\mpChanRole{\mpS}{\roleQ}}
{\stRec{\stFmt{{\mathbf{t}}_{\roleQ}}}{
\roleQi
 \stFmt{\oplus}
 \stLab[2]}}
 \stSeq 
 \stFmt{{\mathbf{t}}_{\roleQ}}
 \stEnvComp
  \stEnvMap{\mpChanRole{\mpS}{\roleQi}}
{\stRec{\stFmt{{\mathbf{t}}_{\roleQi}}}{
\roleQ
 \stFmt{\&}
 \stLab[2]}}
 \stSeq 
 \stFmt{{\mathbf{t}}_{\roleQi}} 
\)}

\smallskip
\noindent
which is $\mpS$-safe and $\mpS$-deadlock-free. 
There is an infinite path starting with $\stEnvNew[E]$, where $\roleQ$ and $\roleQi$ 
continue communicating in session $\mpS$, 
while $\roleP$ and $\rolePi$ never trigger a transition to interact, \ie 

\smallskip
\centerline{\(
\stEnvNew[E] \,\stEnvMoveCommAnnot{\mpS}{\roleQ}{\roleQi}{\stLab[2]}\, \stEnvNew[E] 
\,\stEnvMoveCommAnnot{\mpS}{\roleQ}{\roleQi}{\stLab[2]}\, \cdots
\)}

\smallskip
\noindent
Such a  path is \emph{not} fair for $\mpS$ because 
although message transmission between $\roleP$ and $\rolePi$ is enabled, it will never occur. 
Alternatively, along any fair path of $\stEnvNew[E]$, 
all inputs and outputs are eventually fired on $\mpS$, indicating that $\stEnvNew[E]$ is $\mpS$-live. 

Now consider another typing context: 

\smallskip
\centerline{\(
\stEnvNewi[E] = 
 \stEnvMap{\mpChanRole{\mpS}{\roleP}}
{\stRec{\stFmt{{\mathbf{t}}_{\roleP}}}{
{ \stExtSum{\roleQ}{}{
  \begin{array}{@{}l@{}}
   \stLab[1]
   \stSeq
   \stFmt{{\mathbf{t}}_{\roleP}}, 
   \stLab[2]
   \stSeq
   \stFmt{{\mathbf{t}}_{\roleP}}
  \end{array}}}}}
\stEnvComp
\stEnvMap{\mpChanRole{\mpS}{\roleQ}}
{\stRec{\stFmt{{\mathbf{t}}_{\roleQ}}}{
{ \stIntSum{\roleP}{}{
  \begin{array}{@{}l@{}}
   \stLab[1]
   \stSeq
   \stFmt{{\mathbf{t}}_{\roleQ}}, 
   \stLab[2]
   \stSeq
   \roleR
   \stFmt{\oplus}
   \stLab[2]
   \stSeq
   \stFmt{{\mathbf{t}}_{\roleQ}}
  \end{array}}}}}
 \stEnvComp
  \stEnvMap{\mpChanRole{\mpS}{\roleR}}
{\stRec{\stFmt{{\mathbf{t}}_{\roleR}}}{
\roleQ
 \stFmt{\&}
 \stLab[2]}}
 \stSeq 
 \stFmt{{\mathbf{t}}_{\roleR}}
\)}

\smallskip
\noindent
which is $\mpS$-safe and $\mpS$-deadlock-free. 
$\stEnvNewi[E]$ has fair and live paths, where in $\mpS$, 
 $\stLab[2]$ is transmitted from $\roleQ$ to $\roleP$, and then to $\roleR$. 
 However, 
 there is also a fair path, 
 where in $\mpS$, $\stLab[1]$ is consistently  transmitted from 
 $\roleQ$ to $\roleP$: 
 
 \smallskip
 \centerline{\(
 \stEnvNewi[E] \,\stEnvMoveCommAnnot{\mpS}{\roleQ}{\roleP}{\stLab[1]}\, \stEnvNewi[E] 
\,\stEnvMoveCommAnnot{\mpS}{\roleQ}{\roleP}{\stLab[1]}\, \cdots
\)}

\smallskip
\noindent 
 In this case, $\roleR$ indefinitely awaits an input~($\stLab[2]$ from $\roleQ$) that will never arrive. 
 The path is fair, as the message to $\roleR$ is not enabled -- that is, 
 the action $\stEnvCommAnnotSmall{\roleQ}{\roleR}{\stLab[2]}$ is never selected, so its absence does not violate fairness. 
 Therefore, this path is fair for $\mpS$ but not live, meaning that $\stEnvNewi[E]$ is not $\mpS$-live.
  
  Finally, consider the typing context $\stEnvNewii[E] = 
\stEnvMap{\mpChanRole{\mpS}{\roleP}} 
{\stIntSum{\roleQ}{}{
\stLab[1], 
 \stLab[2]
}}  \stEnvComp 
 \stEnvMap{\mpChanRole{\mpS}{\roleQ}}
{ \stExtSum{\roleP}{}{
\stLab[1], 
 \stLab[3]
}
 }$, which 
 is $\mpS$-live. 
 There is only one path starting from $\stEnvNewii[E]$, \ie 
 
 \smallskip
 \centerline{\(
 \stEnvNewii[E] \,\stEnvMoveCommAnnot{\mpS}{\roleP}{\roleQ}{\stLab[1]}\, 
 \stEnvMap{\mpChanRole{\mpS}{\roleP}} 
{\stEnd}  \stEnvComp 
 \stEnvMap{\mpChanRole{\mpS}{\roleQ}}
{ \stEnd
 }
 \)}
 
 \smallskip
 \noindent
 which is fair and live for $\mpS$; however it is not $\mpS$-safe. 
\end{example}

\paragraph{\bf{Properties by Association}}
To conclude, \Cref{cor:allproperties} asserts that a typing context associated with a global type 
is constructed to guarantee the properties of safety, deadlock-freedom, and liveness.

\begin{restatable}[Safety, Deadlock-Freedom, and Liveness by Association]{theorem}{colProjAll}%
  \label{cor:allproperties}
Let $\gtG$ be a global type, $\stEnvNew$ a typing context, and $\mpS$ a session.  
If $\stEnvNew$ is associated with $\gtG$ for %
$\mpS$:  
 $\stEnvAssoc{\gtG}{\stEnvNew}{\mpS}$, 
  then $\stEnvNew$ is $\mpS$-safe, $\mpS$-deadlock-free, and $\mpS$-live.
\end{restatable}

\begin{example}[Typing Context Properties Guaranteed by Association]
\label{ex:prop_by_proj}
The typing context $\stEnvNew[\text{auth}]$, described in~\Cref{ex:assoc},  
is associated with $\gtG[\text{auth}]$ 
for session $\mpS$, \ie 
$\stEnvAssoc{\gtG[\text{auth}]}{\stEnvNew[\text{auth}]}{\mpS}$. 
Consequently, $\stEnvNew[\text{auth}]$ possesses the desirable properties of 
being $\mpS$-safe, $\mpS$-deadlock-free~(as demonstrated in~\Cref{ex:ctx_safety} and~\Cref{ex:ctx_df}, respectively), 
and $\mpS$-live.
\end{example}

\subsection{Relationships Between Typing Context Properties}
\label{sec:sync:relation}%

We now examine both the relationships between 
typing context properties, and how they relate to association, 
as formalised in~\Cref{lem:liveness-safety-implic}.

\begin{restatable}{theorem}{lemCtxSafetyImplications}
  \label{lem:liveness-safety-implic}%
  For any typing context $\stEnvNew$ and session $\mpS$, the following 
  statements 
  are valid:

  \smallskip%
  \noindent%
  \begin{minipage}{0.6\linewidth}
    \begin{enumerate}[leftmargin=*,nosep,label=(\arabic*),ref=\arabic*]
    \item
    \label{item:liveness-safety-implic:df-safe}%
      $\stEnvDFSessP{\mpS}{\stEnvNew} \mathrel{{\notImpliedBy}\,{\notImplies}}%
      \stEnvSafeSessP{\mpS}{\stEnvNew}$;%
      
    \item
    \label{item:liveness-safety-implic:live-safe}%
      $\stEnvLiveSessP{\mpS}{\stEnvNew} \mathrel{{\notImpliedBy}\,{\notImplies}}%
      \stEnvSafeSessP{\mpS}{\stEnvNew}$;%
      
    \item
    \label{item:liveness-safety-implic:live-df}%
      $\stEnvLiveSessP{\mpS}{\stEnvNew} \mathrel{{\notImpliedBy}\!\!{\implies}}%
      \stEnvDFSessP{\mpS}{\stEnvNew}$;%

    \item
    \label{item:liveness-safety-implic:asso-safe}%
      $\exists \gtG: \stEnvAssoc{\gtG}{\stEnvNew}{\mpS}$ %
      $\mathrel{{\notImpliedBy}\!\!{\implies}}$ %
      $\stEnvSafeSessP{\mpS}{\stEnvNew}$;%
      
    \item
    \label{item:liveness-safety-implic:asso-df}%
      $\exists \gtG: \stEnvAssoc{\gtG}{\stEnvNew}{\mpS}$ %
      $\mathrel{{\notImpliedBy}\!\!{\implies}}$ %
      $\stEnvDFSessP{\mpS}{\stEnvNew}$;%
      
    \item
    \label{item:liveness-safety-implic:asso-live}%
      $\exists \gtG: \stEnvAssoc{\gtG}{\stEnvNew}{\mpS}$ %
      $\mathrel{{\notImpliedBy}\!\!{\implies}}$ %
      $\stEnvLiveSessP{\mpS}{\stEnvNew}$.%
    \end{enumerate}
  \end{minipage}
  \;%
  \begin{minipage}{0.35\linewidth}
    \includegraphics[width=1\linewidth]{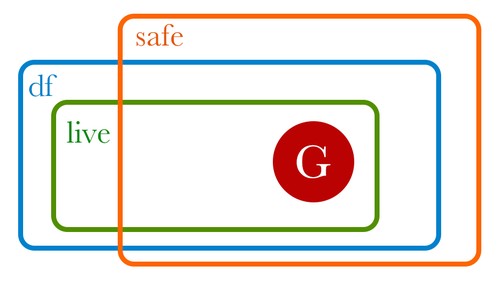}%
  \end{minipage}
\end{restatable}

In the diagram, the ``safe'' set~(resp. ``df'' set, ``live'' set ) 
contains all typing contexts that 
are \linebreak 
$\mpS$-safe~(resp. $\mpS$-deadlock-free, $\mpS$-live). 
The red set %
$\mathbb{\gtG}$ encompasses 
all typing contexts associated with  some global type for~$\mpS$. 
The negated implications stated in~\Cref{lem:liveness-safety-implic} 
are illustrated in~\Cref{ex:ctx_df}, \Cref{ex:ctx_liveness}, and 
\Cref{ex:ctx_not-assoc}, respectively.

\begin{example}[Non-Associated Typing Context, adapted from~\citep{Bernardi2016LMCS,POPL19LessIsMore}]
\label{ex:ctx_not-assoc}
The typing context \linebreak
$\stEnvNew[F] = \stEnvMap{%
    \mpChanRole{\mpS}{\roleP}%
  }{%
   \stT
  }%
  \stEnvComp%
  \stEnvMap{%
    \mpChanRole{\mpS}{\roleQ}%
  }{%
    \roleP \stFmt{\&}
   \stLabFmt{\stLab}(\stT) 
    \stSeq \stEnd%
  }$, 
  with 
  $\stT%
  = \stRec{\stRecVar}{\stOut{\roleQ}{\stLab}{\stRecVar}} \stSeq \stEnd$~(from \citep[Ex.\,1.2]{Bernardi2016LMCS}), 
  is $\mpS$-safe, $\mpS$-deadlock-free, and $\mpS$-live, 
  but \emph{not} associated with any global type for session $\mpS$. 
  This is due to the occurrence of the recursion variable $\stRecVar$ as a payload in $\stT$,  
  which is not allowed by~\Cref{def:subtyping} since the subtyping relation provides no rules for recursion variables;  consequently,  the subtyping condition in~\Cref{def:assoc} is not satisfied.  
\end{example}

\section{Multiparty Session Typing System}
\label{sec:typesystem}
This section presents a type system for the multiparty session $\pi$-calculus~(defined in~\Cref{sec:processes}). 
In~\Cref{sec:type-system:tyrules}, we introduce the typing rules. 
\Cref{sec:subject_reduction,sec:session_fidelity} demonstrate the main properties of typed processes: \emph{subject reduction} and \emph{session fidelity}. 
Finally,~\Cref{sec:typed-process-property} shows how process properties such as  
deadlock-freedom 
and liveness 
can be guaranteed by construction.

\subsection{Typing Rules}%
\label{sec:type-system:tyrules}

\begin{figure}[t!]
  \centerline{\(
  \begin{array}{c} 
  \mpEnvEntails{\mpEnvNew}{\mpNat}{\tyNat}
  \quad 
   \mpEnvEntails{\mpEnvNew}{\mpInt}{\tyInt}
   \quad 
    \mpEnvEntails{\mpEnvNew}{\mpTrue}{\tyBool}
    \quad 
    \mpEnvEntails{\mpEnvNew}{\mpFalse}{\tyBool}
    \quad 
    \mpEnvEntails{\mpEnvNew}{\mpString{}}{\tyString}
    \quad 
    \mpEnvEntails{\mpEnvNew}{\mpUnit{}}{\tyUnit}
     \quad 
     \mpEnvEntails{\mpEnvNew \mpEnvComp \mpFmt{x}\mathbin{\!:\!}\tyGround}{\mpFmt{x}}{\tyGround}
  \\[2mm]
   \inference[]{\mpEnvEntails{\mpEnvNew}{\mpE}{\tyNat}}{\mpEnvEntails{\mpEnvNew}{\mpSucc{\mpE}}{\tyNat}}
   \quad 
   \inference[]{\mpEnvEntails{\mpEnvNew}{\mpE}{\tyInt}}{\mpEnvEntails{\mpEnvNew}{\mpNeg{\mpE}}{\tyNat}}
   \quad 
   \inference[]{\mpEnvEntails{\mpEnvNew}{\mpE}{\tyBool}}{\mpEnvEntails{\mpEnvNew}{\neg \mpE}{\tyBool}}
  \quad
   \inference[]{\mpEnvEntails{\mpEnvNew}{\mpE[1]}{\tyInt}
   &
   \mpEnvEntails{\mpEnvNew}{\mpE[2]}{\tyInt}
   }{\mpEnvEntails{\mpEnvNew}{\mpE[1] < \mpE[2]}{\tyBool}}
   \quad 
    \inference[]{\mpEnvEntails{\mpEnvNew}{\mpE[1]}{\tyGround}
   &
   \mpEnvEntails{\mpEnvNew}{\mpE[2]}{\tyGround}
   }{\mpEnvEntails{\mpEnvNew}{\mpE[1] \otimes \mpE[2]}{\tyGround}}
   \\[1em]
  \hdashline
  \\
  \inference[\iruleMPEnd]
  {\forall \mpC \!\in\! \dom{\stEnvNew} 
&
\stEnvApp{\stEnvNew}{\mpC} = \stEnd}
  { \stEnvEndP{\stEnv}}
  \qquad 
\inference[\iruleMPNil]{
\stEnvEndP{\stEnv}
    }{
      \stJudgeNew{\mpEnvNew}{\stEnvNew}{\mpNil}
      }
\qquad 
       \inference[\iruleMPSub]{
      \stJudgeNew{\mpEnvNew}{\stEnvNew}{\mpP} 
      & 
      \stEnvNew \stSub \stEnvNewi
    }{
     \stJudgeNew{\mpEnvNew}{\stEnvNewi}{\mpP}
    }
    \\[2mm]
    \inference[\iruleMPCall]{
         \forall i \!\in\! 1 \ldots n
         &
          \mpEnvEntails{\mpEnvNew}{\mpE[i]}{\tyGround[i]}
        &
        \stEnvEndP{\stEnv}
    }{
      \stJudgeNew{\mpEnvNew \mpEnvComp \mpEnvMap{\mpX}{\tyGround[1], \ldots, \tyGround[n], \stT[1],\ldots,\stT[m]}}{
        \stEnvNew \stEnvComp 
        \stEnvMap{\mpC[1]}{\stT[1]} \stEnvComp \ldots \stEnvComp \stEnvMap{\mpC[m]}{\stT[m]}
      }{
        \mpCall{\mpX}{\mpE[1],\ldots,\mpE[n], \mpC[1], \ldots, \mpC[m]}
      }
    }
    \\[3mm]
    \inference[\iruleMPDef]{
      \begin{array}{c}
        \stJudgeNew{
          \mpEnvNew \mpEnvComp
          \mpEnvMap{\mpX}{\tyGround[1],\ldots,\tyGround[n], \stT[1], \ldots, \stT[m]}
          \mpEnvComp
          \mpEnvMap{\mpFmt{x_1}}{\tyGround[1]}
          \mpEnvComp \ldots \mpEnvComp \mpEnvMap{\mpFmt{x_n}}{\tyGround[n]}
        }{
          \stEnvMap{y_1}{\stT[1]}
          \stEnvComp \ldots \stEnvComp
          \stEnvMap{y_m}{\stT[m]}
        }{
          \mpP
        }
        \\[0.5mm]
        \stJudgeNew{
          \mpEnvNew \mpEnvComp
          \mpEnvMap{\mpX}{\tyGround[1],\ldots,\tyGround[n], \stT[1], \ldots, \stT[m]}
        }{
          \stEnvNew
        }{
          \mpQ
        }
        \end{array}
    }{
      \stJudge{\mpEnvNew}{
        \stEnvNew
      }{
        \mpDef{\mpX}{
          \mpEnvMap{\mpFmt{x_1}}{\tyGround[1]}, 
          \ldots,
          \mpEnvMap{\mpFmt{x_n}}{\tyGround[n]}, 
          \stEnvMap{y_1}{\stT[1]}, 
          \ldots, 
          \stEnvMap{y_m}{\stT[m]}
        }{\mpP}{\mpQ}
      }
    }
    \\[3mm]
      \inference[\iruleMPPar]{
      \stJudgeNew{\mpEnvNew}{
        \stEnvNew[1]
      }{
        \mpP[1]
      }
      &
      \stJudgeNew{\mpEnvNew}{
        \stEnvNew[2]
      }{
        \mpP[2]
      }%
    }{
      \stJudgeNew{\mpEnvNew}{
        \stEnvNew[1] \stEnvComp \stEnvNew[2]
      }{
        \mpP[1] \mpPar \mpP[2]
      }
    }
    \qquad 
    \inference[\iruleMPIf]{
    \mpEnvEntails{\mpEnvNew}{\mpE}{\tyBool}
    &
    \stJudgeNew{\mpEnvNew}{\stEnvNew}{\mpP}
    &
    \stJudgeNew{\mpEnvNew}{\stEnvNew}{\mpQ}
     }{
     \stJudgeNew{\mpEnvNew}{\stEnvNew}{\mpIf{\mpE}{\mpP}{\mpQ}}
     }
     \\[3mm]
    \inference[\iruleMPBranch]{
        \forall i \!\in\! I
        &
        \left\{
        \begin{array}{c}
        \stJudgeNew{\mpEnvNew \mpEnvComp \stEnvMap{z_i}{\tyGround}}{
          \stEnvNew \stEnvComp
                     \stEnvMap{\mpC}{\stT[i]}
        }{
          \mpP[i]
        } \,\,\,\, \text{if $\stS[i] = \tyGround$}
        \\
       \stJudgeNew{\mpEnvNew}{
          \stEnvNew \stEnvComp \stEnvMap{z_i}{\stT} \stEnvComp 
                     \stEnvMap{\mpC}{\stT[i]}
        }{
          \mpP[i]
        } \,\,\,\, \text{if $\stS[i] = \stT$}
        \end{array}
      \right\}
    }{
      \stJudgeNew{\mpEnvNew}{
        \stEnvNew \stEnvComp \stEnvMap{\mpC}{\stExtSum{\roleQ}{i \in I}{\stChoice{\stLab[i]}{\tyS[i]} \stSeq \stT[i]}}
      }{
        \mpBranch{\mpC}{\roleQ}{i \in I}{\mpLab[i]}{z_i}{\mpP[i]}{}
      }
    }
      \\[3mm]
    \inference[\iruleMPSelV]{
      \mpEnvEntails{\mpEnvNew}{\mpE}{
      \tyGround
      }
      &
      \stJudgeNew{\mpEnvNew}{
        \stEnvNew \stEnvComp \stEnvMap{\mpC}{\stT}
      }{
        \mpP
      }
    }{
      \stJudgeNew{\mpEnvNew}{
        \stEnvNew \stEnvComp \stEnvMap{\mpC}{\stIntSum{\roleQ}{}{\stChoice{\stLab}{\tyGround} \stSeq \stT}}
      }{
        \mpSel{\mpC}{\roleQ}{\mpLab}{\mpE}{\mpP}
      }
    }
   \quad  
    \inference[\iruleMPSelD]{
      \stJudgeNew{\mpEnvNew}{
        \stEnvNew \stEnvComp \stEnvMap{\mpC}{\stT}
      }{
        \mpP
      }
    }{
      \stJudgeNew{\mpEnvNew}{
        \stEnvNew \stEnvComp \stEnvMap{\mpC}{\stIntSum{\roleQ}{}{\stChoice{\stLab}{\stTi} \stSeq \stT}} \stEnvComp 
        \stEnvMap{\mpCi}{\stTi}
      }{
        \mpSel{\mpC}{\roleQ}{\mpLab}{\mpCi}{\mpP}
      }
    }
   \\[3mm]  
       \inference[\iruleMPGlobalRes]{
       \stEnvNewi =
  \setenum{\stEnvMap{\mpChanRole{\mpS}{\roleP}}{\gtProj{\gtG}{\roleP}}}_{%
    \roleP \in \gtRoles{\gtG}}
      &
      \mpS \!\not\in\! \stEnvNew
      &
      \stJudgeNew{\mpEnvNew}{
        \stEnvNew \stEnvComp \stEnvNewi %
      }{
        \mpP
      }
    }{
      \stJudgeNew{\mpEnvNew}{
        \stEnvNew 
      }{
        \mpRes{\stEnvMap{\mpS}{\stEnvNewi}}\mpP
      }
    }
  \end{array}
  \)}%
  \caption{
    Typing rules for expressions (top) and processes (bottom). 
  }
  \label{fig:mpst-rules}
  \label{fig:typing_rules}
  \label{fig:typing_rules_all}
\end{figure}

 Two kinds of typing contexts, as introduced in~\Cref{def:mpst-env}, are used in our type system:
$\mpEnvNew$, which assigns a sequence of basic types and session types 
to each process variable $\mpX$, as well as a basic type to each expression variable $\mpFmt{x}$; 
and $\stEnvNew$, which maps %
channels to session types. 
$\mpEnvNew$ is utilised in judgements for expressions, while 
both $\mpEnvNew$ and $\stEnvNew$ are 
jointly applied in judgements for processes. 
The typing judgements are formulated as:

\medskip
\centerline
{\(
\mpEnvEntails{\mpEnvNew}{\mpE}{\tyGround} \quad \text{and} \quad
\stJudgeNew{\mpEnvNew}{\stEnvNew}{\mpP}
\)}

\medskip
\noindent%
The judgement for expressions is standard: given the expression variables and basic types in $\mpEnvNew$, 
$\mpE$ is of basic type $\tyGround$. For processes, based on the types assigned 
to expressions and process variables in $\mpEnvNew$, 
$\mpP$ uses its channels \emph{linearly} as specified by $\stEnvNew$.

The typing system is defined inductively by the rules depicted in~\Cref{fig:typing_rules_all}. %
 The typing rules for expressions are straightforward, while we focus on elaborating the rules for processes.

Rule \inferrule{\iruleMPEnd} 
introduces a predicate $\stEnvEndP{\cdot}$ on typing
contexts, which denotes the termination of all endpoints (note that $\stEnvEndP{\stEnvEmpty}$ holds).  
This predicate is used in
\inferrule{\iruleMPNil} to type an inactive process $\mpNil$. 
Rule \inferrule{\iruleMPSub} incorporates subtyping within typing contexts. 
Rules~\inferrule{\iruleMPCall} and~\inferrule{\iruleMPDef} handle 
recursive process calls and declarations, respectively. 
Rule~\inferrule{\iruleMPPar} divides the typing context \emph{linearly} into two parts,
each used to type one of the sub-processes.  
Rule~\inferrule{\iruleMPIf} is used for typing conditionals. 
In rule~\inferrule{\iruleMPBranch}, we distinguish whether the payload in each branch is a basic type or a session type. 
Similarly, for typing selections, we apply two rules,  \inferrule{\iruleMPSelV} and \inferrule{\iruleMPSelD},  
to differentiate the payload types. 
Finally, rule \inferrule{\iruleMPGlobalRes} uses 
a typing context derived from a global type via projection to enforce session restriction. 
Note that association can be applied in place of projection in this rule, as a more general approach. %

\begin{example}[Typed Process]
\label{ex:typed_proc}
Consider the processes $\mpP$ and $\mpQ$ from~\Cref{ex:cal_syntax_semantics}, along with 
a type context $\stEnv[H] = \stEnvMap{\mpChanRole{\mpS}{\roleP}}{\stT[H_{\roleP}]}
\stEnvComp 
\stEnvMap{\mpChanRole{\mpS}{\roleQ}}{\stT[H_{\roleQ}]} 
\stEnvComp 
\stEnvMap{\mpChanRole{\mpS}{\roleR}}{\stT[H_{\roleR}]}$, where 
$\stT[H_{\roleP}] = 
\roleQ
 \stFmt{\oplus}
  \stChoice{\stLabi}{\stT[H_{\roleR}]} 
   \stSeq
   \roleR
 \stFmt{\&}
 \stChoice{\stLab}{\stFmtC{int}} 
 \stSeq 
 \stEnd, 
\stT[H_{\roleQ}] = 
\roleP
 \stFmt{\&}
  \stChoice{\stLabi}{\stT[H_{\roleR}]} 
  \stSeq 
  \stEnd, 
 \stT[H_{\roleR}]  = \roleP
 \stFmt{\oplus}
  \stChoice{\stLab}{\stFmtC{int}} 
  \stSeq 
  \stEnd$. 
  The context  $\stEnv[H]$ can type the process $\mpP \mpPar \mpQ$ through the following 
  derivation:
 
 \smallskip
  \centerline{\(%
   \begin{array}{c}
     \inference[\iruleMPPar]{
      \inference[\iruleMPSelD]{%
         \inference[\iruleMPBranch]{
          \inference[\iruleMPNil]{ 
          {
          \stEnvEndP{
           \stEnvMap{\mpChanRole{\mpS}{\roleP}}{\stEnd}
          }}}
          {
          \stJudge{\mpEnvMap{\mpFmt{z}}{\stFmtC{int}}}{\stEnvMap{\mpChanRole{\mpS}{\roleP}}{\stEnd}
          }{\mpNil}}
          }{
         \stJudge{\mpEnvEmpty}{%
            \stEnvMap{\mpChanRole{\mpS}{\roleP}}
            {\roleR 
             \stFmt{\&}
 \stChoice{\stLab}{\stFmtC{int}} 
 \stSeq 
 \stEnd
             }%
          }{%
             \mpBranchSingle{\mpChanRole{\mpS}{\roleP}}{\roleR}{\mpLab}{z}{\mpNil}{}    
          }
          }%
      }%
      {%
        \stJudge{\mpEnvEmpty}{%
          \stEnvMap{\mpChanRole{\mpS}{\roleP}}{\stT[H_{\roleP}]}%
          \stEnvComp%
          \stEnvMap{\mpChanRole{\mpS}{\roleR}}{\stT[H_{\roleR}]}%
        }{%
          \mpP%
        }%
      }
      &%
      \hspace{-1mm}%
      \inference[\iruleMPBranch]{%
         \inference[\iruleMPSelV]{
          \mpEnvEntails{\mpEnvEmpty}{42}{\tyInt}
          \quad  
          \inference[\iruleMPNil]{ 
          {
          \stEnvEndP{
           \stEnvMap{\mpChanRole{\mpS}{\roleQ}}{\stEnd} \stEnvComp 
           \stEnvMap{\mpFmt{z}}{\stEnd}
          }}}
          {
          \stJudge{\mpEnvEmpty}{\stEnvMap{\mpChanRole{\mpS}{\roleQ}}{\stEnd} \stEnvComp 
          \stEnvMap{\mpFmt{z}}{\stEnd}
          }{\mpNil}}
          }{
         \stJudge{\mpEnvEmpty}{%
             \stEnvMap{\mpChanRole{\mpS}{\roleQ}}{\stEnd}%
             \stEnvComp 
              \stEnvMap{\mpFmt{z}}
            {\stT[H_{\roleR}]
             }
          }{%
             \mpSel{z}{\roleP}{\mpLab}{42}{\mpNil}  
          }
          }%
      }%
      {%
        \stJudge{\mpEnvEmpty}{%
          \stEnvMap{\mpChanRole{\mpS}{\roleQ}}{\stT[H_{\roleQ}]}%
        }{%
          \mpQ%
        }%
      }%
    }{%
      \stJudgeNew{\mpEnvEmpty}{%
        \stEnv[H]%
      }{%
        \mpP \mpPar \mpQ%
      }%
    }
   \end{array}
  \)}%

\medskip
\noindent  
 Moreover, 
 since $\stEnv[H]$ is obtained by projecting a global type 
 $\gtG[H] =  \gtCommSingle{\roleP}{\roleQ}{\gtLabi}{\roleP 
 \stFmt{\oplus} \stLabFmt{m(\stFmtC{int})}}{\gtCommSingle{\roleR}{\roleP}{\gtLab}{\stFmtC{int}}{\gtEnd}}$ onto roles in session $\mpS$, 
 \ie $\stEnv[H] =  \setenum{\stEnvMap{\mpChanRole{\mpS}{\roleP}}{\gtProj{\gtG[H]}{\roleP}}}_{%
    \roleP \in \gtRoles{\gtG[H]}}$, 
it follows from \inferrule{\iruleMPGlobalRes} that 
 the process $\mpP \mpPar \mpQ$ is closed under $\stEnv[H]$, \ie $\stJudge{\mpEnvEmpty}{\stEnvEmpty}{\mpRes{\stEnvMap{\mpS}{\stEnv[H]}}\mpP \mpPar \mpQ}$.    
 
\end{example}

\begin{example}[Typed Process of OAuth]
\label{ex:typed_process_auth}
Recall the typing contexts $\stEnv[\text{auth}_{\roleFmt{s}}]$, $\stEnv[\text{auth}_{\roleFmt{c}}]$, and  
 $\stEnv[\text{auth}_{\roleFmt{a}}]$ from \Cref{ex:assoc}, and the processes
$\mpP[\roleFmt{s}]$, $\mpP[\roleFmt{c}]$, and $\mpP[\roleFmt{a}]$ from~\Cref{ex:process_oauth}. 
These contexts enable the typing of the respective processes. 
Therefore, the context $\stEnv[\text{auth}]$ (from~\Cref{ex:assoc}) 
can type the process $\mpP[\roleFmt{s}] \mpPar \mpP[\roleFmt{c}] \mpPar \mpP[\roleFmt{a}]$. 

Furthermore, the typing context 
$\stEnvMap{\mpChanRole{\mpS}{\roleFmt{s}}}{\stT[\roleFmt{s}]}$, using 
the local type $\stT[\roleFmt{s}]$ from~\Cref{ex:types_oath}, can type the process $\mpP[\roleFmt{s}]$, 
since 
$\stEnv[\text{auth}_{\roleFmt{s}}] \stSub \stEnvMap{\mpChanRole{\mpS}{\roleFmt{s}}}{\stT[\roleFmt{s}]}$~(as demonstrated 
in~\Cref{ex:assoc}), 
and $\stJudge{\mpEnvEmpty}{\stEnv[\text{auth}_{\roleFmt{s}}]}{\mpP[\roleS]}$. 
           Similarly, $\stEnvMap{\mpChanRole{\mpS}{\roleFmt{c}}}{\stT[\roleFmt{c}]}$ and 
           $\stEnvMap{\mpChanRole{\mpS}{\roleFmt{a}}}{\stT[\roleFmt{a}]}$
 type 
$\mpP[\roleFmt{c}]$ and $\mpP[\roleFmt{a}]$, respectively. Thus, 
by \inferrule{\iruleMPGlobalRes}, 
 $\mpP[\roleFmt{s}] \mpPar \mpP[\roleFmt{c}] \mpPar \mpP[\roleFmt{a}]$ is closed under 
 $\setenum{\stEnvMap{\mpChanRole{\mpS}{\roleP}}{\gtProj{\gtG[\text{auth}]}{\roleP}}}_{%
    \roleP \in \gtRoles{\gtG[\text{auth}]}}$~(where $\gtG[\text{auth}]$ is from~\Cref{ex:types_oath}). 
  \end{example}

\subsection{Subject Reduction}
\label{sec:subject_reduction}

Subject reduction (type soundness) ensures the preservation of well-typedness under reduction: 
if a process $\mpP$ is typable and $\mpP \!\mpMove\! \mpPi$, then $\mpPi$ remains typable. 
Consequently, no reduction sequence starting from a well-typed process can reach an untypable state.

We present a subject reduction result in~\Cref{lem:subject-reduction}, 
where $\mpP$ is constructed from a global type via association~(\Cref{def:assoc}), 
\ie it is typed by a context associated with the global type.  
\Cref{lem:subject-reduction} serves as the fundamental technical result, establishing such a property as an invariant preserved throughout reduction.

\begin{restatable}[Subject Reduction via Association]{theorem}{lemSubjectReduction}
  \label{lem:subject-reduction}%
  Assume $\stJudgeNew{\mpEnvNew}{\stEnvNew}{\mpP}$ where $\forall \mpS \in \stEnvNew: \exists \gtG[\mpS]: \stEnvAssoc{\gtG[\mpS]}{\stEnvNew[\mpS]}{\mpS}$. 
  If $\mpP \!\mpMove\! \mpPi$, then 
  $\exists \stEnvNewi$
  such that 
  $\stEnvNew \!\stEnvMoveStar\! \stEnvNewi$, 
  $\stJudgeNew{\mpEnvNew}{\stEnvNewi}{\mpPi}$, and $\forall \mpS \in \stEnvNewi: 
  \exists \gtGi[\mpS]:  \gtG[\mpS] \,\gtMoveStar\, \gtGi[\mpS]$ and 
  $\stEnvAssoc{\gtGi[\mpS]}{\stEnvNewi[\mpS]}{\mpS}$. 
\end{restatable}
\begin{proof}
By induction on the derivation of $\mpP \!\mpMove\! \mpPi$~(\Cref{def:mpst-pi-semantics}). 
\qedhere
\end{proof}

Furthermore, since a typing context obtained from a global type via projection 
is a supertype of the context associated with it, 
and a process can be typed using a broader context, 
\Cref{thm:subject_reduction} -- our key subject reduction result -- demonstrates that  
well-typedness is maintained for processes constructed from global types via projection.  
This result recovers  the standard projection-based form of subject reduction and reinforces compatibility with global types as a behavioural invariant through reductions,   ensuring correctness by construction in alignment with the MPST top-down methodology.

\begin{restatable}[Subject Reduction]{theorem}{lemSubjectReductionFinal}
\label{lem:final-subject-reduction}
\label{thm:subject_reduction}
\label{thm:subject-reduction}
 Assume $\stJudgeNew{\mpEnvNew}{\stEnvNew}{\mpP}$,  
  where $\forall \mpS \in \stEnv: \exists \gtG[\mpS]: 
  \stEnv[\mpS] = \setenum{\stEnvMap{\mpChanRole{\mpS}{\roleP}}{\gtProj{\gtG[\mpS]}{\roleP}}}_{%
   \roleP \in \gtRoles{\gtG[\mpS]}}$. 
 If $\mpP \!\mpMove\! \mpPi$, then 
  $\exists \stEnvi$
  such that  $\stJudgeNew{\mpEnvNew}{\stEnvNewi}{\mpPi}$, and 
  $\forall \mpS \in \stEnvi: \exists \gtGi[\mpS]: \gtG[\mpS] \,\gtMoveStar\, \gtGi[\mpS]$ and $\stEnvi[\mpS] = \setenum{\stEnvMap{\mpChanRole{\mpS}{\roleP}}{\gtProj{\gtGi[\mpS]}{\roleP}}}_{%
    \roleP \in \gtRoles{\gtGi[\mpS]}}$. 
\end{restatable}
\begin{proof}
By the definition of association~(\Cref{def:assoc}), subject reduction via association~(\Cref{lem:subject-reduction}), 
and the application of \inferrule{\iruleMPSub} as necessary. 
\qedhere
\end{proof}

As a corollary of~\Cref{lem:final-subject-reduction}, well-typed processes communicate without $\mpErr$ors.

\begin{restatable}[Type Safety]{corollary}{lemTypeSafety}
  \label{cor:type-safety}
  Assume $\stJudgeNew{\emptyset}{\emptyset}{\mpP}$.
  If $\mpP \!\mpMoveStar\! \mpPi$, then %
  $\mpPi$ has no error.%
\end{restatable}

\begin{example}[Subject Reduction]
\label{ex:subject_reduction}
Recall the typed process $\mpP \,\mpPar\, \mpQ$,  
the typing context $\stEnv[H]$, and the global type $\gtG[H]$ from~\Cref{ex:typed_proc}. 
Using rule \inferrule{\iruleMPRedCommChannel}, 
 the process $\mpP \,\mpPar\, \mpQ$ reduces to 
 
 \smallskip
 \centerline{\(
  \mpPi \mpPar \mpQi = \mpBranchSingle{\mpChanRole{\mpS}{\roleP}}{\roleR}{\mpLab}{x}{\mpNil}{}
    \mpPar
    \mpSel{\mpChanRole{\mpS}{\roleR}}{\roleP}{\mpLab}{42}{\mpNil}
    \)}
    
    \smallskip
    \noindent 
 The typing context
  $\stEnv[H]$, with $\stEnvAssoc{\gtG[H]}{\stEnv[H]}{\mpS}$,   
 transitions via rule \inferrule{\iruleTCtxCom} to 
 
 \smallskip
 \centerline{\(
 \stEnvi[H] = \stEnvMap{\mpChanRole{\mpS}{\roleP}}{
     \roleR 
             \stFmt{\&}
 \stChoice{\stLab}{\stFmtC{int}}
 \stSeq 
 \stEnd}
    \stEnvComp 
    \stEnvMap{\mpChanRole{\mpS}{\roleQ}}{\stEnd}
    \stEnvComp 
    \stEnvMap{\mpChanRole{\mpS}{\roleR}}{\roleP 
             \stFmt{\oplus}
 \stChoice{\stLab}{\stFmtC{int}}
 \stSeq 
 \stEnd
 }\)}
 
 \smallskip
 \noindent
 which types $\mpPi \,\mpPar\, \mpQi$ and is associated with the global 
 type 
 $\gtGi[H] = \gtCommSingle{\roleR}{\roleP}{\gtLab}{\stFmtC{int}}{\gtEnd}$, obtained by the transition 
 $\gtG[H] \,\gtMove[\ltsSendRecv{\mpS}{\roleFmt{p}}{\roleFmt{q}}{\gtLabi}]\, \gtGi[H]$.  
  Moreover, the typing context $\stEnvii[H] = 
 \setenum{\stEnvMap{\mpChanRole{\mpS}{\roleP}}{\gtProj{\gtGi[H]}{\roleP}}}_{%
    \roleP \in \gtRoles{\gtGi[H]}}$ also types $\mpPi \,\mpPar\, \mpQi$. 

Note that the typing context $\setenum{\stEnvMap{\mpChanRole{\mpS}{\roleP}}{\gtProj{\gtG[\text{auth}]}{\roleP}}}_{%
    \roleP \in \gtRoles{\gtG[\text{auth}]}}$ and the process $\mpP[\roleFmt{s}] \mpPar \mpP[\roleFmt{c}] \mpPar \mpP[\roleFmt{a}]$ 
from~\Cref{ex:typed_process_auth} also adhere to subject reduction, thereby ensuring type safety.

\end{example}

\subsection{Session Fidelity}
\label{sec:session_fidelity}

\emph{Session fidelity} asserts the converse implication concerning subject reduction: 
if a process $\mpP$ is typed by $\stEnv$, and $\stEnv$ can reduce along session $\mpS$,  
then $\mpP$ can replicate at least one of the reductions performed by $\stEnv$~(though not necessarily all, 
as $\stEnv$ may over-approximate the behaviour of $\mpP$). 

However, this property does \emph{not} hold universally for all well-typed processes. 
For example, a well-typed process may diverge due to unguarded recursion~(\eg $\mpDef{\mpX}{\ldots}{X}{\mpX}$), or may deadlock due to intricate interleavings of communications across multiple sessions~\citep{CDYP2015}.

To address this, and in line  with~\citep{POPL19LessIsMore} and most existing work on session types, %
we establish session fidelity specifically for processes featuring guarded recursion 
and implementing a single multiparty session, realised as a parallel composition of one sub-process per role. 
The guarantees for process properties given in~\Cref{sec:typed-process-property} are demonstrated under the same guarded-recursion, single-session assumption. 

Note that alternative approaches exist for handling interleaved multiparty sessions with delegation~\cite{DBLP:conf/coordination/PadovaniVV14,CDYP2015,DBLP:journals/scp/HeuvelP22,DBLP:conf/ecoop/CicconeDP22,DBLP:journals/jlap/CicconeDP24}. For example,  a decentralised analysis is proposed in~\cite{DBLP:journals/scp/HeuvelP22} to ensure session fidelity and deadlock freedom for implementations with interleaved and delegated multiparty sessions. In addition, the type system introduced in~\cite{DBLP:conf/ecoop/CicconeDP22,DBLP:journals/jlap/CicconeDP24} guarantees termination of well-typed processes under fairness assumptions in the presence of session chaining, nesting, interleaving, delegation, and dynamic session creation. 

The formalisation of our session fidelity, as provided in \Cref{lem:session-fidelity,thm:session-fidelity} below, builds upon the concepts introduced in \Cref{def:unique-role-proc}.

\begin{definition}[from~\citep{POPL19LessIsMore}]
  \label{def:guarded-definitions}
  \label{def:unique-role-proc}
  Assume $\stJudge{\mpEnvEmpty}{\stEnv}{\mpP}$. We say that $\mpP$:
\begin{enumerate}[label={(\arabic*)}, leftmargin=*, nosep]
  \item\label{item:guarded-definitions:stmt}%
    \emph{has guarded definitions}
    if and only if 
    in each process definition in $\mpP$ of the form
    \linebreak
    $\mpDef{\mpX}{
      \mpEnvMap{x_1}{\tyGround[1]}, \ldots, 
      \mpEnvMap{x_n}{\tyGround[n]}, \ldots, 
      \stEnvMap{y_1}{\stT[1]},...,\stEnvMap{y_m}{\stT[m]}}
      {\mpQ}{\mpPi}$,
    for all $i \in 1...m$, 
    a call
    $\mpCall{\mpY}{...,y_i,...}$
    can only occur in $\mpQ$
    as a subterm of 
    $\mpBranch{y_i}{\roleQ}{j \in J}{\mpLab[j]}{z_j}{\mpP[j]}{}$
     or 
    $\mpSel{y_i}{\roleQ}{\mpLab}{\mpD}{\mpPii}{}$
    (\ie after using $y_i$ for input or output);
  \item\label{item:unique-role-proc:stmt}%
    \emph{only plays role $\roleP$ in $\mpS$, by $\stEnv$} if and only if:
    \!\upshape (i)
    $\mpP$ has guarded definitions;\; %
    \!\upshape (ii)
    $\fev{\mpP} \!=\! \emptyset$;\;
    \!\upshape (iii)
    $\fcv{\mpP} \!=\! \emptyset$;\;
    \!\upshape (iv) 
      $\stEnv \!=\!
      \stEnv[0] \stEnvComp \stEnvMap{\mpChanRole{\mpS}{\roleP}}{\stT}$
      with
      $\stT \!\stNotSub\! \stEnd$
      and
      $\stEnvEndP{\stEnv[0]}$; 
    \!\upshape (v) 
      for all subterms
      $\mpRes{\stEnvMap{\mpSi}{\stEnvi}}{\mpPi}$
      in $\mpP$,
      $\stEnvEndP{\stEnvi}$.
  \end{enumerate}
  We say ``\emph{$\mpP$ only plays role $\roleP$ in $\mpS$}''
  iff 
  $\exists\stEnv: \stJudge{\mpEnvEmpty}{\stEnv}{\mpP}$,
  and item~\ref{item:unique-role-proc:stmt} holds.
\end{definition}

In~\Cref{def:guarded-definitions}, item~\ref{item:guarded-definitions:stmt} 
formalises guarded recursion for processes, while 
item~\ref{item:unique-role-proc:stmt} %
defines  a process that plays
exactly \emph{one} role on \emph{one} session. 
It is evident that 
an ensemble of such processes
cannot deadlock
by waiting for each other across multiple sessions. 

\begin{example}[Playing Only Role]
\label{ex:play_role}
Consider the processes $\mpP, \mpQ$, and the typing context 
$\stEnv[H]$ from~\Cref{ex:typed_proc}. 
Observe that $\mpP$ does not only play either $\roleP$ or  $\roleR$ in $\mpS$. 
This is because $\mpP$ can only be typed 
by a context of the form 
$\stEnvMap{\mpChanRole{\mpS}{\roleP}}{\stT[\roleP]} 
\stEnvComp 
\stEnvMap{\mpChanRole{\mpS}{\roleR}}{\stT[\roleR]}$, where 
neither $\stEnvEndP{\stEnvMap{\mpChanRole{\mpS}{\roleP}}{\stT[\roleP]}}$ nor 
$\stEnvEndP{\stEnvMap{\mpChanRole{\mpS}{\roleR}}{\stT[\roleR]}}$ holds, 
thus violating item~\ref{item:unique-role-proc:stmt} of~\Cref{def:unique-role-proc}. 
Conversely,  $\mpQ$ only plays role $\roleQ$ in $\mpS$, 
by $\stEnvMap{\mpChanRole{\mpS}{\roleQ}}{\stT[H_{\roleQ}]}$, 
with all required conditions satisfied.

Additionally,  the processes $\mpP[\roleFmt{s}]$, $\mpP[\roleFmt{c}]$, 
and $\mpP[\roleFmt{a}]$ from~\Cref{ex:typed_process_auth} 
only play roles $\roleFmt{s}$, $\roleFmt{c}$, and $\roleFmt{a}$, respectively,  
in $\mpS$, which can be easily verified. %
\end{example}

We now formalise our session fidelity properties~(\Cref{lem:session-fidelity,thm:session-fidelity}).
Similar to subject reduction in~\Cref{sec:subject_reduction}, \Cref{lem:session-fidelity} utilises a typing context associated with a global type for a specific session $\mpS$ to type the process, while~\Cref{thm:session-fidelity}, as a corollary, applies a typing context projected from a global type,  with both asserting that a process constructed from a global type preserves its structure after reductions.

\begin{restatable}[Session Fidelity via Association]{theorem}{lemSessionFidelity}
  \label{lem:session-fidelity}%
   \label{lem:session-fidelity-association}
  Assume $\stJudge{\mpEnvEmpty\!}{\!\stEnv}{\!\mpP}$, %
  with %
 $\stEnvAssoc{\gtG}{\stEnv}{\mpS}$, 
$\mpP \equiv \mpRes{\mpS[1]}{\ldots \mpRes{\mpS[n]}{\mpBigPar{\roleP \in I}{\mpP[\roleP]}}}$, %
  and $\stEnv = \bigcup_{\roleP \in I}\stEnv[\roleP]$ %
  such that for each \,$\mpP[\roleP]$: %
  (1) $\stJudge{\mpEnvEmpty\!}{\stEnv[\roleP]}{\!\mpP[\roleP]}$, 
  and
  (2) either $\mpP[\roleP] \equiv \mpNil$, %
  or $\mpP[\roleP]$ only plays $\roleP$ in $\mpS$, by $\stEnv[\roleP]$. %
  Then, 
    $\stEnv \!\stEnvMoveWithSession[\mpS]$ implies 
    $\exists \stEnvi, \gtGi, \mpPi$ %
    such that 
    $\stEnv \!\stEnvMoveWithSession[\mpS]\! \stEnvi$, 
    $\gtG \,\gtMove\, \gtGi$, 
    $\mpP \mpMoveStar\! \mpPi$,  %
    and 
    $\stJudge{\mpEnvEmpty\!}{\!\stEnvi}{\mpPi}$, %
    with 
    $\stEnvAssoc{\gtGi}{\stEnvi}{\mpS}$, %
   $\mpPi \equiv \mpRes{\mpSi[1]}{\ldots \mpRes{\mpSi[m]}{\mpBigPar{\roleP \in I}{\mpPi[\roleP]}}}$, %
    and $\stEnvi = \bigcup_{\roleP \in I}\stEnvi[\roleP]$ %
    such that for each $\mpPi[\roleP]$:
    (1) $\stJudge{\mpEnvEmpty\!}{\stEnvi[\roleP]}{\!\mpPi[\roleP]}$, %
    and
    (2) either $\mpPi[\roleP] \equiv \mpNil$, 
    or $\mpPi[\roleP]$ only plays $\roleP$ in $\mpS$, by~$\stEnvi[\roleP]$. %
\end{restatable}
\begin{proof}
By induction on the derivation of $\stEnv \!\stEnvMoveWithSession[\mpS]$~(\Cref{def:typing_context_reduction}). 
\qedhere 
\end{proof}

\begin{restatable}[Session Fidelity]{theorem}{thmSessionFidelity}
\label{lem:final-session-fidelity}
\label{thm:session-fidelity}
 Assume $\stJudge{\mpEnvEmpty\!}{\!\setenum{\stEnvMap{\mpChanRole{\mpS}{\roleP}}{\gtProj{\gtG}{\roleP}}}_{%
   \roleP \in \mpFmt{I}}}{\!\mpP}$, %
  with $\gtRoles{\gtG} \subseteq \mpFmt{I}$ and %
 $\mpP \equiv$ \linebreak 
 $ \mpRes{\mpS[1]}{\ldots \mpRes{\mpS[n]}{\mpBigPar{\roleP \in I}{\mpP[\roleP]}}}$ %
  such that for each \,$\mpP[\roleP]$: %
  (1) $\stJudge{\mpEnvEmpty\!}{\stEnvMap{\mpChanRole{\mpS}{\roleP}}{\gtProj{\gtG}{\roleP}}}{\!\mpP[\roleP]}$, 
  and
  (2) either $\mpP[\roleP] \equiv \mpNil$, %
  or $\mpP[\roleP]$ only plays $\roleP$ in $\mpS$, by $\stEnvMap{\mpChanRole{\mpS}{\roleP}}{\gtProj{\gtG}{\roleP}}$. %
  Then, 
    $\gtG \,\gtMove$ implies 
    $\exists \gtGi, \mpPi$ %
    such that 
    $\gtG \,\gtMove\, \gtGi$, 
    $\mpP \mpMoveStar\! \mpPi$,  %
    and 
    $\stJudge{\mpEnvEmpty\!}{\!\setenum{\stEnvMap{\mpChanRole{\mpS}{\roleP}}{\gtProj{\gtGi}{\roleP}}}_{%
   \roleP \in \mpFmt{I}}}{\mpPi}$, %
    with $\gtRoles{\gtGi} \subseteq I$ and 
    $\mpPi \equiv \mpRes{\mpSi[1]}{\ldots \mpRes{\mpSi[m]}{\mpBigPar{\roleP \in I}{\mpPi[\roleP]}}}$ %
    such that for each $\mpPi[\roleP]$:
    (1) $\stJudge{\mpEnvEmpty\!}{\stEnvMap{\mpChanRole{\mpS}{\roleP}}{\gtProj{\gtGi}{\roleP}}}{\!\mpPi[\roleP]}$, %
    and
    (2) either $\mpPi[\roleP] \equiv \mpNil$, 
    or $\mpPi[\roleP]$ only plays $\roleP$ in $\mpS$, by $\stEnvMap{\mpChanRole{\mpS}{\roleP}}{\gtProj{\gtGi}{\roleP}}$. %

\end{restatable}
\begin{proof}
By the definition of association~(\Cref{def:assoc}), the soundness of association~(\Cref{thm:gtype:proj-sound}), and session fidelity via association~(\Cref{lem:session-fidelity}). 
\qedhere 
\end{proof}

\begin{example}[Guarded Definitions in Session Fidelity, from~\citep{POPL19LessIsMore}] 
\label{ex:guarded_def} 
According to rule \inferrule{\iruleMPDef} in~\Cref{fig:mpst-rules}, %
an unguarded definition  %
$\mpJustDef{\mpX}{\stEnvMap{y}{\stT}}{\mpCall{\mpX}{y}}$ %
can be typed with \emph{any} $\stT$. %
Therefore, we have: %

\smallskip%
\centerline{\(%
  \stJudge{\mpEnvEmpty}{
      \stEnvMap{\mpChanRole{\mpS}{\roleP}}{\roleQ \stFmt{\oplus}{\stLabFmt{m}}}
      \stEnvComp %
     \stEnvMap{\mpChanRole{\mpS}{\roleQ}}{\roleP \stFmt{\&}{\stLabFmt{m}}}
 }
{%
    \mpDefAbbrev{%
      \mpJustDef{\mpX}{\stEnvMap{y}{\stOut{\roleQ}{\stLab}{}}}{%
        \mpCall{\mpX}{y}%
      }%
    }{%
      \mpCall{X}{\mpChanRole{\mpS}{\roleP}}
      \mpPar%
      \mpBranchSingle{\mpChanRole{\mpS}{\roleQ}}{\roleP}{\mpLab}{}{}
  }
  }
\)}

\smallskip
\noindent%
The unguarded process above reduces vacuously by infinitely invoking $\mpX$, %
without aligning with any typing context reduction. %
This highlights the necessity of guarded definitions for session fidelity~(\Cref{lem:session-fidelity,thm:session-fidelity}). 
\end{example}

Observe that the process $\mpP$ from~\Cref{ex:typed_proc} 
does not satisfy \ref{item:unique-role-proc:stmt} of~\Cref{def:unique-role-proc}, 
while $\mpQ$ does, as shown in~\Cref{ex:play_role}. 
However, as demonstrated in~\Cref{ex:typed_proc},  $\mpP \,\mpPar\, \mpQ$ adheres to session fidelity. 
Readers might question the necessity of enforcing the strict condition in session fidelity that requires a process to play exactly one role in a single session. The following example highlights its importance.

\begin{example}[Playing Only Role in Session Fidelity]
\label{ex:play_role_session_fidelity}
Take the processes $\mpP[z] = \mpSel{\mpChanRole{\mpS}{\roleFmt{A}}}{\roleFmt{B}}{\mpLab}{}{\mpSel{\mpChanRole{\mpS}{\roleFmt{C}}}{\roleFmt{D}}{\mpLab}{}{}}$, $\mpQ[z] = 
\mpBranchSingle{\mpChanRole{\mpS}{\roleFmt{D}}}{\roleFmt{C}}{\mpLab}{}{\mpBranchSingle{\mpChanRole{\mpS}{\roleFmt{B}}}{\roleFmt{A}}{\mpLab}{}{}}$, 
and the typing context 
$ \stEnv[z] = \stEnvMap{\mpChanRole{\mpS}{\roleFmt{A}}}{\roleFmt{B} \stFmt{\oplus}{\stLabFmt{m}}}
      \stEnvComp %
     \stEnvMap{\mpChanRole{\mpS}{\roleFmt{B}}}{\roleFmt{A} \stFmt{\&}{\stLabFmt{m}}}
     \stEnvComp 
     \stEnvMap{\mpChanRole{\mpS}{\roleFmt{C}}}{\roleFmt{D} \stFmt{\oplus}{\stLabFmt{m}}}
      \stEnvComp %
     \stEnvMap{\mpChanRole{\mpS}{\roleFmt{D}}}{\roleFmt{C} \stFmt{\&}{\stLabFmt{m}}}
     $, which can be obtained by projecting the global type $\gtG[z] = \gtCommSingle{\roleFmt{A}}{\roleFmt{B}}{\gtLab}{}{\gtCommSingle{\roleFmt{C}}{\roleFmt{D}}{\gtLab}{}{}}$ onto $\mpS$. 
  It is trivial to verify that neither $\mpP[z]$ nor  
     $\mpQ[z]$ satisfies~\ref{item:unique-role-proc:stmt} of~\Cref{def:unique-role-proc}. 
     Moreover, it is evident that  $\mpP[z] \,\mpPar\, \mpQ[z]$ does not satisfy session fidelity, as the process cannot reduce further, whereas $\stEnv[z]$ remains reducible.      
     \end{example}

\subsection{Properties of Typed Processes}
\label{sec:typed-process-property}

We finalise this section by showcasing that processes constructed from global types guarantee desirable run-time properties,
including \emph{deadlock-freedom} and \emph{liveness}, as formalised in~\Cref{def:proc-properties}.  
Deadlock-freedom indicates that if a process cannot reduce, it consists only of inactive sub-processes~($\mpNil$), while 
liveness (\aka ``lock-freedom"~\cite{KobayashiS10Hybrid,Padovani14ICALP}) ensures that every pending input or output action of a process can eventually engage in communication, \ie  there exists a future execution in which it synchronises with the corresponding dual.

\begin{definition}[Runtime Process Properties]
  \label{def:proc-properties}%
  \label{def:proc-deadlock-free}
  \label{def:proc-term}%
  \label{def:proc-nterm}%
  \label{def:proc-df}%
  \label{def:proc-liveness}%
  We say $\mpP$ is:
  \begin{enumerate}[label={(\arabic*)},leftmargin=*,nosep]
  \item
  \label{item:proc-properties:df}
  \emph{deadlock-free} %
    iff %
    $\mpP \!\mpMoveStar\! \mpNotMoveP{\mpPi}$
    implies $\mpPi \equiv \mpNil$; 
        \item
    \label{item:proc-properties:live}
    \emph{live} %
    iff $\mpP \!\mpMoveStar\! \mpPi \!\equiv\! \mpCtxApp{\mpCtx}{\mpQ}$ 
    implies:
    \begin{enumerate}[leftmargin=*,nosep]
    \item\label{item:process-liveness:sel}%
      if  
      $\mpQ = \mpSel{\mpChanRole{\mpS}{\roleP}}{\roleQ}{\mpLab}{\mpW}{\mpQi}$,  
      then %
      there exist $\mpCtxi, I, {\setenum{\mpLab[i]}}_{i \in I}, {\setenum{z_i}}_{i \in I}, {\setenum{\mpR[i]}}_{i \in I}$ such that $\mpLab \in \setenum{\mpLab[i]}_{i \in I}$ and 
      $\mpCtx \mpCtxMoveStar 
      \mpCtxApp{\mpCtxi}{\mpCtxHole \mpPar  \mpBranch{\mpChanRole{\mpS}{\roleQ}}{\roleP}{i \in I}{\mpLab[i]}{z_i}{\mpR[i]}{}}$; 
    \item\label{item:process-liveness:branch}
      if 
      $\mpQ = \mpBranch{\mpChanRole{\mpS}{\roleP}}{\roleQ}{i \in I}{\mpLab[i]}{z_i}{\mpQi[i]}{}$,   
      then there exist $\mpCtxi, k \in I, \mpW, \mpR$ such that 
      $\mpCtx \mpCtxMoveStar \mpCtxApp{\mpCtxi}{\mpCtxHole \mpPar 
      \mpSel{\mpChanRole{\mpS}{\roleQ}}{\roleP}{\mpLab[k]}{\mpW}{\mpR}}$.
    \end{enumerate}
  \end{enumerate}
\end{definition}

\begin{remark}[Correction of Process Liveness in~\citet{POPL19LessIsMore}] 
 In~\cite[Def. 5.1]{POPL19LessIsMore}, a process $\mpP$ is defined to be live iff for every $\mpPi$ such that 
 $\mpP \!\mpMoveStar\! \mpPi \!\equiv\! \mpCtxApp{\mpCtx}{\mpQ}$, 
    \begin{itemize}[nosep]
    \item%
      if  
      $\mpQ = \mpSel{\mpC}{\roleQ}{\mpLab}{\mpW}{\mpQi}$,  
      then %
      $\exists \mpCtxi: \mpPi \mpMoveStar \mpCtxApp{\mpCtxi}{\mpQi}$;  %
    \item%
      if 
      $\mpQ = \mpBranch{\mpC}{\roleQ}{i \in I}{\mpLab[i]}{x_i}{\mpQi[i]}{}$,   
      then 
      $\exists \mpCtxi, k \in I, \mpU: \mpPi \mpMoveStar \mpCtxApp{\mpCtxi}{\mpQi[k]\subst{x_k}{\mpU}}$.
    \end{itemize}

\smallskip
\noindent
In this formulation, the witness context $\mpCtxi$ is unrestricted and need not be related to the original context
$\mpCtx$. 
Consequently,  the liveness condition may be satisfied by selecting an arbitrary context that 
enables the desired reduction.

For example, consider the process $ \mpP =  \mpSel{\mpChanRole{\mpS}{\roleP}}{\roleQ}{\mpLab}{5}{\mpNil}$,  
which is \emph{not} live since no matching input action exists. 
Under the above definition, however, %
one may take $\mpCtxi = \mpCtxHole \mpPar \mpP$ to obtain 
$\mpP \mpMoveStar \mpNil \mpPar \mpP = \mpCtxApp{\mpCtxi}{\mpNil}$, making the liveness notion hold vacuously, despite the absence of any communication.   %

\Cref{def:proc-properties} addresses this issue by requiring $\mpCtx \mpCtxMoveStar \mpCtxi$, thereby enforcing a structural relationship between $\mpCtxi$ and $\mpCtx$. 
This constraint ensures that, under this  liveness formalisation, the execution of an output or input action 
arises from reductions of the process itself, rather than from an arbitrarily extended environment.

Note that the process liveness property adopted in existing work (\eg \cite{DBLP:conf/ecoop/LagaillardieNY22,BSYZ2022,DBLP:conf/esop/BrunD23}) relies on 
the definition in~\cite{POPL19LessIsMore}, and may therefore require corresponding revision. 
\hfill $\blacktriangleleft$
\end{remark}

Finally, we illustrate how a process, typed with a typing context obtained from a global type via projection, ensures both deadlock-freedom and liveness.

\begin{restatable}[Process Deadlock-Freedom, Liveness]{theorem}{lemProcessPropertiesVerif}%
  \label{lem:stenv-proc-properties}
  \label{lem:deadlock-freedom}%
 Assume $\stJudge{\mpEnvEmpty\!}{\!\setenum{\stEnvMap{\mpChanRole{\mpS}{\roleP}}{\gtProj{\gtG}{\roleP}}}_{%
   \roleP \in \gtRoles{\gtG}}}{\!\mpP}$, %
  where 
  $\mpP \equiv \mpBigPar{\roleP \in  \gtRoles{\gtG}}{\mpP[\roleP]}$ %
  and for each $\mpP[\roleP]$, $\stJudge{\mpEnvEmpty\!}{\stEnvMap{\mpChanRole{\mpS}{\roleP}}{\gtProj{\gtG}{\roleP}}}{\!\mpP[\roleP]}$.
  Further, assume that each $\mpP[\roleP]$
  is either $\mpNil$ (up to $\equiv$), %
  or only plays $\roleP$ in $\mpS$, by $\stEnvMap{\mpChanRole{\mpS}{\roleP}}{\gtProj{\gtG}{\roleP}}$. %
  Then, $\mpP$ is deadlock-free and live. 
\end{restatable}

\begin{example}[Typed Process Properties]
\label{ex:proc_properties}
The process $\mpP[\roleFmt{s}] \mpPar \mpP[\roleFmt{c}] \mpPar \mpP[\roleFmt{a}]$ 
from~\Cref{ex:typed_process_auth} %
is both deadlock-free and live, as can be easily verified by applying either~\Cref{def:proc-properties} or~\Cref{lem:stenv-proc-properties}.
\end{example}

\section{Related Work}%
\label{sec:related}

In this section, we explore related work on multiparty session types (MPST). 
We begin by discussing top-down frameworks, with emphasis on their proof methodologies for type soundness (subject reduction), while also including work on the projectability and implementability of global types.  
We then cover developments based on the bottom-up approach. Additionally, we discuss mechanisations of MPST, focusing on projection-based settings. This overview is intended to trace the development of these approaches rather than to provide an exhaustive survey, and therefore not all related work is included.

\subsection{Top-Down Multiparty Session Types}
\label{sec:related:topdown}

The classic top-down MPST framework, with its notions of 
\emph{global types} and \emph{projections}, 
was first introduced in~\citet{HYC08} and fully developed in~\citet{HYC16}, %
where \emph{linearity conditions} 
were incorporated to ensure the well-formedness of global types and the \emph{projectability} of local types. %
Subsequently, \citet{BCDDDY08} proposed %
a simplified MPST system without type-level channel declarations, 
which has since  been widely adopted in most works, both theoretical and practical, 
including ours. Later, \citet{CarboneMSY15,CarboneLMSW16} and \citet{Caires2016Binary} investigated the logical foundations of MPST under restricted classes of global types.

We now classify some related works according to  their treatment of \emph{projection}, \emph{consistency}, and \emph{association}, as well as the correctness of their proofs for the subject reduction theorem.

\newcounter{CounterRelatedTaxonomy}%
\renewcommand{\theCounterRelatedTaxonomy}{\alph{CounterRelatedTaxonomy}}%
\newcommand{\RelTax}[1]{%
  \refstepcounter{CounterRelatedTaxonomy}\label{#1}%
  \textbf{(\theCounterRelatedTaxonomy)}%
}%
\newcommand{\refRelTax}[1]{\textbf{(\ref{#1})}}%
\begin{center}
  \small%
  \setlength\tabcolsep{1mm}%
\begin{tabular}{@{\hskip 0mm}c|c|c|c|c@{\hskip 0mm}}
   \textbf{Papers} &%
  \textbf{Projection} & \textbf{Consistency} & \textbf{Association} & \textbf{Subject Reduction}%
  \\[1mm]%
  \RelTax{item:related-plain}%
  \begin{minipage}{0.2\linewidth}
    ~\cite{HYC08,HYC16,BCDDDY08,CarboneMSY15,CarboneLMSW16,Coppo2015GentleIntroMAPST}%
  \end{minipage}
  &%
  $\leq$ plain & yes & no & correct %
  \\[1mm]
  \hline%
  \RelTax{item:related-full}%
  \begin{minipage}{0.2\linewidth}
   ~\cite{ParameterisedYDBH10,ParameterisedYDBH12,Denielou2012,CHEN2015708,TY2016}%
  \end{minipage}
  &%
  $\geq$ full & no &no & flawed %
  \\[1mm]
  \hline%
  \RelTax{item:related-full-consistent}%
  \begin{minipage}{0.2\linewidth}
    ~\cite{SDHY2017,TY2017}%
  \end{minipage}
  &%
  full & yes (required) & no & correct %
  \\[1mm]
  \hline%
  \RelTax{item:related-association}%
  \begin{minipage}{0.2\linewidth}
    ~\cite{BHYZ2023,lmcs2025,DBLP:series/lncs/YoshidaH24,DBLP:conf/ecoop/HouLY24}%
  \end{minipage}
  &%
  full & no & yes & correct 
  \\[1mm]
  \hline%
  \RelTax{item:related-other}%
  \begin{minipage}{0.2\linewidth}
    ~\cite{DGJPY15,GJPSY2019,Caires2016Binary,DBLP:conf/ecoop/LagaillardieNY22}%
  \end{minipage}
  &%
  full & no & no & correct 
\end{tabular}
\end{center}

\textbf{Row~\refRelTax{item:related-plain}} %
lists works that use \emph{plain} (or stricter) %
global type projection, 
which guarantees \emph{consistency}. 
Consequently, their proofs of type soundness, \ie subject reduction,  are correct, as they rely on consistency as the invariant $\varphi$ to be satisfied in \textbf{(SR3)} in~\cref{sec:intro}. 
However, this plain projection is overly restrictive, excluding many valid protocols -- even a simple one 
such as $\gtG[\text{f}]$ in~\cref{sec:intro} -- at the cost of guaranteeing consistency.

\textbf{Row~\refRelTax{item:related-full}} %
lists works that adopt \emph{full} (or more flexible) %
global type projection, %
originally introduced in~\citet{ParameterisedYDBH10} %
to support a wider range of protocols. %
These works overlook the consistency requirement,  %
and, as shown in \Cref{sec:intro}, subject reduction proofs %
that rely on full projection (without consistency) %
are flawed. %

To ``fix''  these issues within MPST theory, 
the works in \textbf{row \refRelTax{item:related-full-consistent}}   
enforce consistency while keeping full projection, which restricts typability and thus falls back into over-restriction.
In contrast, our proposed proof technique, based on the \emph{association} relation under full projection, 
establishes subject reduction within the top-down MPST framework without losing expressivity. 

\textbf{Row~\refRelTax{item:related-association}} specifies works that 
employ \emph{association} relations similar to the one used in our work, thereby 
guaranteeing subject reduction. 
The notion of association was first  proposed in~\citet{BHYZ2023} and further developed in~\citet{lmcs2025} within an MPST framework with crash-stop failures, 
and later extended in~\citet{DBLP:conf/ecoop/HouLY24} to incorporate time and failure handling. 
In these settings, the association relation is used to establish a sound and complete operational correspondence between global and local semantics, ensuring that key global type properties are preserved in local types through projection. 
Our work, however, develops this idea into the foundation of a new invariant-based proof approach to subject reduction.  
Moreover, the framework of~\cite{BHYZ2023,lmcs2025} is limited to a \emph{single-session} type system with 
first-order session types (\ie without channel passing), and both~\cite{BHYZ2023,lmcs2025} and~\cite{DBLP:conf/ecoop/HouLY24} assume asynchronous communication and adopt a channel-oriented subtyping discipline. As a result, their association relations differ substantially from ours, which is formulated for a synchronous model with a distinct subtyping direction and different applicability.

Additionally, there are several MPST works that fall within the top-down framework but are \emph{neither} based on the classic projection+consistency approach \emph{nor}  employ the association relation as in our work, as shown in \textbf{row~\refRelTax{item:related-other}}. %
\citet{DGJPY15}, with a journal version in~\citet{GJPSY2019}, introduce a  single-session type system without channel passing, similar to that of~\citet{BHYZ2023,lmcs2025}. 
Rooted in global types and their projections, it does \emph{not} require consistency. Such a system is  strictly subsumed by our framework, which, in contrast, supports higher-order types and multiple interleaved sessions. Moreover, their subject reduction proof strategy proceeds by reasoning directly on global types and their semantics, whereas ours relies on the operational correspondence between global types and typing contexts via association, making explicit use of typing contexts in the proof.

\citet{Caires2016Binary} develop %
a theory of multiparty session types %
encoded into binary sessions, with a type system based on linear logic from~\citet{CairesP10} and~\citet{Wadler12}. %
A related multiparty-to-binary session decomposition %
was later studied in~\citet{SDHY2017},  with a crucial difference: 
in \cite{SDHY2017}, consistency is a \emph{necessary} requirement~(formalised in their Theorem 6.3), %
whereas in \cite{Caires2016Binary} it is not, despite supporting full projection and merging.  
This distinction arises because 
the decomposition in \cite{Caires2016Binary} %
introduces a centralised \emph{medium process} (similar to the \emph{arbiter} in \citet{CarboneLMSW16}) %
that receives and forwards all messages %
between processes playing different roles, whereas the decomposition in \cite{SDHY2017} %
preserves the peer-to-peer nature of MPST interactions. This suggests that, when decomposing multiparty choreographies into linear binary interactions, consistency is necessary if and only if no centralised medium process is introduced. Our present work supports binary sessions as a special case, without requiring either consistency or medium processes.

Finally, \citet{DBLP:conf/ecoop/LagaillardieNY22} present an MPST framework with affine communication channels and implicit/explicit cancellation mechanisms, where subject reduction is formalised in the style of~\citet{POPL19LessIsMore}, with typing-context safety as the key invariant. 
Since this property is guaranteed by projection, it suffices to establish correctness. Our approach, by contrast, adopts a different invariant -- association -- and explicitly incorporates both global types and projection into the theorem, aligning more closely with the top-down formulation.  

\paragraph{Projectability and Implementability of Global Types} 
A complementary line of top-down work focuses on the projectability and implementability of global types, addressing the existence and behavioural correctness of projections. 
In this setting, the implementability problem asks whether there exist local specifications for all roles such that
their composition is deadlock-free and generates exactly the executions specified by the global type.

\citet{DBLP:conf/ecoop/Stutz23} establishes decidability of implementability for a class of MPSTs via a reduction to safe realisability of globally cooperative high-level message sequence charts. 
\citet{DBLP:conf/cav/LiSWZ23} introduce an automata-theoretic projection framework that separates synthesis from implementability checking, yielding a practical sound and complete projection operator for general MPSTs with a PSPACE decision procedure. 
\citet{DBLP:conf/esop/StutzD25} further generalise this approach by extending the projection algorithm of~\cite{DBLP:conf/cav/LiSWZ23} to Protocol State Machines, a highly expressive formalism for global protocol specifications, and deriving Communicating State Machines as local specifications.

In contrast, the works discussed earlier in this section that employ projection -- including our own -- do not focus on completeness or decidability of projection. Instead, these approaches typically rely on syntactic projection operators  and study how type soundness can be ensured in the presence of a fixed projection.

\subsection{Bottom-Up Multiparty Session Types}

The first attempt to develop a theory of MPST based on a bottom-up method, rather than the top-down approach with global types and projections, was presented in~\citet{ScalasY18}. Subsequently, \citet{POPL19LessIsMore} introduced a general MPST typing system that does not rely on global types: it ensures desired process properties by checking the corresponding properties of typing contexts. 
These properties are specified in modal $\mu$-calculus formulae and verified using the mCRL2 checker from~\citet{DBLP:conf/tacas/BunteGKLNVWWW19}. 
Unlike the top-down approach, it is not constrained by the projectability or implementability of global types, thereby capturing the full class of well-behaved processes while preserving type soundness. 
However, \citet{DBLP:journals/pacmpl/UdomsrirungruangY25} observe that this gain in typability comes at a higher computational cost. 
In asynchronous MPST with unbounded FIFO queues, the choice of methodology becomes even more crucial: type checking is \emph{undecidable} under the bottom-up approach, as shown by~\citet{POPL19LessIsMore}, but remains decidable in the top-down approach we adopt,  thanks to the decidability of end-point projection, as presented in~\citet{HYC08,HYC16} and \citet{POPL19LessIsMore}.

Over time, this line of work has also inspired a number of extensions.  For example,  \citet{DBLP:conf/ecoop/00020DG21} adapt the bottom-up approach to model actor systems with explicit connection types from~\citet{HY2017}; \citet{BSYZ2022} apply it to account for crash-stop failures;  
\citet{DBLP:conf/esop/BrunD23} broaden it  to address a wider range of fault models;  \citet{DBLP:conf/esop/BrunFD25} extend it into a session-typed multiparty process calculus with replication and first-class roles; 
and~\citet{DBLP:conf/esop/GiuntiY25} develop it further with iso-recursion. 

\subsection{Mechanisations in Multiparty Session Types}

Due to the complexity of MPST theories, numerous mechanisations have been developed across different frameworks to ensure correctness rigorously. To focus the discussion, we restrict attention to works that
employ mechanised proofs in MPST frameworks involving projection.

\citet{DBLP:conf/pldi/Castro-Perez0GY21} introduce Zooid, a domain-specific language embedded in Rocq for certified multiparty communication.  Zooid formalises the syntaxes of global and local types inductively, together with a coinductive tree-based representation supporting a semantic interpretation. 
By defining an unravelling relation from types to trees and a projection operation from global to local specifications, they mechanise the result that projection is preserved under unravelling. 
Moreover, Zooid provides mechanisation proofs of sound and complete correspondence between the labelled transition systems of global and local types, in terms of execution-trace equivalence, thereby ensuring properties 
such as deadlock-freedom and protocol compliance. Additionally, the syntax and semantics of processes, together with a typing system, are formalised, and type preservation is mechanised.

\citet{DBLP:conf/itp/TiroreBC23}, with a journal version in~\citet{DBLP:journals/jar/TiroreBC25}, 
define a computable projection function from global types to local types and provide a Rocq mechanisation of its 
soundness and completeness with respect to a coinductive tree semantics. 
Their work addresses limitations of existing computable projections by introducing a function that is equally expressive as its coinductive, non-computable counterpart, while remaining decidable. However, the mechanisation focuses on projection, independently of any process or typing calculus; consequently, properties such as subject reduction, progress, and type safety for processes are not addressed.

\citet{DBLP:journals/pacmpl/JacobsBK22a} extend earlier work on a compiler for a functional language with binary session types by~\citet{DBLP:conf/esop/TassarottiJ017}, based on a simplified variant of the GV (``Good Variation'') system by~\citet{DBLP:journals/jfp/GayV10}, to MPGV, which enriches a linear lambda calculus with multiparty sessions and supports participant redirection and dynamic thread spawning. Their type system incorporates both global and local types, with projection operations similar to those of~\citet{DBLP:conf/pldi/Castro-Perez0GY21}.  
All guarantees for well-typed processes in MPGV are formally mechanised in Rocq, including type safety, session fidelity, global progress, deadlock-freedom, and leak-freedom.

Additionally, several closely related mechanisation works address subject reduction. 
\citet{DawitThesis} provides a Rocq formalisation of subject reduction for the multiparty session $\pi$-calculus of~\citet{HYC08,HYC16}, including session initialisation and delegation. The type system is based on 
channel-explicit global and local types, with projections derived from~\citet{DBLP:conf/itp/TiroreBC23,DBLP:journals/jar/TiroreBC25}. Channel-explicit types require additional linearity checks to guarantee the projectability of global types.  
In follow-up work, \citet{DBLP:conf/ecoop/TiroreBC25} extend these results by formalising proofs of communication safety and safety preservation in Rocq.

\citet{DBLP:conf/itp/EkiciKY25} develop a Rocq mechanisation of subject reduction and progress for MPST using coinductive reasoning over type trees derived from global and local types. Their approach exploits structural properties of these trees to refine projection accuracy in the presence of the precise subtyping discipline introduced by ~\citet{GJPSY2019}, integrating subtyping into type checking and thereby extending the expressiveness of the type system.

All the mechanisations discussed above rely on plain (or stricter) merging. 

Finally, regarding the mechanisation of implementability of global protocols,  
\citet{DBLP:conf/itp/LiW25} present a Rocq formalisation of the precise implementability
characterisation by~\citet{DBLP:conf/cav/LiSWZ23}, discussed earlier in~\Cref{sec:related:topdown}.   
Their work unifies distinct frameworks, simplifies existing proof arguments, and makes explicit the construction of canonical implementations. The mechanisation further reveals a subtle issue in the semantics of infinite behaviours and shows that the implementability characterisation extends to protocols with infinitely many participants.

\section{Conclusion}
\label{sec:conclusion}
This paper addresses a recent concern in the multiparty session type community: namely, that the top-down approach with mergeability is unsound, or more strongly, 
that global types themselves are inherently problematic, and thus correctness-by-construction cannot be achieved from them. To clarify this, we introduce a new notion, \emph{association}, which relates global types to sets of local types, and apply it to develop a general proof technique for establishing type soundness in top-down MPST frameworks. With this method, we show that a sound typing system can indeed be obtained via endpoint projection with mergeability, thereby demonstrating that global types provide a clear and principled foundation for proving type soundness. Moreover, we establish that the top-down typing discipline, supported by endpoint projection, guarantees type safety, deadlock-freedom, and liveness of session processes by construction.

Future work includes applying the association-based proof technique to MPST with alternative subtyping disciplines and projection methods, as well as to extended variants such as those incorporating probabilistic behaviour or more advanced fault-tolerant models. These directions will further evaluate the robustness and generality of the association method, and provide a unified foundation for establishing type soundness across a broader spectrum of MPST frameworks. In addition, we plan to verify the correctness of our approach through mechanisation, particularly with respect to full merging.

\section*{Acknowledgments}
\noindent 
We thank the reviewers for their detailed and constructive comments and suggestions.  We are particularly grateful to the reviewers, as well as to Jake Masters and Joe Paulus, for prompting a correction to the definition of process liveness (Definition 4.2). 
This work was partially supported by EPSRC grants EP/T006544/2, EP/T014709/2, EP/Y005244/1,
 EP/V000462/1, EP/X015955/1, EP/Z0005801/1; 
 Horizon EU TaRDIS 101093006 (UKRI No.~10066667); Advanced Research and Invention Agency~(ARIA) Safeguarded AI;  and a grant from the Simons Foundation.

\bibliographystyle{elsarticle-harv} 

\bibliography{main,popl19}

\newpage 

\appendix

\renewcommand{\thesection}{\Alph{section}}

\section{Proofs for~\Cref{sec:sessiontypes}}
\label{ch:association_proofs}

\subsection{Roles}
\begin{definition}[Roles in Global Types]
\label{app:def:role_set}
The set of roles in a global type $\gtG$, denoted by $\gtRoles{\gtG}$, is defined inductively as: 

\smallskip
\centerline{\(
\begin{array}{c}
    \gtRoles{\gtComm{\roleP}{\roleQ}{i \in I}{\gtLab[i]}{\tyGround[i]}{\gtG[i]}
    } = 
     \setenum{\roleP, \roleQ} \cup \bigcup\limits_{i \in I}{\gtRoles{\gtG[i]}}
     \\[1em]
        \gtRoles{\gtEnd} = \gtRoles{\gtRecVar} = \emptyset
        \qquad 
        \gtRoles{\gtRec{\gtRecVar}{\gtG}}
     = 
     \gtRoles{\gtG\subst{\gtRecVar}{\gtRec{\gtRecVar}{\gtG}}}
  \end{array}
 \)}

\end{definition}

\subsection{Unfolding}

\begin{definition}[Type Unfolding]
\label{app:def:unfolding}
The \emph{unfolding} of a global type $\gtG$, written $\unfoldOne{\gtG}$,  is defined as: 

\smallskip
\centerline{\(
\unfoldOne{\gtRec{\gtRecVar}{\gtG}} = \gtG\subst{\gtRecVar}{\gtRec{\gtRecVar}{\gtG}} \qquad 
\unfoldOne{\gtG} = \gtG \quad \text{if } \gtG \neq \gtRec{\gtRecVar}{\gtGi}
\)}

\smallskip
\noindent
The \emph{unfolding} of a local type $\stT$, written $\unfoldOne{\stT}$, is defined similarly: 

\smallskip
\centerline{\(
\unfoldOne{\stRec{\stRecVar}{\stT}} = \stT\subst{\stRecVar}{\stRec{\stRecVar}{\stT}} \qquad 
\unfoldOne{\stT} = \stT \quad \text{if } \stT \neq \stRec{\stRecVar}{\stTi}
\)}
\end{definition}

\begin{proposition}
\label{app:lem:unfold-no-rec}
  For a closed, well-guarded global type $\gtG$, $\unfoldOne{\gtG}$ can only be of
  form $\gtEnd$, or $\gtComm{\roleP}{\roleQ}{}{\cdots}{}{}$. 
  For a closed, well-guarded local type $\stT$, $\unfoldOne{\stT}$ can only be
  of form $\stEnd$, $\stIntSum{\roleP}{}{\cdots}$, or
  $\stExtSum{\roleP}{}{\cdots}$.
\end{proposition}
\begin{proof}
  $\gtRecVar$ will not appear since we require closed types.
  $\gtRec{\gtRecVar}{\gtG'}\subst{\gtRecVar}{\gtRec{\gtRecVar}{\gtG'}} \neq
  \gtRec{\gtRecVar}{\gtG'}$ since we require well-guarded types (recursive
  types are contractive).
  Similar argument for local types. 
  \qedhere
\end{proof}

\begin{proposition}
\label{app:prop:roles_unfold}
$\gtRoles{\gtG} = \gtRoles{\unfoldOne{\gtG}}$. 
\end{proposition}
\begin{proof}
Follows directly from~\Cref{app:def:role_set,app:def:unfolding}. 
\qedhere
\end{proof}

\subsection{Subtyping}
\label{sec:app:subtyping}
\begin{lemma}[Subtyping is Reflexive]
\label{app:lem:subtyping:reflexive}
  For any closed, well-guarded local type $\stT$, $\stT \stSub \stT$ holds.
\end{lemma}
\begin{proof}
We construct a relation $R = \setenum{(\stT, \stT)}$.  
It is trivial to show that $R$ satisfies all clauses 
of~\Cref{def:subtyping}, and hence,  $R \subseteq \stSub$. 
\qedhere
\end{proof}

\begin{lemma}[Subtyping is Transitive]
	\label{app:lem:subtyping:transitive}
For any closed, well-guarded local types $\stT[1]$, $\stT[2]$, $\stT[3]$, 
if $\stT[1] \stSub \stT[2]$ and $\stT[2] \stSub \stT[3]$, then $\stT[1] \stSub \stT[3]$ holds. 
\end{lemma}
\begin{proof}
By constructing a relation 
$R = \setcomp{(\stT[1], \stT[3])}{\exists \stT[2] \text{ such that } \stT[1] \stSub \stT[2] \text{ and } \stT[2] \stSub \stT[3]}$, and showing that 
$R \subseteq \stSub$. 
\qedhere
\end{proof}

\begin{lemma}[Subtyping with Unfolding]
	\label{app:lem:subtyping:unfolding}
  For any closed, well-guarded local type $\stT$, 
  \begin{enumerate*}[label=(\arabic*)]
    \item $\unfoldOne{\stT} \stSub \stT$, and 
    \item $\stT \stSub \unfoldOne{\stT}$.
  \end{enumerate*}
\end{lemma}
\begin{proof}
  \begin{enumerate}[label=(\arabic*)]
    \item By $\inferrule{\iruleStSubRecR}$ if $\stT = \stRec{\stRecVar}{\stTi}$. Otherwise, 
	    by reflexivity~(\Cref{app:lem:subtyping:reflexive}).
    \item By $\inferrule{\iruleStSubRecL}$ if $\stT = \stRec{\stRecVar}{\stTi}$. Otherwise, 
	    by reflexivity~(\Cref{app:lem:subtyping:reflexive}).
	    \qedhere
  \end{enumerate}
\end{proof}

\begin{lemma}[Inversion of Subtyping]
	\label{app:lem:subtyping:inversion}
  ~
  \begin{enumerate}
    \item If
      $
       \stIntSum{\roleP}{i \in I}{\stChoice{\stLab[i]}{\stS[i]} \stSeq \stT[i]}
      \stSub \stU$, then
      $\unfoldOne{\stU} =
       \stIntSum{\roleP}{j \in J}{\stChoice{\stLabi[j]}{\stSi[j]} \stSeq
       \stTi[j]}
      $, and $I \subseteq J$,
      and $\forall i \in I: \stLab[i] = \stLabi[i], \stSi[i] \stSub 
      \stS[i]$ and $\stT[i] \stSub \stTi[i]$.
    \item If
      $
       \stExtSum{\roleP}{i \in I}{\stChoice{\stLab[i]}{\stS[i]} \stSeq \stT[i]}
      \stSub \stU$, then
      $\unfoldOne{\stU} =
       \stExtSum{\roleP}{j \in J}{\stChoice{\stLabi[j]}{\stSi[j]} \stSeq
       \stTi[j]}
      $, and $J \subseteq I$,
      and $\forall i \in J: \stLab[i] = \stLabi[i], \stS[i] \stSub
      \stSi[i]$ and $\stT[i] \stSub \stTi[i]$.
    \item If
      $\stU \stSub
       \stIntSum{\roleP}{i \in I}{\stChoice{\stLab[i]}{\stS[i]} \stSeq \stT[i]}
      $, then
      $\unfoldOne{\stU} =
       \stIntSum{\roleP}{j \in J}{\stChoice{\stLabi[j]}{\stSi[j]} \stSeq
       \stTi[j]}
      $, and $J \subseteq I$,
      and $\forall i \in I: \stLab[i] = \stLabi[i], \stS[i] \stSub 
      \stSi[i]$ and $\stTi[i] \stSub \stT[i]$.
    \item If
      $\stU \stSub
       \stExtSum{\roleP}{i \in I}{\stChoice{\stLab[i]}{\stS[i]} \stSeq \stT[i]}
      $, then
      $\unfoldOne{\stU} =
       \stExtSum{\roleP}{j \in J}{\stChoice{\stLabi[j]}{\stSi[j]} \stSeq
       \stTi[j]}
      $, and $I \subseteq J$,
      and $\forall i \in J: \stLab[i] = \stLabi[i], \stSi[i] \stSub
      \stS[i]$ and $\stTi[i] \stSub \stT[i]$.
  \end{enumerate}
\end{lemma}
\begin{proof}
	By~\Cref{app:lem:subtyping:unfolding}, the transitivity of subtyping (\Cref{app:lem:subtyping:transitive}), \Cref{app:lem:unfold-no-rec}, 
	and the definition of subtyping (\Cref{def:subtyping}), in particular rules \inferrule{\iruleStSubIn} and \inferrule{\iruleStSubOut}. 
	\qedhere
\end{proof}

\begin{lemma}
	\label{app:lem:subtyping:merging}
  Given a collection of mergeable local types $\stT[i]$ ($i \in I$).
  For all $j \in I$, $\stMerge{i \in I}{\stT[i]} \stSub \stT[j]$ holds.
\end{lemma}
\begin{proof}
By constructing a relation 
$R =  \setcomp{(\stMerge{i \in I}{\stT[i]}, \stT[j])}{j \in I}$, and showing that  
$R \subseteq \stSub$. 
\qedhere
\end{proof}

\begin{lemma}
\label{app:lem:subtyping:merging_bound}
  Given a collection of mergeable local types $\stT[i]$ ($i \in I$).
  If for all $i \in I$, $\stT[i] \stSub \stU$ for some local type $\stU$,
  then $\stMerge{i \in I}{\stT[i]} \stSub \stU$.
\end{lemma}
\begin{proof}
By constructing a relation 
$R =  \setenum{(\stMerge{i \in I}{\stT[i]}, \stU)}$, and showing that 
$R \subseteq \stSub$. 
\qedhere
\end{proof}

\begin{lemma}
\label{app:lem:subtyping:merging_bound_upper}
  Given a collection of mergeable local types $\stT[i]$ ($i \in I$).
  If for all $i \in I$, $\stU \stSub \stT[i]$ for some local type $\stU$,
  then $\stU \stSub \stMerge{i \in I}{\stT[i]}$.
\end{lemma}
\begin{proof}
By constructing a relation 
$R =  \setenum{(\stU, \stMerge{i \in I}{\stT[i]})}$, and showing that 
$R \subseteq \stSub$. 
(As $\stU \stSub \stT[i]$ holds for all $i \in I$, $\stU$ contains every needed external choice. For the external choices holds that they need to be the same in all $\stT[i]$ and hence $\stU$ has a subset of them.)
\qedhere
\end{proof}

\begin{lemma}
\label{lem:subtype:merge-subty}
Given two collections of mergeable local types $\stS[i],  \stT[i]$ ($i \in I$).
 If for all $i \in I$, $\stS[i] \stSub \stT[i]$, then $\stMerge{i \in I} {\stS[i]} \stSub \stMerge{i \in I}{\stT[i]}$.
\end{lemma}
\begin{proof}
 By constructing a relation $R = \setenum{(\stMerge{i \in I}{\stS[i]}, \stMerge{i \in I}{\stT[i]})}$, and showing 
that $R \subseteq \stSub$. 
\qedhere
\end{proof}

\begin{lemma}[Subtyping of Projection with Unfolding]
\label{app:lem:subtyping_unfolding_projection}
For any closed, well-guarded global type $\gtG$ and role $\roleP$, 
$\gtProj[]{\unfoldOne{\gtG}}{\roleP} \stSub \gtProj[]{\gtG}{\roleP}$ and $\gtProj[]{\gtG}{\roleP} \stSub \gtProj[]{\unfoldOne{\gtG}}{\roleP}$. 
\end{lemma}
\begin{proof}
We proceed by cases on $\gtG$. 
\begin{itemize}
\item Case $\gtG \neq \gtRec{\gtRecVar}{\gtGi}$: by~\Cref{app:def:unfolding}, $\unfoldOne{\gtG} = \gtG$. 
The thesis follows by reflexivity of subtyping~(\Cref{app:lem:subtyping:reflexive}). 
\item Case $\gtG = \gtRec{\gtRecVar}{\gtGi}$: by~\Cref{app:def:unfolding}, $\unfoldOne{\gtG} = \gtGi\subst{\gtRecVar}{\gtRec{\gtRecVar}{\gtGi}}$. 
The proof follows by induction on the structure of $\gtGi$.  
\begin{itemize}
\item Case $\gtGi = \gtEnd$: by the definition of projection~(\Cref{def:global-proj}) and the subtyping rule~\inferrule{\iruleStSubEnd}. 
\item Case $\gtGi =  \gtComm{\roleQ}{\roleR}{i \in I}{\gtLab[i]}{\tyS[i]}{\gtGii[i]}$: subcases are considered on $\roleP$.
\begin{itemize}
\item Subcase $\roleP = \roleQ$ or $\roleP = \roleR$: by the definition of projection~(\Cref{def:global-proj}), the induction hypothesis, 
and the definition of subtyping (\Cref{def:subtyping}). 
\item $\roleP \neq \roleQ$ and $\roleP \neq \roleR$: by the definition of projection~(\Cref{def:global-proj}),  the induction hypothesis,  
the definition of subtyping~(\Cref{def:subtyping}), 
and~\Cref{lem:subtype:merge-subty}. 
\qedhere
\end{itemize} 
\end{itemize}
\end{itemize}
\end{proof}

\subsection{Properties of Global Type and Typing Context Transitions}
\begin{lemma}
\label{app:lem:gt-lts-unfold}
  $\gtG \,\gtMove[\stEnvAnnotGenericSym]\,
  \gtGi$ \;iff\; $\unfoldOne{\gtG} 
  \,\gtMove[\stEnvAnnotGenericSym]\, \gtGi$.
\end{lemma}
\begin{proof}
  By inverting or applying $\inferrule{\iruleGtMoveRec}$ when necessary. 
  \qedhere
\end{proof}

\begin{lemma}[Progress of Global Types]
\label{app:lem:gtype:progress}
  If $\gtG \neq
  \gtEnd$~(where $\gtG$ is a projectable
  global type), 
  then there exists $\gtGi$ such that
  $\gtG \,\gtMove\, 
  \gtGi$.
\end{lemma}
\begin{proof}
  By \Cref{app:lem:gt-lts-unfold}, we only consider unfoldings.
  \begin{itemize}%
\item Case $\unfoldOne{\gtG} = \gtEnd$: the premise does not hold; 
\item Case $\unfoldOne{\gtG} = \gtComm{\roleP}{\roleQ}{i \in I}{\gtLab[i]}{\tyS[i]}{\gtG[i]}$: apply \inferrule{\iruleGtMoveComm} to reduce the global type. 
\qedhere 
\end{itemize}
\end{proof}

\begin{lemma}
\label{lem:global_end}
If $\roleP \notin \gtRoles{\gtG}$, then $\gtProj{\gtG}{\roleP} = \stEnd$.  
\end{lemma}
\begin{proof}
By induction on the structure of $\gtG$: 
\begin{itemize}
\item Case $\gtG = \gtEnd$: trivial as $\gtProj{\gtEnd}{\roleP} = \stEnd$. 
\item Case $\gtG =   \gtComm{\roleQ}{\roleR}{i \in I}{\gtLab[i]}{\tyS[i]}{\gtG[i]}$: 
we have $\roleP \neq \roleQ$, $\roleP \neq \roleR$, and $\forall i \in I: \roleP \notin \gtRoles{\gtG[i]}$. 
 The thesis holds by $\gtProj{\gtG}{\roleP} = \stMerge{i \in I}{\gtProj{\gtG[i]}{\roleP}}$ and the induction hypothesis that 
$\forall i \in I: \gtProj{\gtG[i]}{\roleP} = \stEnd$. 
\item Case $\gtG = \gtRec{\gtRecVar}{\gtGi}$: assume $\gtFv{\gtRec{\gtRecVar}{\gtGi}} \neq \emptyset$.  
By induction hypothesis, $\gtProj{\gtG}{\roleP} = \stRec{\stRecVar}{\gtProj{\gtGi}{\roleP}} = \stRec{\stRecVar}{\stEnd}$, which is an unguarded recursion -- a contradiction. Therefore, the thesis holds directly, as $\gtProj{ \gtRec{\gtRecVar}{\gtGi}}{\roleP} = \stEnd$.
\qedhere  
\end{itemize}
\end{proof}

\begin{lemma}
\label{lem:stenv-red:trivial-1-new}
  If $\stEnv
  \stEnvMoveGenAnnot \stEnvi$, 
  then 
   $\dom{\stEnv} = \dom{\stEnvi}$. 
\end{lemma}
\begin{proof}
By induction on typing context transitions. 
\qedhere 
\end{proof}

\begin{lemma}
\label{lem:stenv-red:trivial-2}
  If $\stEnv
  \stEnvMoveGenAnnot \stEnvi$ and $\dom{\stEnv} =
  \setenum{\mpS}$, then 
    for all  $\mpChanRole{\mpS}{\roleP} \in \dom{\stEnv}$ with 
      $\roleP \notin \ltsSubject{\stEnvAnnotGenericSym}$, %
      $\stEnvApp{\stEnv}{\mpChanRole{\mpS}{\roleP}} =
      \stEnvApp{\stEnvi}{\mpChanRole{\mpS}{\roleP}}$.
\end{lemma}
\begin{proof}
By induction on typing context transitions.  
\qedhere 
\end{proof}

\begin{lemma}[Inversion of Typing Context Transition]
\label{lem:stenv-red:inversion-basic}
~
  \begin{enumerate}
    \item
      If $
        \stEnv
        \stEnvMoveOutAnnot{\roleP}{\roleQ}{\stChoice{\stLab[k]}{\stS[k]}}
        \stEnvi
      $, then 
      $
        \unfoldOne{\stEnvApp{\stEnv}{\mpChanRole{\mpS}{\roleP}}} =
        \stIntSum{\roleQ}{i \in I}{\stChoice{\stLab[i]}{\stS[i]} \stSeq \stT[i]}%
      $, $k \in I$, and 
      $
        \stEnvApp{\stEnvi}{\mpChanRole{\mpS}{\roleP}} = \stT[k]$; 
    \item
      If $
        \stEnv
        \stEnvMoveInAnnot{\roleQ}{\roleP}{\stChoice{\stLab[k]}{\stS[k]}}
        \stEnvi
      $, then 
      $
        \unfoldOne{\stEnvApp{\stEnv}{\mpChanRole{\mpS}{\roleQ}}} =
        \stExtSum{\roleP}{i \in I}{\stChoice{\stLab[i]}{\stS[i]} \stSeq \stT[i]}%
      $, $k \in I$, and 
      $
        \stEnvApp{\stEnvi}{\mpChanRole{\mpS}{\roleQ}} = \stT[k]
      $.  
      \item 
       If $\stEnv \stEnvMoveCommAnnot{\mpS}{\roleP}{\roleQ}{\stLab}  \stEnvi$,  then 
      $
        \unfoldOne{\stEnvApp{\stEnv}{\mpChanRole{\mpS}{\roleP}}} =
        \stIntSum{\roleQ}{i \in I}{\stChoice{\stLab[i]}{\stS[i]} \stSeq \stT[i]}%
      $,   $
        \unfoldOne{\stEnvApp{\stEnv}{\mpChanRole{\mpS}{\roleQ}}} =
        \stExtSum{\roleP}{j \in J}{\stChoice{\stLabi[\!\!j]}{\stSi[j]} \stSeq \stTi[j]}%
      $, and there exist $k \in I$ and $l \in J$ such that  $\stLab[k] = \stLabi[\!\!l] = \stLab$, 
      $\stSi[l] \stSub   \stS[k]$, with 
      $
        \stEnvApp{\stEnvi}{\mpChanRole{\mpS}{\roleP}} = \stT[k]
      $ and $
        \stEnvApp{\stEnvi}{\mpChanRole{\mpS}{\roleQ}} = \stTi[l]
      $. 
  \end{enumerate}
\end{lemma}
\begin{proof}
By applying and inverting \inferrule{\iruleTCtxOut}, \inferrule{\iruleTCtxIn}, and \inferrule{\iruleTCtxRec}~(when necessary).  
\qedhere 
\end{proof}

\begin{lemma}[Determinism of Typing Context Transition]
\label{lem:stenv-red:det}
  If $\stEnv \stEnvMoveGenAnnot \stEnvi$ \,and\, $\stEnv \stEnvMoveGenAnnot
  \stEnvii$, then $\stEnvi = \stEnvii$.
\end{lemma}
\begin{proof}
By induction on typing context transitions. 
\qedhere 
\end{proof}

\begin{lemma}
\label{lem:stenv-red:trivial-5}
~
\begin{enumerate}
\item 
If $\stEnv
        \stEnvMoveOutAnnot{\roleP}{\roleQ}{\stChoice{\stLab[k]}{\stS[k]}}
        \stEnvi
      $, then for any channel with role $\mpC \in \dom{\stEnv}$ with $\mpC \neq \mpChanRole{\mpS}{\roleP}$, 
      $\stEnvApp{\stEnv}{\mpC} = \stEnvApp{\stEnvi}{\mpC}$. 
 
\item If $
        \stEnv
        \stEnvMoveInAnnot{\roleQ}{\roleP}{\stChoice{\stLab[k]}{\stS[k]}} 
        \stEnvi
      $, then for any channel with role $\mpC \in \dom{\stEnv}$ 
      with $\mpC \neq \mpChanRole{\mpS}{\roleQ}$, 
      $\stEnvApp{\stEnv}{\mpC} = \stEnvApp{\stEnvi}{\mpC}$. 
      
\item  If 
$\stEnv \stEnvMoveCommAnnot{\mpS}{\roleP}{\roleQ}{\stLab}  \stEnvi$,  
 then for any channel with role $\mpC \in \dom{\stEnv}$ 
 with $\mpC \neq \mpChanRole{\mpS}{\roleP}$ and 
 $\mpC \neq \mpChanRole{\mpS}{\roleQ}$, 
        $\stEnvApp{\stEnv}{\mpC} = \stEnvApp{\stEnvi}{\mpC}$.
\end{enumerate}
\end{lemma}
\begin{proof}
By induction on typing context transitions.  
\qedhere 
\end{proof}

\subsection{Properties of Association}
\begin{lemma}
\label{app:lem:unfolding_projection_local_subtyping}
  $\stT \stSub \gtProj{\gtG}{\roleP}$ if and only if $\stT \stSub \gtProj{\unfoldOne{\gtG}}{\roleP}$
\end{lemma}
\begin{proof}
Follows directly from~\Cref{app:lem:subtyping_unfolding_projection} and 
transitivity of subtyping~(\Cref{app:lem:subtyping:transitive}). 
\qedhere 
\end{proof}

\begin{proposition}
	\label{app:prop:assoc:global_unfold} %
  $\stEnvAssoc{\gtG}{\stEnv}{\mpS}$
  \,if and only if\,
  $\stEnvAssoc{\unfoldOne{\gtG}}{\stEnv}{\mpS}$.
\end{proposition}
\begin{proof}
Direct from~\Cref{def:assoc}, \Cref{app:prop:roles_unfold}, and~\Cref{app:lem:unfolding_projection_local_subtyping}. 
 \qedhere
\end{proof}

\begin{lemma}[Relating Terminations]
	\label{app:lem:assoc:termination}
If $\gtG = \gtEnd$ and $\stEnvAssoc{\gtG}{\stEnv}{\mpS}$, 
then 
$\forall \mpChanRole{\mpS}{\roleP} \in \dom{\stEnv}: \stEnvApp{\stEnv}{\mpChanRole{\mpS}{\roleP}} = \stEnd$. 
\end{lemma}
\begin{proof}
By the definition of association~(\Cref{def:assoc}), we know that 
$\stEnv = \stEnv[\gtG] \stEnvComp \stEnv[\stEnd]$, where, by the hypothesis $\gtG = \gtEnd$, 
$\dom{\stEnv[\gtG]} = \emptyset$. Hence, $\stEnv =  \stEnv[\stEnd]$, which is the thesis. 
\qedhere 
\end{proof}

\begin{lemma}[Inversion of Projection]
	\label{app:lem:inversion_projection} %
  Given a local type $\stT$, which is a subtype of projection from a global type
  $\gtG$ on a role $\roleP$, 
  \ie $\stT \stSub \gtProj{\gtG}{\roleP}$,
  then:
  \begin{enumerate}%
    \item
      If\;
      $\unfoldOne{\stT}
      =
      \stIntSum{\roleQ}{i \in I}{\stChoice{\stLab[i]}{\stS[i]} \stSeq
        \stT[i]}$, then either
        \begin{enumerate}
        \item
          $\unfoldOne{\gtG} =
            \gtComm{\roleP}{\roleQ}{i \in I'}{\gtLabi[i]}{\tySi[i]}{\gtG[i]}$,
          where $I \subseteq I'$, and for all $i \in I: \stLab[i] = \gtLabi[i]$, 
          $\stSi[i] \stSub \tyS[i]$, and 
	  $\stT[i] \stSub \gtProj{\gtG[i]}{\roleP}$; 
          or,
        \item
          $\unfoldOne{\gtG} =
            \gtComm{\roleS}{\roleT}{j \in
            J}{\gtLabi[j]}{\tySi[j]}{\gtG[j]}$,
          where for all $j \in J: \stT \stSub \gtProj{\gtG[j]}{\roleP}$,
          with $\roleP \neq \roleS$ and $\roleP \neq \roleT$.
        \end{enumerate}
    \item
      If\;
      $\unfoldOne{\stT}
      =
      \stExtSum{\roleQ}{i \in I}{\stChoice{\stLab[i]}{\stS[i]} \stSeq
        \stT[i]}$, then either
      \begin{enumerate}
        \item
          $\unfoldOne{\gtG} =
            \gtComm{\roleQ}{\roleP}{i \in I'}{\gtLabi[i]}{\tySi[i]}{\gtG[i]}$,
          where $I' \subseteq I$, and for all $i \in I': \stLab[i] = \gtLabi[i]$, 
           $\tyS[i] \stSub \stSi[i]$, and  
	   $\stT[i] \stSub \gtProj{\gtG[i]}{\roleP}$; 
          or,
        \item
          $\unfoldOne{\gtG} =
            \gtComm{\roleS}{\roleT}{j \in
            J}{\gtLabi[j]}{\tySi[j]}{\gtG[j]}$,
          where for all $j \in J:  \stT \stSub \gtProj{\gtG[j]}{\roleP}$,
          with $\roleP \neq \roleS$ and $\roleP \neq \roleT$.
      \end{enumerate}
    \item
      If\;
      $\unfoldOne{\stT} = \stEnd$, then $\roleP \notin \gtRoles{\gtG}$. 
  \end{enumerate}
\end{lemma}
\begin{proof}
By~\Cref{app:lem:subtyping:unfolding,app:lem:unfolding_projection_local_subtyping}, transitivity of subtyping~(\Cref{app:lem:subtyping:transitive}), and the definition of global type projection~(\Cref{def:global-proj}). For the subcases (b) in 
items (1) and (2), apply~\Cref{app:lem:subtyping:merging} additionally. 
\qedhere 
\end{proof}

\begin{lemma}[Matching Communication Under Projection]
	\label{app:lem:match_comm_projection} %
  If two local types $\stT, \stU$ are subtypes %
  of an internal choice and an
  external choice with matching roles, obtained via projection from a global
  type $\gtG$, \ie
  $  
    \unfoldOne{\stT} 
    =
      \stIntSum{\roleQ}{i \in I_{\roleP}}{\stChoice{\stLab[i]}{\stS[i]} \stSeq \stT[i]} 
  \stSub \gtProj{\gtG}{\roleP}$
  and
  $  
    \unfoldOne{\stU}
    =
    \stExtSum{\roleP}{j \in I_{\roleQ}}{\stChoice{\stLabi[j]}{\stSi[j]} \stSeq \stTi[j]} 
  \stSub \gtProj{\gtG}{\roleQ}$, then
      $I_{\roleP} \subseteq I_{\roleQ}$, and
      $\forall i \in I_{\roleP}: \stLab[i] = \stLabi[i]$ and $\stSi[i] \stSub \stS[i]$.
\end{lemma}
\begin{proof}
  By induction on items (1) and (2) of \Cref{app:lem:inversion_projection} simultaneously.
  \begin{enumerate}[label=(\alph*)]
    \item
      We have $\unfoldOne{\gtG} = \gtComm{\roleP}{\roleQ}{i \in I}{\gtLabii[i]}{\stSii[i]}{\gtG[i]}$,
      $I_{\roleP} \subseteq I \subseteq I_{\roleQ}$,
      $\forall i \in I_{\roleP}: \stLab[i] = \gtLabii[i]$ and $\stSii[i] \stSub \stS[i]$, and
      $\forall i \in I: \gtLabii[i] = \stLabi[i]$ and $\stSi[i] \stSub \stSii[i]$.
      We have $I_{\roleP} \subseteq I_{\roleQ}$~(by transitivity of $\subseteq$),
      and $\forall i \in I_{\roleP}: \stLab[i] = \stLabi[i]$~(by transitivity of $=$) and $\stSi[i] \stSub \stS[i]$~(by transitivity of
      $\stSub$).
    \item
      We have $\unfoldOne{\gtG} = \gtComm{\roleS}{\roleT}
        {j \in J}{\gtLabii[j]}{\tySii[j]}{\gtG[j]}$,
      where for all $j \in J: \stT \stSub \gtProj{\gtG[j]}{\roleP}$,
      $\stU \stSub \gtProj{\gtG[j]}{\roleQ}$,
      $\setenum{\roleP, \roleQ} \cap \setenum{\roleS, \roleT} = \emptyset$.
      Apply induction on
      $\stT \stSub \gtProj{\gtG[j]}{\roleP}$ and
      $\stU \stSub \gtProj{\gtG[j]}{\roleP}$ on any $j \in J$.
      \qedhere
  \end{enumerate}
\end{proof}

\begin{restatable}[Inversion of Association]{lemma}{lemInvProj}
	\label{app:lem:inversion_association} %
  Let $\stEnvAssoc{\gtG}{\stEnv}{\mpS}$. 

  \begin{enumerate}
    \item
      If\;
      $\unfoldOne{\stEnvApp{\stEnv}{\mpChanRole{\mpS}{\roleP}}}
      =
       \stIntSum{\roleQ}{i \in I}{\stChoice{\stLab[i]}{\stS[i]} \stSeq
        \stT[i]}$, then either
        \begin{enumerate}
        \item
          $\unfoldOne{\gtG} =
            \gtComm{\roleP}{\roleQ}{i \in I'}{\gtLab[i]}{\tySi[i]}{\gtG[i]}$,
          where $I \subseteq I'$, and for all $i \in I: \stLab[i] = \gtLab[i]$, 
          $\stSi[i] \stSub \tyS[i]$, and 
	  $\stT[i] \stSub \gtProj{\gtG[i]}{\roleP}$; 
          or,
        \item
          $\unfoldOne{\gtG} =
            \gtComm{\roleS}{\roleT}{j \in
            J}{\gtLab[j]}{\tySi[j]}{\gtG[j]}$,
          where for all $j \in J: \stEnvApp{\stEnv}{\mpChanRole{\mpS}{\roleP}} \stSub \gtProj{\gtG[j]}{\roleP}$,
          with $\roleP \neq \roleS$ and $\roleP \neq \roleT$.
        \end{enumerate}
        
    \item\label{item:inversion:input}
      If\;
      $\unfoldOne{\stEnvApp{\stEnv}{\mpChanRole{\mpS}{\roleP}}}
      =
     \stExtSum{\roleQ}{i \in I}{\stChoice{\stLab[i]}{\stS[i]} \stSeq
        \stT[i]}$, then either
      \begin{enumerate}
        \item
          $\unfoldOne{\gtG} =
            \gtComm{\roleQ}{\roleP}{i \in I'}{\gtLab[i]}{\tySi[i]}{\gtG[i]}$,
          where $I' \subseteq I$, and for all $i \in I': \stLab[i] = \gtLab[i]$, 
           $\tyS[i] \stSub \stSi[i]$, and  
	   $\stT[i] \stSub \gtProj{\gtG[i]}{\roleP}$; 
          or,
        \item
          $\unfoldOne{\gtG} =
            \gtComm{\roleS}{\roleT}{j \in
            J}{\gtLab[j]}{\tySi[j]}{\gtG[j]}$,
          where for all $j \in J: \stEnvApp{\stEnv}{\mpChanRole{\mpS}{\roleP}} \stSub \gtProj{\gtG[j]}{\roleP}$,
          with $\roleP \neq \roleS$ and $\roleP \neq \roleT$.
      \end{enumerate}
       \item
       \label{item:inversion:end}
      If\;
      $\unfoldOne{\stEnvApp{\stEnv}{\mpChanRole{\mpS}{\roleP}}}
      =
     \stEnd$, then $\roleP \notin \gtRoles{\gtG}$. 
  \end{enumerate}
\end{restatable}
\begin{proof}
  Follows directly  from \Cref{def:assoc} %
   and
   \Cref{app:lem:inversion_projection}.
   \qedhere
\end{proof}

\begin{lemma}[Simultaneous Inversions of Association]
	\label{app:lem:inversion_association_simultaneous} %
  Let $\stEnvAssoc{\gtG}{\stEnv}{\mpS}$. 
  If\;
  $ \unfoldOne{\stEnvApp{\stEnv}{\mpChanRole{\mpS}{\roleP}}}
    =
    \stIntSum{\roleQ}{i \in I_{\roleP}}{\stChoice{\stLab[i]}{\stS[i]} \stSeq \stT[i]}
  $
  and
  $ \unfoldOne{\stEnvApp{\stEnv}{\mpChanRole{\mpS}{\roleQ}}}
    =
    \stExtSum{\roleP}{i \in I_{\roleQ}}{\stChoice{\stLab[i]}{\stSi[i]} \stSeq \stTi[i]}
  $, then either
  \begin{enumerate}
    \item
      $\unfoldOne{\gtG} = \gtComm{\roleP}{\roleQ}{i \in I}{\gtLab[i]}{\stSii[i]}{\gtG[i]}$,
      where $I_{\roleP} \subseteq I \subseteq I_{\roleQ}$,
      $\forall i \in I_{\roleP}: \stSii[i] \stSub \stS[i]$,
      $\forall i \in I: \stSi[i] \stSub \stSii[i]$, 
      $\forall i \in I_{\roleP}:
      \stT[i] \stSub \gtProj{\gtG[i]}{\roleP}$, and 
      $\forall i \in I:
      \stTi[i] \stSub \gtProj{\gtG[i]}{\roleQ}$;
      or,
    \item
      $\unfoldOne{\gtG} = \gtComm{\roleS}{\roleT}
        {j \in J}{\gtLab[j]}{\tySii[j]}{\gtG[j]}$,
      where for all $j \in J: 
      \stEnvApp{\stEnv}{\mpChanRole{\mpS}{\roleP}} \stSub \gtProj{\gtG[j]}{\roleP} 
     $,
      $
      \stEnvApp{\stEnv}{\mpChanRole{\mpS}{\roleQ}} \stSub \gtProj{\gtG[j]}{\roleQ}
      $, and 
      $\setenum{\roleP, \roleQ} \cap \setenum{\roleS, \roleT} = \emptyset$. 
  \end{enumerate}
\end{lemma}
\begin{proof}
By combining cases (1) and (2) from \Cref{app:lem:inversion_association}. %
  Note that case 1(a) is incompatible with case 2(b), since 2(b) requires that
  $\roleP \neq \roleS$.
  Similarly, case 1(b) is
  incompatible with case 2(a). 
  \qedhere
\end{proof}

\subsection{Completeness of Association}
\label{sec:app:completeness_association}
\thmProjCompleteness*

\begin{proof}
By induction on transitions of typing context $\stEnv \stEnvMoveGenAnnot \stEnvi$. Since $\stEnvAnnotGenericSym$ is 
of the form $\stEnvCommAnnotSmall{\roleP}{\roleQ}{\stLab}$, we only need to consider the following two cases. 
 \begin{itemize}[left=0pt, topsep=0pt]
 \item Case \inferrule{\iruleTCtxCom}: 

From the premise, we have: 
\begin{gather}
\stEnvAssoc{\gtG}{\stEnv}{\mpS}
 \\
 \stEnvAnnotGenericSym = \stEnvCommAnnotSmall{\roleP}{\roleQ}{\stLab}
 \\
 \stEnv = \stEnv[1] \stEnvComp \stEnv[2]
  \label{eq:stenvp-comm-type-ctx}
  \\
  \stEnv[1]
    \stEnvMoveOutAnnot{\roleP}{\roleQ}{\stChoice{\stLab}{\stS}}
    \stEnvi[1]
    \label{eq:stenvp-comm-send}
    \\
     \stEnv[2]
    \stEnvMoveInAnnot{\roleQ}{\roleP}{\stChoice{\stLab}{\stSi}}
    \stEnvi[2]
    \label{eq:stenvp-comm-rcv}
    \\
\stEnvi = \stEnvi[1] \stEnvComp \stEnvi[2]
\label{eq:stenvp-comm-type-cxt-cons}
\end{gather} 
By applying~\Cref{lem:stenv-red:inversion-basic} on \eqref{eq:stenvp-comm-send} and \eqref{eq:stenvp-comm-rcv}, 
we have %
\[ 
	\unfoldOne{\stEnvApp{\stEnv[1]}{\mpChanRole{\mpS}{\roleP}}} =
    \stIntSum{\roleQ}{i \in I_{\roleP}}{\stChoice{\stLab[
    i]}{\stS[i]} \stSeq \stT[i]}
    \quad \text{and} \quad %
	\unfoldOne{\stEnvApp{\stEnv[2]}{\mpChanRole{\mpS}{\roleQ}}} =
    \stExtSum{\roleP}{i \in I_{\roleQ}}{\stChoice{\stLab[
    i]}{\stSi[i]} \stSeq \stTi[i]}%
\]
  Moreover, by inverting $\stEnv \,\stEnvMoveCommAnnot{\mpS}{\roleP}{\roleQ}{\stLab}\, \stEnvi$, 
  $\exists k \in (I_{\roleP} \cap I_{\roleQ})$, such that $\stLab[k] = \stLab$
  and $\stSi[k] \stSub \stS[k]$.
  
We perform case analysis on \Cref{app:lem:inversion_association_simultaneous}. %
\begin{itemize}[left=0pt, topsep=0pt]
\item Case 1 of \Cref{app:lem:inversion_association_simultaneous}: %
	  we know
      $\unfoldOne{\gtG} = \gtComm{\roleP}{\roleQ}{i \in I}{\gtLab[i]}{\stSii[i]}{\gtG[i]}$
      where $I_{\roleP} \subseteq I \subseteq I_{\roleQ}$,
      $\forall i \in I_{\roleP}: \stSii[i] \stSub \stS[i]$,
      $\forall i \in I: \stSi[i] \stSub \stSii[i]$,
      $\forall i \in I_{\roleP}:
      \stT[i] \stSub \gtProj{\gtG[i]}{\roleP}$,
      and $\forall i \in I:
      \stTi[i] \stSub \gtProj{\gtG[i]}{\roleQ}$.
        
 Since $k \in (I_{\roleP} \cap I_{\roleQ})$ and $I_{\roleP} \subseteq I$, we
  have $k \in I$. 
 By applying \inferrule{\iruleGtMoveComm} (with
 \Cref{app:prop:assoc:global_unfold}),  %
  the result becomes
  $\gtG[k]$.

  We are left to show association. By
  \Cref{lem:stenv-red:inversion-basic,app:lem:inversion_association}, %
  we have $
    \stEnvApp{\stEnvi[1]}{\mpChanRole{\mpS}{\roleP}} = 
    \stT[k] \stSub
  \gtProj{\gtG[k]}{\roleP} 
  $ and $
    \stEnvApp{\stEnvi[2]}{\mpChanRole{\mpS}{\roleQ}} =
    \stTi[k] \stSub
    \gtProj{\gtG[k]}{\roleQ} 
  $.
  For other roles $\roleR \in \gtRoles{\gtG}$, we have $
    \stEnvApp{\stEnv}{\mpChanRole{\mpS}{\roleR}}
    =
    \stEnvApp{\stEnvi}{\mpChanRole{\mpS}{\roleR}}
  $ by \Cref{lem:stenv-red:trivial-2}, and 
  $\gtProj{\gtG}{\roleR} = 
    \stMerge{i \in I}(\gtProj{\gtG[i]}{\roleR}) 
    \stSub \gtProj{\gtG[k]}{\roleR} 
$ by \Cref{app:lem:subtyping:merging}.
  Since $\stEnvApp{\stEnv}{\mpChanRole{\mpS}{\roleR}} \stSub \gtProj{\gtG}{\roleR}$, 
  we can conclude that
  $\stEnvApp{\stEnvi}{\mpChanRole{\mpS}{\roleR}} \stSub
  \gtProj{\gtG[k]}{\roleR} 
  $ by transitivity of subtyping. 
   Additionally, for all other $\mpChanRole{\mpS}{\roleQi} \in \dom{\stEnv}$ 
   where $\roleQi \notin \gtRoles{\gtG}$, we have 
    $\stEnvApp{\stEnv}{\mpChanRole{\mpS}{\roleQi}} = \stEnvApp{\stEnvi}{\mpChanRole{\mpS}{\roleQi}} = \stEnd$ by 
    $\stEnvAssoc{\gtG}{\stEnv}{\mpS}$ and (3) of~\Cref{lem:stenv-red:trivial-5}.

\item Case 2 of \Cref{app:lem:inversion_association_simultaneous}: %
  we know $\unfoldOne{\gtG} =
  \gtComm{\roleS}{\roleT}%
    {j \in J}{\gtLab[j]}{\tySii[j]}{\gtG[j]}$,
  where for all $j \in J$,
  $ \stEnvApp{\stEnv}{\mpChanRole{\mpS}{\roleP}} \stSub \gtProj{\gtG[j]}{\roleP}$ and
  $\stEnvApp{\stEnv}{\mpChanRole{\mpS}{\roleQ}} \stSub \gtProj{\gtG[j]}{\roleQ}$, and 
  $\setenum{\roleP, \roleQ} \cap \setenum{\roleS, \roleT} = \emptyset$. 
  
  Take an arbitrary index $j \in J$, we construct a typing context $\stEnv[j]$
  such that $\stEnv[j] \stEnvMoveCommAnnot{\mpS}{\roleP}{\roleQ}{\stLab}
  \stEnvi[j]$ and
  $\stEnvAssoc{\gtG[j]}{\stEnv[j]}{\mpS}$.
 
 To construct $\stEnv[j]$, we consider sub-cases for all roles, and show that
  $\stEnvAssoc{\gtG[j]}{\stEnv[j]}{\mpS}$:

  \begin{itemize}[left=0pt, topsep=0pt]
      \item
      For role $\roleS$,
      we know
      from $\stEnvAssoc{\gtG}{\stEnv}{\mpS}$ that
      $ 
                 \stEnvApp{\stEnv}{\mpChanRole{\mpS}{\roleS}} 
		 \stSub 
        \stIntSum{\roleT}{j \in J}{%
          \stChoice{\stLab[j]}{\tySii[j]}
                 \stSeq (\gtProj{\gtG[j]}{\roleS})} 
        =
      \gtProj{\unfoldOne{\gtG}}{\roleS}
      $.
      By inverting $\inferrule{\iruleStSubOut}$ (applying
\Cref{app:lem:subtyping:unfolding} where necessary), we have
      $ \unfoldOne{\stEnvApp{\stEnv}{\mpChanRole{\mpS}{\roleS}}}
        =
        \stIntSum{\roleT}{j \in J_{\roleS}}{%
          \stChoice{\stLab[j]}{\stSiii[j]} \stSeq \stTii[j]
        }
      $, where $J_{\roleS} \subseteq J$,
      and $\forall j \in J_{\roleS}:
      \stTii[j] \stSub \gtProj{\gtG[j]}{\roleS}$.

      To construct $\stEnv[j]$, let
      $ \stEnvApp{\stEnv[j]}{\mpChanRole{\mpS}{\roleS}}
        =
        \stTii[j]
      $ if $j \in J_{\roleS}$ and
      $ \stEnvApp{\stEnv[j]}{\mpChanRole{\mpS}{\roleS}}
        =
        \gtProj{\gtG[j]}{\roleS}
      $ otherwise.
      In either case, we have
      $ 
        \stEnvApp{\stEnv[j]}{\mpChanRole{\mpS}{\roleS}} 
	\stSub \gtProj{\gtG[j]}{\roleS} 
       $, as required.
    
    \item
      For role $\roleT$,
      we know
      from $\stEnvAssoc{\gtG}{\stEnv}{\mpS}$ that
      $
                  \stEnvApp{\stEnv}{\mpChanRole{\mpS}{\roleT}} \stSub 
        \stExtSum{\roleS}{j \in J}{\stChoice{\stLab[j]}{\stSii[j]}
                 \stSeq (\gtProj{\gtG[j]}{\roleT})} 
        =
        \gtProj{\unfoldOne{\gtG}}{\roleT}
      $.
      By inverting $\inferrule{\iruleStSubIn}$ (applying
      \Cref{app:lem:subtyping:unfolding} where necessary), we have
      $ \unfoldOne{\stEnvApp{\stEnv}{\mpChanRole{\mpS}{\roleT}}}
        =
        \stExtSum{\roleS}{j \in J_{\roleT}}{%
          \stChoice{\stLab[j]}{\stSiii[j]} \stSeq \stTii[j]
        }
      $, where $J \subseteq J_{\roleT}$,
      and $\forall j \in J:
      \stTii[j] \stSub \gtProj{\gtG[j]}{\roleT}$.

      To construct $\stEnv[j]$, let
      $ \stEnvApp{\stEnv[j]}{\mpChanRole{\mpS}{\roleT}}
        =
        \stTii[j]
      $, and we have
      $\stEnvApp{\stEnv[j]}{\mpChanRole{\mpS}{\roleT}} \stSub
        \gtProj{\gtG[j]}{\roleT}$, as required.
      
      \item 
      For other roles $\roleR \in \gtRoles{\gtG}$ with $\roleR \notin
  \setenum{\roleS, \roleT}$, 
  their typing context
  entry do not change, \ie 
  $\stEnvApp{\stEnv[j]}{\mpChanRole{\mpS}{\roleR}}
  = \stEnvApp{\stEnv}{\mpChanRole{\mpS}{\roleR}}$.
 We have
  $
    \stEnvApp{\stEnv[j]}{\mpChanRole{\mpS}{\roleR}} =
    \stEnvApp{\stEnv}{\mpChanRole{\mpS}{\roleR}} \stSub
    \gtProj{\gtG}{\roleR} =
\stMerge{j \in J}(\gtProj{\gtG[j]}{\roleR}) \stSub
    \gtProj{\gtG[j]}{\roleR} 
 $ (applying \Cref{app:lem:subtyping:merging}).
  Notice that $\setenum{\roleP, \roleQ} \in \gtRoles{\gtG}$, so they are still able to
  perform the communication action
  $\stEnv[j] \stEnvMoveCommAnnot{\mpS}{\roleP}{\roleQ}{\stLab} \stEnvi[j]$.
  \end{itemize}

 We apply inductive hypothesis on $\stEnv[j]$, and obtain
  $\gtG[j] 
    \stepsto[\ltsSendRecv{\mpS}{\roleP}{\roleQ}{\gtLab}]
    \gtGi[j] $ and
  $ \stEnvAssoc{\gtGi[j]}{\stEnvi[j]}{\mpS}$.

  We can apply $\inferrule{\iruleGtMoveCtx}$ on
  $
      \gtComm{\roleS}{\roleT}%
        {j \in J}{\gtLab[j]}{\stSii[j]}{\gtG[j]} 
    \stepsto[\ltsSendRecv{\mpS}{\roleP}{\roleQ}{\gtLab}]
      \gtComm{\roleS}{\roleT}%
        {j \in J}{\gtLab[j]}{\stSii[j]}{\gtGi[j]}
  $.
  
  We now show $\stEnvAssoc{\gtGi}{\stEnvi}{\mpS}$,
  where
  $\gtGi =
    \gtComm{\roleS}{\roleT}%
      {j \in J}{\gtLab[j]}{\stSii[j]}{\gtGi[j]}$.

  For role $\roleS$, we know that
  $ \unfoldOne{\stEnvApp{\stEnv}{\mpChanRole{\mpS}{\roleS}}}
    =
    \stIntSum{\roleT}{j \in J_{\roleS}}{%
      \stChoice{\stLab[j]}{\stSiii[j]} \stSeq \stTii[j]
    }
  $, where $J_{\roleS} \subseteq J$,
  and $\forall j \in J_{\roleS}:
  \stTii[j] \stSub \gtProj{\gtG[j]}{\roleS}
  $ and $
    \stSii[j] \stSub \stSiii[j]
  $.
  Since $\roleS \notin
  \ltsSubject{\ltsSendRecv{\mpS}{\roleP}{\roleQ}{\gtLab}}$,
  we apply
  \Cref{lem:stenv-red:trivial-2} on $\stEnv$ and $\stEnv[j]$ for all
  $j \in J_{\roleS}$.
  For all $j \in J_{\roleS}$, 
  we have $
    \stTii[j]
    =
    \stEnvApp{\stEnv[j]}{\mpChanRole{\mpS}{\roleS}}
    =
    \stEnvApp{\stEnvi[j]}{\mpChanRole{\mpS}{\roleS}}
  $ (from \Cref{lem:stenv-red:trivial-2}) and $
    \stEnvApp{\stEnvi[j]}{\mpChanRole{\mpS}{\roleS}} \stSub \gtProj{\gtGi[j]}{\roleS}$
    (from inductive hypothesis).
    Therefore, we have $\stTii[j] \stSub \gtProj{\gtGi[j]}{\roleS}$.
  We now apply \Cref{lem:stenv-red:trivial-2} on $\stEnv$,
  which gives
  $ \unfoldOne{\stEnvApp{\stEnv}{\mpChanRole{\mpS}{\roleS}}}
    =
    \unfoldOne{\stEnvApp{\stEnvi}{\mpChanRole{\mpS}{\roleS}}}
    =
    \stIntSum{\roleT}{j \in J_{\roleS}}{%
      \stChoice{\stLab[j]}{\stSiii[j]} \stSeq \stTii[j]
    }
  $.
  We can now apply $\inferrule{\iruleStSubOut}$ to conclude $
    \stEnvApp{\stEnvi}{\mpChanRole{\mpS}{\roleS}} \stSub \gtProj{\gtGi}{\roleS}$.

  For role $\roleT$, we know that
  $ \unfoldOne{\stEnvApp{\stEnv}{\mpChanRole{\mpS}{\roleT}}}
    =
    \stExtSum{\roleS}{j \in J_{\roleT}}{%
      \stChoice{\stLab[j]}{\stSiii[j]} \stSeq \stTii[j]
    }
  $, where $J \subseteq J_{\roleT}$,
  and $\forall j \in J:
  \stTii[j] \stSub \gtProj{\gtG[j]}{\roleT} \stSub \stTii[j]
  $ and $
    \stSiii[j] \stSub \stSii[j]
  $.
  Since $\roleT \notin
  \ltsSubject{\ltsSendRecv{\mpS}{\roleP}{\roleQ}{\gtLab}}$,
  we apply
  \Cref{lem:stenv-red:trivial-2} on $\stEnv$ and $\stEnv[j]$ for all
  $j \in J$.
  For all $j \in J$, 
  we have $
    \stTii[j]
    =
    \stEnvApp{\stEnv[j]}{\mpChanRole{\mpS}{\roleT}}
    =
    \stEnvApp{\stEnvi[j]}{\mpChanRole{\mpS}{\roleT}}
  $ (from \Cref{lem:stenv-red:trivial-2}) and $ 
    \stEnvApp{\stEnvi[j]}{\mpChanRole{\mpS}{\roleT}}   
\stSub
     \gtProj{\gtGi[j]}{\roleT}
    $ (from inductive hypothesis).
    Therefore, we have $\stTii[j] \stSub \gtProj{\gtGi[j]}{\roleT}$.
  We now apply \Cref{lem:stenv-red:trivial-2} on $\stEnv$,
  which gives
  $ \unfoldOne{\stEnvApp{\stEnv}{\mpChanRole{\mpS}{\roleT}}}
    =
    \unfoldOne{\stEnvApp{\stEnvi}{\mpChanRole{\mpS}{\roleT}}}
    =
    \stExtSum{\roleS}{j \in J_{\roleT}}{%
      \stChoice{\stLab[j]}{\stSiii[j]} \stSeq \stTii[j]
    }
  $.
  We can now apply $\inferrule{\iruleStSubIn}$ to conclude $
   \stEnvApp{\stEnvi}{\mpChanRole{\mpS}{\roleT}} \stSub \gtProj{\gtGi}{\roleT} 
  $.

  For other roles $\roleR \in \gtRoles{\gtG}$ (where
  $\roleR \notin \setenum{\roleS, \roleT}$),
  we need to show
    $\stEnvApp{\stEnvi}{\mpChanRole{\mpS}{\roleR}}
\stSub 
  \gtProj{\gtGi}{\roleR}$.
   We know that $
    \gtProj{\gtGi}{\roleR} =
    \stMerge{j \in J}{\gtProj{\gtGi[j]}{\roleR}}
  $. 
  The inductive hypothesis gives
    $\stEnvApp{\stEnvi[j]}{\mpChanRole{\mpS}{\roleR}} 
   \stSub \gtProj{\gtGi[j]}{\roleR}
  $, and then 
we apply \Cref{app:lem:subtyping:merging_bound_upper} to obtain $
     \stEnvApp{\stEnvi[j]}{\mpChanRole{\mpS}{\roleR}} \stSub
	 \stMerge{j \in J}{\gtProj{\gtGi[j]}{\roleR}} 
    = \gtProj{\gtGi}{\roleR} 
  $.
  Note that $
    \stEnvApp{\stEnv[j]}{\mpChanRole{\mpS}{\roleR}} =
    \stEnvApp{\stEnv}{\mpChanRole{\mpS}{\roleR}}
  $ by construction.
  We now apply \Cref{lem:stenv-red:trivial-2} %
  on $\stEnv$ and all $\stEnv[j]$, which
  gives $\stEnvApp{\stEnvi[j]}{\mpChanRole{\mpS}{\roleR}} = \stEnvApp{\stEnv[j]}{\mpChanRole{\mpS}{\roleR}}$ and $\stEnvApp{\stEnv}{\mpChanRole{\mpS}{\roleR}} = \stEnvApp{\stEnvi}{\mpChanRole{\mpS}{\roleR}}$, 
  thus $\stEnvApp{\stEnvi[j]}{\mpChanRole{\mpS}{\roleR}} = \stEnvApp{\stEnvi}{\mpChanRole{\mpS}{\roleR}}$ for all $j$.
  Therefore, we have $
  \stEnvApp{\stEnvi}{\mpChanRole{\mpS}{\roleR}} 
  \stSub \gtProj{\gtGi}{\roleR}$.

\end{itemize}
 
 \item Case \inferrule{\iruleTCtxRec} (possibly with  \inferrule{\iruleTCtxCong}): 
 
By inductive hypothesis and~\Cref{app:prop:assoc:global_unfold}. %
 \qedhere
 \end{itemize}
\end{proof}

 \begin{restatable}{corollary}{colCompleteness}%
 \label{cor:completeness_association}
 Assume that for any session $\mpS \in \stEnv$, there exists a global 
  type $\gtG[\mpS]$ such that $\stEnvAssoc{\gtG[\mpS]}{\stEnv[\mpS]}{\mpS}$. 
  If $\stEnv \stEnvMove \stEnvi$, then for any $\mpS \in \stEnvi$, 
  there exists a global 
  type $\gtGi[\mpS]$ such that $\gtG[\mpS] \,\gtMoveStar\, \gtGi[\mpS]$ and $\stEnvAssoc{\gtGi[\mpS]}{\stEnvi[\mpS]}{\mpS}$.  
\end{restatable}
\begin{proof}
By~\Cref{def:typing_context_reduction}, we have that there exists a label 
$\ltsSendRecv{\mpS}{\roleP}{\roleQ}{\stLab}$ such that $\stEnv %
      \,\stEnvMoveCommAnnot{\mpS}{\roleP}{\roleQ}{\stLab}\,%
      \stEnvi$, and hence, by~\Cref{lem:stenv-red:trivial-1-new}, 
      $\dom{\stEnv} = \dom{\stEnvi}$.  
      We are left to show that for any $\stEnvi[\mpSi]$ with $\mpSi \in \stEnvi$, 
  there exists a global 
  type $\gtGi[\mpSi]$ such that $\gtG[\mpSi] \,\gtMoveStar\, \gtGi[\mpSi]$ and $\stEnvAssoc{\gtGi[\mpSi]}{\stEnvi[\mpSi]}{\mpSi}$. 
 \begin{itemize}%
  \item Case $\stEnvi[\mpS]$: since $\stEnv %
      \,\stEnvMoveCommAnnot{\mpS}{\roleP}{\roleQ}{\stLab}\,%
      \stEnvi$, it is trivial to have $\stEnv[\mpS] %
      \,\stEnvMoveCommAnnot{\mpS}{\roleP}{\roleQ}{\stLab}\,%
      \stEnvi[\mpS]$. Moreover, by $\stEnvAssoc{\gtG[\mpS]}{\stEnv[\mpS]}{\mpS}$ 
      and~\Cref{thm:completeness_association}, we have that there exists 
  $\gtGi[\mpS]$ such that  $\gtG[\mpS] \,\gtMove\, \gtGi[\mpS]$ and $\stEnvAssoc{\gtGi[\mpS]}{\stEnvi[\mpS]}{\mpS}$, as desired. 
  
  \item Case $\stEnvi[\mpSi]$ with $\mpSi \neq \mpS$: 
  we know from~\Cref{lem:stenv-red:trivial-5} that 
  for any $\mpChanRole{\mpSi}{\roleR} \in \dom{\stEnvi}$, 
  $\stEnvApp{\stEnv}{\unfoldOne{\mpChanRole{\mpSi}{\roleR}}} = \stEnvApp{\stEnvi}{\unfoldOne{\mpChanRole{\mpSi}{\roleR}}}$. 
Then, along with $\gtG[\mpSi] \,\gtMoveStar\, \gtG[\mpSi]$ and $\stEnvAssoc{\gtG[\mpSi]}{\stEnv[\mpSi]}{\mpSi}$, the thesis holds. 
\qedhere 
 \end{itemize}
 \end{proof}

\subsection{Soundness of Association}
\label{sec:app:soundness_association}
\thmProjSoundness*

\begin{proof}
   By induction on transitions of global type 
   \(
    \gtG %
    \,\gtMove[\stEnvAnnotGenericSym]\, %
    \gtGi
  \).
 \begin{itemize}[left=0pt, topsep=0pt]
 \item Case \inferrule{\iruleGtMoveComm}: 
 
 From the premise, we have: 
      \begin{gather}
        \stEnvAssoc{\gtG}{\stEnv}{\mpS}
        \label{eq:soundness:commu_assoc_pre}
        \\
        \gtG =
          \gtCommSmall{\roleP}{\roleQ}{i \in I}{\gtLab[i]}{\tyS[i]}{\gtG[i]}
          \\
       \stEnvAnnotGenericSym =
          \ltsSendRecv{\mpS}{\roleP}{\roleQ}{\gtLab[j]}
        \\
        j \in I
        \\
        \gtGi = \gtG[j] 
      \end{gather}
By association~\eqref{eq:soundness:commu_assoc_pre}, we have 
$\stEnvApp{\stEnv}{\mpChanRole{\mpS}{\roleP}} \stSub
\gtProj{\gtG}{\roleP} =  
\stIntSum{\roleQ}{i \in I}{
       \stChoice{\stLab[i]}{\stS[i]} \stSeq (\gtProj{\gtG[i]}{\roleP})
     } $, and \\
     $\stEnvApp{\stEnv}{\mpChanRole{\mpS}{\roleQ}} \stSub
\gtProj{\gtG}{\roleQ} =  
\stExtSum{\roleP}{i \in I}{
       \stChoice{\stLab[i]}{\stS[i]} \stSeq (\gtProj{\gtG[i]}{\roleP})
     } $. 
Then by~\Cref{app:lem:subtyping:inversion} and~\Cref{app:lem:match_comm_projection}, we have 
     $\stEnvApp{\stEnv}{\mpChanRole{\mpS}{\roleP}} = 
    \stIntSum{\roleQ}{i \in I_{\roleP}}{\stChoice{\stLabi[i]}{\stSi[i]} \stSeq \stTi[i]}$  and 
    $\stEnvApp{\stEnv}{\mpChanRole{\mpS}{\roleQ}} = 
    \stExtSum{\roleP}{i \in I_{\roleQ}}{\stChoice{\stLabii[i]}{\stSii[i]} \stSeq \stTii[i]}$, with  
    $I_{\roleP} \subseteq 
    I \subseteq I_{\roleQ}$, and for all $i \in I_{\roleP}: \gtLab[i] = \stLabi[i] = \stLabii[i]$, 
    $\stSii[i] \stSub \stSi[i]$,   
    $\stTi[i] \stSub \gtProj{\gtG[i]}{\roleP}$, and $\stTii[i] \stSub \gtProj{\gtG[i]}{\roleQ}$. 
    
    Now let us choose some 
     $k \in I_{\roleP}$ such that  
     $\stEnvAnnotGenericSymi = \ltsSendRecv{\mpS}{\roleP}{\roleQ}{\gtLab[k]}$ . 
     Furthermore,  we have $\gtG \,\gtMove[\stEnvAnnotGenericSymi]\, \gtGii$ 
     with $\gtGii = \gtG[k]$.  We are left to show that there exists $\stEnvi$ such that $\stEnv
  \stEnvMoveCommAnnot{\mpS}{\roleP}{\roleQ}{\stLab[k]} \stEnvi$ and 
      $\stEnvAssoc{\gtGii}{\stEnvi}{\mpS}$. 

We apply \inferrule{\iruleTCtxOut} on $\mpChanRole{\mpS}{\roleP}$ and
  \inferrule{\iruleTCtxIn} 
  on $\mpChanRole{\mpS}{\roleQ}$, which can be combined via
  \inferrule{\iruleTCtxCom}.
  By applying \inferrule{\iruleTCtxCong} and \inferrule{\iruleTCtxCongBasic} when needed, we have $\stEnv
  \stEnvMoveCommAnnot{\mpS}{\roleP}{\roleQ}{\stLab[k]} \stEnvi$, with  
 $\stEnvApp{\stEnvi}{\mpChanRole{\mpS}{\roleP}} =   \stTi[k]$, 
 $\stEnvApp{\stEnvi}{\mpChanRole{\mpS}{\roleQ}} =   \stTii[k]$, and 
  $\stEnvApp{\stEnvi}{\mpChanRole{\mpS}{\roleR}} =   \stEnvApp{\stEnv}{\mpChanRole{\mpS}{\roleR}}$ if 
  $\roleR \neq \roleP$ and $\roleR \neq \roleQ$.  
  
  Finally, we show that $\stEnvAssoc{\gtGii = \gtG[k]}{\stEnvi}{\mpS}$. 
  For $\roleP$, we have $\stTi[k] = \stEnvApp{\stEnvi}{\mpChanRole{\mpS}{\roleP}}
  \stSub \gtProj{\gtG[k]}{\roleP}$, and similar for $\roleQ$. 
   For $\roleR \neq \roleP \neq \roleQ \in \gtRoles{\gtG}$,
 it follows that 
 $
  \stEnvApp{\stEnv}{\mpChanRole{\mpS}{\roleR}} = 
  \stEnvApp{\stEnvi}{\mpChanRole{\mpS}{\roleR}} \stSub
  \stMerge{i \in I}{\gtProj{\gtG[i]}{\roleR}} =
  \gtProj{\gtG}{\roleR} 
 $.
  By \Cref{app:lem:subtyping:merging}, 
  $
    \stEnvApp{\stEnvi}{\mpChanRole{\mpS}{\roleR}}
    \stSub 
    \stMerge{i \in I}{\gtProj{\gtG[i]}{\roleR}} 
    \stSub 
    \gtProj{\gtG[k]}{\roleR} 
    $, and hence, 
    $\stEnvApp{\stEnvi}{\mpChanRole{\mpS}{\roleR}} \stSub
    \gtProj{\gtG[k]}{\roleR} $
    holds by the transitivity of $\stSub$. 
    For any $\rolePi \in \gtRoles{\gtG}$ where $\rolePi \notin \gtRoles{\gtG[k]}$, by~\Cref{lem:global_end}, $\gtProj{\gtG[k]}{\rolePi} = \stEnd$, implying that  $\stEnvApp{\stEnvi}{\mpChanRole{\mpS}{\rolePi}} = \stEnd$. 
    Additionally, for all other $\mpChanRole{\mpS}{\roleQi} \in \dom{\stEnv}$ where $\roleQi \notin \gtRoles{\gtG}$, we have 
    $\stEnvApp{\stEnv}{\mpChanRole{\mpS}{\roleQi}} = \stEnvApp{\stEnvi}{\mpChanRole{\mpS}{\roleQi}} = \stEnd$ by 
    $\stEnvAssoc{\gtG}{\stEnv}{\mpS}$ and (3) of~\Cref{lem:stenv-red:trivial-5}. 
    
\item Case \inferrule{\iruleGtMoveRec}:  
 
 By inductive hypothesis and $\inferrule{\iruleTCtxRec}$.
 
\item Case \inferrule{\iruleGtMoveCtx}: 

From the premise, we have: 
      \begin{gather}
        \stEnvAssoc{\gtG}{\stEnv}{\mpS}
        \label{eq:soundness:ctx_assoc_pre}
        \\
        \gtG = \gtCommSmall{\roleP}{\roleQ}{i \in I}{\gtLab[i]}{\tyS[i]}{\gtG[i]}
          \label{eq:soundness:ctx_global_pre}
          \\
           \forall i \in I: 
    \gtG[i] 
    \,\gtMove[\stEnvAnnotGenericSym]\, 
    \gtGi[i] 
    \\
     \ltsSubject{\stEnvAnnotGenericSym} \cap \setenum{\roleP, \roleQ} =
    \emptyset
    \\
    \gtGi = 
      \gtCommSmall{\roleP}{\roleQ}{i \in I}{\gtLab[i]}{\stS[i]}{\gtGi[i]}   
      \end{gather}
      
 Let $\stEnvAnnotGenericSym = \stEnvCommAnnotSmall{\roleR}{\roleU}{\stLab}$ 
 with $\setenum{\roleR, \roleU} \cap \setenum{\roleP, \roleQ} = \emptyset$. 
 
 By~\eqref{eq:soundness:ctx_assoc_pre} and~\eqref{eq:soundness:ctx_global_pre}, 
 it follows that for any role $\roleT$ with $\roleT \neq \roleP$ and $\roleT \neq \roleQ$,  
 $\stEnvApp{\stEnv}{\mpChanRole{\mpS}{\roleT}} \stSub \stMerge{i \in I}{\gtProj{\gtG[i]}{\roleT}}$. 
Let $j \in I$ be arbitrary. By~\Cref{app:lem:subtyping:merging} and transitivity of subtyping, we obtain 
 $\stEnvApp{\stEnv}{\mpChanRole{\mpS}{\roleT}} \stSub \gtProj{\gtG[j]}{\roleT}$.

  We define a typing context $\stEnv[j]$ by:  
    $\stEnvApp{\stEnv[j]}{\mpChanRole{\mpS}{\roleP}}
      = \gtProj{\gtG[j]}{\roleP}$,
    $\stEnvApp{\stEnv[j]}{\mpChanRole{\mpS}{\roleQ}}
      = \gtProj{\gtG[j]}{\roleQ}$,
    $\stEnvApp{\stEnv[j]}{\mpChanRole{\mpS}{\roleT}} =
      \stEnvApp{\stEnv}{\mpChanRole{\mpS}{\roleT}}$ for any other role $\roleT$. 
  It follows that  
      $\stEnvAssoc{\gtG[j]}{\stEnv[j]}{\mpS}$. 

Since $\gtG[j] \,\gtMove[\stEnvAnnotGenericSym]\, \gtGi[j]$ and $\stEnvAssoc{\gtG[j]}{\stEnv[j]}{\mpS}$,  by the induction hypothesis there exist 
$\stEnvAnnotGenericSymi = \stEnvCommAnnotSmall{\roleR}{\roleU}{\stLabi}$ and $\stEnvi[j]$ such that 
$\stEnv[j]  \stEnvMoveAnnot{\stEnvAnnotGenericSymi} \stEnvi[j]$. 

By inversion of typing context transition (\Cref{lem:stenv-red:inversion-basic}), together with 
$\stEnvApp{\stEnv[j]}{\mpChanRole{\mpS}{\roleR}} =
      \stEnvApp{\stEnv}{\mpChanRole{\mpS}{\roleR}}$ and 
      $\stEnvApp{\stEnv[j]}{\mpChanRole{\mpS}{\roleU}} =
      \stEnvApp{\stEnv}{\mpChanRole{\mpS}{\roleU}}$, 
      we obtain that     
      $
        \unfoldOne{\stEnvApp{\stEnv}{\mpChanRole{\mpS}{\roleR}}} =
        \stIntSum{\roleU}{k \in K}{\stChoice{\stLabi[\!k]}{\stSi[k]} \stSeq \stTi[k]}%
      $ and 
      $
        \unfoldOne{\stEnvApp{\stEnv}{\mpChanRole{\mpS}{\roleU}}} =
        \stExtSum{\roleR}{l \in L}{\stChoice{\stLabii[\!\!l]}{\stSii[l]} \stSeq \stTii[l]}%
      $, where there exist $k \in K$ and $l \in L$ such that  $\stLabi[\!k] = \stLabii[\!\!\!l] = \stLabi$ and 
      $\stSii[l] \stSub   \stSi[k]$. 
      
  Applying typing context transition rule \inferrule{\iruleTCtxCom} (and \inferrule{\iruleTCtxRec}, when necessary) to $\stEnv$, there exists $\stEnvi$ such that 
      $\stEnv  \stEnvMoveAnnot{\stEnvAnnotGenericSymi} \stEnvi$. 
      
      Finally,  by completeness of association~(\Cref{thm:gtype:proj-comp}), there exists $\gtGii$ such that 
      $\gtG 
    \,\gtMove[\stEnvAnnotGenericSymi]\, 
    \gtGii$ and  $\stEnvAssoc{\gtGii}{\stEnvi}{\mpS}$, as desired.  
    \qedhere

 \end{itemize}
\end{proof}

\subsection{Behavioural Properties Guaranteed by Association}
\label{app:properties_association}
\subsubsection{Safety by Association}

\begin{restatable}{lemma}{lemProjSafe}
  If  $\stEnvAssoc{\gtG}{\stEnv}{\mpS}$, 
  then $\stEnv$ is $\mpS$-safe.%
  \label{lem:safety_association}
\end{restatable}
\begin{proof}
  Let $\predP = \setcomp{\stEnvi}{
    \exists \gtGi:
    \gtG
    \,\gtMoveStar\, 
    \gtGi
    \text{~and~}
    \stEnvAssoc{\gtGi}{\stEnvi}{\mpS}
  }$.

  Take any $\stEnv \in \predP$, we show that $\stEnv$ satisfies all safety
  properties.
  By definition of $\predP$, there exists $\gtGi$ with
  $ \gtG 
     \,\gtMoveStar\, 
    \gtGi 
    \text{~and~}
    \stEnvAssoc{\gtGi}{\stEnv}{\mpS}
  $. 
  We only detail the case that $\gtGi \neq \gtEnd$ as 
  if $\gtGi = \gtEnd$, by~\Cref{app:lem:assoc:termination}, we have that $\stEnv = \stEnv[\stEnd]$, which satisfies all clauses of~\Cref{def:safety_typing_context}.

 \begin{description}[leftmargin=1.15cm, labelindent=1cm, align=right]
    \item[\inferrule{\iruleSafeComm}]
     Since (by hypothesis) 
      $\stEnvMoveAnnotP{\stEnv}{\stEnvOutAnnot{\roleP}{\roleQ}{\stChoice{\stLab}{\stS}}}$ 
      \,and\,   $\stEnvMoveAnnotP{\stEnv}{\stEnvInAnnot{\roleQ}{\roleP}{\stChoice{\stLabi}{\stSi}}}$, 
      by~\Cref{lem:stenv-red:inversion-basic}, we have  
      $
        \unfoldOne{\stEnvApp{\stEnv}{\mpChanRole{\mpS}{\roleP}}} =
        \stIntSum{\roleQ}{i \in I}{\stChoice{\stLab[i]}{\stS[i]} \stSeq \stT[i]}%
      $, $\exists k \in I$ such that $\stLab = \stLab[k]$, and 
      $
        \unfoldOne{\stEnvApp{\stEnv}{\mpChanRole{\mpS}{\roleQ}}} =
        \stExtSum{\roleP}{j \in J}{\stChoice{\stLabi[j]}{\stSi[j]} \stSeq \stTi[j]}%
      $. 
      Then we apply \Cref{app:lem:match_comm_projection} on $
	\stEnvApp{\stEnv}{\mpChanRole{\mpS}{\roleP}} \stSub \gtProj{\gtGi}{\roleP}$
      and
      $\stEnvApp{\stEnv}{\mpChanRole{\mpS}{\roleQ}} \stSub \gtProj{\gtGi}{\roleQ}$ 
      to get $I \subseteq J$, and $\forall i \in I: \stLab[i] 
      = \stLabi[i]$ and $\stSi[i] \stSub \stS[i]$. Consequently, along with $\stLab = \stLab[k]$, 
      by applying \inferrule{\iruleTCtxCom} (and \inferrule{\iruleTCtxRec} as needed), 
      $\stEnvMoveAnnotP{\stEnv}{\stEnvCommAnnotSmall{\roleP}{\roleQ}{\stLab}}$, which is the thesis.

    \item[\inferrule{\iruleSafeRec}]
      Let $\stEnvi$ be constructed from $\stEnv$ with
      $ \stEnvApp{\stEnvi}{\mpChanRole{\mpS}{\roleP}} =
        \stT\subst{\stRecVar}{\stRec{\stRecVar}{\stT}}$.
      By \Cref{app:lem:subtyping:unfolding}, we know $
        \stEnvApp{\stEnvi}{\mpChanRole{\mpS}{\roleP}} \stSub 
        \stEnvApp{\stEnv}{\mpChanRole{\mpS}{\roleP}}
      $, and thus $\stEnvi \stSub \stEnv$.
      By~\Cref{def:assoc} and transitivity of subtyping, we have
      $ \stEnvAssoc{\gtGi}{\stEnvi}{\mpS}$,
      which means that $\stEnvi \in \predP$.
   
    \item[\inferrule{\iruleMoveSession}]
      Let $\stEnv \stEnvMoveGenAnnot \stEnvi$ with 
  $\stEnvAnnotGenericSym = \stEnvCommAnnotSmall{\roleP}{\roleQ}{\stLab}$, meaning that 
  $\stEnv \!\stEnvMoveWithSession[\mpS]\! \stEnvi$.  
  By \Cref{thm:completeness_association}, there exists $\gtGii$
      with $\gtGi
      \,\gtMove[\stEnvAnnotGenericSym]\, 
      \gtGii$ %
      and
      $\stEnvAssoc{\gtGii}{\stEnvi}{\mpS}$.
      By definition of $\predP$, the typing context $\stEnvi$ after transition $\stEnv$ on session $\mpS$  
      is in $\predP$.
      \qedhere
  \end{description}
\end{proof}

\subsubsection{Deadlock-Freedom by Association}

\begin{restatable}{lemma}{lemProjDF}
  If  $\stEnvAssoc{\gtG}{\stEnv}{\mpS}$, 
  then $\stEnv$ is $\mpS$-deadlock-free.%
  \label{lem:df_association}
\end{restatable}
\begin{proof}
By operational correspondence of global type $\gtG$ and typing context
  $\stEnv$ (\Cref{thm:soundness_association,thm:completeness_association}), there exists
  a global type $\gtGi$ such that
  $\gtG \,\gtMoveStar\,  
  \gtNotMove{\gtGi}$, with
  associated typing contexts $\stEnv \!\stEnvMoveWithSessionStar[\mpS]\! 
 \stEnvNotMoveWithSessionP[\mpS]{\stEnvi}$.
  Since no further reductions are possible for the global type $\gtGi$, 
  it must be in the  form of 
  $\gtEnd$~(\Cref{app:lem:gtype:progress}).
Therefore, the thesis  holds  by~\Cref{app:lem:assoc:termination}. 
\qedhere 
\end{proof}

\subsubsection{Liveness by Association}

\begin{restatable}{lemma}{lemProjLive}
  If  $\stEnvAssoc{\gtG}{\stEnv}{\mpS}$, 
  then $\stEnv$ is $\mpS$-live.%
  \label{lem:live_association}
\end{restatable}
\begin{proof}
We want to show that %
  any path starting with $\stEnv$ which is fair for session $\mpS$ is also live for $\mpS$.
  We proceed by contradiction, assuming that there is a fair path for session $\mpS$: $(\stEnv[n])_{n \in N}$ 
where $N = \setenum{0, 1, 2, \ldots}$, $\stEnv[0] = \stEnv$, 
and $\forall n, n+1 \in N: \stEnv[n] \!\stEnvMoveWithSession[\mpS]\!
\stEnv[n+1]$, which is not live for $\mpS$. 
We consider the following two cases. 
\begin{itemize}%
\item  Case 
$\stEnvMoveAnnotP{\stEnv[j]}{\stEnvOutAnnot{\roleP}{\roleQ}{\stChoice{\stLab}{\stS}}}$ with $j \in N$, and  
for any $k \in N$ with $k \geq j$, there does not exist $\stLabi$ %
    such that $\stEnv[k] \stEnvMoveCommAnnot{\mpS}{\roleP}{\roleQ}{\stLabi}
      \stEnv[k+1]$: by operational correspondence of global type $\gtG$
and typing context $\stEnv$~(\Cref{thm:soundness_association,thm:completeness_association}), 
there exists
a global type $\gtG[j]$ such that
$\gtG \,\gtMoveStar\,  
\gtG[j]$ and
$\stEnvAssoc{\gtG[j]}{\stEnv[j]}{\mpS}$. 
Moreover, since $\stEnvMoveAnnotP{\stEnv[j]}{\stEnvOutAnnot{\roleP}{\roleQ}{\stChoice{\stLab}{\stS}}}$, 
by~\Cref{lem:stenv-red:inversion-basic,app:lem:inversion_association},
we have $\unfoldOne{\stEnvApp{\stEnv[j]}{\mpChanRole{\mpS}{\roleP}}} =
        \stIntSum{\roleQ}{i \in I}{\stChoice{\stLab[i]}{\stS[i]} \stSeq \stT[i]}$, and 
  \begin{itemize}%
         \item either $\unfoldOne{\gtG[j]} = 
            \gtComm{\roleP}{\roleQ}{i \in I'}{\gtLab[i]}{\tySi[i]}{\gtGii[i]}$,
          where $I \subseteq I'$, and for all $i \in I: \stLab[i] = \gtLab[i]$, 
          $\stSi[i] \stSub \tyS[i]$, and 
	  $\stT[i] \stSub \gtProj{\gtGii[i]}{\roleP}$. It follows directly that 
          $ 
	  \stEnvApp{\stEnv[j]}{\mpChanRole{\mpS}{\roleQ}}
	  \stSub 
	  \gtProj{\gtG[j]}{\roleQ} 
	  \stSub 
	  \gtProj{\unfoldOne{\gtG[j]}}{\roleQ} 
	  = 
	  \stExtSum{\roleP}{i \in I'}{\stChoice{\stLab[i]}{\stSi[i]} \stSeq \gtProj{\gtGii[i]}{\roleQ}}
	  $. Hence, we have 
$\stEnvMoveAnnotP{\stEnv[j]}{\stEnvInAnnot{\roleQ}{\roleP}{\stChoice{\stLab}{\stSi}}}$  with $\stSi \stSub \stS$, and therefore,  $\stEnv[j] \stEnvMoveCommAnnot{\mpS}{\roleP}{\roleQ}{\stLab}$.  With the fact that $(\stEnv[n])_{n \in N}$ is a fair path for session $\mpS$, we know that there exist some $k$ and $\stLabi$ such that $k \in N$, $k \geq j$, and  
$\stEnv[k] \stEnvMoveCommAnnot{\mpS}{\roleP}{\roleQ}{\stLabi} \stEnv[k+1]$, a desired contradiction.     
          
          \item or 
           $\unfoldOne{\gtG[j]} = \gtComm{\roleS}{\roleT}{l \in
            L}{\gtLab[l]}{\tySi[l]}{\gtGii[l]}$,
          where for all $l \in L:  
	  \stEnvApp{\stEnv[j]}{\mpChanRole{\mpS}{\roleP}} \stSub \gtProj{\gtGii[l]}{\roleP}$,
          with $\roleP \neq \roleS$ and $\roleP \neq \roleT$. 
      By the assumption that for any $k \in N$ with $k \geq j$, there does not exist $\stLabi$ %
    such that $\stEnv[k] \stEnvMoveCommAnnot{\mpS}{\roleP}{\roleQ}{\stLabi}
      \stEnv[k+1]$, $\unfoldOne{\stEnvApp{\stEnv[j]}{\mpChanRole{\mpS}{\roleP}}} =
        \stIntSum{\roleQ}{i \in I}{\stChoice{\stLab[i]}{\stS[i]} \stSeq \stT[i]}$, we know that transmission  
        between $\roleP$ and $\roleQ$ will never occur from $\stEnv[j]$. By operational correspondence, we also conclude that such transmission will never be triggered from $\gtG[j]$. This is a desired contradiction, as by~\Cref{app:lem:inversion_projection},         
        each continuation $\gtGii[l]$ of $\gtG[j]$ must involve transmission between $\roleP$ and $\roleQ$, 
         and therefore, the previous subcase scenario should inevitably occur.  
         
\end{itemize}
\item Case  
$\stEnvMoveAnnotP{\stEnv[j]}{\stEnvInAnnot{\roleQ}{\roleP}{\stChoice{\stLab}{\stS}}}$ with $j \in N$, and  
for any $k \in N$ with $k \geq j$, there does not exist $\stLabi$ %
    such that $\stEnv[k] \stEnvMoveCommAnnot{\mpS}{\roleP}{\roleQ}{\stLabi}
      \stEnv[k+1]$: similar to the previous case. 
\qedhere 
\end{itemize}
\end{proof}

\colProjAll*
\begin{proof}
	Apply~\Cref{lem:safety_association,lem:df_association,lem:live_association}.
	\qedhere 
\end{proof}

\subsection{Relating Typing Context Properties}%
\label{sec:proof:relating:ctx}
\lemCtxSafetyImplications*

\begin{proof}
The negated implications in the statement are 
demonstrated in~\Cref{ex:ctx_df,ex:ctx_liveness,ex:ctx_not-assoc}. 
 Let's now consider the remaining implications. 

 \begin{description}

\item[(\ref{item:liveness-safety-implic:live-df}).]%
Assume $\stEnvLiveSessP{\mpS}{\stEnv}$. We need to 
prove $\stEnvDFSessP{\mpS}{\stEnv}$, \ie 
for any path $\stEnvNew[0] \!\stEnvMoveWithSession[\mpS]\! \stEnvNew[1] \!\stEnvMoveWithSession[\mpS]\! \cdots \!\stEnvMoveWithSession[\mpS]\! \stEnvNew[n]$ with $n > 0$, $\stEnvNew[0] = \stEnvNew$, and 
$\stEnvNotMoveWithSessionP[\mpS]{\stEnvNew[n]}$, it holds that $\forall \mpChanRole{\mpS}{\roleP} \in \dom{\stEnvNew[n]}: 
\stEnvApp{\stEnvNew[n]}{\mpChanRole{\mpS}{\roleP}} = \stEnd$. 

We first show that any such path $\stEnvNew[0] \!\stEnvMoveWithSession[\mpS]\! \stEnvNew[1] \!\stEnvMoveWithSession[\mpS]\! \cdots \!\stEnvMoveWithSession[\mpS]\! \stEnvNew[n]$ is fair for session $\mpS$ by contradiction, assuming it is not fair for $\mpS$. 
In this case, 
there exists $k$ such that $0 \leq k \leq n$ and $\stEnvMoveAnnotP{\stEnvNew[k]}{\ltsSendRecv{\mpS}{\roleP}{\roleQ}{\stLab}}$, while 
for all $j$ and $\stLabi$ with $k \leq j \leq n$, $\stEnvNew[j] \stEnvMoveCommAnnot{\mpS}{\roleP}{\roleQ}{\stLabi} \stEnvNew[j+1]$ does not occur in the path. Hence, by~\Cref{lem:stenv-red:trivial-2}, we have that  $\stEnvApp{\stEnv[n]}{\mpChanRole{\mpS}{\roleP}} =
      \stEnvApp{\stEnv[k]}{\mpChanRole{\mpS}{\roleP}}$ and 
      $\stEnvApp{\stEnv[n]}{\mpChanRole{\mpS}{\roleQ}} =
      \stEnvApp{\stEnv[k]}{\mpChanRole{\mpS}{\roleQ}}$. It follows directly that $\stEnvMoveAnnotP{\stEnvNew[n]}{\ltsSendRecv{\mpS}{\roleP}{\roleQ}{\stLab}}$, contradicting $\stEnvNotMoveWithSessionP[\mpS]{\stEnvNew[n]}$.

We proceed 
by contradiction again, 
assuming that there exists $\mpChanRole{\mpS}{\roleP} \in \dom{\stEnv[n]}$ 
such that 
$\stEnvApp{\stEnv[n]}{\mpChanRole{\mpS}{\roleP}} \not \stSub \stEnd$, \ie 
$\unfoldOne{\stEnvApp{\stEnv[n]}{\mpChanRole{\mpS}{\roleP}}} = %
    \stExtSum{\roleQ}{i \in I}{\stChoice{\stLab[i]}{\stS[i]} \stSeq \stT[i]}$ or 
$\unfoldOne{\stEnvApp{\stEnv[n]}{\mpChanRole{\mpS}{\roleP}}} =  
\stIntSum{\roleQ}{i \in I}{\stChoice{\stLab[i]}{\stS[i]} \stSeq \stT[i]}$.  

We now detail the case where $\unfoldOne{\stEnvApp{\stEnv[n]}{\mpChanRole{\mpS}{\roleP}}} = %
    \stExtSum{\roleQ}{i \in I}{\stChoice{\stLab[i]}{\stS[i]} \stSeq \stT[i]}$, with the other case following similarly. 

    Since $\unfoldOne{\stEnvApp{\stEnv[n]}{\mpChanRole{\mpS}{\roleP}}} = \stExtSum{\roleQ}{i \in I}{\stChoice{\stLab[i]}{\stS[i]} \stSeq \stT[i]}$, we have 
    $\stEnvMoveAnnotP{\stEnv[n]}{\stEnvInAnnot{\roleP}{\roleQ}{\stChoice{\stLab[k]}{\stS[k]}}}$ with $k \in I$. 
     Given that  $\stEnv[0]$ is $\mpS$-live and $\stEnvNew[0] \!\stEnvMoveWithSession[\mpS]\! \stEnvNew[1] \!\stEnvMoveWithSession[\mpS]\! \cdots \!\stEnvMoveWithSession[\mpS]\! \stEnvNew[n]$ is fair for $\mpS$, 
     there exist $j$ and $\stLabi$ such that $n \leq j \leq n$ and $\stEnvNew[j]
\!\stEnvMoveCommAnnot{\mpS}{\roleP}{\roleQ}{\stLabi}\! \stEnvNew[j+1]$. This follows 
$\stEnvNew[n] \!\stEnvMoveWithSession[\mpS]$, 
contradicting $\stEnvNotMoveWithSessionP[\mpS]{\stEnvNew[n]}$.

  \item[(\ref{item:liveness-safety-implic:asso-safe}).]%
  Straightforward from~\Cref{lem:safety_association}.

  \item[(\ref{item:liveness-safety-implic:asso-df}).]%
  Straightforward from~\Cref{lem:df_association}. 
  
  \item[(\ref{item:liveness-safety-implic:asso-live}).]%
  Straightforward from~\Cref{lem:live_association}.
\qedhere 
\end{description}

\end{proof}

\section{Subtyping Properties Regarding  Association}
\label{sec:app:subtyping_within_assoc}

\begin{lemma}
\label{lem:stenv-assoc-sub}
If $\stEnvAssoc{\gtG}{\stEnv}{\mpS}$ and 
$\stEnvi \stSub \stEnv$, 
then $\stEnvAssoc{\gtG}{\stEnvi}{\mpS}$. 
\end{lemma}

\begin{proof}
By the definition of association~(\Cref{def:association}), 
the definition of $\stEnvi \stSub \stEnv$~(\Cref{def:typing_context}), and 
the transitivity of subtyping~(\Cref{app:lem:subtyping:transitive}). 
\qedhere
\end{proof}

\begin{restatable}{lemma}{lemStenvReductionSubAssoc}
    \label{lem:stenv-assoc-reduction-sub}%
    Assume that\, $\stEnvAssoc{\gtG}{\stEnv}{\mpS}$,
    $\stEnvi \stSub \stEnv$ and $\stEnvi \stEnvMoveGenAnnot \stEnvii$ %
    with 
      $\stEnvAnnotGenericSym \!\in\! \setcomp{
      \ltsSendRecv{\mpS}{\roleP}{\roleQ}{\stLab}}{\,\roleP,\!\roleQ \!\in\! \roleSet}$. 
    Then, there is\, $\stEnviii$ such that\; %
    $\stEnv \stEnvMoveGenAnnot \stEnviii$ and $\stEnvii \stSub \stEnviii$.%
  \end{restatable}
\begin{proof}
Since $\stEnvi 
      \!\stEnvMoveCommAnnot{\mpS}{\roleP}{\roleQ}{\stLab}\!%
      \stEnvii$, by applying and inverting \inferrule{\iruleTCtxCom}~(and \inferrule{\iruleTCtxRec} when necessary),  and also using~\Cref{lem:stenv-red:inversion-basic}, we have 
      $ \unfoldOne{\stEnvApp{\stEnvi}{\mpChanRole{\mpS}{\roleP}}} =
         \stIntSum{\roleQ}{i \in I}{\stChoice{\stLab[i]}{\stS[i]} \stSeq \stT[i]}$, 
          $
        \unfoldOne{\stEnvApp{\stEnvi}{\mpChanRole{\mpS}{\roleQ}}} =
        \stExtSum{\roleP}{j \in J}{\stChoice{\stLabi[j]}{\stSi[j]} \stSeq \stTi[j]}%
      $, $\exists k: k \in I$, $k \in J$, $\stLab[k] = \stLabi[k] = \stLab$, 
      $\stEnvApp{\stEnvii}{\mpChanRole{\mpS}{\roleP}} = \stT[k]$, and 
       $\stEnvApp{\stEnvii}{\mpChanRole{\mpS}{\roleQ}} = \stTi[k]$. 
       Furthermore, by~\Cref{lem:stenv-red:trivial-1-new,lem:stenv-red:trivial-2}, 
       it also holds that for all  $\mpChanRole{\mpS}{\roleR} \in \dom{\stEnvi} = \dom{\stEnvii}$ with 
      $\roleR \neq \roleP$ and $\roleR \neq \roleQ$, %
      $\stEnvApp{\stEnvi}{\mpChanRole{\mpS}{\roleR}} =
      \stEnvApp{\stEnvii}{\mpChanRole{\mpS}{\roleR}}$. 
     Observe that by $\stEnvi \stSub \stEnv$,  
     $\unfoldOne{\stEnvApp{\stEnv}{\mpChanRole{\mpS}{\roleP}}} =
         \stIntSum{\roleQ}{i \in I_{\roleP}}{\stChoice{\stLab[i]}{\stSii[i]} \stSeq \stTii[i]}$ where 
         $I \subseteq I_{\roleP}$ and $\forall i \in I: \stSii[i] \stSub \stS[i]$ and 
         $\stT[i] \stSub \stTii[i]$, and 
        $
        \unfoldOne{\stEnvApp{\stEnv}{\mpChanRole{\mpS}{\roleQ}}} =
        \stExtSum{\roleP}{j \in J_{\roleQ}}{\stChoice{\stLabi[j]}{\stSiii[j]} \stSeq \stTiii[j]}%
      $ where 
         $J_{\roleQ} \subseteq J$ and $\forall j \in J_{\roleQ}: \stSi[j] \stSub \stSiii[j]$ and 
         $\stTi[j] \stSub \stTiii[j]$. 

         Now we apply~\Cref{app:lem:match_comm_projection} on $\stEnvAssoc{\gtG}{\stEnv}{\mpS}$, 
          $\unfoldOne{\stEnvApp{\stEnv}{\mpChanRole{\mpS}{\roleP}}}$, and 
           $\unfoldOne{\stEnvApp{\stEnv}{\mpChanRole{\mpS}{\roleQ}}}$, to get 
           $I \subseteq I_{\roleP} \subseteq J_{\roleQ} \subseteq J$, and 
           $\forall i \in I_{\roleP}: \stLab[i] = \stLabi[i]$ and $\stSiii[i] \stSub \stSii[i]$. 
           Consequently, we have $k \in I_{\roleP}$, $\stLab[k] = \stLabi[k] = \stLab$, and 
           $\stSiii[k] \stSub \stSii[k]$, which follows that there exists 
           $\stEnviii$ such that $\stEnv
      \!\stEnvMoveCommAnnot{\mpS}{\roleP}{\roleQ}{\stLab}\!%
      \stEnviii$, $\stEnvApp{\stEnviii}{\mpChanRole{\mpS}{\roleP}} = \stTii[k]$, 
       $\stEnvApp{\stEnviii}{\mpChanRole{\mpS}{\roleQ}} = \stTiii[k]$, and 
       for all  $\mpChanRole{\mpS}{\roleR} \in \dom{\stEnv} = \dom{\stEnvi} = 
       \dom{\stEnvii} = \dom{\stEnviii}$ with 
      $\roleR \neq \roleP$ and $\roleR \neq \roleQ$, %
      $\stEnvApp{\stEnv}{\mpChanRole{\mpS}{\roleR}} =
      \stEnvApp{\stEnviii}{\mpChanRole{\mpS}{\roleR}}$. 
      
      We are left to show that $\stEnvii \stSub \stEnviii$, which is straightforward from 
      $\stEnvApp{\stEnvii}{\mpChanRole{\mpS}{\roleP}} = \stT[k] 
      \stSub \stTii[k] = \stEnvApp{\stEnviii}{\mpChanRole{\mpS}{\roleP}}$, 
      $\stEnvApp{\stEnvii}{\mpChanRole{\mpS}{\roleQ}} = \stTi[k] 
      \stSub \stTiii[k] = \stEnvApp{\stEnviii}{\mpChanRole{\mpS}{\roleQ}}$, and 
      for all  $\mpChanRole{\mpS}{\roleR} \in \dom{\stEnviii} =
       \dom{\stEnvii}$ with 
      $\roleR \neq \roleP$ and $\roleR \neq \roleQ$, %
      $\stEnvApp{\stEnvii}{\mpChanRole{\mpS}{\roleR}} =
      \stEnvApp{\stEnvi}{\mpChanRole{\mpS}{\roleR}} \stSub 
       \stEnvApp{\stEnv}{\mpChanRole{\mpS}{\roleR}} =  
       \stEnvApp{\stEnviii}{\mpChanRole{\mpS}{\roleR}}$. 
       \qedhere
\end{proof}

\begin{restatable}{lemma}{lemStenvReductionSubAssocGeneral}
    \label{lem:stenv-assoc-reduction-sub-gen}%
    Assume that\, $\forall \mpS \in \stEnv: \exists \gtG[\mpS]: 
    \stEnvAssoc{\gtG[\mpS]}{\stEnv[\mpS]}{\mpS}$,
   $\stEnvi \stSub \stEnv$ and $\stEnvi \stEnvMoveGenAnnot \stEnvii$ %
    with:
      $\stEnvAnnotGenericSym \in \setcomp{
     \ltsSendRecv{\mpS}{\roleP}{\roleQ}{\stLab}}{\,\roleP,\!\roleQ \!\in\! \roleSet}$. 
    Then, there is\, $\stEnviii$ such that\; %
    $\stEnv \stEnvMoveGenAnnot \stEnviii$ and $\stEnvii \stSub \stEnviii$.%
  \end{restatable}
  \begin{proof}
  Apply~\Cref{def:typing_context} and~\Cref{lem:stenv-red:trivial-5,lem:stenv-assoc-reduction-sub}. 
  \qedhere 
  \end{proof}

\begin{restatable}{proposition}{lemStenvReductionSubAssocInd}%
  \label{lem:stenv-assoc-reduction-sub-ind}%
  Assume that\, $\stEnvAssoc{\gtG}{\stEnv}{\mpS}$  
  ,\, 
  $\stEnvi \stSub \stEnv$ and $\stEnvi \stEnvMoveAnnot{\stEnvAnnotGenericSym[1]}\cdots\stEnvMoveAnnot{\stEnvAnnotGenericSym[n]} \stEnvii$, with:
    $\forall i \in 1...n: \stEnvAnnotGenericSym[i] \in \setcomp{
        \ltsSendRecv{\mpS}{\roleP}{\roleQ}{\stLab[i]}}{\,\roleP,\!\roleQ \!\in\! \roleSet}$. 
  Then, there is\, $\stEnviii$ such that\; %
  $\stEnv \stEnvMoveAnnot{\stEnvAnnotGenericSym[1]}\cdots\stEnvMoveAnnot{\stEnvAnnotGenericSym[n]} \stEnviii$ and $\stEnvii \stSub \stEnviii$.%
\end{restatable}
\begin{proof}
    By induction on the number of transitions $n$ in $\stEnvi \stEnvMoveAnnot{\stEnvAnnotGenericSym[1]}\cdots\stEnvMoveAnnot{\stEnvAnnotGenericSym[n]} \stEnvii$.
    The base case ($n = 0$ transitions) is immediate: we have $\stEnvi = \stEnvii$,
    hence we conclude by taking $\stEnviii = \stEnv$.
    In the inductive case with $n = m+1$ transitions,
    there is $\stEnvii[0]$ such that $\stEnvi \stEnvMoveAnnot{\stEnvAnnotGenericSym_1}\cdots\stEnvMoveAnnot{\stEnvAnnotGenericSym[m]} \stEnvii[0] \stEnvMoveAnnot{\stEnvAnnotGenericSym[n]} \stEnvii$.
    By the induction hypothesis, there is $\stEnviii[0]$ such that
    $\stEnv \stEnvMoveAnnot{\stEnvAnnotGenericSym[1]}\cdots\stEnvMoveAnnot{\stEnvAnnotGenericSym[m]} \stEnviii[0]$ and $\stEnvii[0] \stSub \stEnviii[0]$. %
    By $\stEnvAssoc{\gtG}{\stEnv}{\mpS}$,  $\stEnv \stEnvMoveAnnot{\stEnvAnnotGenericSym[1]}\cdots\stEnvMoveAnnot{\stEnvAnnotGenericSym[m]} \stEnviii[0]$, and applying completeness of association~(\Cref{thm:gtype:proj-comp}) $m$ times, there exists some $\gtGiii$ such that $\stEnvAssoc{\gtGiii}{\stEnviii[0]}{\mpS}$. 
    Hence, by $\stEnvAssoc{\gtGiii}{\stEnviii[0]}{\mpS}$, $\stEnvii[0] \stEnvMoveAnnot{\stEnvAnnotGenericSym[n]} \stEnvii$, 
    $\stEnvii[0] \stSub \stEnviii[0]$, and~\Cref{lem:stenv-assoc-reduction-sub}, there exists $\stEnviii$
    such that $\stEnviii[0] \stEnvMoveAnnot{\stEnvAnnotGenericSym[n]} \stEnviii$
    and $\stEnvii \stSub \stEnviii$. Therefore, we have
    $\stEnv \stEnvMoveAnnot{\stEnvAnnotGenericSym[1]}\cdots\stEnvMoveAnnot{\stEnvAnnotGenericSym[n]} \stEnviii$ and $\stEnvii \stSub \stEnviii$,
    which is the thesis.
    \qedhere
\end{proof}

\section{Proofs for~\Cref{sec:typesystem}}
\label{app:sec:section_4_proof}

\subsection{Type System Properties}
\begin{restatable}[Broadening]{lemma}{lemBroadening}%
	\label{app:lem:narrowing}
	\label{lem:broading}
    If %
    \,$\stJudgeNew{\mpEnv}{\stEnv}{\mpP}$
    and %
    $\stEnv \stSub \stEnvi$, %
   then %
    $\stJudgeNew{\mpEnv}{\stEnvi}{\mpP}$.
  \end{restatable}
  \begin{proof}
    By induction on the derivation of $\stJudgeNew{\mpEnv}{\stEnv}{\mpP}$,
    we obtain a derivation that concludes $\stJudgeNew{\mpEnv}{\stEnvi}{\mpP}$
    by inserting (possibly vacuous) instances of typing rule \inferrule{\iruleMPSub}~(\Cref{fig:typing_rules}).
\end{proof}

\begin{lemma}[Terminated Typing Context]
\label{app:lem:termination}
If \,$\stEnvEndP{\stEnv}$, then either $\stEnv = \stEnvEmpty$ or $\forall \mpC \in \dom{\stEnv}: \stEnvApp{\stEnv}{\mpC} = \stEnd$. 
\end{lemma}
\begin{proof}
By inversion of \inferrule{\iruleMPEnd} and $\stEnvEndP{\stEnvEmpty}$. 
\qedhere 
\end{proof}

\begin{restatable}[Typing Inversion]{lemma}{lemTypingInversion}
  \label{app:lem:typing_inversion}%
  Assume $\stJudgeNew{\mpEnv}{\stEnvi}{\mpP}$. %
  Then there exists $\stEnv \stSub \stEnvi$ such that 
  \begin{enumerate}[leftmargin=*,label=(\arabic*),ref=\arabic*]
  \item
  \label{item:mpst-typing-inversion:nil}%
    $\mpP = \mpNil$ %
    implies $\stEnvEndP{\stEnv}$;%

  \item
  \label{item:mpst-typing-inversion:def}%
    $\mpP = %
    \mpDef{\mpX}{%
    \mpEnvMap{x_1}{\tyGround[1]}, 
    \ldots, 
    \mpEnvMap{x_n}{\tyGround[n]}, 
    \stEnvMap{y_1}{\stT[1]}, 
    \ldots, 
    \stEnvMap{y_m}{\stT[m]}
    }{\mpPi}{\mpQ}$ %
    implies:%
    \begin{enumerate}[label=(\roman*)]
    \item%
      $\stJudgeNew{\mpEnv \mpEnvComp 
      \mpEnvMap{\mpX}{\stB[1],\ldots,\stB[n],\stT[1],\ldots,\stT[m]}
\mpEnvComp \mpEnvMap{x_1}{\stB[1]} \mpEnvComp \ldots \mpEnvComp \mpEnvMap{x_n}{\stB[n]}}{%
	      \stEnv \stEnvComp
        \stEnvMap{y_1}{\stT[1]}%
        \stEnvComp \ldots \stEnvComp%
        \stEnvMap{y_m}{\stT[m]}%
      }{%
        \mpPi%
      }$,   %
      \\ and  %
    \item%
      $\stJudgeNew{%
        \mpEnv \mpEnvComp%
	\mpEnvMap{\mpX}{\stB[1],\ldots,\stB[n], \stT[1],\ldots,\stT[m]}%
      }{%
        \stEnv%
      }{%
        \mpQ%
      }$;  %
    \end{enumerate}
   
  \item
  \label{item:mpst-typing-inversion:call}%
  $\mpP = \mpCall{\mpX}{\mpE[1],\ldots,\mpE[n],\mpC[1],\ldots,\mpC[m]}$ %
    implies:
    \begin{enumerate}[label=(\roman*)]
\item%
	$\stEnv = \stEnvii, \stEnvMap{c_1}{\stT[1]},\ldots , \stEnvMap{c_m}{\stT[m]}
	$, and
    \item%
	    $
	    \mpEnv \vdash
        \mpEnvMap{\mpE[1]}{\tyGround[1]},\ldots,\mpEnvMap{\mpE[n]}{\tyGround[n]}
      $, and 
    \item%
      $\stEnvEndP{\stEnvii}$; 
    \end{enumerate}
  
  \item
  \label{item:mpst-typing-inversion:res}%
    $\mpP = \mpRes{\stEnvMap{\mpS}{\stEnvii}}\mpPi$ %
    implies:
    \begin{enumerate}[label=(\roman*)]
    \item%
      $\mpS \not\in \stEnv$, and %
    \item $\stEnvii =
  \setenum{\stEnvMap{\mpChanRole{\mpS}{\roleP}}{\gtProj{\gtG}{\roleP}}}_{%
    \roleP \in \gtRoles{\gtG}}$ for some $\gtG$,  and 
        \item%
      $\stJudgeNew{\mpEnv}{%
        \stEnv \stEnvComp \stEnvii%
      }{%
        \mpPi%
      }$;  %
    \end{enumerate}
    
  \item
  \label{item:mpst-typing-inversion:par}%
    $\mpP = \mpP[1] \mpPar \mpP[2]$ %
    implies:
    \begin{enumerate}[label=(\roman*)]
    \item%
      $\stEnv = \stEnv[1] \stEnvComp \stEnv[2]$ %
      such that%
    \item%
      $\stJudgeNew{\mpEnv}{%
        \stEnv[1]%
      }{%
        \mpP[1]%
      }$, and%
    \item%
      $\stJudgeNew{\mpEnv}{%
        \stEnv[2]%
      }{%
        \mpP[2]%
      }$;  %
    \end{enumerate}
  
  \item
  \label{item:mpst-typing-inversion:branch}%
    $\mpP = \mpBranch{\mpC}{\roleQ}{i \in I}{\mpLab[i]}{z_i}{\mpP[i]}{}$ %
    implies:
    \begin{enumerate}[label=(\roman*)]
    \item%
      $\stEnv = \stEnvii, \stEnvMap{\mpC}{
      \stExtSum{\roleQ}{i \in I}{\stChoice{\stLab[i]}{\tyS[i]} \stSeq \stT[i]}}$ such that
    \item%
      $\forall i \!\in\! I:\;%
      \left\{
        \begin{array}{c}
        \stJudgeNew{\mpEnvNew \mpEnvComp \stEnvMap{z_i}{\tyGround}}{
          \stEnvii \stEnvComp
                     \stEnvMap{\mpC}{\stT[i]}
        }{
          \mpP[i]
        } \,\,\,\, \text{if $\stS[i] = \tyGround$}
        \\
       \stJudgeNew{\mpEnvNew}{
          \stEnvii \stEnvComp \stEnvMap{z_i}{\stT} \stEnvComp 
                     \stEnvMap{\mpC}{\stT[i]}
        }{
          \mpP[i]
        } \,\,\,\, \text{if $\stS[i] = \stT$}
        \end{array}
      \right\}$; 
    \end{enumerate}
   
  \item\label{item:mpst-typing-inversion:sel:1}%
    $\mpP = \mpSel{\mpC}{\roleQ}{\mpLab}{\mpE}{\mpPi}$ implies:
    \begin{enumerate}[label=(\roman*)]
	    \item $\stEnv = \stEnvii, \stEnvMap{\mpC}{\stIntSum{\roleQ}{}{\stChoice{\stLab}{\tyGround} \stSeq \stT}}$ such that 

	    \item $\mpEnv \vdash \stEnvMap{\mpE}{\tyGround}$, and 
	    \item $\stJudgeNew{\mpEnv}{\stEnvii \stEnvComp \stEnvMap{c}{\stT}}{\mpPi}$; 
    \end{enumerate}
  \item\label{item:mpst-typing-inversion:sel:2}%
    $\mpP = \mpSel{\mpC}{\roleQ}{\mpLab}{\mpCi}{\mpPi}$ implies:
    \begin{enumerate}[label=(\roman*)]
	    \item $\stEnv = \stEnvii \stEnvComp \stEnvMap{\mpC}{\stIntSum{\roleQ}{}{\stChoice{\stLab}{\stTi} \stSeq \stT}} \stEnvComp \stEnvMap{\mpCi}{\stTi}$ such that 
	    \item $\stJudgeNew{\mpEnv}{\stEnvii \stEnvComp \stEnvMap{c}{\stT}}{\mpPi}$; 
    \end{enumerate}

\item \label{item:mpst-typing-inversion:if}
	$\mpP = \mpIf{e}{\mpQ[1]}{\mpQ[2]}$ implies: 
	\begin{enumerate}[label=(\roman*)]
		\item $\stEnvEntails{\mpEnv}{\mpE}{\tyBool}$, and 
		\item $\stJudgeNew{\mpEnv}{\stEnv}{\mpQ[1]}$, and 
		\item $\stJudgeNew{\mpEnv}{\stEnv}{\mpQ[2]}$. 
	\end{enumerate}
    
  \end{enumerate}
  \end{restatable} 
  \begin{proof}
By inverting $\inferrule{\iruleMPSub}$ and then applying induction on typing rules in \Cref{fig:typing_rules}. %
	  \qedhere
  \end{proof} 

\begin{restatable}[Substitution]{lemma}{lemSubstitution}
 \label{app:lem:substitution}%

\quad 

\begin{itemize}
	  \item If $\stJudgeNew{\mpEnv \mpEnvComp \mpEnvMap{x}{\stB}}{\stEnv}{\mpP}$ 
		  and $\mpEnvEntails{\mpEnv}{\mpV}{\tyGround}$, 
		  then $\stJudgeNew{\mpEnv}{\stEnv}{\mpP\subst{\mpFmt{x}}{\mpV}}$.
	  \item If $\stJudgeNew{\mpEnv}{\stEnv \mpEnvComp \mpEnvMap{\mpFmt{y}}{\stT}}{\mpP}$,
		  then $\stJudgeNew{\mpEnv}{\stEnv \mpEnvComp \mpEnvMap{\mpChanRole{\mpS}{\roleP}}{\stT}}{\mpP \subst{\mpFmt{y}}{\mpChanRole{\mpS}{\roleP}}}$.
  \end{itemize}
\end{restatable}
\begin{proof}
	By induction on typing rules in \Cref{fig:typing_rules}. 
	\qedhere%
  \end{proof}

  \begin{restatable}[Subject Congruence]{lemma}{lemSubjectCongruence}
  \label{lem:subject-congruence}%
  Assume %
  $\stJudgeNew{\mpEnv}{\stEnv}{\mpP}$ %
   and %
  $\mpP \equiv \mpPi$. %
   Then, %
  $\stJudgeNew{\mpEnv}{\stEnv}{\mpPi}$.
\end{restatable}
\begin{proof}
By analysing the cases where $\mpP \equiv \mpPi$ holds, and by applying the inversion of the typing judgements $\stJudgeNew{\mpEnv}{\stEnv}{\mpP}$  and  $\stJudgeNew{\mpEnv}{\stEnv}{\mpPi}$. %
\qedhere 
\end{proof}

\begin{restatable}{lemma}{lemSubTermination}
\label{lem:termination_subtyping}
\label{app:lem:termination_subtyping}
If \,$\stJudgeNew{\mpEnv}{\stEnv \stEnvComp \stEnvi}{\mpP}$, $\stEnvEndP{\stEnvi}$, and 
$\stEnv \stSub \stEnvii$, then $\stJudgeNew{\mpEnv}{\stEnvii}{\mpP}$. 
\end{restatable}
\begin{proof}
By $\mpP \equiv \mpP \mpPar \mpNil$ and~\Cref{lem:subject-congruence}, we have $\stJudgeNew{\mpEnv}{\stEnv \stEnvComp \stEnvi}{\mpP \mpPar \mpNil}$. Inverting \inferrule{\iruleMPNil} and \inferrule{\iruleMPPar}, we obtain $\stJudgeNew{\mpEnv}{\stEnv}{\mpP}$. Hence, the thesis follows by applying \inferrule{\iruleMPSub}. 
\qedhere 
\end{proof}

\subsection{Subject Reduction}
\label{app:sec:subject_reduction_proof}
\lemSubjectReduction*
\begin{proof}
 Let us recap the assumptions:
  \begin{align}
    \label{item:subjred:typed-unused}
    &\stJudgeNew{\mpEnv}{\stEnv}{\mpP}
    \\
    \label{item:subjred:stenv-assoc}
    &
    \forall \mpS \in \stEnv: \exists
    \gtG[\mpS]: \stEnvAssoc{\gtG[\mpS]}{\stEnv[\mpS]}{\mpS} 
    \\
    \label{item:subjred:no-reliable-crash}
    &\mpP \!\mpMove\! \mpPi
  \end{align}
  
  \noindent
  The proof proceeds by induction of the derivation of %
  $\mpP \!\mpMove\! \mpPi$, %
  and when the reduction holds by rule $\inferrule{\iruleMPRedCtx}$, %
  with a further structural induction on the reduction context $\mpCtx$. %
  Most cases hold %
  by inverting the typing $\stJudgeNew{\mpEnv}{\stEnv}{\mpP}$~(\Cref{app:lem:typing_inversion}), %
 and applying the induction hypothesis along with~\Cref{lem:broading}. 
\begin{itemize}%
 \item Case \inferrule{\iruleMPRedCommExp}: 
	 \begin{multline}
    \label{eq:subj-red:comm:p-pi:1}%
    \begin{array}{rcl}%
      \textstyle%
      \mpP &=&%
      \mpBranch{\mpChanRole{\mpS}{\roleP}}{\roleQ}{i \in I}{%
        \mpLab[i]}{z_i}{\mpP[i]}{}
      \,\mpPar\,%
      \mpSel{\mpChanRole{\mpS}{\roleQ}}{\roleP}{\mpLab[k]}{%
        \mpE
      }{\mpQ}%
      \\[1mm]%
      \mpPi &=&%
      \mpP[k]\subst{\mpFmt{z_k}}{
        \mpV
      }%
      \,\mpPar\,%
      \mpQ%
      \quad%
      (k \in I, \eval{\mpE}{\mpV})%
    \end{array}
    \\
    \text{%
      (by inversion of \inferrule{\iruleMPRedCommExp})%
    }%
\end{multline}
 \begin{multline}
    \label{eq:subj-red:comm:p-typing:1}%
    \begin{array}{l}
    \stEnv[\stExtC] \stEnvComp \stEnv[\stIntC] \stSub \stEnv
    \quad\text{such that}\quad%
    \\\qquad
      \begin{array}{l}
        \stJudgeNew{\mpEnv}{%
          \stEnv[\stExtC]%
        }{%
          \mpBranch{\mpChanRole{\mpS}{\roleP}}{\roleQ}{i \in I}{%
            \mpLab[i]}{z_i}{\mpP[i]}{}
        }%
        \\%
        \stJudgeNew{\mpEnv}{%
          \stEnv[\stIntC]%
        }{%
          \mpSel{\mpChanRole{\mpS}{\roleQ}}{\roleP}{\mpLab[k]}{%
            \mpE
          }{\mpQ}
        }%
      \end{array}
    \end{array}
    \\ \text{%
      (by \eqref{eq:subj-red:comm:p-pi:1} %
      and~\Cref{app:lem:typing_inversion}\eqref{item:mpst-typing-inversion:par})  %
    }%
\end{multline}
\begin{multline}
    \label{eq:subj-red:comm:p-branch-typing:1}
    \begin{array}{@{}l@{}}
	     \stEnv[1] \stEnvComp \stEnvMap{\mpChanRole{\mpS}{\roleP}}{\stExtSum{\roleQ}{i \in I}{\stChoice{\stLab[i]}{\stS[i]} \stSeq \stT[i]}} \stSub \stEnv[\stExtC]
    \quad\text{such that}\quad%
    \\\qquad
        \forall i \in I%
	\left\{
      \begin{array}{@{\hskip 0mm}l@{\hskip 0mm}r@{\hskip 0mm}}%
        \stJudgeNew{\mpEnv \mpEnvComp \mpEnvMap{z_i}{\tyGround}}{
		\stEnv[1] \stEnvComp
          \stEnvMap{\mpChanRole{\mpS}{\roleP}}{\stT[i]}
        }{
          \mpP[i]
        } & \text{   if } \tyS[i] = \tyGround
	\\
          \stJudgeNew{\mpEnv}{
		\stEnv[1] \stEnvComp
          \stEnvMap{z_i}{\tyT} \stEnvComp
          \stEnvMap{\mpChanRole{\mpS}{\roleP}}{\stT[i]}
        }{
          \mpP[i]
        } & \text{   if } \tyS[i] = \tyT
      \end{array}
\right\}
    \hspace{-0mm}%
    \end{array}
    \\
    \text{%
      (by \eqref{eq:subj-red:comm:p-typing:1} %
      and~\Cref{app:lem:typing_inversion}\eqref{item:mpst-typing-inversion:branch}) %
    }%
\end{multline}
\begin{multline}
    \label{eq:subj-red:comm:p-sel-typing:1}%
    \begin{array}{@{}l@{}}
	    \stEnv[2],  \stEnvMap{\mpChanRole{\mpS}{\roleQ}}{
          \stIntSum{\roleP}{}{\stChoice{\stLab}{\stB} \stSeq \stT} 
	  \stSub \stEnv[\stIntC]
        }
    \;\;\text{such that}\;\;%
    \\\quad
      \begin{array}{l}
	      \stEnvEntails{\mpEnv}{%
          \mpE
        }{\stB}%
        \text{ and }
        \stJudgeNew{\mpEnv}{%
          \stEnv[2] \stEnvComp \stEnvMap{\mpChanRole{\mpS}{\roleQ}}{\stT}%
        }{%
          \mpQ%
        }%
      \end{array}
    \end{array}
    \\ 
    \text{%
      (by \eqref{eq:subj-red:comm:p-typing:1} %
      and~\Cref{app:lem:typing_inversion}\eqref{item:mpst-typing-inversion:sel:1}) %
    }%
\end{multline}
  Now, notice that:%
  \begin{multline}
    \label{eq:subj-red:comm:stenv-composition:1}%
    \stEnv[1] \stEnvComp \stEnvMap{\mpChanRole{\mpS}{\roleP}}{\stExtSum{\roleQ}{i \in I}{\stChoice{\stLab[i]}{\stS[i]} \stSeq \stT[i]}}%
\stEnvComp \stEnv[2],  \stEnvMap{\mpChanRole{\mpS}{\roleQ}}{\stIntSum{\roleP}{}{\stChoice{\stLab}{\stB} \stSeq \stT}}
	= \stEnvii \stSub \stEnv
	\\
    \text{%
      (by \eqref{eq:subj-red:comm:p-typing:1}, %
      \eqref{eq:subj-red:comm:p-branch-typing:1}, %
      and \eqref{eq:subj-red:comm:p-sel-typing:1})%
    }%
\end{multline}
\begin{flalign}
    &
    \label{eq:subj-red:comm:stenvii-assoc:1}%
    \forall \mpS \in \stEnv: \stEnvAssoc{\gtG[\mpS]}{\stEnvii[\mpS]}{\mpS}
    &
    \text{%
      (by %
      \eqref{item:subjred:stenv-assoc}, 
      \eqref{eq:subj-red:comm:stenv-composition:1},
      and \Cref{lem:stenv-assoc-sub})%
    }%
    \\%
    &
	   \label{eq:subj-red:comm:stenv-i-k-carried-sub:1}%
    k \in I %
    \quad\text{and}\quad%
    \stS[k] \stSub \stB%
    &%
    \text{%
	    (by 
      \eqref{eq:subj-red:comm:stenv-composition:1}, 
	    \eqref{eq:subj-red:comm:stenvii-assoc:1} and
      \Cref{app:lem:match_comm_projection})
    }%
    \\%
    &
    \label{eq:subj-red:comm:stenvii-move-stenviii:1}%
    \stEnvii \stEnvMove \stEnvi%
    = \stEnv[1] \stEnvComp%
    \stEnvMap{\mpChanRole{\mpS}{\roleP}}{\stT[k]} \stEnvComp%
    \stEnv[2] \stEnvComp %
    \stEnvMap{\mpChanRole{\mpS}{\roleQ}}{\stT}%
    &
    \text{%
      (by %
      \eqref{eq:subj-red:comm:stenv-composition:1}, 
      \eqref{eq:subj-red:comm:stenv-i-k-carried-sub:1},  %
      and \Cref{def:typing_context_reduction})%
    }%
   \\%
   &
	\label{eq:subj-red:comm:stenviii-assoc:1}%
    \forall \mpS \in \stEnv: \exists \gtGi[\mpS]: \gtG[\mpS] \,\gtMoveStar\, \gtGi[\mpS] \text{ and }  \stEnvAssoc{\gtGi[\mpS]}{\stEnvi[\mpS]}{\mpS}%
   &
    \text{%
      (by %
	    \eqref{eq:subj-red:comm:stenvii-assoc:1},
      \eqref{eq:subj-red:comm:stenvii-move-stenviii:1}  %
      and \Cref{cor:completeness_association}) 
   }%
  \end{flalign}
  We can now use $\stEnvi$ to type $\mpPi$:
  \begin{flalign}
    &\label{eq:subj-red:pi-branch-cont-typing:1}%
    \stJudgeNew{\mpEnv \stEnvComp \stEnvMap{z_k}{\stB}}{%
      \stEnv[1] \stEnvComp%
      \stEnvMap{\mpChanRole{\mpS}{\roleP}}{\stT[k]}%
    }{%
      \mpP[k]%
    }%
    &\text{%
      (by \eqref{eq:subj-red:comm:stenv-i-k-carried-sub:1}, %
      and \eqref{eq:subj-red:comm:p-branch-typing:1})%
    }%
    \\[1mm]%
    &\label{eq:subj-red:comm-payload-entails:1}%
    \stEnvEntails{\mpEnv}{%
      \mpV
    }{\stB}%
    &
      \text{%
      (by \eqref{eq:subj-red:comm:p-pi:1} and \eqref{eq:subj-red:comm:p-sel-typing:1})}
    \\%
    &\label{eq:subj-red:comm:pi-branch-cont-typing-subst:1}%
    \stJudgeNew{\mpEnv}{%
      \stEnv[1] \stEnvComp%
      \stEnvMap{\mpChanRole{\mpS}{\roleP}}{\stT[k]}%
    }{%
      \mpP[k]\subst{\mpFmt{z_k}}{
        \mpV
      }%
    }%
    &\text{%
      (by \eqref{eq:subj-red:pi-branch-cont-typing:1}, %
      \eqref{eq:subj-red:comm-payload-entails:1} %
      and \Cref{app:lem:substitution})%
    }%
\end{flalign}
\begin{multline}
	\label{eq:subj-red:comm:pi-branch-cont-typing-subst-stenviii:1}%
    \inference[\iruleMPPar]{%
      \begin{array}{l}
      \stJudgeNew{\mpEnv}{%
        \stEnv[1] \stEnvComp%
        \stEnvMap{\mpChanRole{\mpS}{\roleP}}{\stT[k]}%
      }{%
        \mpP[k]\subst{\mpFmt{z_k}}{
          \mpV
        }%
      }%
      \\
      \stJudgeNew{\mpEnv}{%
        \stEnv[2] \stEnvComp \stEnvMap{\mpChanRole{\mpS}{\roleQ}}{\stT}%
      }{%
        \mpQ%
      }%
      \end{array}
    }{%
      \stJudgeNew{\mpEnv}{%
        \stEnvi%
      }{%
        \mpPi%
      }%
    }%
    \\
     \text{%
      (by \eqref{eq:subj-red:comm:pi-branch-cont-typing-subst:1}, %
      \eqref{eq:subj-red:comm:p-sel-typing:1}, %
      \eqref{eq:subj-red:comm:stenvii-move-stenviii:1}, %
      \eqref{eq:subj-red:comm:stenviii-assoc:1},  %
      and \eqref{eq:subj-red:comm:p-pi:1})%
    }%
  \end{multline}
  We conclude this case by showing that there exists some $\stEnviii$ that satisfies the statement:
  \begin{flalign}
    &\label{eq:subj-red:comm:exists-stenvi:1}%
    \stEnv \stEnvMove \stEnviii \text{ and } \stEnvi \stSub \stEnviii
    &\text{%
      (by \eqref{item:subjred:stenv-assoc}, 
      \eqref{eq:subj-red:comm:stenv-composition:1},
      \eqref{eq:subj-red:comm:stenvii-move-stenviii:1}, %
      and~\Cref{lem:stenv-assoc-reduction-sub-gen})%
    }%
    \\%
    &\label{eq:subj-red:comm:exists-stenvi:2_new}%
    \stJudgeNew{\mpEnv}{\stEnviii}{\mpPi}%
    &\text{%
      (by %
      \eqref{eq:subj-red:comm:pi-branch-cont-typing-subst-stenviii:1}
      \eqref{eq:subj-red:comm:exists-stenvi:1}, %
      and \Cref{app:lem:narrowing})%
    }%
     \\%
    &\label{eq:subj-red:comm:stenvi-safe:1}%
    \forall \mpS \in \stEnviii: \exists \gtGiii[\mpS]: \gtG[\mpS] \,\gtMoveStar\, \gtGiii[\mpS] \text{ and }  \stEnvAssoc{\gtGiii[\mpS]}{\stEnviii[\mpS]}{\mpS} 
    &\text{%
      (by \eqref{item:subjred:stenv-assoc}, 
      \eqref{eq:subj-red:comm:exists-stenvi:1},  %
      and~\Cref{cor:completeness_association}) %
    }%
  \end{flalign}

 \item Case \inferrule{\iruleMPRedCommChannel}: 
 \begin{multline}
    \label{eq:subj-red:comm:p-pi:2}%
    \begin{array}{rcl}%
      \textstyle%
      \mpP &=&%
      \mpBranch{\mpChanRole{\mpS}{\roleP}}{\roleQ}{i \in I}{%
        \mpLab[i]}{z_i}{\mpP[i]}{}
      \,\mpPar\,%
      \mpSel{\mpChanRole{\mpS}{\roleQ}}{\roleP}{\mpLab[k]}{%
	      \mpChanRole{\mpSi}{\roleR}
      }{\mpQ}%
      \\[1mm]%
      \mpPi &=&%
      \mpP[k]\subst{\mpFmt{z_k}}{
	      \mpChanRole{\mpSi}{\roleR}
      }%
      \,\mpPar\,%
      \mpQ%
      \quad%
      (k \in I)
    \end{array}
    \\
    \text{%
      (by inversion of \inferrule{\iruleMPRedCommChannel})%
    }%
\end{multline}
\begin{multline}
    \label{eq:subj-red:comm:p-typing:2}%
    \begin{array}{l}
    \stEnv[\stExtC] \stEnvComp \stEnv[\stIntC] \stSub \stEnv
    \quad\text{such that}\quad%
    \\\qquad
      \begin{array}{l}
        \stJudgeNew{\mpEnv}{%
          \stEnv[\stExtC]%
        }{%
          \mpBranch{\mpChanRole{\mpS}{\roleP}}{\roleQ}{i \in I}{%
            \mpLab[i]}{z_i}{\mpP[i]}{}
        }%
        \\%
        \stJudgeNew{\mpEnv}{%
          \stEnv[\stIntC]%
        }{%
          \mpSel{\mpChanRole{\mpS}{\roleQ}}{\roleP}{\mpLab[k]}{%
		  \mpChanRole{\mpSi}{\roleR}
          }{\mpQ}
        }%
      \end{array}
    \end{array}
    \\
    \text{%
      (by \eqref{eq:subj-red:comm:p-pi:2} %
      and~\Cref{app:lem:typing_inversion}\eqref{item:mpst-typing-inversion:par})  %
    }%
\end{multline}
\begin{multline}
	\label{eq:subj-red:comm:p-branch-typing:2}
    \begin{array}{@{}l@{}}
	\stEnv[1] \stEnvComp \stEnvMap{\mpChanRole{\mpS}{\roleP}}{\stExtSum{\roleQ}{i \in I}{\stChoice{\stLab[i]}{\stS[i]} \stSeq \stT[i]}} \stSub \stEnv[\stExtC]
    \quad\text{such that}\quad%
    \\\qquad
        \forall i \in I%
	\left\{
      \begin{array}{@{\hskip 0mm}l@{\hskip 0mm}r@{\hskip 0mm}}%
        \stJudgeNew{\mpEnv \mpEnvComp \mpEnvMap{z_i}{\tyGround}}{
		\stEnv[1] \stEnvComp
          \stEnvMap{\mpChanRole{\mpS}{\roleP}}{\stT[i]}
        }{
          \mpP[i]
        } & \text{   if } \tyS[i] = \tyGround
	\\
          \stJudgeNew{\mpEnv}{
		\stEnv[1] \stEnvComp
          \stEnvMap{z_i}{\tyT} \stEnvComp
          \stEnvMap{\mpChanRole{\mpS}{\roleP}}{\stT[i]}
        }{
          \mpP[i]
        } & \text{   if } \tyS[i] = \tyT
      \end{array}
\right\}
    \hspace{-0mm}%
    \end{array}
    \\
    \text{%
      (by \eqref{eq:subj-red:comm:p-typing:2} %
      and~\Cref{app:lem:typing_inversion}\eqref{item:mpst-typing-inversion:branch}) %
    }%
\end{multline}
\begin{multline}
	\label{eq:subj-red:comm:p-sel-typing:2}%
    \begin{array}{@{}l@{}}
	 \stEnv[2] \stEnvComp  \stEnvMap{\mpChanRole{\mpS}{\roleQ}}{
		    \stIntSum{\roleP}{}{\stChoice{\stLab}{\stT[d]} \stSeq \stT}
	  \stEnvComp \stEnvMap{\mpChanRole{\mpSi}{\roleR}}{\stT[d]}
	  \stSub \stEnv[\stIntC]
        }
    \;\;\text{such that}\;\;%
    \\\quad
      \begin{array}{l}
        \stJudgeNew{\mpEnv}{%
          \stEnv[2] \stEnvComp \stEnvMap{\mpChanRole{\mpS}{\roleQ}}{\stT}%
        }{%
          \mpQ%
        }%
      \end{array}
    \end{array}
    \hspace{-15mm}%
    \\
    \text{%
      (by \eqref{eq:subj-red:comm:p-typing:2} %
      and~\Cref{app:lem:typing_inversion}\eqref{item:mpst-typing-inversion:sel:2}) %
    }%
      \end{multline}
  Now, notice that:%
  \begin{multline}
	  \label{eq:subj-red:comm:stenv-composition:2}%
    \stEnv[1] \stEnvComp \stEnvMap{\mpChanRole{\mpS}{\roleP}}{\stExtSum{\roleQ}{i \in I}{\stChoice{\stLab[i]}{\stS[i]} \stSeq \stT[i]}}%
    \stEnvComp \stEnv[2],  \stEnvMap{\mpChanRole{\mpS}{\roleQ}}{\stIntSum{\roleP}{}{\stChoice{\stLab}{\stT[d]} \stSeq \stT}}
	  \stEnvComp \stEnvMap{\mpChanRole{\mpSi}{\roleR}}{\stT[d]}
	  = \stEnvii 
 	\stSub \stEnv
	\\
	\text{%
      (by \eqref{eq:subj-red:comm:p-typing:2}, %
      \eqref{eq:subj-red:comm:p-branch-typing:2}, %
      and \eqref{eq:subj-red:comm:p-sel-typing:2})%
    }%
\end{multline}
  \begin{flalign}
    &\label{eq:subj-red:comm:stenvii-assoc:2}%
    \forall \mpS \in \stEnv: \stEnvAssoc{\gtG[\mpS]}{\stEnvii[\mpS]}{\mpS}
    &\text{%
      (by %
      \eqref{item:subjred:stenv-assoc}, 
      \eqref{eq:subj-red:comm:stenv-composition:2}
      and \Cref{lem:stenv-assoc-sub})%
    }%
    \\%
    &\label{eq:subj-red:comm:stenv-i-k-carried-sub:2}%
    k \in I %
    \quad\text{and}\quad%
    \stS[k] \stSub \stT[d]%
    &%
    \text{%
      (by %
      \eqref{eq:subj-red:comm:stenv-composition:2},
      \eqref{eq:subj-red:comm:stenvii-assoc:2} and
      \Cref{app:lem:match_comm_projection})
    }%
\end{flalign}
\begin{multline}
	\label{eq:subj-red:comm:stenvii-move-stenviii:2}%
    \stEnv \stEnvMove \stEnvi%
    = \stEnv[1] \stEnvComp%
    \stEnvMap{\mpChanRole{\mpS}{\roleP}}{\stT[k]} \stEnvComp%
    \stEnv[2] \stEnvComp %
    \stEnvMap{\mpChanRole{\mpS}{\roleQ}}{\stT}%
	  \stEnvComp \stEnvMap{\mpChanRole{\mpSi}{\roleR}}{\stT[d]}
	  \\
	  \text{%
      (by %
      \eqref{eq:subj-red:comm:stenv-composition:2},
      \eqref{eq:subj-red:comm:stenv-i-k-carried-sub:2} and
      \Cref{def:typing_context_reduction})%
    }%
\end{multline}
\begin{flalign}
   &\label{eq:subj-red:comm:stenviii-assoc:2}%
    \forall \mpS \in \stEnv: \exists \gtGi[\mpS]: \gtG[\mpS] \,\gtMoveStar\, \gtGi[\mpS] \text{ and }  \stEnvAssoc{\gtGi[\mpS]}{\stEnvi[\mpS]}{\mpS}%
   &\text{%
      (by %
      \eqref{eq:subj-red:comm:stenvii-move-stenviii:2},  %
      \eqref{eq:subj-red:comm:stenvii-assoc:2} and
      \Cref{cor:completeness_association}) 
   }%
  \end{flalign}
  We can now use $\stEnvi$ to type $\mpPi$:
  \begin{flalign}
    &\label{eq:subj-red:pi-branch-cont-typing:2}%
    \stJudgeNew{\mpEnv }{%
      \stEnv[1] \stEnvComp%
      \stEnvMap{z_k}{\stS[k]}
      \stEnvComp 
      \stEnvMap{\mpChanRole{\mpS}{\roleP}}{\stT[k]}%
    }{%
      \mpP[k]%
    }%
    &\text{%
      (by \eqref{eq:subj-red:comm:stenv-i-k-carried-sub:2}, %
      and \eqref{eq:subj-red:comm:p-branch-typing:2})%
    }%
    \\[1mm]%
    &\label{eq:subj-red:pi-branch-cont-typing:2_subtyping}%
    \stJudgeNew{\mpEnv }{%
      \stEnv[1] \stEnvComp%
      \stEnvMap{z_k}{\stT[d]}
      \stEnvComp 
      \stEnvMap{\mpChanRole{\mpS}{\roleP}}{\stT[k]}%
    }{%
      \mpP[k]%
    }%
    &\text{%
      (by \eqref{eq:subj-red:comm:stenv-i-k-carried-sub:2}, %
       \eqref{eq:subj-red:pi-branch-cont-typing:2}, and \inferrule{\iruleMPSub}
       )%
    }%
    \\[1mm]%
    &\label{eq:subj-red:pi-branch-cont-typing-env-subst-def:2}%
    \stEnv[1]  \stEnvComp%
    \stEnvMap{\mpChanRole{\mpS}{\roleP}}{\stT[k]}%
    \stEnvComp \stEnvMap{\mpChanRole{\mpSi}{\roleR}}{\stT[d]}
    \text{\; defined}%
    &\text{%
      (by \eqref{eq:subj-red:comm:p-sel-typing:2},
      \eqref{eq:subj-red:comm:p-branch-typing:2}, %
      and \eqref{eq:subj-red:comm:p-typing:2})%
    }%
    \\%
    &\label{eq:subj-red:comm:pi-branch-cont-typing-subst:2}%
    \stJudgeNew{\mpEnv}{%
      \stEnv[1] \stEnvComp%
      \stEnvMap{\mpChanRole{\mpS}{\roleP}}{\stT[k]}%
      \stEnvComp \stEnvMap{\mpChanRole{\mpSi}{\roleR}}{\stT[d]}
    }{%
      \mpP[k]\subst{\mpFmt{z_k}}{
	      \mpChanRole{\mpSi}{\roleR}
      }%
    }%
    &\text{%
      (by \eqref{eq:subj-red:pi-branch-cont-typing:2_subtyping}, %
      \eqref{eq:subj-red:pi-branch-cont-typing-env-subst-def:2}, %
      and \Cref{app:lem:substitution})%
    }%
\end{flalign}
\begin{multline}
	\label{eq:subj-red:comm:pi-branch-cont-typing-subst-stenviii:2}%
    \inference[\iruleMPPar]{%
      \begin{array}{l}
      \stJudgeNew{\mpEnv}{%
        \stEnv[1] \stEnvComp%
        \stEnvMap{\mpChanRole{\mpS}{\roleP}}{\stT[k]}%
	\stEnvComp \stEnvMap{\mpChanRole{\mpSi}{\roleR}}{\stT[d]}
      }{%
        \mpP[k]\subst{\mpFmt{z_k}}{
		\mpChanRole{\mpSi}{\roleR}
        }%
      }%
      \\
      \stJudgeNew{\mpEnv}{%
        \stEnv[2] \stEnvComp \stEnvMap{\mpChanRole{\mpS}{\roleQ}}{\stT}%
      }{%
        \mpQ%
      }%
      \end{array}
    }{%
      \stJudgeNew{\mpEnv}{%
        \stEnvi%
      }{%
        \mpPi%
      }%
    }%
    \\
     \text{%
      (by \eqref{eq:subj-red:comm:pi-branch-cont-typing-subst:2}, %
      \eqref{eq:subj-red:comm:p-sel-typing:2}, %
      \eqref{eq:subj-red:comm:stenvii-move-stenviii:2}, %
      \eqref{eq:subj-red:comm:stenviii-assoc:2},  %
      and \eqref{eq:subj-red:comm:p-pi:2})%
    }%
  \end{multline}
  We conclude this case by showing that there exists some $\stEnviii$ that satisfies the statement: 
  \begin{flalign}
    &\label{eq:subj-red:comm:exists-stenvi:2}%
    \exists \stEnviii: %
    \stEnv \stEnvMove \stEnviii \text{ and } \stEnvi \stSub \stEnviii
    &\text{%
      (by \eqref{item:subjred:stenv-assoc}, 
      \eqref{eq:subj-red:comm:stenv-composition:2},
      \eqref{eq:subj-red:comm:stenvii-move-stenviii:2}, %
      and~\Cref{lem:stenv-assoc-reduction-sub-gen})%
    }%
    \\%
    &\label{eq:subj-red:comm:exists-stenvi:2_new_D}
    \stJudgeNew{\mpEnv}{\stEnviii}{\mpPi}%
    &\text{%
      (by %
      \eqref{eq:subj-red:comm:pi-branch-cont-typing-subst-stenviii:2}
      \eqref{eq:subj-red:comm:exists-stenvi:2}, %
      and \Cref{app:lem:narrowing})%
    }%
     \\%
    &\label{eq:subj-red:comm:stenvi-safe:2}%
    \forall \mpS \in \stEnviii: \exists \gtGiii[\mpS]: \gtG[\mpS] \,\gtMoveStar\, \gtGiii[\mpS] \text{ and }  \stEnvAssoc{\gtGiii[\mpS]}{\stEnviii[\mpS]}{\mpS} 
    &\text{%
      (by \eqref{item:subjred:stenv-assoc}, 
      \eqref{eq:subj-red:comm:exists-stenvi:2},  %
      and~\Cref{cor:completeness_association}) %
    }%
  \end{flalign}

 \item Case \inferrule{\iruleMPRedCtx}: 
 By inversion of the rule~(\Cref{app:lem:typing_inversion}) and~\Cref{def:calculus_semantics}, we have to prove the statement in the following sub-cases:
\begin{enumerate}[leftmargin=*]
  \item
  \label{item:subj-red:ctx:par}
    $\mpP = \mpQ \mpPar \mpR$ \;\;and\;\; $\mpPi = \mpQi \mpPar \mpR$ \;\;and\;\; $\mpQ \mpMove \mpQi$
  \item
  \label{item:subj-red:ctx:res}
    $\mpP = \mpRes{\mpSi}{\mpQ}$ \;\;and\;\; $\mpPi = \mpRes{\mpSi}{\mpQi}$ \;\;and\;\; $\mpQ \mpMove \mpQi$
  \item
  \label{item:subj-red:ctx:def}
    $\mpP = \mpDefAbbrev{\mpDefD}{\mpQ}$ \;\;and\;\; $\mpPi = \mpDefAbbrev{\mpDefD}{\mpQi}$ \;\;and\;\; $\mpQ \mpMove \mpQi$
\end{enumerate}
Cases \ref{item:subj-red:ctx:par} and \ref{item:subj-red:ctx:def} are easily proved using the induction hypothesis.  Therefore, here we focus on case \ref{item:subj-red:ctx:res}.
\begin{multline}
  \label{eq:subj-red:ctx:res:p-typing}%
  \exists \stEnv[1] \stSub \stEnv, \gtG \;\;\text{such that}\;\;
    \begin{array}{@{}l@{}}
      \stEnv[2] =  
      \setenum{\stEnvMap{\mpChanRole{\mpSi}{\roleP}}{\gtProj{\gtG}{\roleP}}}_{%
    \roleP \in \gtRoles{\gtG}},\, 
      \mpSi \!\not\in\! \stEnv[1], 
     \text{ and }
      \stJudgeNew{\mpEnv}{
        \stEnv[1] \stEnvComp \stEnv[2]
      }{
        \mpQ
      }
    \end{array}
  \\
  \text{%
    (by \ref{item:subj-red:ctx:res} and~\Cref{app:lem:typing_inversion}\eqref{item:mpst-typing-inversion:res})
  }
\end{multline}
  \begin{align}
  &\label{eq:subj-red:ctx:res:assoc:1}%
  \stEnvAssoc{\gtG}{\stEnv[2]}{\mpSi}
  &
  \text{(by \eqref{eq:subj-red:ctx:res:p-typing}, \Cref{def:assoc},  %
 and~\Cref{app:lem:subtyping:reflexive})}
  \end{align}
\begin{multline}
  \label{eq:subj-red:ctx:res:stenvi-stenvsi}%
  \exists \stEnv[3], \stEnv[4], \gtGi\;\;\text{such that}\;
  \left\{\begin{array}{@{}l@{}}
    \mpSi \!\not\in\! \stEnv[3]
    \\ %
    \stEnv[1] \stEnvMoveStar \stEnv[3]
    \\ %
    \stEnv[2] \stEnvMoveStar \stEnv[4]
    \\ %
    \forall \mpS \in \stEnv[3]: \exists \gtGiii[\mpS]:  \gtG[\mpS] \,\gtMoveStar\, \gtGiii[\mpS] \text{ and }  \stEnvAssoc{\gtGiii[\mpS]}{{\stEnv[3]}_{\mpS}}{\mpS}
    \\
    \stEnvAssoc{\gtGi}{\stEnv[4]}{\mpSi}
    \\
    \stJudgeNew{\mpEnv}{\stEnv[3] \stEnvComp \stEnv[4]}{\mpQi}
  \end{array}\right\}
  \\
  \text{%
    (by \eqref{eq:subj-red:ctx:res:p-typing}, \eqref{eq:subj-red:ctx:res:assoc:1}, 
    and inductive hypothesis)
  }
\end{multline}

\begin{multline}
\label{eq:subj-red:ctx:res:large-env}
\exists \stEnv[5] \;\;\text{such that}\;\; 
\stEnv[5] = \setenum{\stEnvMap{\mpChanRole{\mpSi}{\roleP}}{\gtProj{\gtGi}{\roleP}}}_{%
    \roleP \in \gtRoles{\gtGi}} \text{ and } 
    \stJudgeNew{\mpEnv}{\stEnv[3] \stEnvComp \stEnv[5]}{\mpQi}
    \\
    \text{
    (by \eqref{eq:subj-red:ctx:res:stenvi-stenvsi}, \Cref{def:assoc}, and \Cref{lem:termination_subtyping}
    )}
\end{multline}

\begin{multline}
  \label{eq:subj-red:ctx:res:pi-typing}%
  \inference[\iruleMPGlobalRes]{%
    \begin{array}{@{}l@{}}
     \stEnv[5] = \setenum{\stEnvMap{\mpChanRole{\mpSi}{\roleP}}{\gtProj{\gtGi}{\roleP}}}_{%
    \roleP \in \gtRoles{\gtGi}}
      \quad
      \mpSi \!\not\in\! \stEnv[3]
      \quad%
      \stJudgeNew{\mpEnv}{
        \stEnv[3] \stEnvComp \stEnv[5]
      }{
        \mpQi
      }
    \end{array}
  }{%
    \stJudgeNew{\mpEnv}{%
      \stEnv[3]%
    }{%
      \mpPi %
    }%
  }%
  \\
  \text{%
    (by \eqref{eq:subj-red:ctx:res:stenvi-stenvsi},
    \eqref{eq:subj-red:ctx:res:large-env}
    and \ref{item:subj-red:ctx:res})
  }
\end{multline}
Hence, we obtain the thesis by \eqref{eq:subj-red:ctx:res:p-typing}, \eqref{eq:subj-red:ctx:res:stenvi-stenvsi} and \eqref{eq:subj-red:ctx:res:pi-typing}, \Cref{lem:stenv-assoc-reduction-sub-gen}, 
\Cref{app:lem:narrowing}, and \Cref{cor:completeness_association}. 
\item Case \inferrule{\iruleMPRedCall}: Follows from~\Cref{app:lem:typing_inversion}\eqref{item:mpst-typing-inversion:def},~\Cref{app:lem:typing_inversion}\eqref{item:mpst-typing-inversion:call},~\Cref{app:lem:substitution}, and~\Cref{lem:broading}. 

\item Case \inferrule{\iruleMPErrLabel}: Trivial, as $\mpP$ is not typable. 
\item Case \inferrule{\iruleMPRedCongr}: Follows directly from induction hypothesis and~\Cref{lem:subject-congruence}
\item Cases \inferrule{\iruleMPRedCondTrue} and~\inferrule{\iruleMPRedCondFalse}: 
Follow directly from~\Cref{app:lem:typing_inversion}\eqref{item:mpst-typing-inversion:if} and~\Cref{lem:broading}. 
\qedhere 
\end{itemize}
\end{proof}

\lemSubjectReductionFinal*
\begin{proof}
By~\Cref{def:assoc} and \Cref{app:lem:subtyping:reflexive}, 
we have $\forall \mpS \in \stEnv: \exists
   \gtG[\mpS]: \stEnvAssoc{\gtG[\mpS]}{\stEnv[\mpS]}{\mpS}$. 
   Then,  by~\Cref{lem:subject-reduction}, it follows that 
   $\exists \stEnvNewi$
  such that 
  $\stEnvNew \!\stEnvMoveStar\! \stEnvNewi$, 
  $\stJudgeNew{\mpEnvNew}{\stEnvNewi}{\mpPi}$, and 
  $\forall \mpS \in \stEnvNewi: 
  \exists \gtGi[\mpS]: \gtG[\mpS] \,\gtMoveStar\, \gtGi[\mpS]$ and 
  $\stEnvAssoc{\gtGi[\mpS]}{\stEnvNewi[\mpS]}{\mpS}$. %
   
   Let $\stEnvii$ be the typing context defined by  $\forall \mpS \in \stEnvi:  \stEnvii[\mpS] = 
   \setenum{\stEnvMap{\mpChanRole{\mpS}{\roleP}}{\gtProj{\gtGi[\mpS]}{\roleP}}}_{%
    \roleP \in \gtRoles{\gtGi[\mpS]}}$. It remains to show that $\stJudgeNew{\mpEnvNew}{\stEnvii}{\mpPi}$, which follows from $\stJudgeNew{\mpEnvNew}{\stEnvNewi}{\mpPi}$,~\Cref{def:assoc},  and~\Cref{lem:termination_subtyping}.   
    \qedhere 
\end{proof}

\lemTypeSafety*
\begin{proof}
  From the hypothesis $\mpP \!\mpMoveStar\! \mpPi$,
  we know that $\mpP = \mpP[0] \!\mpMove\! \mpP[1] \!\mpMove\! \cdots
  \!\mpMove\! \mpP[n] = \mpPi$ (for some $n$).
  The proof proceeds by induction on $n$.  
  
  The base case for n=0 is straightforward: we have $\mpP = \mpPi$, thus $\mpPi$ is well-typed. Furthermore, since the term $\mpErr$ is not typeable, $\mpPi$ cannot contain such a term.

In the inductive case for n = m+1, we already know (by the induction hypothesis) that $\mpP[m]$ is well-typed. By applying \Cref{thm:subject_reduction}, we can conclude that $\mpP[m+1] = \mpPi$ is also well-typed and does not contain any $\mpErr$ subterms. 
\end{proof}

\subsection{Session Fidelity}
\label{app:sec:session_fidelity_proof}
\lemSessionFidelity*

\begin{proof}
  The proof structure is  based on %
  induction on the derivation of the reduction of $\stEnv$. 
  We infer the contents of $\stEnv$, as well as the shape of $\mpP$ and its
  sub-processes $\mpP[\roleP]$, showing that they can mimic the reduction of
  $\stEnv$. 
  We focus on the case of $\stEnv \stEnvMoveAnnot{\ltsSendRecv{\mpS}{\roleP}{\roleQ}{\stLab}} \stEnvi$, while the other cases follow by applying the induction hypothesis. 
\begin{itemize}%
 \item Case $\stEnv \stEnvMoveAnnot{\ltsSendRecv{\mpS}{\roleP}{\roleQ}{\stLab}} \stEnvi$:  
 in this case, the process $\mpP[\roleP]$ playing role $\roleP$ in session $\mpS$ is a selection on $\mpChanRole{\mpS}{\roleP}$ towards $\roleQ$ (possibly within a process definition); 
 while the process $\mpP[\roleQ]$ playing role $\roleQ$ in session $\mpS$ is a branching 
 on $\mpChanRole{\mpS}{\roleQ}$ from  $\roleP$ (possibly within a process definition). 
 Therefore, by~\inferrule{\iruleMPRedCommExp} or \inferrule{\iruleMPRedCommChannel} in~\Cref{fig:mpst-pi-semantics}, 
 $\mpP$ can correspondingly reduce to $\mpPi$ by transmitting  
 either a value $\mpV$ or a channel endpoint $\mpChanRole{\mpSi}{\rolePi}$ 
 from $\roleP$ to 
 $\roleQ$ in session $\mpS$~(possibly after a finite number of transitions under rule \inferrule{\iruleMPRedCall}). The resulting continuation process $\mpPi$ 
 is typed by $\stEnvi$. 
 The assertion that there exists 
 $\gtGi$ such that $\gtG \gtMove \gtGi$ and $\stEnvAssoc{\gtGi}{\stEnvi}{\mpS}$ 
 follows from 
 $\stEnv \stEnvMoveAnnot{\ltsSendRecv{\mpS}{\roleP}{\roleQ}{\stLab}} \stEnvi$, 
 $\stEnvAssoc{\gtG}{\stEnv}{\mpS}$, 
 and~\Cref{thm:gtype:proj-comp}. 
\qedhere 
  \end{itemize}  
 \end{proof}
 
 \Cref{lem:single-session-persistent} below states that if a process $\mpP$ satisfies the assumptions of session fidelity~(\Cref{lem:session-fidelity}), then all its reductums will also satisfy these assumptions. 
In other words, if $\mpP$ enjoys session fidelity, so will all of its reductums. 

\begin{restatable}{proposition}{lemSingleSessionPersistent}%
    \label{lem:single-session-persistent}%
    Assume\, $\stJudge{\mpEnvEmpty\!}{\!\stEnv}{\!\mpP}$, %
    where %
    $\stEnvAssoc{\gtG}{\stEnv}{\mpS}$, %
   \,$\mpP \equiv \mpRes{\mpS[1]}{\ldots \mpRes{\mpS[n]}{\mpBigPar{\roleP \in I}{\mpP[\roleP]}}}$, %
    \,and\, $\stEnv = \bigcup_{\roleP \in I}\stEnv[\roleP]$ %
    such that, for each $\mpP[\roleP]$, %
    we have\, $\stJudge{\mpEnvEmpty\!}{\stEnv[\roleP]}{\!\mpP[\roleP]}$. %
    \,Further, assume that each $\mpP[\roleP]$
    is either\, $\mpNil$ (up to $\equiv$), %
    or only plays $\roleP$ in $\mpS$, by $\stEnv[\roleP]$. %
    Then,\, $\mpP \mpMove \mpPi$
    \,implies\, $\exists \stEnvi, \gtGi$ %
    such that\, %
    $\stEnv \!\stEnvMoveWithSessionStar[\mpS]\! \stEnvi$ %
    \,and\, %
    $\stJudge{\mpEnvEmpty\!}{\!\stEnvi}{\mpPi}$, %
    \;with\; %
    $\stEnvAssoc{\gtGi}{\stEnvi}{\mpS}$, %
   \,$\mpPi \equiv \mpRes{\mpSi[1]}{\ldots \mpRes{\mpSi[m]}{\mpBigPar{\roleP \in I}{\mpPi[\roleP]}}}$,
    \,and\, $\stEnvi = \bigcup_{\roleP \in I}\stEnvi[\roleP]$ %
    such that, for each $\mpPi[\roleP]$, 
    we have\, $\stJudge{\mpEnvEmpty\!}{\stEnvi[\roleP]}{\!\mpPi[\roleP]}$; %
    \,furthermore, each $\mpPi[\roleP]$
    is $\mpNil$ (up to $\equiv$),
    or only plays $\roleP$ in $\mpS$, by $\stEnvi[\roleP]$.%
\end{restatable}
\begin{proof}
  Straightforward from the proof of \Cref{lem:subject-reduction}, which accounts for all possible transitions from $\mpP$ to $\mpPi$, and in all cases yields the desired properties for its typing context $\stEnvi$. 
\qedhere 
\end{proof}

\thmSessionFidelity*
\begin{proof}
Let $\stEnv = \setenum{\stEnvMap{\mpChanRole{\mpS}{\roleP}}{\gtProj{\gtG}{\roleP}}}_{%
   \roleP \in \mpFmt{I}}$ with $\gtRoles{\gtG} \subseteq I$. It follows directly that 
   $\stEnv = \bigcup_{\roleP \in I}\stEnv[\roleP]$, where, by~\Cref{lem:global_end}, 
   $\stEnv[\roleP] = \stEnvMap{\mpChanRole{\mpS}{\roleP}}{\stEnd}$ if $\roleP \notin \gtRoles{\gtG}$. 
 Then, by~\Cref{def:assoc}, $\stEnvAssoc{\gtG}{\stEnv}{\mpS}$. %
  
  Since $\gtG \,\gtMove$, by~\Cref{thm:soundness_association}, $\stEnv \!\stEnvMoveWithSession[\mpS]$.  
  Moreover, by the assumptions and~\Cref{lem:session-fidelity-association},  $\exists \stEnvi, \gtGi, \mpPi$ %
    such that 
    $\stEnv \!\stEnvMoveWithSession[\mpS]\! \stEnvi$, 
    $\gtG \,\gtMove\, \gtGi$, 
    $\mpP \mpMoveStar\! \mpPi$,  %
    and 
    $\stJudge{\mpEnvEmpty\!}{\!\stEnvi}{\mpPi}$, %
    with 
    $\stEnvAssoc{\gtGi}{\stEnvi}{\mpS}$, %
    $\mpPi \equiv \mpBigPar{\roleP \in I}{\mpPi[\roleP]}$, %
    and $\stEnvi = \bigcup_{\roleP \in I}\stEnvi[\roleP]$ %
    such that for each $\mpPi[\roleP]$:
    (1) $\stJudge{\mpEnvEmpty\!}{\stEnvi[\roleP]}{\!\mpPi[\roleP]}$, %
    and
    (2) either $\mpPi[\roleP] \equiv \mpNil$, 
    or $\mpPi[\roleP]$ only plays $\roleP$ in $\mpS$, by $\stEnvi[\roleP]$. 
    
    We are left to show that: 
    \begin{itemize}
    \item $\gtRoles{\gtGi} \subseteq I$: By $\stEnv \!\stEnvMoveWithSession[\mpS]\! \stEnvi$ 
    and~\Cref{lem:stenv-red:trivial-1-new}, 
    $\dom{\stEnvi} = \dom{\stEnv} = \setenum{\mpChanRole{\mpS}{\roleP}}_{\roleP \in \mpFmt{I}}$. 
   Moreover, since $\stEnvAssoc{\gtGi}{\stEnvi}{\mpS}$,  $\setenum{\mpChanRole{\mpS}{\roleP}}_{\roleP \in \gtRoles{\gtGi}} \subseteq \dom{\stEnvi}$, which implies that   $\gtRoles{\gtGi} \subseteq I$. 
   
    \item $\stJudge{\mpEnvEmpty\!}{\!\setenum{\stEnvMap{\mpChanRole{\mpS}{\roleP}}{\gtProj{\gtGi}{\roleP}}}_{%
   \roleP \in \mpFmt{I}}}{\mpPi}$: By $\stEnvAssoc{\gtGi}{\stEnvi}{\mpS}$,  $\dom{\stEnvi} = \setenum{\mpChanRole{\mpS}{\roleP}}_{\roleP \in \mpFmt{I}}$, and~\Cref{lem:global_end}, $\stEnvi \stSub \setenum{\stEnvMap{\mpChanRole{\mpS}{\roleP}}{\gtProj{\gtGi}{\roleP}}}_{%
   \roleP \in \mpFmt{I}}$. Moreover, by $\stJudge{\mpEnvEmpty\!}{\!\stEnvi}{\mpPi}$ and \inferrule{\iruleMPSub}, it follows that 
   $\stJudge{\mpEnvEmpty\!}{\!\setenum{\stEnvMap{\mpChanRole{\mpS}{\roleP}}{\gtProj{\gtGi}{\roleP}}}_{%
   \roleP \in \mpFmt{I}}}{\mpPi}$. 
   
   \item $\stJudge{\mpEnvEmpty\!}{\stEnvMap{\mpChanRole{\mpS}{\roleP}}{\gtProj{\gtGi}{\roleP}}}{\!\mpPi[\roleP]}$: Follows directly from $\stEnvi =  \bigcup_{\roleP \in I}\stEnvi[\roleP]\stSub \setenum{\stEnvMap{\mpChanRole{\mpS}{\roleP}}{\gtProj{\gtGi}{\roleP}}}_{%
   \roleP \in \mpFmt{I}}$, $\stJudge{\mpEnvEmpty\!}{\stEnvi[\roleP]}{\!\mpPi[\roleP]}$, and \inferrule{\iruleMPSub}. 
   
   \item Either $\mpPi[\roleP] \equiv \mpNil$, 
    or $\mpPi[\roleP]$ only plays $\roleP$ in $\mpS$, by $\stEnvMap{\mpChanRole{\mpS}{\roleP}}{\gtProj{\gtGi}{\roleP}}$: Follows directly from $\stJudge{\mpEnvEmpty\!}{\stEnvMap{\mpChanRole{\mpS}{\roleP}}{\gtProj{\gtGi}{\roleP}}}{\!\mpPi[\roleP]}$,  
    $\stEnvi[\roleP] \stSub \stEnvMap{\mpChanRole{\mpS}{\roleP}}{\gtProj{\gtGi}{\roleP}}$, and 
   the fact that either $\mpPi[\roleP] \equiv \mpNil$, 
    or $\mpPi[\roleP]$ only plays $\roleP$ in $\mpS$, by $\stEnvi[\roleP]$. 
   \qedhere
    \end{itemize}
\end{proof}

\subsection{Process Properties}
\label{app:sec:process_properties_proof}
\begin{restatable}[Process Deadlock-Freedom]{lemma}{lemProcessDF}%
  \label{lem:stenv-proc-df}
  \label{lem:deadlock-freedom}%
 Assume $\stJudge{\mpEnvEmpty\!}{\!\setenum{\stEnvMap{\mpChanRole{\mpS}{\roleP}}{\gtProj{\gtG}{\roleP}}}_{%
   \roleP \in \gtRoles{\gtG}}}{\!\mpP}$, %
  where 
  $\mpP \equiv \mpBigPar{\roleP \in  \gtRoles{\gtG}}{\mpP[\roleP]}$ %
  and for each $\mpP[\roleP]$, $\stJudge{\mpEnvEmpty\!}{\stEnvMap{\mpChanRole{\mpS}{\roleP}}{\gtProj{\gtG}{\roleP}}}{\!\mpP[\roleP]}$.
  Further, assume that each $\mpP[\roleP]$
  is either $\mpNil$ (up to $\equiv$), %
  or only plays $\roleP$ in $\mpS$, by $\stEnvMap{\mpChanRole{\mpS}{\roleP}}{\gtProj{\gtG}{\roleP}}$. %
  Then, $\mpP$ is deadlock-free. 
\end{restatable}
\begin{proof}
Let $\stEnv = \setenum{\stEnvMap{\mpChanRole{\mpS}{\roleP}}{\gtProj{\gtG}{\roleP}}}_{%
   \roleP \in \gtRoles{\gtG}}$. It follows directly that 
   $\stEnv = \bigcup_{\roleP \in \gtRoles{\gtG}}\stEnv[\roleP]$, where 
   $\stEnv[\roleP] = \stEnvMap{\mpChanRole{\mpS}{\roleP}}{\gtProj{\gtG}{\roleP}}$. 
 By~\Cref{def:assoc}, $\stEnvAssoc{\gtG}{\stEnv}{\mpS}$. Moreover, by~\Cref{cor:allproperties}, 
   $\stEnv$ is $\mpS$-deadlock-free.

  Consider any $\mpPi$ such that $\mpP \!\mpMoveStar\! \mpNotMoveP{\mpPi}$, with 
  $\mpP = \mpP[0] \!\mpMove\! \mpP[1] \!\mpMove\! \cdots \!\mpMove\! \mpP[n] = \mpNotMoveP{\mpPi}$ (for some $n$), where each reduction $\mpP[i] \!\mpMove\! \mpP[i+1]$ ($i \!\in\! 0...n\!-\!1$). 
  By~\Cref{lem:single-session-persistent}, we know that each $\mpP[i]$ is well-typed
  and its typing context $\stEnv[i]$ satisfies $\stEnv \stEnvMoveWithSessionStar[\mpS] \stEnv[i]$;
  moreover, $\mpP[i]$ adheres to the single-session requirements of \Cref{lem:session-fidelity}.
  Now,  observe that since the process $\mpP[n] = \mpNotMoveP{\mpPi}$ cannot reduce further, %
  by the contrapositive of \Cref{lem:session-fidelity}, 
  we obtain $\stEnvNotMoveWithSessionP[\mpS]{\stEnv[n]}$. 
  Furthermore, since $\stEnv$ is $\mpS$-deadlock-free,   
  by~\Cref{def:stenv-deadlock-free},   
  we have
  $\forall \mpChanRole{\mpS}{\roleP} \!\in\! \stEnv[n]$: $\stEnvApp{\stEnv[n]}{\mpChanRole{\mpS}{\roleP}} = \stEnd$. 
 Therefore, by \inferrule{\iruleMPNil},
  we have $\mpPi \equiv \mpNil$, 
  which (by~\Cref{def:proc-properties}\ref{item:proc-properties:df}) is the statement.
\end{proof}

\begin{restatable}[Process Liveness]{lemma}{lemProcessLive}%
  \label{lem:stenv-proc-live}
   Assume $\stJudge{\mpEnvEmpty\!}{\!\setenum{\stEnvMap{\mpChanRole{\mpS}{\roleP}}{\gtProj{\gtG}{\roleP}}}_{%
   \roleP \in \gtRoles{\gtG}}}{\!\mpP}$, %
  where 
  $\mpP \equiv \mpBigPar{\roleP \in  \gtRoles{\gtG}}{\mpP[\roleP]}$ %
  and for each $\mpP[\roleP]$, $\stJudge{\mpEnvEmpty\!}{\stEnvMap{\mpChanRole{\mpS}{\roleP}}{\gtProj{\gtG}{\roleP}}}{\!\mpP[\roleP]}$.
  Further, assume that each $\mpP[\roleP]$
  is either $\mpNil$ (up to $\equiv$), %
  or only plays $\roleP$ in $\mpS$, by $\stEnvMap{\mpChanRole{\mpS}{\roleP}}{\gtProj{\gtG}{\roleP}}$. %
  Then, $\mpP$ is live. 
\end{restatable}

\begin{proof}
   
   Let $\stEnv = \setenum{\stEnvMap{\mpChanRole{\mpS}{\roleP}}{\gtProj{\gtG}{\roleP}}}_{%
   \roleP \in \gtRoles{\gtG}}$. It follows directly that 
   $\stEnv = \bigcup_{\roleP \in \gtRoles{\gtG}}\stEnv[\roleP]$, where 
   $\stEnv[\roleP] = \stEnvMap{\mpChanRole{\mpS}{\roleP}}{\gtProj{\gtG}{\roleP}}$. 
 By~\Cref{def:assoc}, $\stEnvAssoc{\gtG}{\stEnv}{\mpS}$. Moreover, by~\Cref{cor:allproperties}, 
   $\stEnv$ is $\mpS$-live.

    The proof proceeds by contradiction: assume that $\mpP$ is \emph{not} live.
    Since (by hypothesis) each parallel component of $\mpP$ only plays one role $\roleP$ in session $\mpS$,
    there are $\mpPi, \mpCtx, \mpQ$ such that
    $\mpP = \mpP[0] \!\mpMove\! \mpP[1] \!\mpMove\! \cdots \!\mpMove\! \mpP[n] = \mpPi \!\equiv\! \mpCtxApp{\mpCtx}{\mpQ}$ where either:
    \begin{itemize}%
    \item%
      $\mpQ = \mpSel{\mpChanRole{\mpS}{\roleP}}{\roleQ}{\mpLab}{\mpW}{\mpQi}$ %
      (for some $\mpLab, \mpW, \mpQi$), and %
      $\not\exists \mpCtxi$: %
      $\mpCtx \mpCtxMoveStar \mpCtxi$ and $\mpPi \!\mpMoveStar\! \mpCtxApp{\mpCtxi}{\mpQi}$.\quad%
      By~\Cref{lem:single-session-persistent}, we know that each $\mpP[i]$ is well-typed
      and its typing context $\stEnv[i]$ is such that $\stEnv \stEnvMoveWithSessionStar[\mpS] \stEnv[i]$;
      moreover, each $\mpP[i]$ satisfies the single-session requirements of \Cref{lem:session-fidelity}.
      Therefore, $\mpPi$ satisfies the single-session requirements of \Cref{lem:session-fidelity},
      and is typed by some $\stEnvi$ such that $\stEnv \stEnvMoveWithSessionStar[\mpS] \stEnvi$. 
       Hence, by inversion of typing, $\mpQ$ is typed by some $\stEnvi[\roleP]$ (part of $\stEnvi$)
      where $\stEnvApp{\stEnvi[\roleP]}{\mpChanRole{\mpS}{\roleP}}$ is a (possibly recursive) internal choice
      towards $\roleQ$, including a choice $\stChoice{\stLab}{\stS}$ (where $\stS$ types the message payload $\mpW$). 
      Therefore, we have $\stEnvMoveAnnotP{\stEnvi}{\stEnvOutAnnot{\roleP}{\roleQ}{\stChoice{\stLab}{\stS}}}$. 
      
      Now, recall that (for the sake of the proof by contradiction) there is no $\mpCtxi$ with 
      $\mpCtx \mpCtxMoveStar \mpCtxi$ such that a reduction of  $\mpPi$ reaches $\mpCtxApp{\mpCtxi}{\mpQi}$; 
      that is,  
       the top-level selection of $\mpQ$ cannot be fired. 
    Hence,  there is at least
      one fair path beginning with $\stEnvi$
    that never fires a transmission label 
      $\ltsSendRecv{\mpS}{\roleP}{\roleQ}{\stLabi}$ (for any $\stLabi$).
      But then, such a fair path starting from $\stEnvi$ is not live, and furthermore, 
      (by \Cref{def:typing-ctx-live}) we obtain that $\stEnv$ is \emph{not} live, a desired contradiction; 

    \item%
      $\mpQ = \mpBranch{\mpChanRole{\mpS}{\roleP}}{\roleQ}{i \in I}{\mpLab[i]}{x_i}{\mpQi[i]}{}$ %
      (for some $I$, $\mpLab[i], \mpFmt{x_i}, \mpQi[i]$), and 
      $\not\exists \mpCtxi, k \!\in\! I, \mpU$:\, %
      $\mpCtx \mpCtxMoveStar \mpCtxi$ and $\mpPi \mpMoveStar \mpCtxApp{\mpCtxi}{\mpQi[k]\subst{x_k}{\mpU}}$. %
      The proof is similar to the previous case, and reaches a similar contradiction. 
       \qedhere 
    \end{itemize} 
   \end{proof}

\lemProcessPropertiesVerif*
\begin{proof}
  Directly from~\Cref{lem:stenv-proc-df,lem:stenv-proc-live}.  
  \qedhere 
  \end{proof}

\end{document}